\documentclass[a4paper,10pt,accepted=2026-04-11,aps,onecolumn,prx]{quantumarticle}
\pdfoutput=1
\usepackage{silence}
\usepackage{amsmath,mathtools,amsthm,amssymb}
\usepackage{tabularx}
\usepackage{tabularray}
\usepackage{monomialvis}
\usepackage{graphicx}
\usepackage{tikz}
\usepackage{color}
\usepackage[sort&compress,numbers]{natbib}
\usepackage{bbold}
\usepackage{enumitem}
\usepackage{physics}
\usepackage{comment}
\usepackage{array}
\usepackage{graphics}
\usepackage{wrapfig}
\graphicspath{{figures/}}
\usepackage{biolinum}
\usepackage{bbm}
\usepackage{dsfont}
\usepackage{mathrsfs}
\usepackage{mathdots}
\usepackage{enumitem}
\usepackage{pifont}

\pgfplotsset{compat=newest,compat/show suggested version=false}
\newcommand\sbullet[1][.5]{\mathbin{\vcenter{\hbox{\scalebox{#1}{$\bullet$}}}}}

\definecolor{myrefcolor}{rgb}{0.067,0.5,0.5}
\definecolor{myurlcolor}{rgb}{0.1,0,0.9}

\usepackage{xcolor}
\definecolor{deepblue}{rgb}{0.1, 0.1, 0.8} 
\definecolor{darkgreen}{rgb}{0.0, 0.4, 0.0} 
\definecolor{richmaroon}{rgb}{0.6, 0.0, 0.0} 

\usepackage{setspace}

\definecolor{myrefcolor}{rgb}{0.067,0.5,0.5}

\usepackage[
    breaklinks,
    pdftex,
    colorlinks=true,
    linkcolor=myrefcolor,
    citecolor=myrefcolor,
    urlcolor=myrefcolor
]{hyperref}

\usepackage[capitalize]{cleveref}

\newcommand{\pur}[0]{\mathrm{Pur}}

\newtheorem*{theorem*}{Theorem}
\newtheorem*{corollary*}{Corollary}
\newtheorem*{definition*}{Definition}
\newcounter{thm}
\newtheorem{theorem}[thm]{Theorem}
\newtheorem{lemma}[thm]{Lemma}

\newtheorem{proposition}[thm]{Proposition}
\newtheorem{definition}[thm]{Definition}
\newtheorem{problem}[thm]{Problem}

\newtheorem{open}{Open Question}
\newtheorem{example}{Example}
\newtheorem{corollary}[thm]{Corollary}

\theoremstyle{remark}
\newtheorem{remark}[thm]{Remark}

\DeclareMathOperator{\stab}{STAB}
\usepackage{algorithm}
\usepackage{algpseudocode}

\newcommand{\symftwo}{\mathrm{Sym}_0(\mathbb{F}_2^{m\times m})}

\DeclareMathOperator{\de}{d\!}
\DeclareMathOperator{\sym}{sym}
\DeclareMathOperator{\cl}{Cl}
\DeclareMathOperator{\com}{Com}

\DeclareMathOperator{\diag}{diag}
\DeclareMathOperator{\otoc}{OTOC}

\newcommand{\haar}[0]{\operatorname{Haar}}

\definecolor{airforceblue}{rgb}{0.36, 0.54, 0.66}
\newcommand{\be}{\begin{equation}\begin{aligned}\hspace{0pt}}
\newcommand{\bbb}{\begin{equation*}\begin{aligned}}
\newcommand{\ee}{\end{aligned}\end{equation}}
\newcommand{\eee}{\end{aligned}\end{equation*}}
\definecolor{burgundy}{rgb}{0.5, 0.0, 0.13}

\newcommand{\eqt}[1]{\stackrel{\mathclap{ \mbox{\scriptsize #1}}}{=}}
\newcommand{\leqt}[1]{\stackrel{\mathclap{ \mbox{\scriptsize #1}}}{\leq}}

\newcommand{\geqt}[1]{\stackrel{\mathclap{ \mbox{\scriptsize #1}}}{\geq}}
\newcommand{\even}{\mathrm{Even}(\mathbb{F}_{2}^{k\times m})}
\newcommand{\symf}{\symftwo}

\allowdisplaybreaks
\makeatletter
\renewcommand{\subsubsection}{\@startsection{subsubsection}{3}{0pt}%
  {1.5ex plus 1ex minus .2ex}%
  {1ex plus .2ex}%
  {\centering\bfseries}}
\makeatother

\begin{document}

\title{A complete theory of the Clifford commutant}

 \author{Lennart Bittel}
 \affiliation{Dahlem Center for Complex Quantum Systems, Freie Universit\"at Berlin, 14195 Berlin, Germany}
 \author{Jens Eisert}
  \affiliation{Dahlem Center for Complex Quantum Systems, Freie Universit\"at Berlin, 14195 Berlin, Germany}
  \affiliation{Helmholtz-Zentrum Berlin für Materialien und Energie, Berlin, Germany}
 
 \author{Lorenzo Leone}
  \affiliation{Dahlem Center for Complex Quantum Systems, Freie Universit\"at Berlin, 14195 Berlin, Germany}

\author{Antonio A.~Mele}
  \affiliation{Dahlem Center for Complex Quantum Systems, Freie Universit\"at Berlin, 14195 Berlin, Germany}
  
\author{Salvatore F.~E.~Oliviero}
 \affiliation{Scuola Normale Superiore di Pisa, Piazza dei Cavalieri 7, IT-56126 Pisa, Italy}
   \affiliation{Dahlem Center for Complex Quantum Systems, Freie Universit\"at Berlin, 14195 Berlin, Germany}

\begin{abstract}
The Clifford group plays a central role in quantum information science. It is the building block for many error-correcting schemes and matches the first three moments of the Haar measure over the unitary group—a property that is essential for a broad range of quantum algorithms, with applications in pseudorandomness, learning theory, benchmarking, and entanglement distillation. At the heart of understanding many properties of the Clifford group lies the Clifford commutant: the set of operators that commute with $k$-fold tensor powers of Clifford unitaries.
Previous understanding of this commutant has been limited to relatively small values of $k$, constrained by the number of qubits $n$. In this work, we develop a complete theory of the Clifford commutant. Our first result provides an explicit orthogonal basis for the commutant and computes its dimension for arbitrary $n$ and $k$. We also introduce an alternative and easy-to-manipulate basis formed by isotropic sums of Pauli operators. We show that this basis is generated by products of permutations— which generate the unitary group commutant— and at most three other operators. Additionally, we develop a \emph{graphical calculus} allowing a diagrammatic manipulation of elements of this basis. These results enable a wealth of applications: among others, we characterize all \emph{measurable} magic measures and identify optimal strategies for stabilizer property testing, whose success probability also offers an operational interpretation to stabilizer entropies. Finally, we show that these results also generalize to multi-qudit systems with prime local dimension.
\end{abstract}

\maketitle

\section{Introduction}

The Clifford group is fundamental to quantum information science, with
wide-ranging applications in a wealth of both theoretical and practical domains. These include notions of quantum error
correction and fault tolerance~\cite{gottesman_stabilizer_1997}, quantum simulation~\cite{aaronson_improved_2004,nest2009classicalsimulationquantumcomputation,smith2023cliffordbasedcircuitcuttingquantum,begušić2023simulatingquantumcircuitexpectation},
randomized benchmarking~\cite{Magesan_2011,Eisert_2020}, as well as in the context of resource theories~\cite{veitch_resource_2014, leone_stabilizer_2022,leone_stabilizer_2022}. It consists of unitary transformations that map Pauli operators to Pauli operators (up to a phase) under conjugation. This structure enables efficient classical simulation of Clifford circuits~\cite{aaronson_improved_2004,gottesman_heisenberg_1998}, as one can efficiently track the Heisenberg evolution of Pauli operators through the circuit gate-by-gate, making them a standard benchmark for quantum devices.  

Beyond their efficient simulability, Clifford circuits exhibit a wealth of interesting mathematical properties. In particular, the uniform distribution over the multi-qubit Clifford group constitutes an exact unitary $ 3 $-design \cite{Gross_2007}, meaning it replicates the first three moments of the Haar measure over the unitary group. This property is central to quantum learning, such as classical shadows~\cite{Huang_2020,Elben_2022,PhysRevLett.133.020602} and randomized benchmarking~\cite{Magesan_2011,Eisert_2020}—as well as in readings of quantum error correction~\cite{gottesman_stabilizer_1997}, the resource theory of magic~\cite{veitch_resource_2014, leone_stabilizer_2022}, and entanglement distillation~\cite{Dehaene_2003,gu_2025}. Consequently, random Clifford operators are fundamental in the field of quantum information science.

A useful way to characterize the effects of applying random Clifford unitaries is through the Clifford average. From a technical point of view, the Clifford average of any polynomial function $ f(C) $ (or a function that can be approximated by a polynomial) reduces to computing the so-called \textit{twirling operator}. This operator corresponds to the uniform average of an operator $ O_f $, which encodes the function $ f $, over $ k $ tensor powers of the Clifford group $ \mathcal{C}_n $
\begin{align}
\Phi_{\cl}^{(k)}(O_f) \coloneqq \frac{1}{|\mathcal{C}_n|} \sum_{C\in\mathcal{C}_n} C^{\otimes k} O_f C^{\dag\otimes k} \,.\label{eq:twirlingintro}
\end{align}
While this definition provides a formal procedure to compute the Clifford average, direct evaluation becomes infeasible as both the tensor powers $ k $ and the number of qubits $ n $ grow. A way to circumvent this obstacle is through a key observation: the twirling operator $ \Phi_{\cl}^{(k)}(\cdot) $ commutes with the $ k $-fold tensor power of any Clifford operators $ C^{\otimes k} $. Consequently, its image belongs to the $ k $-th order commutant: the set of linear operators that commute with $ k $ tensor powers of any Clifford unitary.

The commutant is a vector space that is also closed under multiplication (i.e., it is an algebra), a fact that follows directly from its definition. Therefore, a basis of the commutant is sufficient to fully characterize any of its elements, including the image of the twirling operator. Consequently, given the widespread use of random Clifford circuits, understanding the $ k $-th order commutant of the Clifford group is of fundamental importance.

For the unitary group, the commutant is well understood and is spanned by permutation operators~\cite{collins_integration_2006,collins_moments_2003}. However, for the Clifford group, the analysis is significantly more involved, as its commutant contains substantially more elements. Major progress has been made in the seminal work of Gross, Nezami, and Walter~\cite{gross_schurweyl_2019}, which provided a characterization of the Clifford group commutant for $k \leq n+1$. In this work, we extend their results to arbitrary values of $n$ and $k$, a generalization that is particularly desirable given that higher-order commutants play a central role in applications such as deriving concentration bounds of the ``L\'evy's lemma'' type—where one typically requires $k \sim 2^n$—as well as in the construction of high-order unitary designs, which often demand $k \sim \operatorname{poly}(n)$. Notably, even for small values of $k$, practical computational techniques remain limited.

In particular, in this work, we develop a complete and practical theory of the Clifford group commutant, presenting three main results:  
\begin{enumerate}[label=(\roman*)]
    \item For any $ k  \in \mathbb{N}$, we identify a minimal generating set of the algebra, consisting of at most three additional elements beyond those required for the unitary group. Specifically, we show that the $k$-th commutant of the Clifford group is generated by products of permutations (which also span the $ k $-th commutant of 
    the unitary group), together with the elements
\begin{equation}
\sum_{P \in \mathbb{P}_n} P^{\otimes 4}, \quad 
\sum_{P \in \mathbb{P}_n} P^{\otimes 6}, \quad \text{and, if $ k $ is divisible by 4,} \quad 
\sum_{P \in \mathbb{P}_n} P^{\otimes k},
\end{equation}
where $ \mathbb{P}_n $ denotes the Pauli basis on $n$ qubits. This also demonstrates that uniform sums of Pauli operators raised to tensor powers $\sum_{P \in \mathbb{P}_n} P^{\otimes 2m}$ with $m\in \mathbb{N}$ suffice to generate the Clifford commutant.
As a particular case of this finding, one can recover the main result of Ref.~\cite{zhu_clifford_2016}, where it has been shown that $ \sum_{P \in \mathbb{P}_n} P^{\otimes 4} $ has been sufficient, together with permutations, to generate the $ (k = 4) $-th Clifford group commutant. Surprisingly, we show that at most two more operators are sufficient to generate the Clifford commutant for any $ k $.
    \item We explicitly construct an orthogonal basis from first principles, and explicitly compute the dimension of the commutant for any $n,k\in\mathbb{N}$.
    \item We introduce a \textit{graphical calculus}, a useful tool for {\em practically} analyzing and manipulating the operators in the commutant.  
    In particular, this method allows us to derive explicit and compact formulas for Clifford averages for $ k=8 $-th moments averages with direct applications, despite the basis containing nearly 10 million elements.
\end{enumerate}
We explore applications of our theoretical results. Specifically, in the context of magic-state resource theory, we show that any `measurable' magic monotone~\cite{veitch_resource_2014,Leone_2024} lies within the commutant of the Clifford group. Furthermore, we demonstrate that the characterization of the Clifford group commutant is essential for analyzing the property testing of stabilizer states—a well-studied problem in the quantum learning literature~\cite{gross_schurweyl_2019,arunachalam2024polynomialtimetoleranttestingstabilizer,grewalImprovedStabilizerEstimation2023}.

\section{Overview of the results}
In this section, we provide a high-level overview of the main contributions of this work, guiding the reader through the structure of this work.
We begin in~\cref{sub:theoryINTRO} by discussing the core contribution: a complete characterization of the Clifford group commutant, both from algebraic and vector space perspectives. Next, in~\cref{sub:graphINTRO}, we introduce a graphical calculus that plays a central role in our analysis, which is developed in detail later in the manuscript.
In~\cref{sub:appINTRO}, we present an overview of the applications of our framework, specifically in magic-state resource theory and the property testing of stabilizer states.
While our primary focus is the Clifford commutant for $n$ qubits, our results naturally extend to the case of $n$ qudits with local prime dimension, as outlined in~\cref{sec:quditcase}.
Finally, in~\cref{sec:priorwork}, we discuss related prior works.
This concludes the high-level summary of our results, which should help orient the reader before diving into the more detailed analysis in the following sections.
Before proceeding, we introduce some notation that will be useful throughout the work.

\vspace{1em}
\textbf{Notation.} We denote by $ \mathbb{F}_2^k $ the $ k $-dimensional vector space over the finite field $ \mathbb{F}_2 $. For a vector $ x \in \mathbb{F}_2^k $, $ |x| $ represents the Hamming weight of $ x $, i.e., the number of non-zero components. Let $ \mathcal{H} \simeq \left( \mathbb{C}^2 \right)^{\otimes n} $ be the Hilbert space of $ n $ qubits, and $ \mathcal{H}^{\otimes k} $ its $ k $-fold tensor power. $ \mathcal{B}(\mathcal{H}^{\otimes k}) $ denotes the set of operators on $ \mathcal{H}^{\otimes k} $, and $ \mathcal{U}(2^n) $ the unitary group acting on $ \mathcal{H} $. The Pauli basis is $ \mathbb{P}_n \coloneqq \{ I, X, Y, Z \}^{\otimes n} $. The Clifford group $ \mathcal{C}_n $ is the group of unitaries that map the Pauli basis to itself under conjugation up to phase, i.e., 
$\mathcal{C}_n \coloneqq \{ C \in \mathcal{U}(2^n) : C P C^\dagger \in \pm \mathbb{P}_n \text{ for all } P \in \mathbb{P}_n \}$.
The $ k $-th commutant of the Clifford group is defined as
\begin{align}
    \com(\mathcal{C}_n^{\otimes k}) \coloneqq \{ O \in \mathcal{B}(\mathcal{H}^{\otimes k}) : C^{\dagger \otimes k} O C^{\otimes k} = O \, \text{ for all } C \in \mathcal{C}_n \}.
    \label{eq:introcomm}
\end{align}
Throughout the manuscript, we often use $ d \coloneqq 2^n $ to denote the Hilbert space dimension of an $ n $-qubit system.

\subsection{Full theory of the Clifford group commutant: generators and orthogonal basis} 
\label{sub:theoryINTRO}
We now provide a comprehensive characterization of the Clifford commutant. Here, we first present a generating set for the commutant, demonstrating that every operator in the commutant can be written as a linear combination of products of elements from this set. Then, we construct an orthogonal basis for the $k$-th order commutant, valid for any $k \in \mathbb{N}$. 
It is important to note that this subsection serves as an overview of the main results. Later in the manuscript, specifically in \cref{sec:comclif,sec:Paulimon,sec:3gen}, we will present additional properties of the commutant and its basis elements, which are not summarized here for the sake of brevity.

\vspace{1em}
\textbf{Generating set for the Clifford group commutant.} Two key ideas underpin our results. First, the commutant of the Clifford group should naturally be describable in terms of Pauli operators, as they inherently define its symmetry—specifically, for any two Pauli operators $ P, Q \in \mathbb{P}_n $, we have $ C^{\dag} P C \propto Q $ for some Clifford unitary $ C $. Second, it is well known that applying a random Clifford unitary to a Pauli operator $P$ maps it uniformly at random to another non-identity Pauli operator. Consequently, to account for this behavior, the operators spanning the commutant of the Clifford group in \cref{eq:introcomm} must include isotropic sums over Pauli operators. To this end, we introduce various types of \textit{Pauli monomials}—uniform sums of Pauli operators that remain invariant under the Clifford group. 
\begin{definition*}[Primitive Pauli monomials. Informal version of \cref{def:primitivepaulimonomials}] Let $n,k\in\mathbb{N}$. Let $v\in\mathbb{F}_{2}^{k}$ be a vector satisfying $|v|=0\pmod2$. We define a primitive Pauli monomial as an operator on $\mathcal{B}(\mathcal{H}^{\otimes k})$ defined as
\begin{align}\label{def:primitivepaulimonomialsintro}
\Omega(v)
\coloneqq
\frac{1}{d}\sum_{P\in\mathbb{P}_n}P^{\otimes v},
\end{align}
where $P^{\otimes v}\coloneqq\bigotimes_{i=1}^{k}P^{ v_i}$.
\end{definition*}

For example, when the vector $ v $ satisfies $ |v| = 2 $, the corresponding primitive Pauli monomial reduces to a swap operator. Given a binary vector, denoted as $ v_{i,j} $, with non-zero $ i $-th and $ j $-th components, we define $T_{(i\,j)} \coloneqq \Omega(v_{i,j})$,  
where $ T_{(i\,j)} $ is the swap operator between the $ i $-th and $ j $-th tensor copy. Thus, primitive Pauli monomials generalize swap operators to vectors with higher Hamming weight.  

For the vector $ (1,1,1,1,0,\dots,0) $, the corresponding Pauli monomial coincides with the projector  
\begin{align}  
    \Omega(1,1,1,1,0,\dots,0) = \frac{1}{d} \sum_{P \in \mathbb{P}_n} P^{\otimes 4}\otimes \mathbb{1}^{\otimes k-4},  
\end{align}  
originally introduced in Ref.~\cite{zhu_clifford_2016}, during the early stages of understanding the Clifford group commutant. Specifically, such an operator, together with all permutations (which generate the commutant of the unitary group~\cite{collins_moments_2003,collins_integration_2006}), has been shown to be the sole element required to \emph{gracefully} generate the $ k=4 $-th commutant of the Clifford group.

Our first main result is that only two additional primitive Pauli monomials, besides $ \Omega(v_4) $ and permutation operators, suffice to \emph{gracefully} generate the full commutant for any $ k \in \mathbb{N}$. This will be the focus of the following theorem.
\begin{theorem*}[Algebraic structure of the commutant. Informal version of \cref{th:algebraicstructurecommutantofclifford}]\label{th:1} Let $n,k\in\mathbb{N}$. Let $\Omega_l$ be the primitive Pauli monomial associated to $(\underbrace{1,\dots,1}_{\text{$l$ times}},0, \cdots,0)$. The following facts hold:
\begin{enumerate}[label=(\Alph*)]
\item The $k$-th order commutant of the Clifford group is generated, up to permutations, by $\Omega_4$, $\Omega_6$, and, if $k$ is divisible by $4$, also $\Omega_k$, referred to as \emph{fundamental primitive} Pauli monomials.

\item The above constitutes a minimal set of generators for every $k$ and $n$.

\item Any product of primitive Pauli monomials is equal to $d^{\alpha} \Omega$, where $\alpha \in \mathbb{N}$ and $\Omega$ is the product of at most $k + 1$ primitive Pauli monomials as defined in \cref{def:primitivepaulimonomialsintro}.
\end{enumerate}
\end{theorem*}
Thus, while the $k$-th order commutant of the unitary group—defined analogously to \cref{eq:introcomm}—can be generated by permutation operators (which themselves are generated by at most $k - 1$ swap operators), the mere addition of the operators corresponding to the primitive Pauli monomials $\Omega_4$, $\Omega_6$, and $\Omega_k$ suffices to generate the $k$-th order commutant of the Clifford group.

\vspace{1em}

\textbf{An orthogonal basis for the $k$-th .} 
To take a step further and construct an orthogonal basis, we observe that any Pauli operator $ P \in \mathbb{P}_n^{\otimes k} $ can be decomposed as $ P = \phi \prod_{i=1}^{m} P_{i}^{\otimes v_i} $, where $ \{P_1, \ldots, P_m\} $ is a set of algebraically independent Pauli operators, $ m \leq k $, $ \phi \in \{\pm1, \pm i\} $, and the vectors $V \coloneqq \{v_1, \ldots, v_m\} \in \mathbb{F}_2^{k \times m}$ are linearly independent. Given such a decomposition, we associate to $\{P_1,\ldots,P_m\}$ the (binary) adjacency matrix $G$ of their anticommutation graph, where vertices label the Pauli operators and an edge connects two vertices if the corresponding Pauli operators anticommute. The concept of anticommutation graph has also recently emerged as a crucial tool in other contexts of quantum information~\cite{anschuetz2024boundsgroundstateenergy, aguilar2024classificationpauliliealgebras, king2024triplyefficientshadowtomography, chen2024optimaltradeoffsestimatingpauli, bao2024toleranttestingstabilizerstates}.

Hence, each decomposition of $P$ determines a pair $(V,G)$. Crucially, different choices of algebraically independent generating sets for the same Pauli operator are related by the action of a matrix (gauge) $A \in \operatorname{GL}(\mathbb{F}_2^{m \times m})$, acting as $(V, G) \mapsto (V A, A^{-1} G A^{-T})$, without affecting the reconstructed operator $P \in \mathbb{P}_n^{\otimes k}$. Here $A^{-T}$ denotes the inverse transpose. We therefore define the equivalence class
\[
[V, G] \coloneqq \left\{ (V', G') \,\middle|\, \exists A \in \operatorname{GL}(\mathbb{F}_2^{m \times m}) \text{ such that } V' = V A,\ G' = A^{-1} G A^{-T} \right\},
\]
which is uniquely associated with $P$. In particular, we define the map $\mathcal{M}: P \mapsto \mathcal{M}(P) = [V,G]$. Together with a choice of algebraically independent generators $\{P_1,\ldots,P_m\}$, the equivalence class $[V,G]$ uniquely specifies $P$. We are now ready to introduce a second type of \textit{Pauli monomials}, which will provide a complete basis for the $k$-th order commutant of the Clifford group for any $n$.

\begin{definition*}[Independent graph-based Pauli monomials] 
The independent graph-based Pauli monomials are defined as the isotropic sum
\begin{align}\label{def:indeppaulimonomialsINTRO}
\mho_{I}([V, G]) = \frac{1}{|S_{[V, G]}|} \sum_{\substack{P \in \mathbb{P}_n^{\otimes k} \\ \mathcal{M}(P) = [V, G]}} \varphi(P)\, P \,,
\end{align}
where $S_{[V, G]} \coloneqq \{ P \in \mathbb{P}_n^{\otimes k} \,:\, \mathcal{M}(P) = [V, G] \}$, and $\varphi(P) \coloneqq \tr\big(P\, T_{(k\,k{-}1)} T_{(k{-}1\,k{-}2)} \cdots T_{(2\,1)} T_{(1\,k)} \big)/d$, with $T_{(i\,j)}$ denoting the swap operator between the $i$-th and $j$-th tensor copies.
\end{definition*}
It is not difficult to verify that $ \mho_{I}([V,G]) \in \com(\mathcal{C}_n^{\otimes k}) $ for any choice of $ V $ and $ G $.
We also show that $ \mho_{I}([V,G]) $ is non-zero if and only if $ \operatorname{rank}_2(G) \geq 2(m-n) $. Here, $\rank_2$ denotes the rank on $\mathbb F_2$.
Proving that these non-zero elements form an orthogonal basis of the commutant for all $ n $ and $ k $, as shown in \cref{sec:comclif}, leads to the following main result.

\begin{theorem*}[Basis of the commutant. Informal version of \cref{th:fullcommutantnk,cor:dimcom}]\label{th:B} 
For any $ n,k \in \mathbb{N}$, the following holds:
\begin{itemize}
    \item The set of non-zero operators $ \mho_{I}([V,G]) $ in \cref{def:indeppaulimonomialsINTRO} forms an orthogonal basis for $ \com(\mathcal{C}_n^{\otimes k}) $.
    \item The dimension of the commutant is given by 
    \begin{align}
\dim(\com(\mathcal{C}_n^{\otimes k}))&\simeq\begin{cases}
        2^{\frac{k^2-3k}{2}}& 2n\geq k-1\\
        2^{2kn-2n^2-3n}& 2n< k-1
        \end{cases},
\end{align} 
where $\simeq$ means up to constant factors. We refer to \cref{eq:dimcommutant} in \cref{cor:dimcom} for the complete expression.
\end{itemize}

\end{theorem*}
This theorem provides a complete characterization of the commutant of the Clifford group, as well as explicitly computing its dimension, which has previously been known only for $ k \leq n+1 $ \cite{gross_schurweyl_2019}. 

We also provide new insights into the commutant basis introduced in Ref.~\cite{gross_schurweyl_2019}. In \cref{ssec:equi_basis}, we demonstrate that this basis exactly coincides with the one formed by products of primitive Pauli monomials, defined in \cref{def:primitivepaulimonomialsintro} and referred to as \textit{reduced Pauli monomials} (see \cref{def:paulimonomials}), for which we establish a range of important properties.

Furthermore, we establish the existence of an invertible Fourier-like transformation between independent graph-based Pauli monomials and reduced Pauli monomials (see \cref{Sec:invertiblemapspaulimonomials}). Although the set of reduced Pauli monomials remains linearly independent for $k \leq n+1$, linear dependencies emerge when $k > n+1$. Our work characterizes these dependencies, which are captured by the condition $\mho_{I}([V,G]) = 0$, holding if and only if $\operatorname{rank}_2(G) < 2(m - n)$. In particular, for $k \leq n+1$, this condition $\operatorname{rank}_2(G) \geq 2(m - n)$ is automatically satisfied, which naturally explains why the set of reduced Pauli monomials forms a linearly independent set in this regime. Additional properties of these operators are discussed in \cref{sec:comclif,sec:Paulimon,sec:commutantCliffordgen}, and we encourage the reader to consult these sections for further details.

\subsection{Graphical calculus}
\label{sub:graphINTRO}
The graphical calculus provides a powerful yet reader-friendly method for manipulating Pauli monomials, central objects in the commutant of the Clifford group. In general, a product of \textit{primitive} Pauli monomials takes the form (see \cref{def:paulimonomials})
\begin{align}\label{eq:paulimonomialsintro}
\Omega(V, M) \coloneqq \frac{1}{d^m} \sum_{P_1, \ldots, P_m \in \mathbb{P}_n}  
P_1^{\otimes v_1} P_2^{\otimes v_2} \cdots P_m^{\otimes v_m} \times \left( \prod_{\substack{i, j \in [m] \\ i < j}} \chi(P_i, P_j)^{M_{i,j}} \right),
\end{align}
where $ \Omega(V,M) $ is identified by the columns $ \{v_j\}_{j=1}^m $ of the binary matrix $ V \in \mathbb{F}_{2}^{k \times m} $ along with the matrix $ M \in \mathbb{F}_{2}^{m \times m} $, which encodes the phase information and is assumed to be symmetric with a zero diagonal. Here, $ \chi(P_i,P_j) = 1 $ if $ [P_i,P_j] = 0 $, and $ \chi(P_i,P_j) = -1 $ otherwise. We call the operators in \cref{eq:paulimonomialsintro} \textit{Pauli monomials} because they are isotropic sums over Pauli operators in $ \mathbb{P}_n $ and arise from products of \textit{primitive} Pauli monomials, as defined in \cref{def:primitivepaulimonomialsintro}. The set of Pauli monomials with $V$ being full rank form the set of reduced Pauli monomials and provide a basis for the commutant, provided that $ k \leq n+1 $ (see \cref{lem:linearindependency}). As discussed above, if instead $ k > n+1 $, linear dependencies arise. Given that Pauli monomials form a basis for the $k$-th commutant, we refer to general elements of the commutant as \textit{Pauli polynomials}.

We can represent the operator in \cref{eq:paulimonomialsintro} diagrammatically using the following rules:  
\begin{itemize}
    \item The columns $ \{v_j\}_{j=1}^m $ are represented sequentially, with each column represented as a vector of black and white dots, corresponding to entries 1 and 0, respectively (see the following example for clarity).  
    \item Each off-diagonal element $ M_{i,j} $ of the matrix $ M $ is represented by a line (denoted as a \emph{phase}) connecting the columns $ i $-th and $ j $-th if $ M_{i,j} = 1 $, with no line drawn if $ M_{i,j} = 0 $ (again, see the next example).
\end{itemize}
The power of this graphical representation lies in its ease of manipulation, which, at a mathematical level, corresponds to applying elementary transformations $ A \in \operatorname{GL}(\mathbb{F}_{2}^{m \times m}) $ that leave the overall operator $ \Omega(V,M) $ (see \cref{th:gaussOP}) unchanged. These transformations are straightforward in the graphical representation but become significantly more cumbersome when working directly with products of primitive Pauli monomials.

We now illustrate an example that demonstrates both a non-trivial manipulation of Pauli monomials and its graphical representation, within the $ k = 6 $-th commutant $ \com(\mathcal{C}_n^{\otimes 6}) $. The operator 
\begin{align}
\Omega\left(\begin{pmatrix}
        1 & 0  \\
        1 & 0 \\
        1 & 1 \\
        1 & 1 \\
        0 & 1 \\
        0 & 1 
    \end{pmatrix},
    \begin{pmatrix}
        0 & 1  \\
        1 & 0 
    \end{pmatrix}\right)=\frac{1}{d^2}\sum_{P_1,P_2}P^{\otimes 2}\otimes (PQ)^{\otimes 2}\otimes Q^{\otimes 2}\chi(P_1,P_2) \label{eq:intropaulimonom}
\end{align}
belongs to $\com(\mathcal{C}_n^{\otimes 6})$.
Using the simple rules outlined above, we can represent the operator in terms of the diagram
\setmonomialscale{3.5mm}  
\begin{equation}
    \Omega\left(\begin{pmatrix}
        1 & 0  \\
        1 & 0 \\
        1 & 1 \\
        1 & 1 \\
        0 & 1 \\
        0 & 1 
    \end{pmatrix},
    \begin{pmatrix}
        0 & 1  \\
        1 & 0 
    \end{pmatrix}\right) = \monomialdiagram{6}{{1,2,3,4},{3,4,5,6}}{0:1} \label{eq:diagramintro}
    \, .
\end{equation}
The advantage of this approach lies in a set of rules, proven analytically and developed in this manuscript, particularly in \cref{th:graphrules}, which facilitate the manipulation of Pauli monomials and demonstrate their equivalences. While hiding the non-trivial manipulation of the diagram in \cref{eq:diagramintro}, we present one such equivalence
below,

\setmonomialscale{3.5mm}  
\begin{equation}
     \monomialdiagram{6}{{1,2,3,4},{3,4,5,6}}{0:1} 
    \quad\eqt{\text{(i)}} \quad \monomialdiagram{6}{{1,2,3,4,5,6},{1,2},{3,4},{5,6}}{} \quad=\quad \Omega_{6}T_{(12)}T_{(34)}T_{(56)}\, ,\label{eq:equivalenceintro}
\end{equation}
where (i) encapsulates the non-trivial manipulation of Pauli monomials using graphical calculus. \cref{eq:equivalenceintro} shows that the Pauli monomial in \cref{eq:diagramintro} is, up to permutations, exactly equal to $ \Omega_6 $ defined above. Although finding such relations is non-trivial, the graphical calculus allowed us, among many other technical proofs, to characterize the commutant up to $ k = 8 $, up to left/right multiplication by permutations, as detailed in \cref{sec:cliffordcommutantexample}. Pauli monomials, up to left/right multiplication by permutations, are particularly relevant when dealing with tensor powers of pure states.

\subsection{Applications: Clifford-Weingarten Calculus, Magic-state resource theory and property testing}\label{sub:appINTRO}

\textbf{Magic-state resource theory.} As an application of this framework, we analyze the deep connection between magic-state resource theory and the commutant of the Clifford group.
Magic-state resource theory is based on a key dichotomy: while stabilizer operations can be implemented fault-tolerantly in a direct manner—through transversal gates in many settings~\cite{campbell_bound_2010}—the fault-tolerant implementation of non-stabilizer operations is considered significantly more challenging~\cite{PhysRevLett.102.110502}. Consequently, stabilizer operations are designated as the \textit{free operations} within this resource-theoretic framework, with stabilizer states identified as the corresponding \textit{free states}. In contrast, non-Clifford operations and non-stabilizer states are considered \textit{resources}. Additionally, the Gottesman--Knill theorem establishes that quantum circuits composed exclusively of Clifford gates can be efficiently simulated on a classical computer~\cite{aaronson_improved_2004}. However, the inclusion of non-Clifford components markedly increases the difficulty of classical simulation, further highlighting the resourceful nature of such elements in the context of quantum advantage.

A central challenge in any resource theory lies in identifying an appropriate resource monotone—a function that faithfully quantifies the resource content of a given state. Typically, such monotones are positive scalar functions. Although general-purpose resource monotones exist and can be applied across different resource theories~\cite{chitambar_quantum_2019}, they often suffer from computational or experimental intractability. For a resource theory to be useful in practice—especially in the realms of quantum many-body physics and quantum computation—it is essential to identify resource measures that are both experimentally measurable and computationally tractable.

In this context, \emph{stabilizer entropies}~\cite{leone_stabilizer_2022} have recently emerged as the only known magic monotones~\cite{Leone_2024} that are experimentally accessible in an efficient manner~\cite{haugEfficientQuantumAlgorithms2024,olivieroMeasuringMagicQuantum2022}. 
The set of primitive Pauli monomials is deeply connected to magic-state resource theory. To see this, it suffices to note that the stabilizer $ k $-R\'enyi entropies $ M_{k}(\psi) $ can be expressed in terms of $ 2k $-order primitive Pauli monomials (\cref{def:primitivepaulimonomialsintro}) as
\be
M_{k}(\psi) \coloneqq \frac{1}{1-k} \log \tr(\Omega_{2k} \psi^{\otimes 2k})\,.\label{eq:stabilizerentropiesintro}
\ee
In the manuscript, we further explore the connection between magic-state resource theory and the structure of the Clifford commutant.

First, it is worth noting that any (polynomial) function that quantifies magic should be invariant under Clifford operations. 
This invariance property implies that the operators characterizing such functions must lie in the commutant of the Clifford group for some order $ k $, as shown in \cref{lem:measurablefunctionbelongtothecommutant}.
Given this, it is natural to extend the definition of stabilizer purities - the argument of the logarithm in the definition of stabilizer entropy in \cref{eq:stabilizerentropiesintro}—to any Pauli monomial $ \Omega(V,M) $ arising from products of primitive Pauli monomials. 
We therefore introduce the \textit{generalized stabilizer purities} as
\begin{equation}
\Delta_{\Omega}(\psi) \coloneqq \Re\tr(\Omega \psi^{\otimes k})\,, \label{eq:generalizedstabilizerpurityintro}
\end{equation}
and demonstrate that $ \Delta_{\Omega}(\psi) $ satisfies the fundamental properties required of a magic measure: it is naturally Clifford-invariant, and $ \Delta_{\Omega}(\psi) = 1 $ if and only if $ \psi $ is a stabilizer state. 
Since any measurable magic monotone is fully characterized by an operator lying in the commutant of the Clifford group, it must be expressible as a linear combination of the generalized stabilizer purities introduced in \cref{eq:generalizedstabilizerpurityintro}.

We summarize our findings in the following informal theorem: all magic monotones (restricted to pure states) that admit an (unbiased) experimental measurement procedure are fully specified by operators belonging to the commutant of the Clifford group. Thus, all magic measures are ultimately expectation values of \emph{Pauli polynomials}.

\begin{theorem*}[Magic-state resource theory and the commutant. Informal summary of \cref{Sec:magicstateresourcetheory}] The following facts hold:
\begin{enumerate}[label=(\Alph*)]
    \item Any measurable magic monotone is fully specified by an operator belonging to the $k$-th order commutant of the Clifford group.
    \item Any measurable magic monotone is a linear combination of generalized stabilizer purities $\Delta_{\Omega}(\psi)\coloneqq\tr(\Omega\psi^{\otimes k})$ where $\Omega$ is a product of primitive Pauli monomials, defined in \cref{def:primitivepaulimonomialsintro}, obeying the following fundamental property $\Delta_{\Omega}(\psi)=1$ if and only if $\psi$ is a pure stabilizer state.
\end{enumerate}
\end{theorem*}
The above result suggests that every measurable magic monotone can be related to stabilizer entropies~\cite{leone_stabilizer_2022,Leone_2024}, which constitute ``primitive'' stabilizer purities. One of the few measurable magic metrics introduced in the literature beyond stabilizer entropies is the so-called Bell magic $B(\psi)$~\cite{haugbellmagic} (defined in \cref{def:bellmagic}). 

Here, we provide a proof of this connection by establishing lower and upper bounds 
\begin{equation}
2M_{3}(\psi) \le B(\psi) \le 4M_{3}(\psi),
\end{equation}
on the Bell magic, as derived in \cref{th:stabentropyownsbellmagic} using the graphical calculus developed above. Moreover, we express $B(\psi)$ explicitly as the expectation value of a specific Pauli monomial evaluated on $\psi^{\otimes 8}$ (see~\cref{lem:bellmagicpaulimonomials}).

\vspace{1em}
\textbf{Property testing.} We apply our framework to analyze the commutant structure of the Clifford group; we derive meaningful insights into property testing for stabilizer states (\cref{sec:propertytesting}). This topic has recently attracted considerable attention in the quantum information community~\cite{montanaro2018surveyquantumpropertytesting,gross_schurweyl_2019,bu2023stabilizertestingmagicentropy,grewalImprovedStabilizerEstimation2023,arunachalam2024polynomialtimetoleranttestingstabilizer,bao2024toleranttestingstabilizerstates,chen2024stabilizerbootstrappingrecipeefficient,hinsche2024singlecopystabilizertesting}.
Property testing of stabilizer states aims to determine, with high probability, whether an unknown quantum state is close to or far from the stabilizer set (e.g., in trace distance), based on measurements performed on multiple copies of the state. Assuming the state is pure, the task reduces to deciding whether its stabilizer fidelity~\cite{veitch_resource_2014,chen2024stabilizerbootstrappingrecipeefficient} exceeds or falls below given thresholds. Mathematically, the goal is to construct a two-outcome POVM $ \{E, I - E\} $ that can reliably distinguish between these two cases.
A key observation from the literature~\cite{montanaro2018surveyquantumpropertytesting,gross_schurweyl_2019,hinsche2024singlecopystabilizertesting}, which follows naturally from symmetry arguments, is that the optimal stabilizer-testing POVM $ E $, acting on $ k $ copies of the unknown state, must lie within the $ k $-th order commutant of the Clifford group. This establishes the commutant as the central object to analyze when designing efficient testing protocols. 
Using our framework, it immediately follows that no stabilizer testing strategy can succeed when fewer than six copies of the unknown state are available, a fact previously noted in Refs.~\cite{Damanik2018Optimality,gross_schurweyl_2019}. We show how this conclusion arises directly from the fact that, for $ k \leq 5 $, the commutant of the Clifford group consists solely (up to permutations) of primitive Pauli monomials that are proportional to projectors, rather than unitaries.
While we establish the impossibility of testing with fewer than six copies (i.e., the success probability of any strategy, up to exponentially small corrections in the number of qubits, is $1/2$, equivalent to random guessing), we also show that when access to six copies is available, no strategy outperforms the POVM associated with $\Omega_6$. This POVM can be efficiently implemented using Bell sampling~\cite{haugEfficientQuantumAlgorithms2024} and is intimately related to the $(k=3)$-stabilizer entropy in \cref{eq:stabilizerentropiesintro}. This, in turn, provides an \emph{operational interpretation} to the $(k=3)$-stabilizer entropy. Namely, given $6$ copies of a state $\psi$, the optimal (and achievable) success probability for testing whether $\psi$ is a stabilizer state or not is given by
\begin{equation}
p_{\text{succ}}^{\text{optimal}} = \frac{1}{2} + \frac{1}{4} \qty(1 - 2^{-2M_3(\psi)}).
\end{equation}

We can also generalize this result to multiple copies. In \cref{prop:bestpovmisaunitary}, we prove
that for any fixed $ k = O(n^{2-\gamma})$ with any $\gamma>0$, the optimal stabilizer-testing POVM must be a linear combination of Pauli monomials corresponding to unitary operators. These findings establish a meaningful connection between stabilizer purities and stabilizer property testing, highlighting the relationship between these two areas under the umbrella of the Clifford group commutant.

\subsection{Generalization to qudits}\label{sec:quditcase}
All of the results discussed above can be generalized to the qudit setting. While the main focus of this work is on qubits, in \cref{sec:genquditsapp} we generalize our findings to systems of $n$ qudits of local dimension $q$. Crucially, we assume that $q$ is a prime number, which ensures that the generalized Pauli operators—also known as Weyl operators—share the same eigenvalue structure. This enables a meaningful generalization of the multi-qudit Clifford group, analogous to the qubit case. For simplicity, we adopt the same notation used in the qubit setting: $\mathbb{P}_n$ and $\mathcal{C}_n$ now denote the Pauli group and the Clifford group on $n$ qudits, respectively, and $d = q^n$ denotes the dimension of the corresponding Hilbert space.

Analogous to the qubit case, the independent graph-based Pauli monomials introduced in \cref{def:indeppaulimonomialsINTRO} form an orthogonal basis for the commutant for any number of qudits $n$ and tensor power $k$, provided that $\rank_q(G) \ge 2(m - n)$, where $\rank_q$ is the rank of $\mathbb{F}_q$ matrices in the finite field $\mathbb{F}_q$. Here, the matrices $V \in \mathbb{F}_q^{k \times m}$ and $G \in \mathbb{F}_q^{m \times m}$ are defined over the finite field $\mathbb{F}_q$, with $G$ encoding the commutation structure of the Weyl operators. From this construction, one can explicitly compute the dimension of the commutant, given by
\begin{equation}
\dim(\mathcal{C}_n^{\otimes k}) \simeq 
\begin{cases}
q^{\frac{k^2 - 3k}{2}} & \text{if } 2n \ge k - 1, \\
q^{2kn - 2n^2 - 3n} & \text{if } 2n < k - 1,
\end{cases}
\end{equation}
where the full expression is provided in \cref{eq:dimensioncommutantcliffordgroupqudit}.

Using a Fourier-like transform, we define the qudit analogue of Pauli monomials (introduced in \cref{eq:paulimonomialsintro}) as \emph{Weyl monomials}, which may also admit a graphical calculus for manipulation—though more involved than in the qubit case. As in the qubit setting, the Weyl monomial basis can be generated from products of primitive Weyl monomials, defined analogously to \cref{def:primitivepaulimonomialsintro}, which therefore provide generators for the algebra of the commutant.

Although we do not explore specific applications in the qudit regime, we expect that nearly all of our conclusions regarding magic-state resource theory and stabilizer property testing naturally extend to the qudit case. Notably, as observed in Ref.~\cite{Turkeshi_2025}, Weyl monomials allow one to define a family of qudit generalizations of stabilizer entropies, as introduced in \cref{eq:stabilizerentropiesintro}. In particular, when restricted to primitive Weyl monomials, these entropies have been shown to be monotones under magic-state-free operations, as discussed in Ref.~\cite{Turkeshi_2025}.

\subsection{Prior works}
\label{sec:priorwork}
In this section, we contextualize our results by discussing previous literature. The Clifford group commutant has garnered significant attention because, although Clifford operators form classically simulable quantum circuits, their randomness properties are remarkable. The uniform distribution over multiqudits Clifford unitaries is a unitary $2$-design for prime-dimensional qudits~\cite{divincenzo_quantum_2002, dankert_efficient_2005, dankert_exact_2009} and can be shown to be $3$-designs for qubits~\cite{zhu_multiqubit_2017, webb_clifford_2016}. While demonstrating that the Clifford group forms a unitary $3$-design does not necessarily require understanding its commutant for $k=3$, this understanding becomes fundamental for larger values of $k$.

As noted in Ref.~\cite{gross_schurweyl_2019}, early connections to the commutant of the Clifford group can already be found in the work of Nebe, Rains, and Sloane on self-dual codes and Clifford-Weil groups~\cite{nebe_2006}. Their work studies invariant spaces associated with tensor powers of the Clifford group in a broad algebraic setting. From the perspective of quantum information theory, the commutant of the Clifford group arises naturally within this framework, although explicit descriptions of its structure were not the primary focus of their work. Beyond $k=3$, one of the first significant results is presented in a series of works, see Refs.~\cite{zhu_clifford_2016, helsen_representations_2018}, which proved that the Clifford group is not a $4$-design but fails \textit{gracefully} to be a unitary $4$-design by requiring only an additional element to generate the commutant of the Clifford group for $k=4$. This result marked the beginning of new insights into the theory of nonstabilizerness or magic, indicating that $k=4$ copies are the minimal number sufficient for quantifying magic, leading to the inception of stabilizer entropies~\cite{leone_stabilizer_2022}. However, until then, the commutant of the Clifford group for higher values of $k$ remained unknown.

 Significant progress has been made in Ref.~\cite{gross_schurweyl_2019}. Building on key insights from Nebe et al.~\cite{nebe2000invariantscliffordgroups}, Gross, Nezami, and Walter exploited the fundamental observation that the Clifford commutant can be characterized by examining Clifford orbits — specifically, via the Clifford twirling introduced in \cref{eq:twirlingintro} applied to Pauli operators, an approach that had been tentatively explored in~\cite{Van_den_Nest_2005,zhu_clifford_2016}. This key observation greatly simplifies the characterization of the commutant. Using this strategy, the authors developed a basis — corresponding to our \emph{independent graph-based Pauli monomials} basis — to compute the dimension of the commutant for $k \leq n+1$. Additionally, they introduced an alternative basis for $k \leq n+1$ using the notion of Lagrangian subspaces, which coincides with the basis of the reduced Pauli monomials (see \cref{def:reducedpaulimonomials}), leading to a number of interesting applications.

Our work builds on all the insightful and fundamental results discussed above. We extend the results of Ref.~\cite{gross_schurweyl_2019} by characterizing the commutant for any $n$ and $k$, fully characterizing the \emph{independent graph-based Pauli monomials}, exploring some of their properties, and computing its dimension. We also present an alternative way of understanding the basis of the commutant initially introduced in Ref.~\cite{gross_schurweyl_2019} using Pauli monomials, providing fruitful insights, and showing how the two bases of the commutant can be transformed into each other via a Fourier-like transform. We conclude by demonstrating that the Clifford group continues to gracefully fail to be a unitary $k$-design for any $k$, requiring at most three additional generating elements compared to the full unitary group, generalizing the seminal findings of Ref.~\cite{zhu_clifford_2016}.

\subsection{Structure of this work} 
In the remainder of the manuscript, we present a self-contained and pedagogical exposition of the Clifford group commutant through our framework, covering all relevant details in a systematic manner.  

Specifically, we begin with \cref{sec:preliminaries}, which introduces essential notation and foundational concepts necessary for the proofs throughout this work.  
In \cref{sec:comclif}, we establish one of our main results: a simple construction of a basis for the $k$-th order commutant of the $n$-qubit Clifford group for arbitrary $n$ and $k$, using only the algebraic properties of Pauli operators. While conceptually straightforward—relying on independent graph-based Pauli monomials (\cref{def:indeppaulimonomials})—this basis may not be the most convenient for direct manipulation.  

To address this, \cref{sec:Paulimon} introduces the broader notion of \textit{Pauli monomials}, defined as various isotropic sums of Pauli operators. We develop a \textit{graphical calculus}, a diagrammatic framework for handling these operators, and demonstrate invertible transformations between different Pauli monomial bases. This 
formalism allows us to systematically explore alternative bases for the commutant.  

In \cref{sec:commutantCliffordgen}, we prove our second main result that provides a minimal generating set for the commutant of the Clifford group. This reveals the algebraic structure of the commutant through the lens of the semigroup of Pauli monomials.  
The generalization of our results from qubits to qudits is outlined in~\cref{sec:quditcase}.

\cref{sec:applications} is devoted to applications and examples, demonstrating the practical implications of our results.  
Specifically, in \cref{Sec:haaraveragecliffordbeyond}, we develop tools and techniques for computing Haar averages over the Clifford group—arguably the most natural motivation in quantum information for studying the Clifford group commutant. We introduce the so-called \textit{Clifford-Weingarten calculus}, inspired by the analogous Weingarten calculus used for Haar averages over the unitary group. Additionally, we provide asymptotic behaviors for large $n$ ($n \gg \sqrt{k}$), which significantly simplify the computation of the twirling operator in \cref{eq:twirlingintro} for relatively low values of $k$.  
In \cref{Sec:magicstateresourcetheory}, we explore the deep connection between the commutant of the Clifford group and magic-state resource theory, while in \cref{Sec:haaraveragecliffordbeyond}, we discuss its relation to stabilizer property testing.
Finally, in \cref{sec:cliffordcommutantexample}, we explicitly construct the commutant of the Clifford group up to $k=8$, leveraging the  framework of graphical calculus. In this context, we also demonstrate how to perform high-order Clifford averages through a series of practical calculations and examples. Furthermore, we provide a detailed \texttt{Mathematica} notebook containing Clifford-Weingarten matrices encoded up to $k=8$, enabling readers to experiment with these tools using their own examples~\cite{leone_mathematica_2023}.

\section{Preliminaries}\label{sec:preliminaries}

We use the notation $[n]$ to denote the set $[n]\coloneqq\{1,\dots,n\}$, where $n\in\mathbb{N}$.
The finite field $\mathbb{F}_2$ consists of the elements $\{0, 1\}$ with addition and multiplication defined modulo 2. The vector space $\mathbb{F}_2^k$ comprises all $k$-dimensional vectors over $\mathbb{F}_2$. 
The Hamming weight of a vector $x \in \mathbb{F}_2^k$, denoted by $|x|$, is defined as $|x| \coloneqq \sum_{i=1}^k x_i.$
The set of $ k \times m $ binary matrices over $ \mathbb{F}_2 $ is denoted by $ \mathbb{F}_2^{k \times m} $. The rank of a matrix is given by the number of linearly independent columns (or equivalently, row vectors). We denote the rank of a matrix $ A $ over $ \mathbb{F}_2 $ as $ \mathrm{rank}_2(A) $.

The set $ \symftwo $ denotes the set of all symmetric $ m \times m $ matrices over $ \mathbb{F}_2 $, having a null diagonal. That is,
\begin{equation}
\symftwo \coloneqq \left\{ M \in \mathbb{F}_2^{m \times m} : M^T = M\,,\, M_{j,j}=0\,\,\forall j\in[m] \right\}.
\end{equation}
Similarly, we define set $\even$ as the set of binary matrices $V\in\mathbb{F}_{2}^{k\times m}$ with $m\in[k]$ with column vectors $V_{\sbullet,i}$ of even Hamming weight
\begin{equation}
\mathrm{Even}(\mathbb{F}_{2}^{k\times m})\coloneqq\{V\in\mathbb{F}_{2}^{k\times m}\,:\, |V_{\sbullet,i}|=0\bmod 2\quad\forall i\in[m]\}.
\end{equation}

\begin{lemma}[Cardinality of $\symf$ ~\cite{Mathf2Dim}]\label{lem:cardinalitysymf} The cardinality of $\symf$ matrices of $\mathrm{rank}_2$ equal to $2r$ is given by 
    \begin{align}
        N_0(m,r)&=\prod_{i=1}^r \frac{2^{2i-2}}{2^{2i}-1}\times \prod_{i=0}^{2r-1}(2^{m-i}-1)\\
        \nonumber
        &=2^{2mr-r-2r^2}\prod_{i=1}^r \frac{(1-2^{-m+i-1})(1-2^{-m+r+i-1})}{1-2^{-2i}}\,.
    \end{align}
    The total number of $\symf$ matrices is given by $2^{m(m-1)/2}$.

\end{lemma}
\begin{lemma}[Number of subspaces with even Hamming weight]\label{lem:gaussiancoeff}
The number of $ m $-dimensional subspaces $ V \subset \mathbb{F}_2^k $ such that all vectors in $ V $ have even Hamming weight, i.e., for all $ v \in V $, $ |v| = 0 \pmod{2} $, is given by the Gaussian 
binomial coefficient 
\begin{align}
        \binom{k-1}{m}_2&\coloneqq\prod_{j=1}^{m}\frac{2^{k-j}-1}{2^j-1}
        =2^{km-m^2-m}\prod_{j=1}^{m}\frac{1-2^{-k+j}}{1-2^{-j}}\,.
\end{align}
The maximal dimension is $m\leq k-1$.
\end{lemma}
\begin{proof}
The dimension of the subspace of $ \mathbb{F}_2^k $ consisting of all vectors with even Hamming weight is $ k-1 $. Also, all even weight subspaces are themselves subspaces of this subspace. The proof is concluded by using the operational definition of the Gaussian binomial coefficient $ \binom{k-1}{m}_2 $ which counts the number of these subspaces (see Ref.~\cite{chebolu2021gaussianbinomialcoefficientsgroup}).
\end{proof}

Let us also set the notation for the general linear group over $\mathbb{F}_2^{m\times m}$
\begin{equation}
\mathrm{GL}( \mathbb{F}_2^{m\times m}) \coloneqq \left\{ A \in \mathbb{F}_2^{m \times m} \mid \text{$A$ is invertible} \right\}
\end{equation}
that consists of all the invertible matrices over $\mathbb{F}_2^{m\times m}$.

\subsection{Haar measure}  
We denote by $ \mathcal{H} $ the Hilbert space of $ n $ qubits, that is, $ \mathcal{H} \cong \left(\mathbb{C}^{2}\right)^{\otimes n} $, and consider $ k $ copies of it, $ \mathcal{H}^{\otimes k} $. Throughout this manuscript, we use $ d \coloneqq 2^n $ to represent the dimension of the Hilbert space for an $ n $-qubit system. The set of Hermitian operators on $ \mathcal{H}^{\otimes k} $ is denoted by $ \mathcal{B}(\mathcal{H}^{\otimes k}) $, and the unitary group on $ \mathcal{H} $ is denoted by $ \mathcal{U}(2^n) $.
Let $ \mathcal{G}_n \subseteq \mathcal{U}(2^n) $ be a continuous or discrete subgroup of the unitary group. The Haar measure on $ \mathcal{G}_n $ is defined as the unique probability measure that is both left- and right-invariant. Throughout this manuscript, we denote the Haar measure as $ \mu_{\mathrm{H}} $. 
Mathematically, the Haar measure satisfies the  invariance property~\cite{watrous_theory_2018,collins_moments_2003,collins_integration_2006,kliesch_theory_2021,Mele_2024}  
\begin{align}
    \int \de \mu_H(G) = 1, \quad \int \de \mu_H(G) \, f(G) = \int \de \mu_H(G) \, f(GG') = \int \de \mu_H(G) \, f(G'G),
    \label{eq:leftrightinvariance}
\end{align}
for all $ G' \in \mathcal{G}_n $ and any arbitrary function $ f $.  

For discrete groups, such as the Clifford group, the Haar measure reduces to a sum over delta functions, and integration simplifies to a normalized summation over the group's elements.  
With this formalism in place, we now introduce a central object of this work: the twirling super-operator.
\begin{definition}[Twirling (super-)operator over a group $\mathcal{G}_n$] 
Let $O \in \mathcal{B}(\mathcal{H}^{\otimes k})$, and $\mu_H$ be the Haar measure over $\mathcal{G}_n$. The twirling super-operator is defined as
\begin{equation}
\Phi_{\mathcal{G}_n}^{(k)}(O) \coloneqq \int \de \mu_H(G) \, G^{\otimes k} O G^{\dag\otimes k}\,.
\label{twirlingofoperatorOunitary}
\end{equation}
The super-operator $\Phi_{\mathcal{G}_n}(\cdot)$ is sometimes referred to as the \textit{$k$-fold channel} of the group $\mathcal{G}_n$ applied to the operator $O$, or simply the \textit{Haar average} of $O$ over $\mathcal{G}_n$.  
\end{definition}
A crucial tool in analyzing the twirling super-operator is the so-called commutant, which we now define.
\begin{definition}[$k$-th order commutant of a group $\mathcal{G}_n$]\label{def:commutantofagroup}
Let $ k \in \mathbb{N}$ and let $\mathcal{G}_n$ be a subgroup of the unitary group $\mathcal{U}_n$. The $k$-fold commutant of the group $\mathcal{G}_n$ is the set defined as
\begin{equation}\label{eq:comm}
\com(\mathcal{G}_n^{\otimes k}) \coloneqq \{ O \in \mathcal{B}(\mathcal{H}^{\otimes k}) \mid U_g^{\dagger \otimes k} O U_g^{\otimes k} = O, \, \forall U_g \in \mathcal{G}_n \}.
\end{equation}
\end{definition}

The commutant is a subspace of linear operators, and moreover, the product of two operators in the commutant remains in the commutant. Thus, the commutant is an \emph{algebra}.
The following lemma establishes the fundamental principles for computing the $k$-fold channel $\Phi_{\mathcal{G}_n}^{(k)}$, linking it to the commutant of the subgroup $\mathcal{G}_{n}^{\otimes k}$.
\begin{lemma}[The image of the twirling operator belongs to the commutant]\label{lem:twirlingbelongstothecommutant} 
Let $\mathcal{G}_n \subseteq \mathcal{U}(2^n)$ be any subgroup of the unitary group, and let $\Phi_{\mathcal{G}_n}^{(k)}(\cdot)$ be the twirling operator. For any operator $O \in \mathcal{B}(\mathcal{H}^{\otimes k})$, we have $\Phi_{\mathcal{G}_n}^{(k)}(O) \in \com(\mathcal{G}_{n}^{\otimes k})$.
\end{lemma}
The proof follows directly from the invariance of the Haar measure. 
For twirling operations over a finite group, one can establish the following general property.
\begin{lemma}[Effect of group twirling over finite groups]\label{lem:twirl_finite}
    Let $ G $ be a finite group, $ x \in X $ an element of a vector space, and $(g,x)\mapsto g \star x \in X $ the group action. Then, it holds that
    \begin{align}
        \frac{1}{|G|} \sum_{g\in G} g \star x = \frac{1}{|\mathrm{orb}(x)|} \sum_{y\in \mathrm{orb}(x)} y\,,
    \end{align}
    where the orbit of $ x $ is defined as $ \mathrm{orb}(x) \coloneqq \{ y \in X \mid \exists g \in G : y = g \star x \} $.
\end{lemma}
\begin{proof}
    We can express the group twirling as 
    \begin{align}
        \frac{1}{|G|} \sum_{g\in G} g \star x = \frac{1}{|G|} \sum_{y\in \mathrm{orb}(x)} |\mathrm{costab}(y, x)| \, y\,,
    \end{align}
    where $\mathrm{costab}(y, x) \coloneqq \{ g \in G \mid y = g \star x \}$ denotes the set of elements in $ G $ that map $ x $ to $ y $.
    Now, let $ g \in \mathrm{costab}(y, x) $. Then, we can rewrite the set as
    \begin{align}
        \mathrm{costab}(y, x) = g^{-1}   \mathrm{costab}(x, x) = g^{-1}  \mathrm{stab}(x),
    \end{align}
    where the multiplication acts entry-wise, and we have used that $\mathrm{stab}(x)\coloneqq \mathrm{costab}(x,x)$. 
    This implies that $|\mathrm{costab}(y, x)|$ is the same for all $ y \in \mathrm{orb}(x)$. Using the orbit-stabilizer theorem~\cite{fulton_representation_2004}, which states that $|\mathrm{stab}(x)| = \frac{|G|}{|\mathrm{orb}(x)|}$, the result follows.
\end{proof}

\subsection{Pauli operators} 
 In this section, we provide a brief introduction to Pauli operators and their relevant properties, which will be used throughout the manuscript.
Let us introduce the qubit Pauli matrices $ \{I,X,Y,Z\} $ in the standard way as
\begin{equation}
I =
\begin{pmatrix}
    1 & 0 \\
    0 & 1
\end{pmatrix}, \quad
X =
\begin{pmatrix}
    0 & 1 \\
    1 & 0
\end{pmatrix}, \quad
Y =
\begin{pmatrix}
    0 & -i  \\
    i &  0
\end{pmatrix}, \quad
Z =
\begin{pmatrix}
    1 & 0 \\
    0 & -1 
\end{pmatrix}.
\end{equation}
The Pauli matrices are Hermitian, and the Pauli group on $ n $-qubits is defined as the $ n $-fold tensor product of the Pauli matrices multiplied by a phase factor from $ \{-1, 1, -i, +i\} $. We denote the Pauli group modulo phases as $ \mathbb{P}_n \coloneqq \{I, X, Y, Z\}^{\otimes n} $, which, seen as an operator basis, is often referred to as the Pauli basis.
Let us define the bitstring representation of Pauli operators, which enables efficient manipulation of Pauli operators.
\begin{definition}[bitstring representation of a Pauli]
\label{def:bitstring}
Let $ P \in \mathbb{P}_n $ be a Pauli operator. The \emph{bitstring representation} of $ P $ is the binary vector 
\begin{equation}
b(P) \coloneqq (b_1, b_2, \dots, b_{2n})^T \in \mathbb{F}_2^{2n}
\end{equation}
which uniquely determines the Pauli operator by the relation
\begin{align}
    P =\phi(P) \times X_1^{b_1} Z_1^{b_2} \otimes X_2^{b_3} Z_2^{b_4} \otimes \cdots \otimes X_n^{b_{2n-1}} Z_n^{b_{2n}}\,,
\end{align}  
with $\phi(P)\coloneqq (i)^{b(P)^T J b(P)}=(i)^{\# Y(P)}$,
where the symplectic inner product of $ x, y \in \mathbb{F}_2^{2n} $ is defined as 
\begin{align}
\label{eq:symplinner}
    a^T J b &= \sum_{j=1}^{n}  a_{2j} b_{2j-1},
    \quad \text{where } J \coloneqq \bigoplus_{j=1}^{n} \begin{pmatrix} 0 & 0 \\ 1 & 0 \end{pmatrix}.
\end{align}
\end{definition}
In the following, we summarize some well-known properties of the bitstring representation of Pauli operators. Specifically, for any two Pauli operators $ P, Q \in \mathbb{P}_n $, the multiplication can be represented in the following way:
$\phi PQ = R$ where $b(R) = b(P) + b(Q)\pmod 2$ and $ \phi =\phi(R)\phi(P)^{-1}\phi(Q)^{-1} (-1)^{b(P)^TJ b(Q)}$.
It also holds that $\phi^2=\chi(P, Q)$ (because of \cref{le:trivialPQ}).
As such, we have
\begin{align}
[P, Q] = 0 \quad \text{if and only if} \quad b(P)^T (J+J^T)\, b(Q) = 0 \pmod 2\,.
\label{eq:commsymp}
\end{align}
\begin{lemma}
    For $\boldsymbol P=(P_1,\dots,P_m)\in\mathbb P_n^{m}$, with $Q=\phi(Q,\boldsymbol P) P_1P_2\cdots P_m$, it holds that 
    \begin{align}
        \phi(Q,\boldsymbol P)&=\phi(Q)\prod_{i=1}^m \phi^{-1}(P_i)\times (-1)^{\sum_{i< j} b(P_i)^TJ b(P_j)}\\
        &=i^{\#Y(Q)-\sum_{i=1}^m \#Y(P_i)}\times (-1)^{\sum_{i< j} b(P_i)^TJ b(P_j)}.
        \nonumber
    \end{align}
\end{lemma}
\begin{proof}
    This follows from applying the product formula multiple times and using $b(PQ)=b(P)+b(Q)\pmod 2$.
\end{proof}
\begin{definition}[Algebraically independent Pauli operators]
Let $ \boldsymbol{P} \coloneqq (P_1, P_2, \dots, P_m) \subseteq \mathbb{P}_n^m $ be an ordered collection of Pauli operators. The set of Pauli operators $ \boldsymbol{P} $ is said to be algebraically independent if no operator $ P_i \in \boldsymbol{P} $ can be expressed as a product of the remaining operators in $ \boldsymbol{P} $ (up to a phase factor). 
\end{definition}
We call
$B_{\boldsymbol P}=(b(P_1),\dots b(P_m))\in \mathbb F_2^{2n\times m}$ the bitstring matrix of the collection $B_{\boldsymbol P}$.
A collection of Pauli operators \( \boldsymbol{P} \) is algebraically independent if and only if the matrix \( B_{\boldsymbol{P}} \) has full rank over \( \mathbb{F}_2 \), i.e., the bitstring representations \( b(P_1), \dots, b(P_m) \) are linearly independent in \( \mathbb{F}_2^{2n} \).

To simplify the notation, we often denote a subset of Pauli operators by a bold character, as in the previous definition.
\begin{definition}
Given a vector $ v \in \mathbb{F}_2^{k} $ and a Pauli operator $ P \in \mathbb{P}_n $, we define the tensor product of Pauli operators weighted by the vector $ v $ as
\begin{equation}
P^{\otimes v} \coloneqq P^{v_1} \otimes P^{v_2} \otimes \cdots \otimes P^{v_k},
\end{equation}
where $ P^{v_i} $ denotes the Pauli operator $ P $ raised to the power corresponding to the $ i $-th entry of the vector $ v $. In particular, if $ v_i = 0 $, then $ P^{v_i} = \mathbb 1 $, where $ \mathbb 1 $ denotes the identity operator. We also define the binary matrix representation $B(P^{\otimes v})=v b(P)^T \in \mathbb{F}_2^{k\times 2n}$, where the copies represent rows and the system columns. 
\end{definition}
Note that the binary matrix representation $B(P^{\otimes v})=v b(P)^T$ corresponds to the vectorization of the bitstring: $ b(P^{\otimes v}) = v \otimes b(P)$.

We now introduce the concept of \emph{anticommutation graph} associated with a subset of Pauli operators, which is a common tool in other contexts of quantum information theory~\cite{anschuetz2024boundsgroundstateenergy,king2024triplyefficientshadowtomography,chen2024optimaltradeoffsestimatingpauli,aguilar2024classificationpauliliealgebras,bao2024toleranttestingstabilizerstates,diaz2023showcasingbarrenplateautheory,west2025nogotheoremssublineardepthgroup}. To define it formally, we first establish the notation for graphs and their adjacency matrices.

\begin{definition}[Graph and its adjacency matrix]
\label{def:graph}
    A \emph{graph} $ G $ is a pair $ G = (W, E) $, where
    \begin{itemize}
        \item $ W = \{w_1, w_2, \ldots, w_m\} $ is a finite set of elements called \emph{vertices} (also referred to as \emph{nodes}),
        \item $ E \subseteq \{(w_i, w_j) \mid w_i, w_j \in W, w_i \neq w_j \} $ is a set of edges, where each edge is an unordered pair of distinct vertices from $ W $, i.e., $ (w_i, w_j) $ represents an 
        edge between $ w_i $ and $ w_j $.
    \end{itemize}
    The \emph{adjacency matrix} $ A \in \mathbb{F}_2^{m \times m} $ of a graph $ G = (W, E) $ is a (symmetric) matrix where the rows and columns correspond to the vertices $ w_1, w_2, \ldots, w_m $. The entries $ A_{i,j} $ are defined as
    \begin{equation}
    A_{i,j} = \begin{cases}
        1, & \text{if } (w_i, w_j) \in E, \\
        0, & \text{otherwise.}
    \end{cases}
    \end{equation}
\end{definition}

\begin{definition}[anticommutation graph]
\label{def:anticommgraph}
    Given a collection of Pauli operators $ \mathbf{P} \coloneqq (P_1, \dots, P_m ) \in \mathbb{P}_n^{\times m} $, we define its associated \emph{anticommutation graph} as the graph given by the pair $ (\mathbf{P}, E) $, where
    \begin{itemize}
        \item The  vertices are given by $ \mathbf{P} \coloneqq (P_1, \dots, P_m ) $,
        \item The set of edges is defined by
        \begin{align}
             E \coloneqq \{ (P_i, P_j) \mid P_i, P_j \in \mathbf{P}, [P_i, P_j] \neq 0 \},
        \end{align}
        where an edge connects two vertices $ P_i $ and $ P_j $ if their corresponding Pauli operators anticommmute. We denote by $ \mathcal{A}(\mathbf{P}) $ the adjacency matrix of \emph{anticommutation graph} associated with $ \mathbf{P} $. Note that, by construction, $\mathcal{A}(\boldsymbol{P})\in\symf$.
    \end{itemize}
\end{definition}
The element $ \mathcal{A}(\mathbf{P})_{i,j} $ can be expressed in terms of the bitstring representation of the Pauli operators in $ \mathbf{P}$ as  
\begin{align}
    \mathcal{A}(\mathbf{P})_{i,j} &= b(P_i)^T (J+J^{T}) b(P_j), \label{eq:anticomm1}
\end{align}
as it follows from \cref{eq:commsymp} and the definition of anticommutation graph
(see also Ref.~\cite{aguilar2024classificationpauliliealgebras}).
Thus, we have
\begin{align}
\label{eq:anticommJ}
    \mathcal{A}(\boldsymbol{P}) &= B_{\boldsymbol{P}}^T (J+J^T) B_{\boldsymbol{P}},
\end{align}
where we recall that $ B_{\boldsymbol{P}} $ denotes the matrix whose columns correspond to the bitstring representations of the Pauli operators in $ \boldsymbol{P} $.
We now present the following lemma, which will be crucial in subsequent sections.
\begin{lemma}[Decomposition of Pauli operators into tensor products]
\label{le:uniquenessPV}
For any Pauli operator $ Q \in \mathbb{P}_n^{\otimes k} $, there exists a unique integer $ m \leq k $ (referred to as the \emph{order}) such that $ Q $ can be expressed in terms of linearly independent binary vectors $ v_1, \dots, v_m \in \mathbb{F}_2^{k} $ and algebraically independent Pauli operators $ \mathbf{P} \coloneqq ( P_1, \dots, P_m)  \in \mathbb{P}_n^m $, satisfying  
\begin{align}
    Q = \phi P_1^{\otimes v_1} P_2^{\otimes v_2} \cdots P_m^{\otimes v_m},
    \label{eq:QP1P2}
\end{align}
for some phase factor $ \phi \in \{\pm 1, \pm i\} $.
Moreover, any two such decompositions of the same operator $Q$ are related by an invertible transformation $ A \in \mathrm{GL}(\mathbb{F}_{2}^{m\times m}) $, acting as  
\begin{align}
    V \mapsto V'=V A, \quad B_{\boldsymbol P} \mapsto B_{\boldsymbol P'}=B_{\boldsymbol P} A^{-T},
\end{align}
where  
\begin{itemize}
    \item $ V \in \mathbb{F}_2^{k \times m} $ is the matrix whose columns are $ v_1, \dots, v_m $,
    \item $ B_{\boldsymbol P} \in \mathbb{F}_2^{2n \times m} $ is the matrix whose columns correspond to the bitstring representations of the Pauli operators in $ \mathbf{P} $,
    \item $ A^{-T} $ denotes the inverse transpose of $ A $.
\end{itemize}
Equivalently, $ Q $ can be written as $Q \propto P_1^{\prime \otimes v_1^{\prime}} P_2^{\prime \otimes v_2^{\prime}} \cdots P_m^{\prime \otimes v_m^{\prime}},$ where $ v_i^{\prime} $ is the $ i $-th column of $ V A $, and $ P_i^{\prime}$ is the Pauli operator whose bitstring representation is encoded in the $ i $-th column of $ B_{\boldsymbol P} A^{-T} $, i.e., $ P_i^{\prime} \propto \prod^m_{j=1} P_j^{A^{-1}_{i,j}} $, up to an overall phase factor.

\end{lemma}

\begin{proof}
The bitstring representation of $ Q \in \mathbb{P}^{\otimes k}_n$, denoted as $ b(Q) \in \mathbb{F}_2^{2nk}  \cong \left(\mathbb{F}_2^{2n}\right)^{\otimes k} $, can be mapped to a matrix $ B(Q) \in \mathbb{F}_2^{k \times 2n} $ via the vectorization isomorphism, where each row corresponds to the Pauli bitstring in one of the $ k $ copies.
For a Pauli operator $ P \in \mathbb{P}_n $ and a vector $ v \in \mathbb{F}_2^k $, the bitstring representation of $ P^{\otimes v} $  via the vectorization isomorphism is given by
\begin{align}
    B(P^{\otimes v}) \coloneqq v b(P)^T.
\end{align}
Furthermore, given a collection of Pauli operators $ \boldsymbol P=(P_1,\dots P_m) \in \mathbb{P}_n^{\times k} $ and vectors $V=(v_1,\dots,v_m) \in \mathbb{F}_2^{k\times m} $,  
by the vectorization isomorphism, this leads to the matrix expression  
\begin{align}
  B(P_1^{\otimes v_1} P_2^{\otimes v_2} \cdots P_m^{\otimes v_m}) 
  = \sum_{j=1}^{m} v_j b(P_j)^T = V B_{\mathbf{P}}^T\,.
\end{align}
This implies that the claim of the lemma, which asserts the existence of a decomposition of $ Q $ into tensor products of Pauli operators, can be shown by decomposing $ B(Q) $ in the form
\begin{equation}
    B(Q) = V B_{\mathbf{P}}^T,
\end{equation}
where $ V \in \mathbb{F}_2^{k \times m} $ and $ B_{\mathbf{P}} \in \mathbb{F}_2^{m \times 2n} $ are matrices with linearly independent columns. The linearly independent columns of $ V $ correspond to the vectors $ v_1, \dots, v_m $, while the linearly independent columns of $ B_{\mathbf{P}} $ correspond to the bitstring representations of the Pauli operators $ P_1, \dots, P_m $, ensuring that these operators are algebraically independent.

Furthermore, to ensure that $ m $ is unique, we take $ m $ to be $ m = \mathrm{rank}(B(Q)) $, which represents the number of linearly independent rows of $ B(Q) $ over $ \mathbb{F}_2 $. Consequently, we can choose $ V $ and $ B_{\mathbf{P}} $ to be the \emph{minimal rank decompositions} of $ B(Q) $.
This decomposition can be achieved by performing Gaussian elimination, writing $ B(Q) $ as 
\begin{align}
    B(Q) = A \begin{pmatrix}
    X \\
    0_{(k- m) \times 2n}
\end{pmatrix},    
\end{align}
where $ X \in \mathbb{F}_2^{m \times 2n} $ has maximal rank $ m = \mathrm{rank}(B(Q)) $, $ 0_{(k - m) \times 2n} \in \mathbb{F}_2^{(k - m) \times 2n} $ is the zero matrix, and $ A \in \mathrm{GL}(\mathbb{F}_2^{k, k}) $ is the invertible matrix resulting from Gaussian elimination on the rows. Thus, by taking $ V $ to be the matrix formed by the first $ m $ columns of $ A $ (which are linearly independent, since $ A \in \mathrm{GL}(\mathbb{F}_2^{k, k}) $), and defining $ B_{\mathbf{P}} \coloneqq X^T $, we obtain the desired decomposition
\begin{align}
   B(Q) = A \begin{pmatrix}
    X \\
    0_{k-m , 2n}
\end{pmatrix} = V B_{\mathbf{P}}^T.
\end{align}

Moreover, this decomposition is not unique. Specifically, we have $ B(Q) = V B_{\mathbf{P}}^T = \left(V A\right) \left(A^{-1} B_{\mathbf{P}}^T\right) $ for any $ A \in \mathrm{GL}(\mathbb{F}_2^{m \times m}) $, which implies that also $ V' \coloneqq V A $ and $ B'_{\mathbf{P}} \coloneqq B_{\mathbf{P}} A^{-T} $ also give rise to a valid decomposition. This establishes the desired transformation, with all minimal $ m $-rank transformations related by this gauge transformation. Furthermore, all the possible decompositions must be of this form: in fact, the condition $ \mathrm{span}(V) = \mathrm{span}(B(Q)) = \mathrm{span}(V') $ must hold, meaning that the column vectors of $ V $, $ B(Q) $, and $ V' $ span the same vector space, regardless of the specific decomposition. Since all bases are related by an invertible linear transformation, we conclude that $ V' = V A $ for some invertible matrix $ A $, which is captured by the gauge transformation.

\end{proof}

\begin{lemma}[Canonical graph form]\label{lem:can_pauli_graph}  
Let $ \boldsymbol{P} \coloneqq (P_1, \dots, P_m) \in \mathbb{P}_n^m $ be a set of Pauli operators, and define $ r \coloneqq \mathrm{rank}_2(\mathcal{A}(\boldsymbol{P}))/2 $, where $ \mathcal{A}(\boldsymbol{P}) $ is the adjacency matrix of the anticommutation graph associated with $ \boldsymbol{P} $. Then, there exists a transformation $ P_i \mapsto Q_i $ given by  
\begin{align}  
\label{eq:defMap}
    Q_i \coloneqq \prod_{j=1}^{m} P_j^{A_{j,i}}, \quad \text{for all } i \in [m],  
\end{align}  
for some $ A \in \mathrm{GL}(\mathbb{F}_2^{m \times m}) $, such that the adjacency matrix of the anticommutation graph $ \mathcal{A}(\boldsymbol{Q}) $ associated with $ \boldsymbol{Q} \coloneqq \{Q_1, \dots, Q_m\} $ takes the block-diagonal form  
\begin{align}  
    \mathcal{A}(\boldsymbol{Q}) = \bigoplus_{i=1}^{r}  
    \begin{pmatrix}  
        0 & 1 \\  
        1 & 0  
    \end{pmatrix}  
    \oplus 0_{m-2r,m-2r}.  
\end{align}  
\end{lemma}

\begin{proof}
The proof follows inductively. We begin by selecting two Pauli operators $ P_i $ and $ P_j $ which anticommmute. Next, for $ \alpha \in [m] \setminus \{i,j\} $, if $ P_\alpha $ anticommmutes with $ P_i $, we map $ P_\alpha \mapsto P_\alpha P_j $ so that it now commutes with $ P_i $. After this transformation, if $ P_\alpha $ anticommmutes with $ P_j $, we map $ P_\alpha \mapsto P_\alpha P_i $, ensuring that it now commutes with $ P_j $ as well. This ensures that $ P_i $ and $ P_j $ anticommmute with each other but commute with any other Pauli operator in the new set.  
Repeating this process for all other pairs of anticommmuting Pauli operators results in a new set $ \boldsymbol{Q} $ that has the desired structure of the graph after relabeling.

Let $ B_{\boldsymbol{P}} $ be the binary matrix whose columns are the bitstring representations of the Pauli operators in $ \boldsymbol{P} $. Similarly, let $ B_{\boldsymbol{Q}} $ be the corresponding matrix for $ \boldsymbol{Q} $.
Since the operation $ P_\alpha \mapsto P_\alpha P_j $ is invertible, there must exist an invertible mapping between the associated bitstring matrices $ B_{\boldsymbol{P}} $ and $ B_{\boldsymbol{Q}} $, meaning that $B_{\boldsymbol{Q}} = B_{\boldsymbol{P}} A,$ for some $ A \in \mathrm{GL}(\mathbb{F}_2^{m \times m}) $.  
Thus, $ \boldsymbol{Q} $ can be expressed in terms of $ \boldsymbol{P} $ as in \cref{eq:defMap}.  
This implies that the adjacency matrix of the anticommutation graph of $ \boldsymbol{Q} $, $ \mathcal{A}(\boldsymbol{Q}) $, satisfies  
\begin{align}
    \mathcal{A}(\boldsymbol{Q}) = (B_{\boldsymbol{P}} A)^T (J+J^{T}) (B_{\boldsymbol{P}} A) = A^T \mathcal{A}(\boldsymbol{P}) A,
\end{align}
where $ J + J^{T} $ is the standard symplectic matrix (see \cref{eq:anticommJ}).  
Furthermore, since  
\begin{align}
    \mathrm{rank}_2(\mathcal{A}(\boldsymbol{Q})) = \mathrm{rank}_2(A^T \mathcal{A}(\boldsymbol{P}) A) = \mathrm{rank}_2(\mathcal{A}(\boldsymbol{P})),
\end{align}
the rank of the anticommutation graph is preserved.
\end{proof}

So from the previous Lemma it follows this standard fact regarding the normal form of skew-symmetric binary matrices:

Let \( M \in \symftwo \) be a symmetric binary matrix with vanishing diagonal. Then there exists an invertible matrix \( A \in \mathrm{GL}(\mathbb{F}_2^{m \times m}) \) such that
\begin{align}
A^T M A = \bigoplus_{i=1}^r 
\begin{pmatrix}
0 & 1 \\
1 & 0
\end{pmatrix}
\oplus 0_{m - 2r},
\end{align}
where \( \operatorname{rank}(M) = 2r \).

\subsection{Clifford group} 
The Clifford group can be defined as follows:
\begin{definition}[Clifford group]
    The Clifford group $\mathcal{C}_n$ is a subgroup of the unitary group $\mathcal{U}(2^n)$, consisting of all unitary operators that map the Pauli basis $\mathbb{P}_n$ onto itself under conjugation
    \begin{equation}
    \mathcal{C}_n \coloneqq \{ C \in \mathcal{U}(2^n) \mid C P C^{\dagger} \in \mathbb{P}_n \text{ up to a $\pm$ phase}, \, \forall P \in \mathbb{P}_n \}\,.
    \end{equation}
\end{definition}
It is straightforward to verify that the set of Clifford operators forms a group under matrix multiplication.
\begin{definition}[Stabilizer states] 
The set of $n$-qubit stabilizer states consists of all pure quantum states that can be obtained from the computational basis state vector $\ket{0}^{\otimes n}$ via the action of a Clifford unitary $C \in \mathcal{C}_n$. This set is denoted as $\stab(n)$.
\end{definition}

We now state a lemma that will be used extensively in our analysis.
Although the lemma can be shown as a corollary of the Witt's extension lemma~\cite{lamIntroductionQuadraticForms2005}, here we provide a proof for self-consistency and pedagogical clarity.
\begin{lemma}[Clifford action on multiple Pauli operators]\label{lem:alg_transforms}
    For two lists of algebraically independent Pauli operators $\boldsymbol P = (P_1, \dots, P_m)$ and $\boldsymbol Q = (\pm Q_1, \dots, \pm Q_m) \in (\pm \mathbb{P}_n)^{\times m}$, where $\boldsymbol Q$ represents any choice of signed Pauli operators, there exists a Clifford $C$ that maps $\boldsymbol{P}$ to $\boldsymbol{Q}$ under conjugation (i.e., $C \boldsymbol{P} C^\dagger = \boldsymbol{Q}$) if and only if $\mathcal{A}(\boldsymbol{P}) = \mathcal{A}(\boldsymbol{Q})$.
\end{lemma}
\begin{proof}
The fact that the two sets $\boldsymbol P$ and $\boldsymbol Q$ must have the same anticommutation graph when mapped into each other by a Clifford follows directly from the unitarity of Clifford operators.
The non-trivial aspect of the proof is the existence of a Clifford operator, given that $\mathcal{A}(\boldsymbol{P}) = \mathcal{A}(\boldsymbol{Q})$, which we outline below.
The proof is constructive. 
First of all, note that it cannot be $m > 2n$, because otherwise the Pauli operators in $\boldsymbol{P}$ must necessarily be algebraically dependent. Thus, let us first assume that $m = 2n$. This means that both sets form a generating set of the Pauli group $\mathbb{P}_n$. Let us denote $P_{\alpha} \coloneqq \prod_{i=1}^{m}P_i^{\alpha_i}$, where $\alpha \in \mathbb{F}_2^m$, and analogously $Q^{\alpha} \coloneqq \prod_{i=1}^{m}Q_i^{\alpha_i}$.
We can define the linear map $ \mathcal{C}: \mathcal{B}(\mathcal{H}) \to \mathcal{B}(\mathcal{H}) $
as
\begin{align}
    \mathcal{C}(X) = \frac{1}{d} \sum_{\alpha \in \{0,1\}^m} \tr(P^\dagger_{\alpha} X) Q_\alpha.
\end{align}
Notice that $ \mathcal{C}(P_{\alpha}) = Q_{\alpha} $ as desired, so we are just left to show that $ \mathcal{C} $ is a unitary transformation in order to conclude that it corresponds to a Clifford operator, thus finishing the proof.

To do this, it is sufficient to show that the corresponding Choi state, defined as 
\begin{equation}
\rho_\mathcal{C}\coloneqq  \mathcal{C}\otimes\mathcal{I}\left(\frac{1}{\sqrt{d}}\sum^d_{i,j=1}\ket{i,i}\bra{j,j}\right)\in\mathcal{B}\left(\mathcal{H}^{\otimes 2}\right), 
\end{equation}
is a pure state and that $\tr_1(\rho_\mathcal{C})=\frac{\mathbb{1}}{d}$, where $\tr_1$ denotes the partial trace over the first copy of $\mathcal{H}$ (to ensure trace preservicity)~\cite{watrous_theory_2018}. It is not difficult to be convinced that the Choi state of the map $\mathcal{C}(\cdot)$ is $\rho_\mathcal{C}=\frac{1}{d^2}\sum_{\alpha\in\{0,1\}^m} Q_\alpha \otimes P_{\alpha}^* $. To show that $\rho_\mathcal{C}$ is a pure state, let us compute the purity. The following chain of identities holds
    \begin{align}
        \rho_\mathcal{C}^2&=\frac{1}{d^4}\sum_{\alpha,\beta\in\{0,1\}^m} (Q_\alpha \otimes P_{\alpha}^*)( Q_\beta \otimes P_{\beta}^*)\\
        \nonumber
        &=\frac{1}{d^4}\sum_{\alpha,\beta\in\{0,1\}^m} Q_\alpha Q_\beta \otimes P^*_{\alpha}P^*_{\beta} \\
        \nonumber
        &=\frac{1}{d^4}\sum_{\alpha,\beta\in\{0,1\}^m} Q_{\alpha+\beta} \otimes P^*_{\alpha+\beta}\\
         \nonumber
        &=\frac{2^{m}}{d^{4}}\sum_{\alpha\in\{0,1\}^m}Q_{\alpha}\otimes  P^*_{\alpha}=\frac{2^m}{d^2}\rho_\mathcal{C},
        \nonumber
    \end{align}
where the third step holds because $P_i$ and $Q_i$ have the same anticommutation relations, meaning the signs arise from the commutation relations when reordering cancel. This means that for $m=2n$, the Choi-state of $\mathcal{C}$ is a pure state. As $\tr_1(\rho_\mathcal{C})=\frac{\mathbb 1}{d}$, one can conclude that $\mathcal{C}(\cdot)$ is a unitary and, therefore, a Clifford operation, since it maps Pauli operators to Pauli operators.

In the case when $ m < 2n $, it is sufficient to add $ 2n - m $ algebraically independent Pauli operators to both sets $ \boldsymbol{P} $ and $ \boldsymbol{Q} $, completing them to generating sets $ \bar{\boldsymbol{P}} $ and $ \bar{\boldsymbol{Q}} $, respectively, for the full Pauli group $\mathbb{P}_n$, while ensuring that
\begin{align}
\mathcal{A}(\bar{\boldsymbol{P}}) = \mathcal{A}(\bar{\boldsymbol{Q}}).
\end{align}
This allows us to proceed with the previous proof.
We can perform this completion step due to \cref{lem:can_pauli_graph}, as it suffices to only consider the canonical graph
\begin{align}
    \mathcal{A}(\boldsymbol{P}) = \bigoplus_{i=1}^{r}  
    \begin{pmatrix}  
        0 & 1 \\  
        1 & 0  
    \end{pmatrix}  
    \oplus 0_{m-2r,m-2r},  
\end{align}  
as $\mathcal{A}(\boldsymbol{P})\mapsto A\mathcal{A}(\boldsymbol{P}) A^T$  ($A\in\mathrm{GL}(\mathbb{F}_2^{m\times m}$) only corresponds to a generator change of the algebra generated by $\boldsymbol P$ and $\boldsymbol Q$, meaning the algebra is the same). Adding Pauli operators to $\boldsymbol P$ and $\boldsymbol Q$  as in the proof of \cref{lem:can_pauli_graph} to maintain the canonical forms until $m=2n$ ensures that \begin{align}
    \mathcal{A}(\bar{\boldsymbol{P}}) = \mathcal{A}(\bar{\boldsymbol{Q}})=\bigoplus_{i=1}^{n}  
    \begin{pmatrix}  
        0 & 1 \\  
        1 & 0  
    \end{pmatrix}  
\end{align} holds. Afterwards the transformation ($A$) can be inverted.
\end{proof}

Following \cref{def:commutantofagroup}, we denote the commutant of $k$-tensor powers of the Clifford group as $\com(\mathcal{C}_n^{\otimes k})$, we denote the twirling super-operator over the Clifford group (see \cref{twirlingofoperatorOunitary}) as
\begin{align}\label{def:twirlingoperatoroverthecliffordgroup}
\Phi_{\cl}^{(k)}(\cdot) \coloneqq \frac{1}{|\mathcal{C}_n|} \sum_{C \in \mathcal{C}_n} C^{\otimes k} (\cdot) C^{\dagger \otimes k}.
\end{align}
According to \cref{lem:twirlingbelongstothecommutant}, the twirling channel projects any operator onto the $k$-th order commutant of the group. Thus, $\Phi_{\cl}^{(k)}(O)$ for any operator $O$ can be expressed as a linear combination of the basis elements of the $k$-th order Clifford commutant. Finding a basis for the commutant is,  therefore, crucial for computing such object and, consequently, the average properties of random Clifford circuits—an essential component in numerous applications in quantum information~\cite{low2010pseudorandomnesslearningquantumcomputation,huang_classical_2020, Magesan_2011, Eisert_2020, Elben_2022, schuster2025randomunitariesextremelylow, haferkamp2020QuantumHomeopathyWorks}.
Since the Clifford group $ \mathcal{C}_n $ is a subgroup of the unitary group $ \mathcal{U}(2^n) $, it follows that the commutant of the unitary group is contained in the commutant of the Clifford group, i.e., 
\begin{align}
    \com(\mathcal{U}(2^n)^{\otimes k}) \subseteq \com(\mathcal{C}_n^{\otimes k}).    
\end{align}
It has been shown that the two commutants are equal, i.e., $ \com(\mathcal{U}(2^n)^{\otimes k}) = \com(\mathcal{C}_n^{\otimes k}) $, only if $ k \leq 3 $, making the Clifford group a unitary 3-design, but not a 4-design~\cite{zhu_clifford_2016}.
Additionally, the $ k $-th commutant of the unitary group is spanned by the so-called permutation operators acting on $ k $ copies of the Hilbert space, which we define in the next subsection.

Let us conclude the section, by introducing the function  
\begin{align}\label{def:chidef}
\chi(A, B) \coloneqq \frac{1}{d} \tr(A B A^\dagger B^\dagger),
\end{align}  
where $ A $ and $ B $ are operators. This function satisfies the easy-to-verify symmetry properties
\begin{align}
\label{eq:trivialpropCHI}
\chi(A, B) &= \chi(\phi_A A, \phi_B B), \\  
\chi(A, B) &= \chi(B^\dagger, A) = \chi(A^\dagger, B^\dagger), \\  
\chi(A, B) &= \chi(UAU^{\dagger},UBU^{\dagger}),
\end{align} 
where $ \phi_A, \phi_B \in \{\pm 1, \pm i\} $ and $U$ is any unitary operator.

\begin{lemma}
\label{le:trivialPQ}
Let $ P, Q $ be operators from the set $\mathbb{P}_{n}$. The following properties hold:
\begin{enumerate}
    \item $\chi(P, Q)$ can be written as
    \begin{align}
        \chi(P, Q) = 
        \begin{cases} 
        \ \ 1, & \text{if } [P, Q] = 0, \\ 
        -1, & \text{if } \{P, Q\} = 0,
        \end{cases}
    \end{align}
    and we have that $ PQP = \chi(P, Q) Q $.
    \item It holds that $\sum_{Q\in\mathbb{P}_n}\chi(P, Q)= d^2\delta_{P,\mathbb{1}}.$

    \item If $R$ belongs to $ \mathbb{P}_n $ up to a phase factor, we have $\chi(R, PQ) = \chi(R, P) \chi(R, Q)$.
    \item If $ K = \sqrt{\chi(P, Q)} PQ $, then $ K $ belongs to $ \mathbb{P}_n $ up to a $ \pm 1 $ factor, and
    \begin{align}
        \chi(P, Q) = \chi(P, K)  = \chi(K, P) = \chi(Q, K) = \chi(K, Q) .
    \end{align}
\end{enumerate}
\end{lemma}

\begin{proof}
The first claim follows directly from the definition of $ \chi(P, Q) $ and the commutation relations of Pauli operators.
The second claim follows from the fact that
\begin{align}
    \sum_{Q\in\mathbb{P}_n} \chi(P, Q) = \frac{1}{d}\sum_{Q\in\mathbb{P}_n}  \tr(QPQP) = \tr(P)^2=d^2 \delta_{P,\mathbb{1}},
\end{align}
where in the third step we have used  that $\frac{1}{d^2}\sum_{\mathbb{P}_{n}}QPQ=\frac{\tr(P)}{d} \mathbb{1}$, which follows from the fact that the uniform distribution over the Pauli basis is a $1$-design.

For the third claim, we have
\begin{align}
    \chi(R, PQ) &= \frac{1}{d} \tr(R PQ R^\dagger (PQ)^\dagger) = \frac{1}{d} \tr(RPQRQP) = \frac{1}{d} \chi(R, Q) \tr(RP R^\dagger P^\dagger) = \chi(R, Q)\chi(R, P),
\end{align}
where in the first step we have used  the definition of $ \chi $, and in the third step we have used  that $ QRQ = \chi(R, Q) R $.

For the last claim, define $ K' = \phi PQ $, where $ \phi \in \{ \pm 1, \pm i \} $ is a phase factor such that $ K' \in \mathbb{P}_n $ up to a $ \pm 1 $ factor. 
Since $ K'^2 = \mathbb{1} $, we have
\begin{align}
    \mathbb{1} = K'^2 = (\phi PQ)^2 = \phi^2 P Q P Q = \phi^2 \chi(P, Q) \mathbb{1}.
\end{align}
This implies $ \phi^2 \chi(P, Q) = 1 $, so $ \phi = \pm \sqrt{\chi(P, Q)} $. Thus, $ K' = \sqrt{\chi(P, Q)} PQ $ is in $ \mathbb{P}_n $, up to a $ \pm 1 $ factor.

Finally, $ \chi(P, Q) = \chi(P, K)  $ follows from the first item in this list, along with the fact that $ [P, K] $, up to a $ \pm 1 $ factor, is given by
\begin{align}
 [P, K] \propto [P, PQ] &= P^2 Q - PQP = Q - \chi(P, Q) Q = (1 - \chi(P, Q)) Q,
\end{align}
i.e., $ P $ and $ Q $ commute if and only if $ P $ commutes with $ K $.
We have $ \chi(P, K) = \chi(K, P) $ due to \cref{eq:trivialpropCHI} 
and the fact that $K$ and $P$ are Hermitian. 
Similarly, it follows $\chi(P, Q) =\chi(Q, K)$.
\end{proof}

\subsection{Permutation operators}

We denote $ S_k $ as the group of permutations of $ k $ objects, also known as the \textit{symmetric group}~\cite{fulton_representation_2004}. The symmetric group $ S_k $ is a discrete group of order $ |S_k| = k! $. For $ \pi \in S_k $, we define $ T_{\pi} \in \mathcal{B}(\mathcal{H}^{\otimes k}) $ as the unitary representation of $ \pi $ on $ k $ copies of the Hilbert space of $ n $-qubits, given by the following definition.

\begin{definition}[Permutation operators]
    Given $ \pi \in S_k $, an element of the symmetric group $ S_k $, the permutation matrix $ T_{\pi} $ is the unitary matrix satisfying:
    \begin{align}
    T_{\pi} \ket{\psi_1} \otimes \cdots \otimes \ket{\psi_k} = \ket{\psi_{\pi^{-1}(1)}} \otimes \cdots \otimes \ket{\psi_{\pi^{-1}(k)}},
    \label{def:permutation}
    \end{align}
    for all $ \ket{\psi_1}, \dots, \ket{\psi_k} \in \mathbb{C}^{d} $.
\end{definition}
Permutation operators are linearly independent for $k\le d$, and linearly dependent otherwise.

With this definition in hand, we can state the following result, often referred to as Schur-Weyl duality, which provides a full characterization of the commutant of the unitary group.
The following lemma characterizes a basis for the commutant of the unitary group as permutation operators $T_{\pi}$ with $\pi\in S_k$.
\begin{lemma}[Commutant of the unitary group]\label{lem:commutantunitary} Let $\mathcal{U}(d)$ be the unitary group, $\com(\mathcal{U}(d))$ the $k$-th order commutant of the unitary group, and $S_k$ the symmetric group. Then
\be
\com(\mathcal{U}(d)^{\otimes k})=\operatorname{span}\{T_{\pi}\,:\, \pi\in S_k\}, 
\ee
i.e., permutation operators provide a basis of the commutant of the unitary group. For $k\le d$, then the dimension of the commutant coincides with the number of permutation operators, i.e., $\dim(\com(\mathcal{U}(d)^{\otimes k}))=k!$. For the proof, we refer to Refs.~\cite{collins_moments_2003,collins_integration_2006}.
\end{lemma}

Note that from the definition of permutation operators, it follows that $ T_{\sigma} T_{\pi} = T_{\sigma \pi} $ and  $ T_{\pi^{-1}} = T^\dagger_{\pi} $. 
Equivalently, we can write the permutation matrix as
\begin{align}
T_{\pi} = \sum_{i_1, \dots, i_k \in [d]^k} \ketbra{i_{\pi^{-1}(1)}, \ldots, i_{\pi^{-1}(k)}}{i_1, \dots, i_k}.
\end{align}
Thus, we have the property
\begin{align}
T_{\pi} \left( A_1 \otimes \cdots \otimes A_k \right) T^\dagger_{\pi} = A_{\pi^{-1}(1)} \otimes \cdots \otimes A_{\pi^{-1}(k)},
\end{align}
for $ A_1, \dots, A_k \in \mathcal{B}(\mathcal{H}) $.

Another notable property of permutation operators $ T_{\pi} $ is that they factorize on $ n $-qubits. Specifically, for any $ \pi \in S_k$, one can express $ T_{\pi} $ as a product of permutation operators 
\begin{align}
    T_{\pi} = t_{\pi}^{\otimes n}
\end{align}
acting on individual qubits,
where $ t_{\pi}\coloneqq \sum_{i_1, \dots, i_k \in \{0,1\}^k} \ketbra{i_{\pi^{-1}(1)}, \ldots, i_{\pi^{-1}(k)}}{i_1, \dots, i_k}$ is the operator representation of $\pi$ on the single-qubit Hilbert space $\mathcal{B}(\mathbb{C}^2)$.
Swap operators are defined as the permutation operators $ T_{(ij)} $, where $ \pi = (ij) $ corresponds to swapping the $ i $-th and $ j $-th registers.  
We can also express swap operators in the Pauli basis. For example,  
\begin{align}
\label{eq:swapoperator}
    T_{(1\,2)} = \frac{1}{d} \sum_{P \in \mathbb{P}_n} P^{\otimes 2} \otimes \mathbb{1}_{d^{k-2}}. \,
\end{align}
A generating set of permutations consists of the $ k-1 $ neighboring swap operations $ \{T_{(i\,i+1)}\}_{i \in [k-1]} $. 
 
In fact, the minimal number of swaps $ m$ required to generate a given permutation $ \pi $ follows from its cycle notation: since a cycle of length $ l $ can be generated by $ l-1 $ swaps, we obtain $m = k - \#\text{cycles}\le k-1$.

\vspace{1em}
\subsection{Haar average over the unitary group}
We now review the Haar average over the full unitary group and introduce the key techniques for computing the twirling operator when $ \mathcal{G}_n = \mathcal{U}(2^n) $. For a more comprehensive tutorial on this topic, we refer to other references (e.g., see Refs.~\cite{kliesch_theory_2021,Mele_2024}).
For simplicity, we denote the twirling operator over the full unitary group as $ \Phi_{\haar}^{(k)}(\cdot) $, following the common convention in the literature, where the Haar measure over the unitary group is often referred to simply as the \textit{Haar measure}. Exploiting \cref{lem:commutantunitary}, we can write the twirling operator $\Phi_{\haar}^{(k)}(O)$ in \cref{twirlingofoperatorOunitary} as
\be
\Phi_{\haar}^{(k)}(O)=\sum_{\pi\in S_k}c_{\pi}(O)T_{\pi}\,.
\label{expressionintermsofpermutaitonsTpi}
\ee
To determine the explicit form of the coefficients, we now introduce the so-called Weingarten calculus~\cite{weingarten_asymptotic_1978}.
\begin{lemma}[Weingarten calculus]\label{cor:weingartencalculushaar}
Let $ \Phi_{\haar}^{(k)}(O) $ be the $ k $-fold twirling of an operator $ O \in \mathcal{B}(\mathcal{H}^{\otimes k}) $, and let $ S_k $ be the symmetric group. Then, we have
\begin{equation}
    \Phi_{\haar}^{(k)}(O) = \sum_{\pi,\sigma \in S_k} (\boldsymbol{\Lambda}^{-1})_{\pi,\sigma} \operatorname{tr}(T_{\sigma}O) T_{\pi},
\end{equation}
where the coefficients $ (\boldsymbol{\Lambda}^{-1})_{\pi,\sigma} $ are the Weingarten functions, defined through the (pseudo-)inverse of the symmetric matrix $ \boldsymbol{\Lambda} $ with components 
\begin{equation}
    \boldsymbol{\Lambda}_{\pi,\sigma} = \operatorname{tr}(T_{\pi}T_{\sigma}) = d^{|\pi\sigma|}.
\end{equation}
Here, $ |\pi| $ denotes the \textit{length} of the permutation, defined as the number of cycles in the cycle decomposition of $ \pi $.
\end{lemma}

\begin{proof}
Multiplying both sides of \cref{expressionintermsofpermutaitonsTpi} by a permutation operator $ T_{\pi} $ and taking the trace, we obtain
\begin{equation}
    \operatorname{tr}(\Phi_{\haar}^{(k)}(O)T_{\sigma}) = \sum_{\pi\in S_k} c_{\pi}(O) \operatorname{tr}(T_{\pi}T_{\sigma}).
\end{equation}
From this expression, it is evident that $ \boldsymbol{\Lambda} $ is a symmetric matrix.
Using \cref{twirlingofoperatorOunitary} and the fact that $ T_{\pi} \in \com(\mathcal{U}(2^n)^{\otimes k}) $, we observe that $\operatorname{tr}(\Phi_{\haar}^{(k)}(O)T_{\sigma}) = \operatorname{tr}(O T_{\sigma}).$
Now, defining the $ k! \times k! $ matrix $ \boldsymbol{\Lambda} $ with components $ \Lambda_{\pi,\sigma} = \operatorname{tr}(T_{\pi}T_{\sigma}) $, we can invert the above relation to solve for the coefficients $ c_{\pi}(O) $, to get
\begin{equation}
    c_{\pi}(O) = \sum_{\sigma \in S_k} (\boldsymbol{\Lambda}^{-1})_{\pi,\sigma} \operatorname{tr}(T_{\sigma} O).
    \label{coefficientsweingartencalculus}
\end{equation}
\end{proof}
The matrix $ \boldsymbol{\Lambda} $ is invertible for $ k \leq d $, since it is the Gram matrix of the linearly independent family $\{T_\pi\}_{\pi\in S_k}$ with respect to the Hilbert--Schmidt inner product. When $k> d$, it is still possible to employ the Moore-Penrose pseudoinverse.

Hence, after obtaining the matrix $\boldsymbol{\Lambda}$—which is independent of the specific choice of operator $O$—fully characterizing the twirling channel as given in \cref{twirlingofoperatorOunitary} requires computing the $k!$ coefficients $\tr(O T_{\sigma})$ for every $\sigma \in S_k$.  

Lastly, let us introduce the asymptotic behavior of the Weingarten functions (see \cref{cor:weingartencalculushaar}) in the limit of large $d$ and $k = o(d)$. In many cases--if not most--we are primarily interested in studying the asymptotic behavior of averaged quantities as $d$ grows. Notably, Don Weingarten has been the first to analyze the asymptotic behavior of Weingarten functions in Ref.~\cite{weingarten_asymptotic_1978}, and later, Collins~\cite{collins_integration_2006} provided a general formula for integration over the unitary group $\mathcal{U}(2^n)$.

\begin{lemma}[Asymptotics of Weingarten functions~\cite{weingarten_asymptotic_1978}] \label{asymptoticweingarten} 
Let $\boldsymbol{\Lambda}$ be the $k!\times k!$ matrix with components $\boldsymbol{\Lambda}_{\pi,\sigma}=\tr(T_{\pi}T_{\sigma})$. Then, in the limit of large $d$, the components of the inverse matrix $\boldsymbol{\Lambda}^{-1}$ behave as
\be
(\boldsymbol{\Lambda}^{-1})_{\pi,\pi}&=&d^{-k}+O(d^{-(k+2)}),\\
(\boldsymbol{\Lambda}^{-1})_{\pi,\sigma}&=&O(d^{-(k+|\pi\sigma|)}),\quad\pi\neq\sigma.
\ee    
\end{lemma}
The above asymptotics provide an extremely useful tool for applications (see also Ref.~\cite{liu_entanglement_2018}).
In the following proposition, we present an example of a Haar average over the unitary group in the context of the orbit of a pure quantum state, which serves as a useful illustration for the subsequent discussion on averaging over the Clifford group. Additionally, in \cref{app:haaraverage}, we provide a user-friendly guide on how to perform such averages for relatively low values of $k$.  

An example of particular interest arises when the operator $O \in \mathcal{B}(\mathcal{H}^{\otimes k})$ in \cref{twirlingofoperatorOunitary} is chosen as $O = \ketbra{\psi_0}^{\otimes k}$, where we consider $k$ identical copies of a  state vector $\ket{\psi_0}$.

\begin{proposition}[Haar orbit of pure states]\label{Sec: orbitofhaarstates}
Let $d$ and $k$ be natural numbers. Consider a pure quantum state $\psi_0$. Throughout this manuscript, we denote the uniform (Haar) measure over pure states in the Hilbert space by $\de\psi$. Specifically, we have the  expression
\begin{align}
    \Psi_{\haar}^{(k)} \coloneqq \Phi_{\haar}^{(k)}(\psi_0^{\otimes k}) = \int \de\psi \, \psi^{\otimes k} = \frac{\Pi_{\sym}}{\tr(\Pi_{\sym})}
\end{align}
for the Haar average,
where $\Pi_{\sym} \coloneqq \frac{1}{k!} \sum_{\pi} T_{\pi}$, and
\begin{align}
    \tr(\Pi_{\sym}) = \binom{d+k-1}{d-1} = \frac{1}{k!} d (d+1)(d+2) \cdots (d+k-1).
\end{align}
\begin{proof}
The Haar orbit of a pure state vector $\ket{\psi_0}$ is independent of the chosen specific state, due to the invariance of the Haar measure. To determine the coefficients $c_{\pi}(\psi_0^{\otimes k})$ in \cref{expressionintermsofpermutaitonsTpi}, we can make use of the symmetric property of $\Psi_{\haar}^{(k)}$, which is invariant under left and right multiplication by permutation operators $\Psi_{\haar}^{(k)} = T_{\pi} \Psi_{\haar}^{(k)} = \Psi_{\haar}^{(k)} T_{\pi}.$
This symmetry implies that the coefficients $c_{\pi}(\psi_0^{\otimes k})$ must be equal for all $\pi$, i.e., $c_{\pi} (\psi_0^{\otimes k}) = N^{-1},$ where $N$ is a normalization constant. Importantly, this step does not require the permutations to be linearly independent.
Thus, we can express $\Psi_{\haar}$ as $\Psi_{\haar} = N^{-1} \sum_{\pi} T_{\pi}.$
To determine the normalization constant $N$, we impose the condition $\tr(\Psi_{\haar}) = 1$, which gives us the final result. The dimension of the symmetric subspaces can be verified explicitly (see, e.g. Ref.~\cite{harrow_church_2013}).
\end{proof}
\end{proposition}

\section{Characterization of the Clifford commutant from first principles}\label{sec:comclif}
In this section, we present a rigorous construction of the $k$-th order commutant of the Clifford group from first principles. Our approach is based solely on the fundamental properties of Pauli operators, introduced in \cref{sec:preliminaries}, ensuring a clear and self-contained derivation. 

To this end, we begin by defining an equivalence class associated with a Pauli operator acting on multiple tensor copies of the system. This class is characterized by a \emph{decomposition matrix} $V$ and the \emph{anticommutation graph} $G$ corresponding to the Pauli operators appearing in the decomposition.

\begin{definition}[Equivalence class of Pauli operators]\label{def:equivalencerelationVG} Let $P\in\mathbb{P}_n^{\otimes k}$ be a Pauli operator. Let $P_1,\ldots, P_m\in\mathbb{P}_n$ be algebraically independent Paulis and $V\in\mathbb{F}_{2}^{k\times m}$ with columns $V=(v_1,\ldots, v_m)$ such that $P\propto P_1^{\otimes v_1}\cdots P_m^{\otimes v_m}$. Let $G\in\mathbb{F}_{2}^{m\times m}$ be the anticommutation graph of $P_1,\ldots, P_m$. We define the equivalence class of the pair $(V,G)$ as
\begin{align}
[V,G] \coloneqq \{(V',G') \mid \exists\, A \in \operatorname{GL}(\mathbb{F}_{2}^{m\times m}) \text{ such that } V' = V A,\, G' = A^{-1} G A^{-T} \},
\end{align}
where $A^{-T}$ denotes the inverse transpose of $A$. Furthermore, we can define the map $\mathcal M: \mathbb{P}_n^{\otimes k}\rightarrow \{[V,G]\}_{V,G}$ which maps $ P \mapsto [V,G]$.
\end{definition}
We now verify that the above definition is well posed, i.e., the mapping $ P \mapsto [V,G] $ is consistent.
\begin{proof}  
By \cref{le:uniquenessPV}, any Pauli operator $ P \propto P_1^{\otimes v_1} \cdots P_m^{\otimes v_m} $ is associated with a matrix $ V \in \mathbb{F}_{2}^{k \times m} $ whose columns are $ ( v_1, \dots, v_m ) $, and a matrix $ B_{\boldsymbol{P}} $ whose columns encode the bitstring representation of the Pauli operators $ \mathbf{P} \coloneqq (P_1, \dots, P_m) \in \mathbb{P}_n^{\times k} $. Any other valid decomposition, say $ (V', B_{\boldsymbol{P}'}) $, must be related to $ (V, B_{\boldsymbol{P}}) $ through an invertible transformation $ A \in \operatorname{GL}(\mathbb{F}_{2}^{m \times m}) $, such that $ V' = V A $ and $ B_{\boldsymbol{P}'} = B_{\boldsymbol{P}} A^{-T} $ (\cref{le:uniquenessPV}). Since the adjacency matrix of the anticommutation graph associated with $ \mathbf{P} $ is $ G \coloneqq B_{\boldsymbol{P}}^T (J+J^T) B_{\boldsymbol{P}} $ (by \cref{eq:anticommJ}), it follows that under the transformation $ B_{\boldsymbol{P}'}  = B_{\boldsymbol{P}} A^{-T} $, the corresponding adjacency matrix becomes $ G' = A^{-1} G A^{-T} $. This confirms that $ [V, G] $ is well defined.  
\end{proof}
\subsection{An orthogonal basis for the commutant: the algebraic independent graph-based basis}
We now present the first of our main results, namely an orthogonal basis for the $k$-th order commutant of the Clifford group for any $k$ and $n$, denoting number of qubits. In subsequent sections, we will refer to this basis as \emph{algebraically independent graph-based basis} (see \cref{def:indeppaulimonomials}). 
\begin{theorem}[Orthogonal basis for the Clifford commutant]\label{th:fullcommutantnk}
    Let $\mathcal{M}\colon P \in \mathbb{P}_n^{\otimes k} \mapsto \mathcal{M}(P) = [V,G]$ be the map from the Pauli operators $\mathbb{P}_n^{\otimes k}$ to the equivalence class $[V,G]$, defined in \cref{def:equivalencerelationVG}. Then, the set of operators
\begin{equation}\label{eq:omhI}
    \mho_{I}([V,G]) \coloneqq \frac{1}{|S_{[V,G]}|} \sum_{P\in S_{[V,G]}}\varphi(P) P
\end{equation}
where $S_{[V,G]}=\{Q\in \mathbb{P}_n^{\otimes k}| \mathcal{M}(Q)=[V,G]\}$, $\varphi(P)\coloneqq\tr(PT_{(k\cdots 21)})$ with $V \in \mathrm{Even}(\mathbb{F}_2^{k \times m})$ and $G \in \symftwo$ with $\rank_2(G) \geq 2(m - n)$ (which is equivalent to requiring that $\mho_{I}([V,G]) \neq 0$), forms an orthogonal basis for $\com(\mathcal{C}_n^{\otimes k})$ for any $n, k$.
\end{theorem}

\begin{proof} Our proof strategy is similar to that used in the proof of Lemma 4.8 in Ref.~\cite{gross_schurweyl_2019}, which exploits the fact that a basis of the commutant is given by the image of the twirling operator acting on Pauli operators.

In particular, we start by noticing that the image of the Clifford twirling super-operator $\Phi_{\cl}^{(k)}(\cdot)$ defined in \cref{def:twirlingoperatoroverthecliffordgroup} is precisely the commutant of the Clifford group.
So that we can consider the operator basis of Pauli operators $\mathbb{P}_{n}^{\otimes k}$ and construct the commutant as
\begin{align}
\com(\mathcal{C}_n^{\otimes k})=\operatorname{span}\{\Phi_{\cl}^{(k)}(P)\,:\, P\in\mathbb{P}_{n}^{\otimes k}\}\,.
\end{align}
Using \cref{lem:twirl_finite}, we obtain
\begin{align}
    \Phi_{\cl}^{(k)}(P) &= \frac{1}{|\mathcal{C}_n|} \sum_{C \in \mathcal{C}_n} C^{\otimes k} P C^{\dagger \otimes k} 
    = \frac{1}{|\mathrm{orb}(P)|} \sum_{Q \in \mathrm{orb}(P)} Q\,,
\end{align}
where $\mathrm{orb}(P) \coloneqq \{ Q \in \pm \mathbb{P}_n \mid \exists \, C \in \mathcal{C}_n : Q = C^{\otimes k} P C^{\dagger \otimes k} \}$ is the orbit of $ P $ under the Clifford group which maps Pauli operators to Pauli operators up to phases. 

We observe that \( \Phi_{\cl}^{(k)}(P) \) can be zero (for instance, if \( -P \in \operatorname{orb}(P) \), then \( \Phi_{\cl}^{(k)}(P) = 0 \)). Thus, we seek necessary and sufficient conditions under which the twirl is non-zero. 
To this end, we apply \cref{le:uniquenessPV} and write \( P = \phi P_1^{\otimes v_1} \cdots P_m^{\otimes v_m} \), with associated matrices \( V \), \( B_{\boldsymbol{P}} \), and anticommutation graph \( G \) associated to the tuple \( \boldsymbol{P} = (P_1, \ldots, P_m) \).

We now show that the following conditions are equivalent:
\begin{align}
\Phi_{\cl}^{(k)}(P) \neq 0 
\quad \Longleftrightarrow \quad 
\varphi(P) \neq 0 
\quad \Longleftrightarrow \quad 
V \in \mathrm{Even}(\mathbb{F}_2^{k \times m}),
\label{eq:condition}
\end{align}
where we define
\begin{align}
\label{eq:phase_of_twirl}
\varphi(P) \coloneqq \frac{1}{d} \tr\left( P T_{(k \cdots 21)} \right),
\end{align}
which is invariant under Clifford conjugation, since \( T_{(k \cdots 21)} \in \com(\mathcal{C}_n^{\otimes k}) \) (due to the fact that it lies in the unitary group commutant).

First, we show that \( \varphi(P) \neq 0 \) implies \( \Phi_{\cl}^{(k)}(P) \neq 0 \). Indeed, we have \( \varphi(P) = \varphi\left( \Phi_{\cl}^{(k)}(P) \right) \), because \( \varphi \) is linear and Clifford-invariant. Therefore, if \( \varphi(P) \neq 0 \), then \( \Phi_{\cl}^{(k)}(P) \neq 0 \).

Next, we show that \( \varphi(P) \neq 0 \) is equivalent to \( V \in \mathrm{Even}(\mathbb{F}_2^{k \times m}) \). We can write
\begin{align}
\varphi(P)
= \frac{1}{d} \tr\left( P T_{(k \cdots 21)} \right) 
= \frac{1}{d} \phi \tr\left( \prod_{\alpha=1}^{k} \prod_{i=1}^{m} P_i^{(v_i)_\alpha} \right),
\end{align}
where we used the identity \( \tr\left( O_1 \otimes \cdots \otimes O_k \, T_{(k \cdots 1)} \right) = \tr(O_1 \cdots O_k) \), valid for any operators \( O_1, \ldots, O_k \) (see \cref{def:permutation}). The trace is non-zero if and only if the total product of Pauli operators is proportional to the identity, which holds when \( \sum_{i,j} (v_i)_j \, b(P_i) \equiv \sum_i |v_i| \, b(P_i) \equiv 0 \pmod{2} \). Since the \( P_i \) are algebraically independent (see \cref{le:uniquenessPV}), this condition is equivalent to requiring \( |v_i| \equiv 0 \pmod{2} \) for all \( i \), i.e., \( V \in \mathrm{Even}(\mathbb{F}_2^{k \times m}) \).

Finally, we show that \( \Phi_{\cl}^{(k)}(P) \neq 0 \) implies \( V \in \mathrm{Even}(\mathbb{F}_2^{k \times m}) \) (and therefore \( \varphi(P) \neq 0 \)). Suppose, for contradiction, that \( V \notin \mathrm{Even}(\mathbb{F}_2^{k \times m}) \). Then there exists an index \( i \) such that \( |v_i| \equiv 1 \pmod{2} \). Consider a Clifford operator that maps \( P_i \mapsto -P_i \) while leaving all \( P_j \) with \( j \neq i \) unchanged. Under this conjugation, \( P \mapsto -P \), which implies \( \Phi_{\cl}^{(k)}(P) = -\Phi_{\cl}^{(k)}(P) \), and thus \( \Phi_{\cl}^{(k)}(P) = 0 \), a contradiction.

We have therefore established the equivalence claimed in Eq.~\eqref{eq:condition}.

We now show that the elements of \( \operatorname{orb}(P) \) correspond to the elements of \( S_{[V,G]} = \{ Q \in \mathbb{P}_n^{\otimes k} \mid \mathcal{M}(Q) = [V,G] \} \), each multiplied by a specific sign \( \pm 1 \), denoted \( \theta(P,Q) \), depending on \( P \).

Given a Pauli operator written as \( P = \phi P_1^{\otimes v_1} \cdots P_m^{\otimes v_m} \), the action of a Clifford unitary \( C \in \mathcal{C}_n \) on \( P \) takes the form
\begin{align}
C^{\otimes k} P C^{\dagger \otimes k}
= \phi \, C^{\otimes k} \left( P_1^{\otimes v_1} \cdots P_m^{\otimes v_m} \right) C^{\dagger \otimes k}
= \phi \prod_{i=1}^m \left( C^{\otimes k} P_i^{\otimes v_i} C^{\dagger \otimes k} \right).
\end{align}
and, therefore, both $V$ and $G$ are preserved by the action of $C$. Thus, \( \mathcal{M}(C^{\otimes k} P C^{\dagger \otimes k}) = \mathcal{M}(P) \), and so \( C^{\otimes k} P C^{\dagger \otimes k} \in S_{[V,G]} \), up to a global sign.
Conversely, for any \( Q \in S_{[V,G]} \), there exists a Clifford unitary \( C \in \mathcal{C}_n \) such that \( C^{\otimes k} P C^{\dagger \otimes k} = \theta(P,Q) Q \), where \( \theta(P,Q) \in \{ \pm 1 \} \). Hence, the Clifford orbit of \( P \) consists of all the elements of \( S_{[V,G]} \), each weighted by a sign \( \theta(P,Q) \).

We can therefore express the Clifford twirl of \( P \) as
\begin{align}
\Phi_{\cl}^{(k)}(P)
= \frac{1}{|\operatorname{orb}(P)|} \sum_{R \in \operatorname{orb}(P)} R
= \frac{1}{|S_{[V,G]}|} \sum_{Q \in S_{[V,G]}} \theta(P,Q) Q,
\label{eq:twirlingcalculationproofmain}
\end{align}
where \( [V,G] = \mathcal{M}(P) \).
 
We can now observe that $\theta(P,Q)=\varphi(Q)/\varphi(P)=\varphi^*(P)\varphi(Q)$. To see this, consider a Clifford $C$ and $C^{\otimes k}PC^{\dag\otimes k}=\theta Q$ with $\theta=\pm 1$. Since $\phi(P)$ is Clifford invariant, we have $\phi(P)=\phi(C^{\otimes k}PC^{\dag\otimes k})=\theta\phi(Q)$, which implies $\theta=\phi(P)/\phi(Q)=\varphi(Q)/\varphi(P)=\varphi^{*}(P)\varphi(Q)$ proving \cref{eq:twirlingcalculationproofmain}.

We can, therefore, write 
\begin{align}
    \Phi_{\cl}^{(k)}(P)=\varphi^*(P)\mho_I([V,G])\quad\text{where}\quad [V,G]=\mathcal{M}(P).
\end{align}
Thus, we have established that the set of operators $ \mho_I([V,G])$ generates the commutant. Next, we demonstrate that these operators are orthogonal with respect to the Hilbert-Schmidt inner product,  
\begin{align}
    \tr( \mho_I([V,G])^\dagger \mho_I([V',G']) ) = 0 \quad \text{for} \quad [V,G] \neq [V',G'].
\end{align}
In fact, we have
\begin{align}\label{eq:two_norm_of_mho}
\tr(\mho_I([V,G])\mho_I^\dagger([V',G'])) =\frac{1}{|S_{[V,G]}||S_{[V',G']}|} \sum_{\substack{P\in S_{[V,G]}\\P'\in S_{[V',G']}}} \varphi(P) \varphi^{-1}(P') \tr(PP'^\dag) =\frac{d^{k}\delta_{[V,G],[V',G']}}{|S_{[V,G]}|},
\end{align}  
since the trace \( \tr(P P'^\dagger) \) vanishes unless \( P = P' \), which requires \( P \) and \( P' \) to belong to the same orbit. That is, the contribution is nonzero only when \( [V, G] = [V', G'] \), because orbits are disjoint.

This result demonstrates that the operators $\mho_I([V,G])$ are linearly independent, provided that $\mho_I([V,G]) \neq 0$. Since $P \in \mathbb{P}_n^{\otimes k}$ serves as an operator basis, it follows that the set of non-zero, independent graph-based monomials forms a basis of the commutant for any $n$ and $k$.  

The remainder of the proof is now devoted to determining the conditions on $V, G$ under which $\mho_I([V,G])$ is non-zero, which is the case if and only if the set $S_{[V,G]}$ is not empty.
First we note, that this property is independent of $V$, as long as $V\in\mathrm{Even}(\mathbb{F}_{2}^{k\times m})$.
As such, it reduces to asking whether a matrix $G \in \symf$ can represent a valid anticommutation graph of Pauli operators on $n$ qubits. By \cref{lem:can_pauli_graph}, we can take a 
representant of $[V,G]$, where 
\begin{align}
G\coloneq\bigoplus_{i=1}^r\begin{pmatrix}
    0&1\\1&0
\end{pmatrix}\oplus 0_{m-2r,m-2r}
\end{align}  
and where $r\coloneqq \mathrm{rank}_2(G)/2$. Naturally, we can construct a set of Pauli operators that satisfy the anticommutation graph by choosing the set $\{X_i,Z_i\}_{i\leq r} \cup \{Z_{i}\}_{r+1 \leq i\leq m}$. This requires $n\geq m-r$ many qubits.
If instead $n\leq m-r$, we have that $S_{[V,G]}$ must be the empty set, as it must contain $m-r$ algebraically independent and commuting Pauli operators, which requires at least $m-r$ many qubits. 
As such, $r\geq m-n$ needs to be satisfied for $\mho_I([V,G])$ to not vanish. This concludes the proof.
\end{proof}

Another way to express \( \mho_I([V,G]) \), up to a global phase factor, is as follows. Given \( V \in \even \) and \( G \in \symf \), one can define the operator
\begin{align}
\mho_I(V, G) \coloneqq \frac{1}{|S_{[V,G]}|} \sum_{\substack{\mathbf{P} \in \mathbb{P}_n^{\times m} \\ \mathcal{A}(\mathbf{P}) = G \\ \mathbf{P} \text{ alg. ind.}}} \prod_{j=1}^m P_j^{\otimes v_j},
\end{align}
where the sum runs over algebraically independent tuples of Pauli operators \( \mathbf{P} = (P_1, \dots, P_m) \) such that \( \mathcal{A}(\mathbf{P}) = G \). By construction, \( \mho_I(V, G) \) is proportional to \( \mho([V,G]) \), as defined in the previous theorem, with a proportionality factor given by a global phase in \( \{ \pm 1, \pm i \} \), which depends on the choice of representative \( V \) for the equivalence class \( [V,G] \).

\subsection{Dimension of the Clifford commutant}
\cref{th:fullcommutantnk} provides a basis for the commutant for any $k, n$, by leveraging the fact that the image of the Clifford twirling on the Pauli group $\mathbb{P}_n^{\otimes k}$ yields an orthogonal basis of the commutant. The crucial element of the proof is the identification of nonzero orbits through the equivalence classes $[V,G]$ into which each Pauli operator $P \in \mathbb{P}_n^{\otimes k}$ falls. Specifically, the vector space $V$ and the anticommutation graph $G$ associated with $P$ correspond to the only invariants of the Clifford orbit of $P$. By establishing a necessary and sufficient condition for $\Phi_{\cl}^{(k)}(P) \neq 0$ through the equivalence class $[V,G]$, we can directly determine the dimension of the commutant, as stated in the following theorem.

\begin{theorem}[Dimension of the Clifford group commutant]\label{cor:dimcom} The dimension of the Clifford group commutant for any $n$ and $k$ is given by the  expression
\begin{align}\label{eq:dimcommutant}
\dim(\com(\mathcal{C}_n^{\otimes k}))&=\sum_{m=0}^{k-1}\sum_{r=\max(m-n,0)}^{\lfloor m/2\rfloor} 2^{(k+2r-m-1)m-r-2r^2}\prod_{j=1}^{m}\frac{1-2^{-k+j}}{1-2^{-j}}\prod_{i=1}^r \frac{(1-2^{-m+i-1})(1-2^{-m+r+i-1})}{1-2^{2i}}\\
&\simeq\begin{cases}
        2^{\frac{k^2-3k}{2}}& 2n\geq k-1,\\
        2^{2kn-2n^2-3n}& 2n< k-1,
        \end{cases}
        \nonumber
\end{align} 
where `$\simeq$' means `up to factors' bounded with constants from above and below in the subsequent proof. In particular, if $n\geq k-1$, the dimension is independent of $n$, and is equal to $\dim(\com(\mathcal{C}_n^{\otimes k}))=\prod_{i=0}^{k-2}(2^i+1)$.
\begin{proof} 
Using \cref{th:fullcommutantnk}, to determine the dimension of the commutant it is sufficient to count the number of equivalence classes $[V,G]$ corresponding to an operator $\mho_I([V,G])\neq 0$. First, the number of binary vector subspaces of $\mathbb{F}_{2}^{k}$ spanned by $m$ vectors with even Hamming weight is given by the Gaussian Binomial coefficient $\binom{k-1}{m}_2$, see \cref{lem:gaussiancoeff}, which we can rewrite as
    \begin{align}
        \binom{k-1}{m}_2
        =2^{km-m^2-m}\prod_{j=1}^{m}\frac{1-2^{-k+j}}{1-2^{-j}}\,.
    \end{align}
    for future convenience. This determines the number of matrices $V$, up to the linear transformation $A\in\operatorname{GL}(\mathbb{F}_{2}^{m \times m})$, see \cref{def:equivalencerelationVG}.
    Thus, we have:
\begin{align}
\dim(\com(\mathcal{C}_n^{\otimes k})) 
= \sum_{m=0}^{k-1} \binom{k-1}{m}_2 \times (\# \text{ of allowed } m \times m \text{ graphs } G).
\end{align}
    We are now left to count the number of allowed anticommutation graphs $G$ on $n$ qubits. The requirement for this is computed in \cref{th:fullcommutantnk}, namely $n\geq m-r$, where $2r=\mathrm{rank}_2(G)$.
The number $N_0(m,r)$  of $m$-dimensional symmetric matrices with zero diagonal with a given even binary rank $2r$ is given by \cref{lem:cardinalitysymf}:
    \begin{align}
        N_0(m,r)
        &=2^{2mr-r-2r^2} \prod_{i=1}^r \frac{(1-2^{-m+i-1})(1-2^{-m+r+i-1})}{1-2^{-2i}}
    \end{align}
    which exactly counts the number of graphs.
Therefore, to determine the dimension of the $k$-th order commutant of the Clifford group, we use that this corresponds to the number of all even Hamming weight vector spaces of dimension $m$ times the number different anticommutation graphs of $m$ Pauli operators allowed on $n$ qubits. Hence:
    \begin{align}\label{eq1dimensioncommutant}
        \dim(\com(\mathcal{C}_n^{\otimes k}))&=\sum_{m=0}^{k-1}\binom{k-1}{m}_2\sum_{r=\max(m-n,0)}^{\lfloor m/2\rfloor}
        N_0(m,r)\\
        &=\sum_{m=0}^{k-1}\sum_{r=\max(m-n,0)}^{\lfloor m/2\rfloor} 2^{(k+2r-m-1)m-r-2r^2}\prod_{j=1}^{m}\frac{1-2^{-k+j}}{1-2^{-j}}\prod_{i=1}^r \frac{(1-2^{-m+i-1})(1-2^{-m+r+i-1})}{1-2^{-2i}}\nonumber\\
        &\eqqcolon \sum_{m=0}^{k-1}\sum_{r=\max(m-n,0)}^{\lfloor m/2\rfloor}(\#\mho_I)(k,m,r)\nonumber
    \end{align}
    where
    \begin{align}
        (\#\mho_I)(k,m,r)=c_{k,m,r}\times2^{(k+2r-m-1)m-r-2r^2},
    \end{align}
     and $c_{k,m,r}$ is the factor arising from the product terms in \cref{eq1dimensioncommutant}. In \cref{eq1dimensioncommutant}, we remark that whenever $m-n\le\lfloor\frac{m}{2}\rfloor$, then $(\#\mho_I)(k,m,r)\neq0$. This condition is equivalent to requiring that $m\le 2n$.  Hence, overall $m\le \max\{k-1,2n\}$. We now determine simplified lower and upper bounds.
    By setting $r=\lfloor m/2\rfloor-x$
    and defining $m_{0}=\min(k-1,2n)$ to choose $m=m_0-y$, we arrive at the following expression:
    \begin{align}
        (\#\mho_I)(k,m\! =\! m_{0}-y,r\! =\!\lfloor m/2\rfloor-x)
         &= 2^{(k-3/2)m_0-m_0^2/2}\!\times\! 2^{(m_0-k+3/2)y-y^2/2}\!\times\! 2^{(-1)^{m_0-y}x-2x^2}\!\times\! c_{k,m_0-y,\lfloor (m_0-y)/2\rfloor-x}.
    \end{align}
    As can be seen, the number of elements drops exponentially in $x$ and $y$. This allows us to bound the overall size of the commutant by computing the prefactor explicitly
    \begin{align}
        &\sum_{y=0}^{m_0} 2^{(m_0-k+3/2)y-y^2/2}\sum_{x=0}^{\lfloor (m_0-y)/2\rfloor}2^{(-1)^{m_0-y}x-2x^2}c_{k,m_0-y,\lfloor (m_0-y)/2\rfloor-x}\nonumber\\
        &=\sum_{y=0}^{m_0} 2^{(m_0-k+3/2)y-y^2/2}\sum_{x=0}^{\lfloor (m_0-y)/2\rfloor}2^{(-1)^{m_0-y}x-2x^2}\nonumber\\ 
        &\times \prod_{j=1}^{m_0-y}\frac{1-2^{-k+j}}{1-2^{-j}}\prod_{i=1}^{\lfloor (m_0-y)/2\rfloor-x} \frac{(1-2^{-m_0+y+i-1})(1-2^{-m_0+y+\lfloor (m_0-y)/2\rfloor+i-1-x})}{1-2^{-2i}}\\
        &\leqt{(i)} \sum_{y=0}^{\infty} 2^{y/2-y^2/2} \times  \sum_{x=0}^{\infty}2^{x-2x^2}\times \prod_{j=1}^{\infty}\frac{1}{1-2^{-j}}\times\prod_{i=1}^{\infty} \frac{1}{1-2^{-2i}}\nonumber\\
        &\leqt{(ii)} 20.3.
        \nonumber
    \end{align}
    In (i) we have used that $m_0\le k-1$, hence $m_0-k+3/2\le 1$ and then we upper bounded $(1-2^{-m_0+y+i-1})(1-2^{-m_0+y+\lfloor (m_0-y)/2\rfloor+i-1-x})\le 1$; further we noticed that the arguments of the products are increasing functions greater then $1$. In (ii) we have used a numerical solver.
    Let us now turn to discussing the lower bound. We can use only the term when $y=x=0$
    \begin{align}
        \prod_{j=1}^{m_0}\frac{1-2^{-k+j}}{1-2^{-j}}\prod_{i=1}^{\lfloor m_0/2\rfloor} \frac{(1-2^{-m_0+i-1})(1-2^{-m_0+\lfloor m_0/2\rfloor+i-1})}{1-2^{2i}} 
        \geqt{(i)} \prod_{j=1}^{m_0}\frac{1-2^{-m_0+j-1}}{1-2^{-j}}\prod_{i=1}^{\infty} (1-2^{i})\geqt{(ii)} 0.28\,.
    \end{align}
    In (i) we have used the fact that $0\le m_0\le k-1$ and in (ii) we have employed  a numerical solver. This gives the following bounds
    \begin{align}
        0.28 \times  2^{(k-3/2)m_0-m_0^2/2 }\leq \dim(\com(\mathcal{C}_n^{\otimes k}))\leq 20.3\times  2^{(k-3/2)m_0-m_0^2/2}\,.
    \end{align}
We now discriminate two cases, namely $2n< k-1$ and $2n\ge k-1$, to determine $m_0=\min\{2n,k-1\}$, leading to:
    \begin{align}
        0.56\times 
        2^{k^2/2-3k/2}&\leq \dim(\com(\mathcal{C}_n^{\otimes k}))\leq 40.6\times  2^{k^2/2-3k/2} & 2n\geq k-1\nonumber\\
        0.28\times  2^{2kn-2n^2-3n }&\leq \dim(\com(\mathcal{C}_n^{\otimes k}))\leq 20.3\times  2^{2kn-2n^2-3n } &2n<k-1\,.
    \end{align}
    In the case where $n \geq k-1$, all graphs are allowed, and thus there is no restriction on their rank. Consequently, the number of distinct basis elements $\mho_I$ is solely a function of $k$ and $m$. This results in a much simpler expression for the dimension of the commutant, given by
\begin{equation}\label{simpleexpressioncommutatndimension}
\dim(\mathcal{C}_n^{\otimes k}) = \sum_{m=0}^{k-1} \binom{k-1}{m}_2 2^{\frac{m(m-1)}{2}} = \prod_{i=0}^{k-2} (2^i + 1), \quad \text{if } n \geq k-1
    \end{equation}
using the Gaussian binomial theorem. We also observe that \cref{simpleexpressioncommutatndimension} coincides with the dimension of the commutant derived in 
Ref.\  \cite{gross_schurweyl_2019}.

\end{proof}
\end{theorem}
\subsection{Properties of the orthogonal basis elements}
The following lemma establishes some key properties of the operators $\mho_I$ that provide a orthogonal basis for the Clifford group commutant for any $n,k$.
\begin{lemma}[Properties of the operators $\mho_{I}$]\label{lem:orb_ofPauli} 
Let $\mho_{I}([V,G])$ be the set of operators defined in \cref{th:fullcommutantnk}, with $V\in\even$ and $G\in\symf$ with $\rank_2(G)\ge2(m-n)$. They obey the following properties:
\begin{enumerate}[label=(\roman*)]    
    \item $\tr(\mho_I([V,G]))=0$ for all $V\neq 0$;
    \item $\|\mho_I([V,G])\|_{\infty}\le 1$;
    \item $i^{\mathrm{rank}_2(G)}\mho_I([V,G])$ is Hermitian;
    \item For $S_{[V,G]}=\{Q\in \mathbb{P}_n^{\otimes k}| \mathcal{M}(Q)=[V,G]\}$, it holds that 
        \begin{align}
    |S_{[V,G]}|
        =2^{2mn-m^2/2+m/2}\prod_{j=r-(m-n)+1}^{n} (1-4^{-j}).
        \end{align}
    \item $\|\mho_{I}([V,G])\|_2=\sqrt{\frac{2^{nk}}{|S_{[V,G]}|}}=2^{(k/2-m)n-m/4+m^2/4}\times \prod_{i=0}^{m-r-1} \sqrt{1-2^{-2(n-i)}}$.
\end{enumerate}

\end{lemma}
\begin{proof}
Condition (i) and (ii) can be obtained by noticing that $\mho_{I}([V,G])=\varphi^{*}(P)\Phi_{\cl}^{(k)}(P)$ with $\mathcal{M}(P)=[V,G]$. For condition (iii), we note that $\varphi(P)\mho_I([V,G])$ is Hermitian, for $P\in S_{[V,G]}$.

Condition (v) follows from (iv) and \cref{eq:two_norm_of_mho}. Indeed, computing the $2$-norm reduces to computing the size of the set $S_{[V,G]}=\{Q\in\mathbb{P}_{n}^{\otimes k}\,|\,\mathcal{M}(Q)=[V,G]\}$.  Note that by \cref{lem:alg_transforms}, the size of the set is equal to
\begin{align}
    |S_{[V,G]}|=|\{\boldsymbol P \in \mathbb{P}_n^{\times k} \, : \text{ $\boldsymbol P $ is an algebraic independent Pauli set and } \mathcal{A}(\boldsymbol P)=G\}|
\end{align}
we can use the invariance of the equivalence class under $A\in\mathrm{GL}(\mathbb{F}_{2}^{m\times m})$ and \cref{lem:can_pauli_graph} to rewrite $G\mapsto G'=A^{-1}GA^{-T}$ such that 
\begin{equation}
G'=\bigoplus_{i=1}^r \begin{pmatrix}
    0&1\\1&0
\end{pmatrix}\oplus 0_{m-2r,m-2r}
\end{equation}
with $\mathrm{rank}(G)=2r$. One particular set of Pauli operators that satisfies the anticommutation graph $G$ is consisting of $(Z_i,X_i)_{i\in [r]}$ pairs, followed by $(Z_{r+i})_{i\in m-2r}$ stabilizers. This consideration helps in computing $|\mathrm{orb}(P)|$. To count them, we first start by selecting the first anticommuting pair of $Z_1$-like and $X_1$-like operators. There are $(4^n-1)\times 4^n/2$ many options. As any further selected Pauli operators need to commute with the previous ones to ensure algebraic independency. This restricts our choices effectively reducing the number of qubits by $1$ each time. As such, we have that there are
\begin{align}\label{eq:norm2proof1}
    \prod_{i=0}^{r-1} (2^{2(n-i)}-1)\times 2^{2(n-i)-1}
\end{align}
many options to select algebraically independent pairs of $X$-like, $Z$-like operators. We now count the $Z$ components. 
Here we have effectively $m'=m-2r$ many $Z$ type Pauli operators on $n'=n-r$ many qubits. This reduces the counting on selecting algebraically independent Pauli operators all commuting with the previous ones, giving 
\begin{align}\label{eq:norm2proof2}
    \prod_{i=0}^{m'-1} (2^{2n'-i}-2^{i}).
\end{align}
Multiplying \cref{eq:norm2proof1,eq:norm2proof2} gives the desired result.
\end{proof}

\section{A natural and easy-to-manipulate basis for the Clifford commutant: Pauli monomials}\label{sec:Paulimon}
While in \cref{sec:comclif} we derived the commutant of the Clifford group for any $n,k$, here we aim to develop a framework that allows us to handle the commutant in a more direct and simplified manner. In particular, this section introduces a fundamental concept of our work: Pauli monomials. These are constructed as products of isotropic sums of Pauli operators. In the following, we familiarize the reader with various types of Pauli monomials, which will serve as a key tool for constructing a simple basis for the commutant of the Clifford group. 

The underlying intuition is that, since the Clifford group leaves invariant the set of Pauli operators, all operators that commute with every Clifford element must be constructed as isotropic sums over Pauli operators. 

To begin with, we introduce the simplest and most intuitive type of Pauli monomials, referred to as \textit{primitive Pauli monomials}. These monomials encapsulate many essential properties of the full commutant and, as we demonstrate in \cref{th:algebraicstructurecommutantofclifford}, generate the entire commutant.

\begin{definition}[Primitive Pauli monomials]\label{def:primitivepaulimonomials}  
Let $v \in \mathbb{F}_{2}^{k}$ be a vector satisfying $|v| \in 2\mathbb{N}$. Define the operator  
\be
\Omega(v) \coloneqq \frac{1}{d} \sum_{P \in \mathbb{P}_n} P^{\otimes v}.
\ee  
\end{definition}
The primitive Pauli monomials satisfy the following properties.
\begin{lemma}[Properties of primitive Pauli monomials]\label{lem:propertiesprimitivepauli}
The primitive Pauli monomials satisfy the following properties:
\begin{enumerate}[label=\Alph*)]
    \item If $|v|/2 = 0 \pmod{2}$ ($v$ is even), then $d^{-1} \Omega(v)$ is a projector of rank $d^{k-2}$.
    \item If $|v|/2 = 1 \pmod{2}$  ($v$ is odd), then $\Omega(v)$ is a Hermitian unitary operator.
    \item If $v, w \in \mathbb{F}_2^k$ satisfy $v \cdot w = 0 \pmod{2}$, then $[ \Omega(v), \Omega(w)] = 0$.
\end{enumerate}
\end{lemma}

\begin{proof}
That $\Omega(v)$ is Hermitian follows directly from the Pauli operators being Hermitian. We can compute
\begin{align}
    \Omega(v)^2 &= \frac{1}{d^2} \sum_{P, Q \in \mathbb{P}_n} (PQ)^{\otimes v} = \frac{1}{d^2} \sum_{P, R \in \mathbb{P}_n} \big(\frac{1}{\sqrt{\chi(P, R)}} R\big)^{\otimes v},
\end{align}
where we have defined $R \coloneqq \sqrt{\chi(P, Q)} PQ$ and we have used that $\chi(P, Q) = \chi(P, R)$ (because of \cref{le:trivialPQ}).
Thus, we have
\begin{align}
    \Omega(v)^2 &= \frac{1}{d^2} \sum_{P, R \in \mathbb{P}_n} \chi(P, R)^{|v|/2} R^{\otimes v} \\&=\begin{cases}
     \frac{1}{d^2} \sum_{P, R \in \mathbb{P}_n} R^{\otimes v} = \sum_{R \in \mathbb{P}_n} R^{\otimes v} = d \Omega(v),&|v|\in 4\mathbb N,\\
         \frac{1}{d^2} \sum_{R \in \mathbb{P}_n} \bigg(\sum_{P \in \mathbb{P}_n} \chi(P, R)\bigg) R^{\otimes v} = \sum_{R \in \mathbb{P}_n} \delta_{R, \mathbb{1}} R^{\otimes v}= \mathbb{1} &|v|\in 4\mathbb N+2,
         \nonumber
    \end{cases}
\end{align}
where we have used that $\sum_{P \in \mathbb{P}_n} \chi(P, R)= d^2 \delta_{R, \mathbb{1}}$.
For this, property A) and property B) follow directly.
Property C) follows straightforwardly from
\begin{align}
    P^{\otimes v_1} Q^{\otimes v_2} 
    &= (-1)^{v_1 \cdot v_2} Q^{\otimes v_2} P^{\otimes v_1},
\end{align}
where the commutator vanishes for $v_1 \cdot v_2 = 0 \pmod{2}$.

\end{proof}

Properties \textit{A)} and \textit{B)} were already established in Remark 3.9 of~\cite{gross_schurweyl_2019}.

From now on, we classify Pauli monomials as \textit{even} (projectors) or \textit{odd} (unitaries) depending on whether $ \frac{|v|}{2} $ is even or odd, respectively. Importantly, this classification is based on the parity of $ \frac{|v|}{2} $, rather than on whether $ |v| $ itself is even or odd.
We remark that a specific type of Pauli monomial is given by the swap operator, which corresponds to an even vector $ v $; see \cref{eq:swapoperator}.

We are now ready to introduce the \textit{Pauli monomials}, which generalize the primitive Pauli monomials introduced in \cref{def:primitivepaulimonomials} and arise from their products, as rigorously shown below.
\begin{definition}[Pauli monomials]\label{def:paulimonomials}
Let $ k,m \in \mathbb{N}$, $V\in\mathrm{Even}(\mathbb{F}_{2}^{k\times m})$ and $M\in\symftwo$. A Pauli monomial, denoted $ \Omega(V, M) \in \mathcal{B}(\mathcal{H}^{\otimes k}) $, is defined as
\begin{align}
\Omega(V, M) \coloneqq \frac{1}{d^m} \sum_{\boldsymbol{P} \in \mathbb{P}_n^m}  
P_1^{\otimes v_1} P_2^{\otimes v_2} \cdots P_m^{\otimes v_m}\times \left( \prod_{\substack{i, j \in [m] \\ i < j}} \chi(P_i, P_j)^{M_{i,j}} \right),
\end{align}
where  $ \chi(P_i, P_j)$ is defined in \cref{le:trivialPQ}. 
\end{definition}
Thus, the primitive Pauli monomial $ \Omega(v) $ for $ v \in \mathbb{F}^k_2 $, as defined in \cref{def:primitivepaulimonomials}, can be written as $ \Omega(v) = \Omega((v), 0) $, where $ (v) $ is the single-column matrix containing $ v $ as its column. Consequently, this structure hints at the role of the primitive Pauli monomials as fundamental building blocks of the set of Pauli monomials. Given a Pauli monomial, we refer to the linear combination of Pauli monomials as a \textit{Pauli polynomial}.

An intriguing characteristic of Pauli monomials is that, while they are expressed as a sum of tensor products of Pauli operators, they can still be decomposed into a tensor product over individual qubits, as stated by the following lemma.
 
    \begin{lemma}[Tensor product structure of Pauli monomials]\label{lem:monomialsfactorizeonqubits}
Pauli monomials factorize as tensor products over qubits. Specifically, for any $ \Omega(V, M)  $, we have 
\begin{align}
    \Omega(V, M) &= \big( \omega(V, M) \big)^{\otimes n}, \\
    \omega(V, M) &\coloneqq \frac{1}{2^m} 
    \sum_{\boldsymbol{P} \in \mathbb{P}_1^m} 
    \prod_{i=1}^m P_i^{\otimes v_i} 
    \prod_{1 \leq i < j \leq m} \chi(P_i, P_j)^{M_{i,j}},
\end{align}
\begin{proof}
We have
\begin{align}
    \big( \omega(V, M) \big)^{\otimes n} 
    &= \frac{1}{(2^{n})^m} \bigotimes_{k=1}^n 
    \left( \sum_{\boldsymbol{P}^{(k)} \in \mathbb{P}_1^m} 
    \prod_{b=1}^m P_b^{(l) \otimes v_b} 
    \prod_{1 \leq i < j \leq m} \chi(P_i^{(l)}, P_j^{(l)})^{M_{i,j}} \right)  \\
    &= \frac{1}{(2^{n})^m} 
    \sum_{\boldsymbol{P}^{(1)}, \dots, \boldsymbol{P}^{(n)} \in \mathbb{P}_1^m} 
    \bigotimes_{l=1}^n 
    \left( \prod_{b=1}^m P_b^{(l) \otimes v_b} 
    \prod_{1 \leq i < j \leq m} \chi(P_i^{(l)}, P_j^{(l)})^{M_{i,j}} \right)\\
    &= \frac{1}{(2^{n})^m} 
    \sum_{\boldsymbol{P}^{(1)}, \dots, \boldsymbol{P}^{(n)} \in \mathbb{P}_1^m} 
    \left( \prod_{l=1}^n \prod_{1 \leq i < j \leq m} \chi(P_i^{(l)}, P_j^{(k)})^{M_{i,j}} \right) 
    \prod_{b=1}^m \bigotimes_{l=1}^n P_b^{(l) \otimes v_b}.
    \nonumber
\end{align}
Using the property $\chi(P_1, Q_1) \chi(P_2, Q_2) = \chi(P_1 \otimes P_2, Q_1 \otimes Q_2)$ valid for Pauli operators $P_1, P_2, Q_1, Q_2$, we get
\begin{align}
    \big( \omega(V, M) \big)^{\otimes n} 
    &= \frac{1}{(2^{n})^m} 
    \sum_{\boldsymbol{P}^{(1)}, \dots, \boldsymbol{P}^{(n)} \in \mathbb{P}_1^m} 
    \prod_{1 \leq i < j \leq m} 
    \chi\left( \bigotimes_{l=1}^n P_i^{(l)}, \bigotimes_{k=1}^n P_j^{(l)} \right)^{M_{i,j}} 
    \prod_{b=1}^m \bigotimes_{l=1}^n P_b^{(l) \otimes v_b}\\
    &= \frac{1}{(2^{n})^m} 
    \sum_{\boldsymbol{P} \in \mathbb{P}_n^{\times m}} 
    \prod_{1 \leq i < j \leq m} \chi(P_i, P_j)^{M_{i,j}} 
    \prod_{b=1}^m P_b^{\otimes v_b},
    \nonumber
\end{align}
where in the last step we collected the sums into a single sum over $ \mathbb{P}_n^{\times m} = \big(\{I, X, Y, Z\}^{\otimes n}\big)^{\times m} $.
This matches the definition of $ \Omega(V, M) $, concluding the proof.
\end{proof}
\end{lemma}

We now show that Pauli monomials belong to the $k$-th order commutant of the Clifford group. 
\begin{lemma}[Pauli monomials belong to the Clifford commutant]\label{lem:paulimonomialbelongtothecommutant}
    Every Pauli monomial $ \Omega(V, M)  $ belongs to the $ k $-th order commutant of the Clifford group.
    \begin{proof}
        For all $ C \in \mathcal{C}_n $, we have
\begin{align}
            C^{\dagger \otimes k} \Omega(V,M) C^{\otimes k} 
            &= \frac{1}{d^m} \sum_{\boldsymbol{P} \in \mathbb{P}_n} (C^{\dagger} P_1 C)^{\otimes v_1} \cdots (C^{\dagger} P_m C)^{\otimes v_m} \times \left( \prod_{\substack{i, j \in [m] \\ i < j}} \chi(P_i, P_j)^{M_{i,j}} \right) \\
            \nonumber
            &\eqt{\text{(i)}} \frac{1}{d^m} \sum_{\boldsymbol{Q} \in \mathbb{P}_n} Q_1^{\otimes v_1} \cdots Q_m^{\otimes v_m} \times \left( \prod_{\substack{i, j \in [m] \\ i < j}} \chi(C^{\dagger} Q_i C, C^{\dagger} Q_j C)^{M_{i,j}} \right) \\
              \nonumber
            &\eqt{\text{(ii)}} \frac{1}{d^m} \sum_{\boldsymbol{Q} \in \mathbb{P}_n} Q_1^{\otimes v_1} \cdots Q_m^{\otimes v_m} \times \left( \prod_{\substack{i, j \in [m] \\ i < j}} \chi(Q_i, Q_j)^{M_{i,j}} \right) \\
            &= \Omega(V,M),
              \nonumber
        \end{align}
where in step \text{(i)} we have used  that $ C^{\dagger} P_j C = \pm Q_j $ for each $ j\in[m] $, and since $ |v_j| = 0 \pmod{2} $, it follows that $(C^{\dagger} P_j C)^{\otimes v_j} = Q_j^{\otimes v_j}. $
        In step \text{(ii)}, we have used  the fact that if two matrices commute, their adjoint actions with respect to any unitary also commute.
    \end{proof}
\end{lemma}

\subsection{Graphical calculus to manipulate Pauli monomials}\label{sec:graphcalc}
As a useful tool to enhance the reader's understanding of Pauli monomials, we introduce a graphical representation that simplifies their manipulation.
This graphical calculus has applications beyond our focus on the Clifford commutant's basis elements, proving useful in any context involving products of isotropic sums over tensor products of Pauli operators.

\subsubsection{Pauli monomial diagrams}

A Pauli monomial $ \Omega(V,M) $ can be identified from the columns $ \{v_j\}_{j=1}^m $ of the binary matrix $ V $, along with the matrix $ M $, which encodes the `phase' information, see \cref{def:paulimonomials}. To facilitate this identification, we now introduce in the following definition a graphical representation for keeping track of the columns of $ V $ and for the entries of $ M $. For rapid understanding, refer to the example immediately following.
\begin{definition}[Pauli monomial diagram]
Let $ \Omega(V,M) $ denote a Pauli monomial, where $ V\in\even $, and $ M\in\symf $. The graphical representation is defined as

\begin{itemize}
    \item The columns $ \{v_j\}_{j=1}^m $ are depicted sequentially, where each column is represented as a vector of black and white dots corresponding to entries 1 and 0, respectively (see the following example for clarity).
    \item Each off-diagonal element $ M_{i,j} $ of the matrix $ M $ is represented by a line (denoted \emph{phase}) connecting the $ i $-th and $ j $-th columns if $ M_{i,j} = 1 $, with no line drawn if $ M_{i,j} = 0 $ (again, see the next example).
\end{itemize}
\end{definition}

\begin{example}
Consider $ V \in \even $ and $ M \in \symf $, where $ k = 6 $ and $ m = 3 $, and their corresponding Pauli monomial $\Omega(V,M)$, as follows: 
\setmonomialscale{3.5mm}
\begin{align}
    \Omega(V,M) &= \Omega\left( \begin{pmatrix}
        1 & 1 & 0 \\
        1 & 1 & 0 \\
        1 & 1 & 0 \\
        1 & 1 & 0 \\
        0 & 1 & 1 \\
        0 & 1 & 1
    \end{pmatrix},
    \begin{pmatrix}
        0 & 0 & 1 \\
        0 & 0 & 1 \\
        1 & 1 & 0
    \end{pmatrix} \right) = \monomialdiagram{6}{{1,2,3,4},{1,2,3,4,5,6},{5,6}}{0:2,1:2},
\end{align}

where $ \Omega(V,M) $ can be written explicitly as
\begin{align}
     \Omega(V,M)
    &= \frac{1}{d^m} \sum_{\boldsymbol{P} \in \mathbb{P}_n} 
    \left( P_1^{\otimes v_1} P_2^{\otimes v_2} \cdots P_m^{\otimes v_m} \right)
    \prod_{\substack{j, l \in [m] \\ j < l}} \chi(P_j, P_l)^{M_{j,l}} \\
    \nonumber
    &= \frac{1}{d^3} \sum_{P_1, P_2, P_3 \in \mathbb{P}_n}  
    \left( P_1^{\otimes (1,1,1,1,0,0)} P_2^{\otimes (1,1,1,1,1,1)} P_3^{\otimes (0,0,0,0,1,1)} \right) 
    \chi(P_1, P_3) \chi(P_2, P_3) \\
    &= \frac{1}{d^3} \sum_{P_1,P_2,P_3} P_1 P_2 \otimes P_1 P_2 \otimes P_1 P_2 \otimes P_1 P_2 \otimes P_2 P_3 \otimes P_2 P_3 \times \chi(P_1,P_2) \chi(P_2,P_3) \quad.
    \nonumber
\end{align}
\end{example}
Utilizing this diagrammatic representation, as we will see, will facilitate the execution of operations on Pauli monomials.

\subsubsection{Substitution rules}
Given a Pauli monomial $ \Omega(V, M) $, it is useful to establish rules for manipulating the matrices $ V $ and $ M $ in such a way that the Pauli monomial remains invariant. In particular, we focus on the case where Gaussian column operations (such as adding and swapping columns) are applied to $ V $. Indeed, as we will see in full generality in \cref{th:gaussOP}, $\Omega(V, M) = \Omega(VA, M(A))$ for any $A \in \mathrm{GL}(\mathbb{F}_{2}^{m \times m})$, where $M(A)$ denotes a suitable transformation of the matrix $M$. At the level of $V$, this transformation corresponds to performing Gaussian column operations.

To start, we investigate the case where two nearest-neighbor columns of $ V $ are added, and we explore how the matrix $ M $ must be adjusted to maintain the invariance of the Pauli monomial under this transformation. Throughout, we will label a vector ($v$) odd or even based on the halved Hamming weight $\frac{|v|}{2}$.

\begin{theorem}[Equivalence rule for adding columns in Pauli monomials]  
\label{th:rule1}  
Let $ a_{\pm} \in [m] $ denote either $ a+1 $ or $ a-1 $, where $ a $ is a fixed index.  
Let $ V \in \even $ be a matrix with columns $\{v_j\}_{j=1}^m$, where each $ |v_j| = 0 \pmod{2} $, and let $ M \in \symftwo $ be a symmetric matrix with a zero diagonal.

Define $ V' \in \even $, a matrix with columns $\{v'_j\}_{j=1}^m$, identical to those of $ V $, except for the column $ v'_{a_{\pm}} $, which is defined as  
\begin{equation}
    v'_{a_{\pm}} = v_a + v_{a_{\pm}},
\end{equation}
where the addition is taken modulo 2.

Let $ M' \in \symftwo $ be a symmetric matrix with a zero diagonal, whose elements are identical to those of $ M $, except for the following modifications:
\begin{align}
\label{eq:M1rule1}
    M'_{a,a_{\pm}} &= M_{a,a_{\pm}} + \frac{|v_a|}{2}, \\
    M'_{j,a_{\pm}} &= M_{j,a_{\pm}} + M_{j,a}, \quad \text{for any } j \in [m] \text{ such that } j \notin \{a, a_{\pm}\},
    \label{eq:M2rule1}
\end{align}
and their corresponding symmetric entries, ensuring that $ M' $ remains symmetric.
The substitution $ (V, M) \mapsto (V', M') $ leaves the corresponding Pauli monomial invariant, i.e., 
\begin{equation}
    \Omega(V, M) = \Omega(V', M').
\end{equation}

\begin{proof}  
Let us analyze the case in which $ a_\pm = a+1 $; the case $ a_\pm = a-1 $ follows analogously.  
We recall the definition of a Pauli monomial:  
\begin{align}  
\Omega(V, M) = \frac{1}{d^m} \sum_{P_1, \ldots, P_m \in \mathbb{P}_n}  
\left( P_1^{\otimes v_1} P_2^{\otimes v_2} \cdots P_m^{\otimes v_m} \right)  
\prod_{\substack{j, l \in [m] \\ j < l}} \chi(P_j, P_l)^{M_{j,l}}.  
\end{align}  
Focusing on terms relevant to the substitution, we isolate the summation over $ P_a $ and $ P_{a+1} $, to get  
\begin{align}  
\sum_{P_a, P_{a+1}} P_a^{\otimes v_a} P_{a+1}^{\otimes v_{a+1}} \chi(P_a, P_{a+1})^{M_{a,a+1}}  
\prod_{j \notin \{a,a+1\}} \chi(P_j, P_a)^{M_{a,j}} \chi(P_j, P_{a+1})^{M_{a+1,j}}.  
\label{eq:sumPij}  
\end{align}  
Define the operator $ R $ as  
\begin{align}  
R = P_a P_{a+1} \sqrt{\chi(P_a, P_{a+1})},  
\end{align}  
which belongs to the Pauli basis up to a $\pm 1$ phase (see \cref{le:trivialPQ}, Property 3).  
Inverting the definition of $ R $, we can rewrite $ P_a $ (up to a $\pm 1$ phase) as  
\begin{align}  
P_a \propto \sqrt{\chi(P_a, P_{a+1})} R P_{a+1}  
= \sqrt{\chi(R, P_{a+1})} R P_{a+1},  
\end{align}  
where, in the last step, we have  used Property 3 of \cref{le:trivialPQ}.  

Substituting this expression for $ P_a $ into \cref{eq:sumPij}, recognizing that the global $ \pm 1 $ phase does not affect the result since $ |v_a| = 0 \pmod{2} $, and summing over $ R $ instead of $ P_a $, we obtain  
\begin{align*}
&\sum_{R, P_{a+1}}\! \left( R P_{a+1} \sqrt{\chi(R, P_{a+1})} \right)^{\otimes v_a} P_{a+1}^{\otimes v_{a+1}} \chi(R, P_{a+1})^{M_{a,a+1}}  \!\times\!\!\! \prod_{j \notin \{a, a+1\}}\!\!\! \chi(P_j, R P_{a+1} \sqrt{\chi(R, P_{a+1})})^{M_{a,j}} \chi(P_j, P_{a+1})^{M_{a+1,j}} \\
&\eqt{(i)} \sum_{R, P_{a+1}} \chi(R, P_{a+1})^{|v_a|/2 + M_{a,a+1}} R^{\otimes v_a} P_{a+1}^{\otimes (v_a + v_{a+1})}   \times \prod_{j \notin \{a, a+1\}} \chi(P_j, R P_{a+1} \sqrt{\chi(R, P_{a+1})})^{M_{a,j}} \chi(P_j, P_{a+1})^{M_{a+1,j}} \\
&\eqt{(ii)} \sum_{R, P_{a+1}} \chi(R, P_{a+1})^{|v_a|/2 + M_{a,a+1}} R^{\otimes v_a} P_{a+1}^{\otimes (v_a + v_{a+1})}    \times \prod_{j \notin \{a, a+1\}} \chi(P_j, R)^{M_{a,j}} \chi(P_j, P_{a+1})^{M_{a+1,j} + M_{a,j}}\\
&\eqt{(iii)} \sum_{P_a, P_{a+1}}  P_a^{\otimes v_a} P_{a+1}^{\otimes (v_a + v_{a+1})} \chi(P_a, P_{a+1})^{|v_a|/2 + M_{a,a+1}}   \times \prod_{j \notin \{a, a+1\}} \chi(P_j, P_a)^{M_{a,j}} \chi(P_j, P_{a+1})^{M_{a+1,j} + M_{a,j}},  
\end{align*}
where in step (i) we have used  the identity  
\begin{align}  
\left( \sqrt{\chi(R, P_{a+1})} \right)^{|v_a|} = \chi(R, P_{a+1})^{|v_a|/2},  
\end{align}  
and the fact that $ P_{a+1}^{\otimes v_a} P_{a+1}^{\otimes v_{a+1}} = P_{a+1}^{\otimes (v_a + v_{a+1})} $.  
In step (ii), we 
have used  the property  
\begin{align}  
\chi(P_j, R P_{a+1} \sqrt{\chi(R, P_{a+1})})^{M_{a,j}}  
= \chi(P_j, R P_{a+1})^{M_{a,j}}  
= \chi(P_j, R)^{M_{a,j}} \chi(P_j, P_{a+1})^{M_{a,j}},  
\end{align}  
where, in the last step, we applied $\chi(P, QR) = \chi(P, Q) \chi(P, R)$ (see \cref{le:trivialPQ}, Property 2).  In step (iii), we just renamed  the variable from $ R $ to $ P_a $.  
\end{proof}
\end{theorem}

After introducing the notion of the sum between columns, it is useful to establish a rule that determines the parity of the sum of two columns based on the parity (even or odd) of the individual columns.
\begin{lemma}[Parity of the sum of two columns]\
\label{le:paritysum}
Let $v_1 $ and $v_2$ be two columns as defined above. The following relation holds:  
\begin{align}
    \frac{1}{2}|v_1+v_2| = \frac{1}{2}|v_1| + \frac{1}{2}|v_2| + v_1\cdot v_2 \pmod{2}.
\end{align}
\end{lemma}

This relation leads to the following rules:
\begin{itemize}
    \item If $v_1 \cdot v_2 = 0 \pmod{2}$, then $v_1 + v_2$ is odd if exactly one of $v_1$ or $v_2$ is odd; otherwise, it is even.
    \item If $v_1 \cdot v_2 = 1 \pmod{2}$, then $v_1 + v_2$ is even if exactly one of $v_1$ or $v_2$ is even; otherwise, it is odd.
\end{itemize}
\begin{proof}
The quantity $|v_1+v_2|$ counts the number of dots in $v_1+v_2$, that is the sum of dots in which $v_1$ and $v_2$ do not overlap.
That is, the number of dots in $v_1+v_2$ is
\begin{align}
    |v_1+v_2|=|v_1|+|v_2|-2v_1\cdot v_2,
\end{align}
from which the claim follows. 
\end{proof}
From now on, we say that two columns $v_1$, $v_2$ commute if $v_1\cdot v_2 = 0 \pmod 2$ (which also implies that the respective Pauli monomials commute, see \cref{lem:propertiesprimitivepauli}).

 The rules introduced in \cref{th:rule1} provide the moves on such diagrams that leave the Pauli monomial invariant. Here we illustrate the graphical implementation of this rule.
\begin{example}[Graphical implementation of \cref{th:rule1}]
Consider the following Pauli monomial $ \Omega(V,M) $:
\setmonomialscale{3.5mm}
\begin{align}
\label{eq:diagramex}
    \Omega(V,M) = \monomialdiagram{6}{{1,2,3,4,5,6},{3,4,5,6},{2,5},{1,2,4,6}}{0:1,1:2,0:3} \quad .
\end{align}
We now demonstrate how the equivalence relation shown in \cref{th:rule1} can be implemented graphically. Specifically, \cref{th:rule1} describes how the matrix $ M $ should change when we add nearest-neighbor columns of $ V $, to ensure that the Pauli monomial remains invariant.

If we add the first column to the second, the following `phases addition' applies:
\begin{itemize}
    \item \textbf{Add a phase if the vector being added is odd}: A phase is added between the first and second columns if $ |v_1| = 2 \pmod{4} $, as dictated by \cref{eq:M1rule1} of \cref{th:rule1}. (Note that if a phase already exists between these columns, it is removed, since $ 1 + 1 = 0 \pmod{2} $.)
    \item \textbf{Add shared phases between columns}: Phases are added between the second column and any other column that shares a phase with the first column, as dictated by \cref{eq:M2rule1} of \cref{th:rule1}.
\end{itemize}
Hence, the resulting equivalent diagram obtained by adding the first column to the second is
\setmonomialscale{3.5mm}
\begin{align}
    \Omega(V,M) = \monomialdiagram{6}{{1,2,3,4,5,6},{3,4,5,6},{2,5},{1,2,4,6}}{0:1,1:2,0:3} \quad = \monomialdiagram{6}{{1,2,3,4,5,6},{1,2},{2,5},{1,2,4,6}}{1:2,1:3,0:3} \quad.
\end{align}
Similarly, if we had instead added the second column to the first in \cref{eq:diagramex}, we would have obtained:
\begin{align}
    \Omega(V,M) = \monomialdiagram{6}{{1,2,3,4,5,6},{3,4,5,6},{2,5},{1,2,4,6}}{0:1,1:2,0:3}= \monomialdiagram{6}{{1,2},{3,4,5,6},{2,5},{1,2,4,6}}{0:1,0:2,0:3,1:2} \quad .
\end{align}
\end{example}

Let us now consider the case where two nearest-neighbor columns of $ V $ are swapped and investigate how the matrix $ M $ must be adjusted to ensure the invariance of the Pauli monomial $ \Omega(V, M) $. Although this rule can be derived from the previous one (e.g., by sequentially adding the first column to the second, then the second to the first, and again the first to the second), we provide a direct proof here for simplicity and clarity.
\begin{theorem}[Swapping columns in Pauli monomials]
\label{th:rule2}
Let $ a \in [m-1] $.  
Let $ V \in \even $ have columns $\{v_j\}_{j=1}^m$, where $ |v_j| = 0 \pmod{2} $, and let $ M \in \symftwo $ be a symmetric matrix with a null diagonal.  

Define $ V' \in \mathbb{F}_2^{k \times m} $ with the columns $\{v'_j\}_{j=1}^m$ identical to those of $ V $, except that columns $ a $ and $ a+1 $ are swapped, i.e., 
\begin{equation}
v'_{a} = v_{a+1} \quad \text{and} \quad v'_{a+1} = v_{a}.
\end{equation}

Define $ M' \in \symftwo $, a symmetric matrix with a null diagonal, whose elements are defined as those of $ M $, apart from the element:  
\begin{align}
\label{eq:Mrule2}
    M'_{a,a+1} &= M_{a,a+1} + v_a\cdot v_{a+1}, \\
    M'_{j,a+1} &= M_{j,a}\quad \text{and} \quad M'_{j,a} = M_{j,a+1} \quad \text{for any } j \in [m] \text{ such that } j \notin \{a, a+1\},
\end{align}
and its corresponding symmetric entries.

Then the substitution $(V, M) \mapsto (V', M')$ leaves the corresponding Pauli monomial invariant  
\begin{equation}
    \Omega(V, M) = \Omega(V', M').
\end{equation}

\begin{proof}
The invariance of the Pauli monomial follows directly from the fact that:  
\begin{equation}
P^{\otimes v_a}_a P^{\otimes v_{a+1}}_{a+1} = \chi(P_a, P_{a+1})^{(v_a\cdot v_{a+1})} P^{\otimes v_{a+1}}_{a+1} P^{\otimes v_a}_a,
\end{equation}
where $ \chi(P_a, P_{a+1}) $ is the phase factor arising from the commutation relation between $P_a$ and $ P_{a+1} $.  
Thus, swapping $ v_a $ and $ v_{a+1} $ in $ V $ introduces the phase adjustment $v_a \cdot v_{a+1}$ into the matrix $ M' $. 
\end{proof}
\end{theorem}
\begin{example}[Graphical implementation of \cref{th:rule2}]
In this example, we show the graphical implementation of the equivalence relation in \cref{th:rule2}, which deals with the operation of swapping two nearest-neighbor columns.

Specifically, according to \cref{eq:Mrule2}, if we swap two nearest-neighbor columns (say $ v_2 $ and $ v_3 $), an additional phase is introduced between these two columns if their overlap is one modulo 2 (i.e., $  v_2\cdot v_3 = 1 \pmod{2} $), and each of these columns retains its previous phases when they are swapped.

Consider $ \Omega(V,M) $ as defined in the previous example. If we swap the second and third columns, we obtain
\setmonomialscale{3.5mm}
\begin{align}
    \Omega(V,M) = \monomialdiagram{6}{{1,2,3,4,5,6},{3,4,5,6},{2,5},{1,2,4,6}}{0:1,1:2,0:3} 
    = \monomialdiagram{6}{{1,2,3,4,5,6},{2,5},{3,4,5,6},{1,2,4,6}}{0:2,0:3} \quad .
\end{align}
\end{example}

With these two theorems in hand, we are now confident that arbitrary Gaussian-like column operations (i.e., addition and swapping of arbitrary columns) can be performed on the matrix $ V $, with corresponding modifications to the matrix $ M $, which can be derived by iteratively applying the previous two rules. Alternatively, general Gaussian operations can be tracked explicitly through the following result:

\begin{theorem}[Arbitrary Gaussian-like column operations]\label{th:gaussOP}
    Let $ V \in \even $ and $ M \in \symftwo $, and let $ \Omega(V,M) $ denote the corresponding Pauli monomial. Then, for any matrix $ A \in \mathrm{GL}(\mathbb{F}_2^{m\times m}) $, it holds that  
    \begin{align}
        \Omega(V,M) = \Omega(VA, M(A))\,,
    \end{align}
    where $ M(A)\in \symftwo $ is a function of the matrix $A$, determined as follows. Define  
    \begin{align}
        H &\coloneqq V^T V, \\
        w_i &\coloneqq \frac{|v_i|}{2}, \\
        \Lambda_{i,j} &\coloneqq M_{i,j} + w_i \delta_{i,j} + H_{i,j} \delta_{i<j},  
    \end{align}
    where all operations are $\pmod 2$. Then, for any $ A \in \mathrm{GL}(\mathbb{F}_{2}^{m\times m}) $, we obtain the transformations  
    \begin{align}
        \Lambda(A) &= A^T \Lambda A, \\
        H(A) &= \Lambda(A) + \Lambda(A)^T, \\
        w_i(A) &= \Lambda_{i,i}(A), \\
        M_{i,j}(A) &= \Lambda_{i,j}(A) \delta_{i>j} + \Lambda_{j,i}(A) \delta_{j>i}.
    \end{align}
\end{theorem}

    \begin{proof}
       First we note that the transformation $(w,M,H)\leftrightarrow \Lambda$ is bijective, meaning we can transfer between them arbitrarily. Next, we note the action of $A$ is decomposable $B^{T}\Lambda(A)B=\Lambda(AB)$. We can also use that every $A\in \mathrm{GL}(\mathbb{F}_2^{m\times m})$ can be decomposed into $A=\prod_{\alpha}(\mathbb{1}+E_{i_\alpha,i_\alpha+(-1)^{s_\alpha}})$ for some lists of integers and signs $(i_\alpha,s_\alpha)$, where $E_{i,j}$ is the canonical matrix basis. Using the argument above it suffices to show that the generating transformation  
       \begin{align}
           (\mathbb{1}+E_{i\pm 1,i})\Lambda (\mathbb{1}+E_{i,i\pm 1})&=\Lambda+E_{i\pm 1,i}\Lambda+\Lambda E_{i,i\pm 1}+E_{i\pm 1,i}\Lambda E_{i,i\pm 1}\\
       \end{align}
       correctly preserves $(w,M,H)$. This corresponds to the elementary transformation of adding the $i$-th column vector to the $i\pm 1$-th. For $w$, we have
       \begin{align}
           w_j\mapsto(A^T\Lambda A)_{j,j}=\Lambda_{j,j}+\delta_{j,i\pm 1}(\Lambda_{i,j}+\Lambda_{j,i}+\Lambda_{i,i})=w_j+\delta_{j,i\pm 1}(H_{i,j}+w_{i})
       \end{align}
       meaning that $w_j$ is unchanged, unless the $j$-th column vector is added onto the $(j\pm1)$-th, where it changes if the one added is an odd weight or the two primitives do not commute, as it is the requirement, see \cref{th:rule1}.
       For $M$, we have for $j> k$
       \begin{align}
           M_{jk}&\mapsto M_{jk}+\delta_{j,i\pm 1}\Lambda_{i,k}+
           \delta_{k,i\pm 1}\Lambda_{j,i}+
          \delta_{j,i\pm1}\delta_{k,i\pm1}\Lambda_{i,i}\\
           &=M_{jk}+\delta_{j,i\pm 1}(\delta_{i,k}w_i+\delta_{i>k}M_{ik}) +
           \delta_{k,i\pm 1}(\delta_{j,i}w_i+\delta_{j>i}M_{j,i})\\
           &=M_{jk}+w_i(\delta_{j,i\pm 1}\delta_{i,k}+\delta_{k,i\pm 1}\delta_{j,i})+\delta_{j,i\pm1}M_{i,k}+
           \delta_{k,i\pm1}M_{j,i}
       \end{align}
       which corresponds to adding the phases onto the $i\pm 1$-th column and row as well as changing the phase between $i\pm1$ if the respective $w_{i}$ odd weight primitive as required. Finally, for $H$, we already know that $H(A)=A^TV^TVA=A^T(\Lambda+\Lambda^T)A$ transforms adjointly meaning that  $H$ and $\Lambda+\Lambda^T$ transform identically.
    \end{proof}
The above theorem establishes that a Pauli monomial remains invariant under column transformations $ A \in \mathrm{GL}(\mathbb{F}_2^{m \times m}) $ when applied appropriately and explicitly describes how these operations affect the matrix $ M $.
Therefore, following the same approach as in \cref{def:equivalencerelationVG}, we can define an equivalence relation for Pauli monomials. In fact, the substitution rules discussed above correspond precisely to exploring the equivalence class of a Pauli monomial. 

We now present a theorem that encompasses several important properties derived from the aforementioned basic rules. These properties serve as the foundation for establishing the following key fact: primitive Pauli monomials generate the full commutant of the Clifford group.

\begin{theorem}[Graphical moves for Pauli monomials]
\label{th:graphrules}
Consider a Pauli monomial $\Omega(V,M)$, where $V \in \even$ has columns $\{v_j\}_{j=1}^m$. The following properties hold:
\begin{enumerate}[label=(\arabic*)]
    \item Arbitrary Gaussian-elimination operations can be performed on the columns of the Pauli monomial diagram, by appropriately adjusting the phases.

    \item If there exists an `odd' column $v$ (i.e., $|v| = 2 \pmod{4}$) in the diagram, then it is possible to remove all phases attached to $v$ and reposition $v$ to any desired position by suitably modifying the remaining columns of $V$ and $M$.

    \item If we have linear dependence between the columns, by performing Gaussian-elimination operations on these columns, we can reduce $\Omega(V,M)$ to $\Omega(V',M')$, which satisfies one of the following two cases:
    \begin{itemize}
        \item There is a column $v = 0$ with no phase attached to it in the transformed diagram $\Omega(V',M')$. In this case, we can further reduce it to
        \be
        \Omega(V',M') = d \, \Omega(V_r, M_r),
        \ee
        where $V_r  \in \mathbb{F}_2^{k \times m-1}$ and $M_r\in \mathbb{F}_2^{m-1 \times m-1}$ are the matrices obtained by removing the respective column (and its associated rows) corresponding to $v$.
        \item There is a column $v = 0$ with one or more phases attached to it in the transformed diagram $\Omega(V',M')$. In this case, we have
        \be
        \Omega(V',M') = \Omega(V_r, M_r),
        \ee
        where $V_r  \in \mathbb{F}_2^{k \times m-2}$ and $M_r\in \mathbb{F}_2^{m-2 \times m-2}$ are the matrices obtained by removing both the zero column $v$ and the column connected to it through the phase  (and their associated rows).
    \end{itemize}

\item If there are two even non-identical columns that share a phase and commute (i.e., overlap in an even number of dots), they can be transformed into four odd columns without a phase. More specifically, the matrix $V$ containing two even, commuting columns that share a phase can be converted into a new matrix $V' \in \mathbb{F}_2^{k \times (m+2)}$, where these columns are replaced by four odd commuting columns that do not share any phases.
\end{enumerate}
\end{theorem}

\begin{proof}
We prove the properties one by one:
\begin{enumerate}[label=(\arabic*)]
    \item This is a restating of \cref{th:gaussOP} in diagram terminology, where Gaussian elimination on the columns of the diagram corresponds to permissible transformations that preserve the Pauli monomial structure, up to phases.

    \item Assume there is an odd column $v$ in the diagram, possibly with some phases attached. Begin by swapping this column to the leftmost position. During this process, additional phases may appear, but these are inconsequential for the argument.

    Once $v$ is in the first position, swap it with the second column. If a phase exists between the first and second columns after this swap, add the second column $v$ to the first. Since $v$ is odd ($|v| = 2 \pmod{4}$), this operation generates a new phase between the first and second columns, which cancels the pre-existing phase.

    Repeat this process by swapping $v$ from the second position to the third, and so on, always adding the now-swapped $v$ to the left column if a phase exists between them. By continuing this procedure all the way to the rightmost position, $v$ will have no attached phases and will occupy the far-right position.

    To move $v$ to any desired position, reverse the process: bring $v$ leftward, adding to the right to removing phases as necessary to ensure that no phases remain attached to $v$ at each step. This completes the proof of the second property.

    Below is an example of such a procedure. Without loss of generality, assume that the column is already positioned on the leftmost side and we position it on the rightmost side without phases attached to it:
    \setmonomialscale{3.5mm}
    \begin{align}
        \monomialdiagram{4}{{1,2},{2,3},{1,2,3,4},{3,4}}{0:1,0:2,0:3} 
        \quad\textbf{=}\quad
        \monomialdiagram{4}{{2,3},{1,2},{1,2,3,4},{3,4}}{1:2,1:3} 
        \quad\textbf{=}\quad
        \monomialdiagram{4}{{2,3},{1,2,3,4},{1,2},{3,4}}{1:2,2:3} 
        \quad\textbf{=}\quad
        \monomialdiagram{4}{{2,3},{3,4},{1,2},{3,4}}{2:3,1:3} 
        \quad\textbf{=}\quad
        \monomialdiagram{4}{{2,3},{3,4},{3,4},{1,2}}{1:2,2:3} 
        \quad\textbf{=}\quad
        \monomialdiagram{4}{{2,3},{3,4},{1,2,3,4},{1,2}}{1:2}.
    \end{align}

    \item If there is linear dependence between the columns, we can, by virtue of property 1, perform Gaussian-elimination moves to transform the diagram into an equivalent form $\Omega(V',M')$, where at least one column is $v = 0$. 

There are three possible cases to consider: the column $v = 0$ has no phases attached to it, exactly one phase, or more than one phase. We analyze these cases one by one and explicitly show that the last case can always be reduced to one of the first two.
This will conclude the proof.
\begin{itemize}
    \item No attached phases: If a column $v = 0$ has no connected phases, it corresponds to a free sum contributing only a global factor. By \cref{def:paulimonomials}, this implies:
\begin{equation}
\Omega(V',M') = d \, \Omega(V_r, M_r),
\end{equation}
where $V_r$ and $M_r$ are obtained by removing the respective columns and rows of $v$. The factor $d$ arises from the $d^2$-element sum, normalized by $1/d$.

\item One attached phase: If a column $v = 0$ has exactly one connected phase, the relation 
\begin{equation}
\frac{1}{d} \sum_P \chi(P, Q) = d \, \delta_{Q, \mathbb{1}}
\end{equation}
ensures that the sum associated to the connected column vanishes. Consequently, we have
\begin{equation}
\Omega(V',M') = \Omega(V_r, M_r),
\end{equation}
where $V_r$ excludes both the column $v = 0$ and the column associated with the phase.

\item Multiple attached phases: Without loss of generality, assume the first column is $v = 0$, and subsequent columns have phases attached to it in order, followed by the columns that do not share any phase with the first zero column. This can be achieved without loss of generality through Gaussian elimination moves (column swaps). We will focus only on the phases interacting with the first zero column and disregard those between other columns.

For example, consider the following configuration:
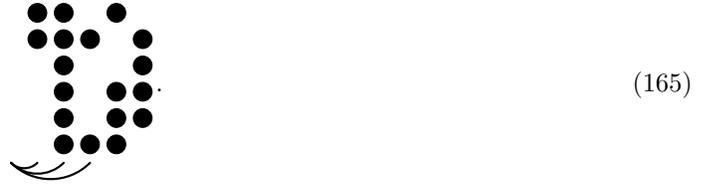
\setmonomialscale{3.5mm}
\begin{align}
    \monomialdiagram{6}{{x},{1,2},{1,2,3,4,5,6},{2,6},{1,4,5,6},{2,3,4,5}}{0:1,0:2,0:3}. 
\end{align}
   
We begin by handling the rightmost column that has a phase with the first zero column. Applying the graphical addition rule, we add this rightmost column to its left neighbor, thereby canceling the phase between them. After this step, the diagram becomes (neglecting to draw any possible new phase generated between columns that are not connected to the first zero column for simplicity):
\setmonomialscale{3.5mm}
\begin{align}
    \monomialdiagram{6}{{x},{1,2},{1,3,4,5},{2,6},{1,4,5,6},{2,3,4,5}}{0:1,0:3}.
\end{align}

Next, we swap the rightmost column with the one on the left, yielding the diagram:
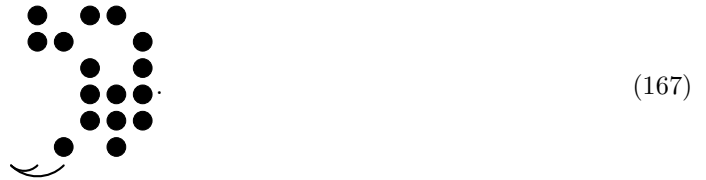
\setmonomialscale{3.5mm}
\begin{align}
    \monomialdiagram{6}{{x},{1,2},{2,6},{1,3,4,5},{1,4,5,6},{2,3,4,5}}{0:1,0:2}.
\end{align}

At this point, we are in a similar situation as initially, but with one fewer column. By induction, we repeat this process until only one column remains connected to $v = 0$. For instance, the progression would proceed as
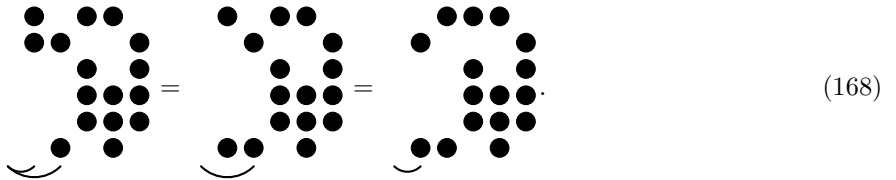
\setmonomialscale{3.5mm}
\begin{align}
    \monomialdiagram{6}{{x},{1,2},{2,6},{1,3,4,5},{1,4,5,6},{2,3,4,5}}{0:1,0:2}=\monomialdiagram{6}{{x},{1,6},{2,6},{1,3,4,5},{1,4,5,6},{2,3,4,5}}{0:2}  = \monomialdiagram{6}{{x},{2,6},{1,6},{1,3,4,5},{1,4,5,6},{2,3,4,5}}{0:1}.
\end{align}
At this stage, the system has effectively reduced to the case of a single attached phase, which we have already addressed.
\end{itemize}

\item If there are two even terms that share a phase and commute (overlap in an even number of dots modulo two), we can transform them into $4$ odd commuting columns with no phase between them.

Let us explain the reason why. First, we can assume these terms to be nearest neighbors (due to Gaussian elimination moves). We can `add for free' two identical two-dot columns with no phase next to these columns. This follows because it corresponds to applying two identical swap operators consecutively, which is equivalent to multiplying the monomial by the identity.  

We are free to place such two-dot columns in any position, and the following argument remains essentially unchanged, provided that the added columns commute with both commuting primitives. In the example below, we place them where the two original columns both have dots (or both have no dots). Such a position exists because the columns are assumed to overlap in an even number of points.

An example for such setup looks like 
\begin{align}
    \monomialdiagram{6}{{1,2,3,4},{3,4,5,6}}{0:1}\quad=\quad\monomialdiagram{6}{{1,2,3,4},{3,4,5,6},{3,4},{3,4}}{0:1}.
\end{align}

Now, the idea is to sum these two-dot columns to the left, changing the parity of the neighboring columns.  
In detail, we proceed as follows:  
We add the first two-dot column to the column on its left, introducing a phase between these two columns (because the two-dot column is odd).  
For example, in the diagram above, we would obtain  
\begin{align}
    \monomialdiagram{6}{{1,2,3,4},{3,4,5,6},{3,4},{3,4}}{0:1} = \monomialdiagram{6}{{1,2,3,4},{5,6},{3,4},{3,4}}{0:1,1:2}.
\end{align}
By \cref{le:paritysum}, this changes the parity of the column we added to (i.e., the second column) from even to odd.  

Next, we add the second column (now odd) to the first column (even). Since these two columns still overlap in an even number of dots (modulo two), the first column becomes odd after this transformation (by \cref{le:paritysum}). Furthermore, this introduces an additional phase between the first and second columns (due to adding an odd column), while the second column retains its phases with the third column.  

In the example, we obtain  
\begin{align}
    \monomialdiagram{6}{{1,2,3,4},{5,6},{3,4},{3,4}}{0:1,1:2}=\monomialdiagram{6}{{1,2,3,4,5,6},{5,6},{3,4},{3,4}}{1:2,0:2}.
\end{align}

At this point, we have four (commuting) odd columns, but we still need to remove the introduced phases.

To remove the phases, recall that there is now a phase between the first column and the third, and between the second column and the third.  
We begin by adding the second column (odd) to the third (odd). Since they commute, the third column becomes even after this addition (by \cref{le:paritysum}), and we remove the phase (because adding an odd column add an additional phase).  

In the example, this yields 
\begin{align}
  \monomialdiagram{6}{{1,2,3,4,5,6},{5,6},{3,4},{3,4}}{1:2,0:2}=\monomialdiagram{6}{{1,2,3,4,5,6},{5,6},{3,4,5,6},{3,4}}{0:2}.
\end{align}

Now, we are left with only the phase between the first and the second columns, but the third column is now even. To address this, we swap the first column to the second position. Since the columns commute, this introduces no additional phase.  
Finally, we add the second column (odd) to the third (even), removing the remaining phase. After this, the third column becomes odd (by \cref{le:paritysum}), and we achieve the desired configuration with no phases. We also observe that all columns commute, as the operations of swapping and summing even columns preserve the fact that they commute.

In the example, we obtain  
\begin{align}
    \monomialdiagram{6}{{1,2,3,4,5,6},{5,6},{3,4,5,6},{3,4}}{0:2} = \monomialdiagram{6}{{5,6},{1,2,3,4,5,6},{3,4,5,6},{3,4}}{1:2} = \monomialdiagram{6}{{5,6},{1,2,3,4,5,6},{1,2},{3,4}}{}.
\end{align}

\end{enumerate}
\end{proof}

\subsection{Properties of Pauli monomials: reduction, normal form and inner product}
Having outlined the substitution rules for Pauli monomials, we now use them to establish their fundamental properties. We begin by introducing the notion of \emph{reduced Pauli monomials}.

\begin{definition}[Reduced Pauli monomial]\label{def:reducedpaulimonomials}
    Let $ \Omega(V, M)  $ be a Pauli monomial with $V\in\even$ and $M\in \symf$. We say that $ \Omega(V, M) $ is a \textit{reduced Pauli monomial} if and only if $\rank(V)=m$, i.e., the column vectors of $ V $ are linearly independent.
\end{definition}
At this point, one might question the reasoning behind introducing an extra set of Pauli monomials; however, this is justified by the following lemma.
\begin{lemma}[Pauli monomial reduction]\label{lem:everypaulimonomialcanbereduced}
    Every Pauli monomial is proportional to a reduced Pauli monomial. More precisely, for any $ \Omega(V,M) $ with $V\in\even$, there exists a reduced Pauli monomial $ \Omega(V', M')$ with $V'\in\mathrm{Even}(\mathbb{F}_{2}^{k\times m'})$ with $m'\le m$, such that 
    \begin{equation}
    \Omega(V, M) = d^{\alpha} \Omega(V', M'),
    \end{equation}
    where $\alpha \in \mathbb{N}$ is a natural number such that $ \alpha \leq m - m' $.
\begin{proof}
    The statement descends from  graphical rule (3) of \cref{th:graphrules}. Specifically, in the presence of linear dependencies, we can perform Gaussian elimination to produce zero columns, which, without loss of generality, can be positioned at the beginning of the matrix. By then applying rule (3) from \cref{th:graphrules}, we can further reduce the number of columns, repeating this process until no linear dependencies remain. This yields the desired result.
\end{proof}
\end{lemma}
\cref{lem:everypaulimonomialcanbereduced} is crucial because it implies the existence of an essential set of Pauli monomials, to which all other monomials are proportional to, thereby allowing us to focus our analysis exclusively on the set of reduced Pauli monomials.

It is evident that reduced Pauli monomials differ based on the number of independent column vectors. Consequently, we can classify Pauli monomials according to the rank of the associated reduced Pauli monomials, leading to the concept of \emph{order}.

\begin{definition}[Order of Pauli monomials]\label{def:orderpaulimonomial}
    Let $ \Omega(V, M)$ be a Pauli monomial. Let $\Omega(V',M')$ be the reduced Pauli monomial corresponding to $\Omega$, i.e., there exists $\alpha$ such that $\Omega(V,M)=d^{\alpha}\Omega(V',M')$. The \textit{order} of $\Omega$, hereby denoted as $m(\Omega)$, is defined as $m(\Omega)\coloneqq\rank(V')$.
\end{definition}

\begin{definition}[Set of reduced Pauli monomials]\label{def:distinctpaulimonomials} We define the set of reduced Pauli monomials as
\be
\mathcal{P}\coloneqq\{\Omega(V,M)\,|\, V\in\even\,:\,\rank(V)=m\,,\,M\in\symf\,,\,m\in[k-1]\}\,.
\ee
\end{definition}

\begin{remark}[Representative Pauli monomials]
   Two remarks about \cref{def:distinctpaulimonomials} are in order. First, by the definition of the set, only a single representative of the equivalence class $ [V,M] \coloneqq \{(VA, M(A)) \,:\, A \in \operatorname{GL}(\mathbb{F}_{2}^{m\times m})\} $ is sufficient to specify $\mathcal{P}$, similarly to \cref{def:equivalencerelationVG}. Here, $ M(A) $ is defined in \cref{th:gaussOP}. Second, since we require $ \operatorname{rank}(V) = m $, the corresponding Pauli monomial $ \Omega(V,M) $ is reduced, as all the column vectors in $ V $ are linearly independent.
\end{remark}

The following lemma determines the value of the trace of a general Pauli monomial.

\begin{lemma}[Trace of Pauli monomials]\label{lem:tracesofpaulimonomials}
    Let $\Omega$ be a Pauli monomial with $\beta$ many iteratively removable linear dependencies, corresponding to $m$ many column vectors. It holds that
    \begin{align}
        \tr(\Omega)=d^{k-m+2\beta} .
    \end{align}
    \begin{proof}
    Let us first assume that $\beta=0$, i.e., the Pauli monomial is reduced. For every $\Omega$ there exists a transformation that puts $\Omega$ in a pivotal structure, i.e., $(V,M)\mapsto(V',M')$, where $V'$ contains pivotal vectors with no linear dependencies. As such, we can perform the trace starting from the first column of $V$. Since $(v_{1})_{1}=1$ and $(v_{i})_1=0$ for $i\neq 1$, this sets $P_1$ to the identity, since $\tr(P)=d^k\delta_{I,P}$ and there are no linear dependencies. Hence, the only term contributing is the all identity term. As $\tr(\mathbb{1})=d^k$, then we have $\tr(\Omega)=d^{k-m}$. Now, let us assume that $\beta>0$. By \cref{lem:everypaulimonomialcanbereduced}, every Pauli monomial can be reduced as $\Omega=d^{\alpha}\Omega_{R}$, with $\Omega_R\in\mathcal{P}$. Then, $\tr(\Omega)=d^{\alpha}\tr(\Omega_{R})=d^{k-(m(\Omega)-\alpha)}$, where $m(\Omega)$ denotes the order of $\Omega$ according to \cref{def:orderpaulimonomial}. Iterating rule 3 \cref{th:graphrules} for $\beta$ many linear dependencies, we know that $m(\Omega)-\alpha=m-2\beta$. Therefore, the result follows.
\end{proof}
\end{lemma}

After defining the concept of Pauli monomials and their reduction to a more compact set of components, comprehending their features remains crucial. To this end, let us introduce the notion of unitary and projective Pauli monomials.

\begin{lemma}[Unitary and projective Pauli monomials]\label{le:unitaryprojmon}

Let $ \Omega(V, 0)\in\mathcal{P} $ be a reduced Pauli monomial, where $ V \in \mathbb{F}_2^{k \times m} $ with columns $ \{v_j\}_{j=1}^m $. The following hold:
\begin{itemize}
    \item If $ V $ consists exclusively of even column vectors satisfying $ v_i \cdot v_j = 0 \pmod{2} $ for $ i \neq j $, then $ \Omega(V, 0) $, called \emph{projective Pauli monomial} and denoted $ \Omega_P(V) \coloneqq \Omega(V, 0) $, satisfies $ \Omega_P(V) / d^m $ is a projector of rank $ d^{k - 2m} $\footnote{The projective Pauli monomials of Lemma~\ref{le:unitaryprojmon} coincide with the CSS code projectors identified by Gross, Nezami, and Walter~\cite{gross_schurweyl_2019} in their Lagrangian subspace description of the Clifford commutant. For example, the operator $\Omega_4$ is proportional to a CSS code projector onto a subspace of dimension $d^{n-2}$.}.
    \item If $ V $ consists exclusively of odd column vectors, then $ \Omega(V, 0) $, called  \emph{unitary Pauli monomial} and denoted $ \Omega_U(V) \coloneqq \Omega(V, 0) $, is a unitary operator.
\end{itemize}
We denote the set of reduced Pauli monomials which are unitary as $\mathcal{P}_U$, and reduced Pauli monomials which are projectors as $\mathcal{P}_P$.
\begin{proof}
    We can express $ \Omega(V, 0) $ as
    \begin{align}
        \Omega(V, 0) = \Omega(v_1) \cdots \Omega(v_m),
    \end{align}
    where $ \{v_j\}^m_{j=1} $ are the columns of $ V $.
    If the vectors $ \{v_j\}^m_{j=1} $ are all even, then, by \cref{lem:propertiesprimitivepauli}, each $ d^{-1} \Omega(v_j) $ is a projector (property A of \cref{lem:propertiesprimitivepauli}). Additionally, since $ v_i \cdot v_j = 0 \pmod{2} $ for $ i \neq j $, these projectors commute with each other (property C of \cref{lem:propertiesprimitivepauli}). Since the product of commuting projectors is itself a projector, we can conclude that $ \Omega(V, 0)/d^m $ is a projector. The rank of this projector is given by its trace:
    \begin{align}
        \frac{1}{d^m} \tr(\Omega(V, 0)) &= \frac{1}{d^{2m}} \sum_{\boldsymbol{P} \in \mathbb{P}_n}  
        \tr(P_1^{\otimes v_1} P_2^{\otimes v_2} \cdots P_m^{\otimes v_m}) \\
        &= d^{k - 2m},
        \nonumber
    \end{align}
where the only non-zero contribution arises when $ P_1^{\otimes v_1} P_2^{\otimes v_2} \cdots P_m^{\otimes v_m} $ equals the identity. This occurs only if all Pauli matrices \( P_1, \dots, P_m \) are equal to the identity, which follows from the linear independence of \( v_1, \dots, v_m \)~\footnote{Indeed, the bitstring representation of the product is \( b(P_1^{\otimes v_1} \cdots P_m^{\otimes v_m}) = \sum_{j=1}^m b(P_j) \otimes v_j \), which must vanish, since it corresponds to the identity (i.e., the zero vector). Projecting this sum onto canonical basis elements of the first tensor factor and using linear independence of the \( v_j \), we conclude that \( b(P_j) = 0 \) for all \( j \).}.

If the vectors $ \{v_j\}^m_{j=1} $ are all odd, then, by \cref{lem:propertiesprimitivepauli} (property B), each $ d^{-1} \Omega(v_j) $ is a unitary. Since the product of unitaries is also a unitary, $ \Omega(V, 0) $ is a unitary operator.
\end{proof}
\end{lemma}

It is clear that when $V=(v)$, this notion reduces to what is shown in \cref{lem:propertiesprimitivepauli} for primitive Pauli monomials.

Given this characterization of the Pauli monomials, we can provide a general decomposition for every Pauli monomial, which we refer to as \textit{normal form}. For the sake of simplicity, we restrict ourselves to the case of reduced Pauli monomials (\cref{def:reducedpaulimonomials}), although it also holds for general Pauli monomials due to \cref{lem:everypaulimonomialcanbereduced}.

\begin{lemma}[Normal form]\label{lem:normalform}
    Let $\Omega(V, M) \in \mathcal{P}$ be any reduced Pauli monomial. It can be decomposed as a product of a projective Pauli monomial and a unitary Pauli monomial as introduced in \cref{le:unitaryprojmon}. Specifically:
    \begin{align}
        \Omega(V, M) = \Omega_P(V') \times \Omega_U(V''),
    \end{align}
    where $\Omega_P(V')$ and $\Omega_U(V'')$ denote a projective and a unitary Pauli monomial, respectively. 

\begin{proof}
We consider a general reduced Pauli monomial $\Omega(V, M)$ along with its associated Pauli diagram. The goal is to partition $V$ into a left part ($V'$), representing the \emph{projective} component, and a right part ($V''$), representing the \emph{unitary} component. The left part will consist of even commuting columns with no phases, while the right part will contain odd columns with no phases. The procedure is as follows:

\begin{itemize}
    \item \textbf{Step 1: Move odd columns to the rightmost positions.}  
    Examine the columns of $V$. For each odd column $(\vert v \vert =  2 \pmod 4)$, move it to the rightmost position among the remaining columns, ensuring no phase is attached, as guaranteed by rule $(2)$ of \cref{th:graphrules}. Fix this column in place, mark it as settled, and exclude it from further consideration. Continue this process until no odd columns remain in $V$.  

    This process terminates in at most $m$ iterations, where $m$ is the total number of columns in $V$. If all columns are moved to the right, then $\Omega(V, M)$ is a unitary Pauli monomial, i.e., $V' = 0$.  

    \item \textbf{Step 2: Ensure the left part contains only even commuting columns.}  
    After Step 1, the remaining columns on the left are even $(\vert v\vert=0 \pmod 4)$. If any of these columns do not commute, add one to another as prescribed by \cref{le:paritysum}. By doing so, any odd column resulting from the sum can be moved to the rightmost position without a phase, as in Step 1. Repeat this process until all remaining columns on the left are even and commute with each other.

    \item \textbf{Step 3: Eliminate phases between even commuting columns.}  
If the left part contains even commuting columns with non-zero phases, proceed as follows (otherwise, we can conclude the proof). First, bring the two columns with a phase between them close to each other and position such column pairs in the rightmost part of the left part. 
Next, apply rule (4) of \cref{th:graphrules} to transform the two even commuting columns into four commuting odd columns with no phases between them. These four odd columns may share phases with other columns in the left part. To remove these phases, take each of the four odd columns in turn and use rule (1) of \cref{th:graphrules} to eliminate the phase, keeping the column fixed in its position. Since these columns commute with each other, they can then be moved to the right part as part of the \emph{unitary} component.  

After this process, the left part loses two columns, while the right part gains four odd columns. Any odd columns created in the left part during this process are moved to the rightmost positions as in Step 1. Repeat this procedure until the left part contains only even commuting columns with no phases.  
\end{itemize}

This iterative procedure ensures that $V'$ and $V''$ are properly partitioned into the projective and unitary components, respectively, satisfying the conditions of the lemma. The process terminates in a finite number of steps, as each iteration strictly reduces the number of columns of the left part.  
This concludes the proof.
\end{proof}
\end{lemma}
Thus, we immediately obtain the following.
\begin{corollary}[Primitive Pauli monomials generate Pauli monomials]
\label{cor:primitivegen}
Any reduced Pauli monomial \( \Omega \) can be written as a product of primitive Pauli monomials.
\end{corollary}

Therefore, any reduced Pauli monomial $\Omega \in \mathcal{P}$ can be expressed as $\Omega = \Omega_P \Omega_U$, where $\Omega_P \in \mathcal{P}_P$ and $\Omega_U \in \mathcal{P}_U$. The following proposition determines the order of these respective monomials (see \cref{def:orderpaulimonomial}). 

\begin{proposition} Let $\Omega$ be a reduced Pauli monomial, reading $\Omega=\Omega_P\Omega_U$ in normal form, where $\Omega_P$ (resp. $\Omega_U$) are the corresponding projective (resp. unitary) Pauli monomials. Then, 
their order satisfies (cf.\ \cref{def:orderpaulimonomial})
\begin{align}
m(\Omega_P)&=\dim(\mathrm{kern}(\Lambda(\Omega)),\\
m(\Omega_U)&=\dim(\Lambda)-\dim(\mathrm{kern}(\Lambda(\Omega))=\mathrm{rank}(\Lambda(\Omega)),
\end{align}
where $\Lambda(\Omega)$ is defined in \cref{th:gaussOP}.
\end{proposition}
\begin{proof}
      In \cref{lem:normalform}, we showed that any reduced monomial $\Omega$ can be expressed as the product of a projective and a unitary monomial, $\Omega = \Omega_P \Omega_U$. In particular, $\Omega_P$ is reduced, whereas $\Omega_U$ may not be.

We observe that for step 3 of \cref{lem:normalform} to succeed, it suffices to add two new permutations commuting with any other primitive. This eliminates the phase from the even commuting columns, effectively rendering them unitary. A convenient strategy is to introduce two additional registers (i.e., increasing the number of copies from $k$ to $k+2$) and place the new permutations there. By construction, these columns commute with all others. In this setting, we can write $\Lambda \mapsto \Lambda \oplus \mathbb{1}_2$.

Now, by \cref{lem:normalform}, there exists a matrix $A \in \mathrm{GL}(\mathbb{F}_2^{m \times m})$ such that the transformation $\Lambda \oplus \mathbb{1}_2 \mapsto A^T( \Lambda \oplus \mathbb{1}_2 ) A$ yields an upper triangular matrix. In the unitary part of the normal form, the unitary primitives have a $1$ on the diagonal, implying that the corresponding columns are linearly independent. In contrast, the projectors have a $0$ on the diagonal and commute with one another. Consequently, the corresponding columns are all-zero columns, meaning that all and only the projective primitives contribute to the kernel dimension. Hence, $m(\Omega_P) = \dim \ker\left( A^T( \Lambda \oplus \mathbb{1}_2 ) A \right) = \dim \ker( \Lambda \oplus \mathbb{1}_2 ) = \dim \ker( \Lambda )$. Therefore, we have $m(\Omega_U) = m(\Omega) - m(\Omega_P) = m - \dim\left( \ker(\Lambda) \right)$, which proves the statement.
\end{proof}

As a consequence of the normal form, we have the following corollary about the Hilbert-Schmidt product between Pauli monomials.

\begin{lemma}[Hilbert--Schmidt inner product between Pauli monomials]\label{lem:hilbertschmidtproductofpaulimonomials}
Let \( \Omega \coloneqq \Omega(V, M) \) and \( \Omega' \coloneqq \Omega(V', M') \) be two reduced Pauli monomials in \( \mathcal{P} \). Then their Hilbert--Schmidt inner product satisfies:
\begin{align}
    \Tr(\Omega^\dagger \Omega') =
    \begin{cases}
        d^k & \text{if } \Omega = \Omega', \\
        \le d^{k-2} & \text{if } \Omega \ne \Omega' \text{ with } V = V' \text{ and } M \ne M', \\
        \le d^{k-1} & \text{if } \Omega \ne \Omega' \text{ with } V \ne V'.
    \end{cases}
\end{align}

\end{lemma}

\begin{proof}
The case \( \Omega = \Omega' \) follows directly by expressing \( \Omega \) in its normal form \( \Omega = \Omega_U \Omega_P \).

We now consider the case \( \Omega \neq \Omega' \), which we split into two subcases.
\begin{itemize}
    \item \textbf{Case 1:} \( V \ne V' \).\\
    Note that \( \Omega^\dagger \Omega' \) is again a (not necessarily reduced) Pauli monomial. Applying the reduction procedure, we obtain \( \Omega^\dagger \Omega' = d^\alpha \, \Omega(V + V', M_r) \), for some integer \( \alpha \le \dim(V \cap V') \), and some phase matrix \( M_r \). Here, \( \Omega(V + V', M_r) \) is in reduced form. Taking the trace gives
    \begin{align}
        \Tr(\Omega^\dagger \Omega') = d^\alpha \Tr(\Omega(V + V', M_r)) = d^\alpha \cdot d^{k - \dim(V + V')} = d^{k - \dim(V + V') + \alpha}.
    \end{align}
    Since \( \alpha \le \dim(V \cap V') \), this yields the upper bound \( \Tr(\Omega^\dagger \Omega') \le d^{k - \dim(V + V') + \dim(V \cap V')} \). Now, \( \dim(V + V') > \dim(V \cap V') \), so \( \dim(V + V') - \dim(V \cap V') \ge 1 \), and we conclude \( \Tr(\Omega^\dagger \Omega') \le d^{k - 1} \), as claimed.

    \item \textbf{Case 2:} \( V = V' \) but \( M \ne M' \).\\
    The previous bound would yield \( \Tr(\Omega^\dagger \Omega') \le d^k \), whereas we want something smaller. Since \( V = V' \), we have
    \begin{align}
        \Tr(\Omega^\dagger \Omega') = d^{k - \dim(V) + \alpha}.
    \end{align}
    To ensure this is at most \( d^{k - 1} \), we need \( \alpha < \dim(V) \). Assume that \( M' \ne 0 \). In the reduction of \( \Omega^\dagger \Omega' \), the nontrivial phases in \( M' \) survive until the final steps. When a zero column is removed as per Rule 3 in Theorem~\ref{th:graphrules}, we fall into the second case of the rule, which strictly reduces \( \alpha \) compared to \( \dim(V \cap V') = \dim(V) \). Specifically, from the rule we would get \( \alpha \le \dim(V)-2 \), and thus \( \Tr(\Omega^\dagger \Omega') \le d^{k - 2}  \).

    If instead \( M' = 0 \) but \( M \ne 0 \), we can invoke cyclicity of the trace: \( \Tr(\Omega^\dagger \Omega') = \Tr(\Omega' \Omega^\dagger) \), and apply the same reasoning to \( \Omega' \Omega^\dagger \), reducing to the previous case.
\end{itemize}

This completes the proof.
\end{proof}
From this follows that Pauli monomials are \emph{approximately orthogonal} with respect to the Hilbert-Schmidt inner product, similar to the permutations operators.

\begin{lemma}
\label{lem:stabmonomial}
Let \( \Omega \) be a Pauli monomial.  
If \( \psi \) is a stabilizer state, then \( \Tr(\Omega\, \psi^{\otimes k}) = 1 \).
\end{lemma}

\begin{proof}
Since \( \psi \) is a stabilizer state, there exists a Clifford unitary \( C \) such that \( \psi = C \ketbra{0} C^\dagger \). Using cyclicity of the trace and the fact that Pauli monomials are in the Clifford commutant, we have
\begin{align}
\Tr(\Omega\, \psi^{\otimes k}) 
&= \Tr(C^{\dagger \otimes k} \Omega\, C^{\otimes k} \ketbra{0}^{\otimes k}) 
= \Tr(\Omega\, \ketbra{0}^{\otimes k}).
\end{align}
Write \( \Omega = \Omega(V, M) \) in pivotal form, with columns \( v_1, \dots, v_m \) forming an upper-triangular binary matrix. Then
\begin{align}
\Tr(\Omega\, \ketbra{0}^{\otimes k}) 
= \frac{1}{d^m} \sum_{P_1, \dots, P_m} \Tr(P_1^{\otimes v_1} \cdots P_m^{\otimes v_m} \ketbra{0}^{\otimes k}) \cdot \prod_{1 \le i < j \le m} \chi(P_i, P_j)^{M_{i,j}},
\end{align}
where \( \chi(P_i, P_j) \) are the phase factors.

Now evaluate this trace row by row starting from the last non-zero one. At each row, we take the overlap with the zero state and use that only Pauli operators composed of \( I \) and \( Z \) have nonzero expectation on \( \ket{0} \). Each such term contributes \( \bra{0} P \ket{0} = 1 \), and only \( d \) of the \( d^2 \) Pauli matrices can contribute per index. Thus, for each row, we gain a factor \( d \), and since there are \( m \) rows, the total is \( d^m \), exactly canceling the \( 1/d^m \) prefactor. Note that the phase factors were all $1$ (regardless of $M$) since all Pauli composed only by Identity and $Z$ commute with each other. Hence, we conclude.
\end{proof}

\subsection{Alternative basis for the commutant: Graph and algebraically independent Pauli monomials}\label{sec:mhosandstuff}

In this section, we define two alternative representations of Pauli polynomials, namely graph-based Pauli monomials and algebraically independent Pauli monomials (which we already introduced in \cref{sec:comclif}). Here, instead of summing over all Pauli operators, we restrict the sum to only consider Pauli operators with a certain anticommutation relation and/or Pauli operators that are algebraically independent (i.e., cannot be written as products of each other). These representations have the advantage that they form Hilbert-Schmidt orthogonal operators, which is a useful property for twirling and the Weingarten calculus. They also have drawbacks, most notably, they do not factorize on qubits, meaning that the structure is way more dependent on the number of qubits. In this section, first, we define the different representations and show that they can be bijectively transformed into each other. 
\vspace{1em}
\subsubsection{Definition of alternative basis}
Let us start by introducing the definition of Graph-based Pauli monomial. We note that we use different scaling conventions here, as this allows for a better relation between various bases of the commutant.
\begin{definition}[Graph-based Pauli monomials]
\label{def:graphbasedpauli}
Let $ G \in \symf $, and $ V \in\even $. The graph-based Pauli monomials are defined as
\begin{equation}
\mho(V, G) \coloneqq \frac{1}{d^m} \sum_{\substack{\boldsymbol{P} \in \mathbb{P}_n \\ \mathcal{A}(\boldsymbol{P}) = G}} 
\prod_{j=1}^m P_j^{\otimes v_j},
\end{equation}
where $ \mathcal{A}(\boldsymbol{P}) $ denotes the adjacency matrix of the anticommutation graph (as defined in \cref{def:anticommgraph}) associated with the set $ \{P_1, \ldots, P_m\} $.  We denote the set of graph-based Pauli monomials as $\hat{\mathcal{P}}$, to reflect the fact that the two sets are connected by a binary Fourier transform, as we will show later in \cref{lem:linear-transformation}. 
\end{definition}

An equivalent formulation of $ \mho(V, G) $ is
\begin{align}
    \label{eq:alternativewaygraphbasedmonomials}
\mho(V, G) = \frac{1}{d^m} \sum_{\boldsymbol{P} \in \mathbb{P}_n} 
\prod_{j=1}^m P_j^{\otimes v_j} \times \left( \prod_{\substack{i, j \in [m] \\ i < j}} \frac{1 + (-1)^{G_{i,j}} \chi(P_i, P_j)}{2} \right)
\end{align}
which can be verified as the latter term gives $\delta_{\mathcal A(\boldsymbol P)=G}$.
We conclude our set of definitions on graph-based Pauli monomials, by introducing the notion of algebraically independent graph-based Pauli monomials.

\begin{definition}[Independent graph-based Pauli monomials]\label{def:indeppaulimonomials} 
Let $ G \in \symf $ be a symmetric binary matrix, and $ V \in \even $. We define independent graph-based Pauli monomials as

  \begin{align}
        \mho_I(V,G)&= \frac{1}{|S_{[V,G]}|}\sum_{\substack{\mathbf{P} \in \mathbb{P}_n^{\times m} \\ \mathcal{A}(\mathbf{P}) = G,\\\mathbf P \textrm{ alg. ind. }}} 
\prod_{j=1}^m P_j^{\otimes v_j}
    \end{align}
i.e., the sum runs over the set of algebraically independent Pauli operators $P=(P_1,\ldots, P_m)$.
Notice that $\mho_I(V,G)\propto\mho_I([V,G])$ holds, introduced in \cref{sec:comclif}. In particular, the proportionality factor is a phase $\{\pm1,\pm i\}$ that is uniquely determined by a choice of the representative $V$.
\end{definition}

\subsubsection{Transformations of different Pauli monomials}\label{Sec:invertiblemapspaulimonomials}
Thus far, we have introduced various notions of Pauli monomials. Since the set of monomials remains the same, a transformation between these definitions must exist. In the following lemmas, we show that such transformations do exist.
\begin{lemma}[Linear transformation between Pauli and graph-based monomials]
\label{lem:linear-transformation}
Let $ \Omega(V, M) $ denote a Pauli monomial (\cref{def:paulimonomials}) and $ \mho(V, G) $ a graph-based Pauli monomial (\cref{def:graphbasedpauli}). The following `Fourier transform' expresses the linear transformation between these two representations:
\begin{align}
    \Omega(V, M) &= \sum_{\substack{G \in \symftwo}} (-1)^{\sum_{i<j} G_{i,j} M_{i,j}} \mho(V, G), \label{eq:fromgraphtomonomials} \\
    \mho(V, G) &= \frac{1}{2^{\binom{m}{2}}}\sum_{\substack{M \in \symftwo}}(-1)^{\sum_{i<j} G_{i,j} M_{i,j}} \Omega(V, M). \label{eq:frommonomialstograph}
\end{align}
Here, the summation in both equations runs over all symmetric $ m \times m $ matrices $ M $ or $ G $ with zero diagonal entries.
\end{lemma}

\begin{proof}
To prove the lemma, we verify the validity of the transformations. We first confirm \cref{eq:fromgraphtomonomials} as
\begin{align}
    \sum_{\substack{G \in \symftwo : \\ \mathrm{diag}(G) = 0}} (-1)^{\sum_{i<j} G_{i,j} M_{i,j}} \mho(V, G)
    &\eqt{(i)}\!\!\!\!\sum_{\substack{G \in \symftwo}} \sum_{\substack{\mathbf{P} \in \mathbb{P}_n^{\times m}}} \prod_{b=1}^{m} P_b^{\otimes v_b}\!\!\! \prod_{\substack{i, j \in [m] \\ i < j}} \frac{1 + (-1)^{G_{i,j}} \chi(P_i, P_j)}{2} (-1)^{G_{i,j} M_{i,j}} \\
    \nonumber
    &=\!\!\!\!\!\! \sum_{\substack{\mathbf{P} \in \mathbb{P}_n^{\times m}}} \prod_{b=1}^{m} P_b^{\otimes v_b} \sum_{\substack{G \in \symftwo}} \prod_{\substack{i, j \in [m] \\ i < j}} \frac{1 + (-1)^{G_{i,j}} \chi(P_i, P_j)}{2} (-1)^{G_{i,j} M_{i,j}} \\
    \nonumber
    &\eqt{(ii)} \sum_{\substack{\mathbf{P} \in \mathbb{P}_n^{\times m}}} \prod_{b=1}^{m} P_b^{\otimes v_b} \prod_{\substack{i, j \in [m] \\ i < j}} \frac{1 + \chi(P_i, P_j) + (-1)^{M_{i,j}}(1 - \chi(P_i, P_j))}{2} \\
    \nonumber
    &\eqt{(iii)} \sum_{\substack{\mathbf{P} \in \mathbb{P}_n^{\times m}}} \prod_{b=1}^{m} P_b^{\otimes v_b} \prod_{\substack{i, j \in [m] \\ i < j}} \chi(P_i, P_j)^{M_{i,j}}
    \nonumber\\
    &= \Omega(V, M).
    \nonumber
\end{align}
In step (i), we have used \cref{eq:alternativewaygraphbasedmonomials}. In step (ii), we 
have employed the  identity
\begin{align}
    \sum_{\substack{G \in \symftwo : \\ \mathrm{diag}(G) = 0}} \prod_{\substack{i, j \in [m] \\ i < j}} \left( 1 + (-1)^{G_{i,j}} \chi(P_i, P_j) \right) (-1)^{G_{i,j} M_{i,j}} 
    &= \prod_{\substack{i, j \in [m] \\ i < j}} \sum_{G_{i,j} \in \{0,1\}} \left( 1 + (-1)^{G_{i,j}} \chi(P_i, P_j) \right) (-1)^{G_{i,j} M_{i,j}} \\
    \nonumber
    &= \prod_{\substack{i, j \in [m] \\ i < j}} \left( \left( 1 + \chi(P_i, P_j) \right) + \left( 1 - \chi(P_i, P_j) \right) (-1)^{M_{i,j}} \right).
\end{align}
Finally, in step (iii), we have used  the case distinction
\begin{align}
   \left( 1 + \chi(P_i, P_j) \right) + \left( 1 - \chi(P_i, P_j) \right) (-1)^{M_{i,j}}
   & = \begin{cases}
        2 & \text{if } M_{i,j} = 0, \\
        2\chi(P_i, P_j) & \text{if } M_{i,j} = 1,
    \end{cases}
    \\&= 2 \chi(P_i, P_j)^{M_{i,j}}.
\end{align}
Similarly, we verify \cref{eq:frommonomialstograph}
as
\begin{align}
    \frac{1}{2^{\binom{m}{2}}}\!\!\sum_{\substack{M \in \symftwo : \\ \mathrm{diag}(M) = 0}} \!\!\!\!\!\!\!\!\!\Omega(V, M) \times (-1)^{\sum_{i<j} G_{i,j} M_{i,j}} 
    \!\!&=\!\! \frac{1}{2^{\binom{m}{2}}}\!\sum_{\substack{M \in \symftwo : \\ \mathrm{diag}(M) = 0}} \frac{1}{d^m} \sum_{\mathbf{P} \in \mathbb{P}_n^{\times m}}\prod_{b=1}^{m} P_b^{\otimes v_b} \!\!\!\prod_{\substack{i, j \in [m] \\ i < j}} \chi(P_i, P_j)^{M_{i,j}} (-1)^{G_{i,j} M_{i,j}} \\
    \nonumber
    &= \frac{1}{d^m} \sum_{\mathbf{P} \in \mathbb{P}_n^{\times m}} \prod_{b=1}^{m} P_b^{\otimes v_b}\!\!\! \sum_{\substack{M \in \symftwo : \\ \mathrm{diag}(M) = 0}} \prod_{\substack{i, j \in [m] \\ i < j}} \frac{\chi(P_i, P_j)^{M_{i,j}} (-1)^{G_{i,j} M_{i,j}}}{2} \\
    \nonumber
    &= \frac{1}{d^m} \sum_{\substack{\mathbf{P} \in \mathbb{P}_n^{\times m}}} \prod_{b=1}^{m} P_b^{\otimes v_b}\!\!\! \prod_{\substack{i, j \in [m] \\ i < j}} \frac{1 + (-1)^{G_{i,j}} \chi(P_i, P_j)}{2} \\
    &= \mho(V, G),
    \nonumber
\end{align}
where in the final step, we have again used \cref{eq:alternativewaygraphbasedmonomials}. This completes the proof.
\end{proof}
\begin{lemma}\label{lemm:bijectiveommho}
The sets of operators \( \{ \mho_I(V, G) \} \) and \( \{ \mho(V, G) \} \) span the same vector space. In particular, there exists a bijective linear transformation between the two families.
\end{lemma}
\begin{proof}
First, we observe that the definition of graph-based monomials in \cref{def:graphbasedpauli} and the definition of independent graph-based monomials in \cref{def:indeppaulimonomials} are identical (up to a constant), except for the additional requirement in the latter that the set of Pauli operators must be algebraic independent. Consequently, starting from $\mho(V, G)$, we can write  
\begin{align}
\mho(V, G) 
= \frac{|S_{[V, G]}|}{d^m} \mho_I(V, G) 
+ \frac{1}{d^m} \sum_{\substack{\boldsymbol{P} \in \mathbb{P}_n^{\times m} \text{ not alg. ind.} \\ \mathcal{A}(\boldsymbol{P}) = G}} \prod_{i \in [m]} P_i^{\otimes v_i}.
\end{align}

Let us analyze the second term. If the set of Pauli operators $\boldsymbol{P}$ is not algebraically independent, there exists at least one linear relation within the set $\boldsymbol{P}$. As such, we can express $P_k = \theta \prod_{i \neq k} P_i^{a_i}$, where $\theta \in \{\pm i, \pm 1\}$. Notably, $\theta^2$ depends only on the anticommutation graph $G$ (as we show in \cref{le:trivialPQ}) of $\boldsymbol{P}$ and, since $V \in \even$, only $\theta^2$ appears. We can, therefore, rewrite
    \begin{align}
       \prod_{i\in [m]} P_i^{\otimes v_i} =\phi(G)\prod_{i\in [m]\backslash k} P_i^{\otimes (v_i+a_i v_k)}
    \end{align}
    where the proportionality factor depends on the anticommutation graph $G$ originating from permuting the Pauli operators in the product to the correct spot. We can reiterate the argument in the presence of more than one algebraic relations. We can,  therefore, write
    \begin{align}
        \mho(V,G)=\frac{\Phi_V |S_{[V,G]}|}{d^m}\mho_I(V,G) +\sum_{a\in \text{lin relations}} \frac{\phi_a(G)|S_{[V_a,G_a]}|}{d^{m}}\mho_I(V_a,G_a)\label{eq:linearrelations}
    \end{align}
    where $m_a<m$ and $\phi_a(G)$ only depends on the anticommutation graph $G$. \cref{eq:linearrelations} proves that there exists a map $\{\mho_I(V,G)\}\rightarrow \{\mho(V,G)\}$. This map is also bijective. To prove it, we note that for $m=0$, we have that $\mho_I(0,0)=\mho(0,0)=\mathbb{1}$. Then, by induction, we simply note that if one can generate all $\mho_I(V,G)$ from $\mho(V,G)$ for $m=m'-1$, we can invert \cref{eq:linearrelations} and get all the $\mho_I(V,M)$ for $m=m'$:
    \begin{align}
        \frac{\Phi_V|S_{[V,G]}|}{d^m}\mho_I([V,G]) =\mho(V,G)-\sum_{a} \frac{\phi_a(G)|S_{[V_a,G_a]}|}{d^{m}}\mho_I(V_a,G_a).
    \end{align}
    This shows that we can generate all the set $\{\mho_{I}(V,G)\}$ from $\hat{\mathcal{P}}$, and, therefore,  the map is bijective.
\end{proof}

\begin{corollary}\label{cor:spanningeverything} The sets $\mathcal{P}$, $\hat{\mathcal{P}}$,$\{\mho_I([V,G])\}_{[V,G]}$ generate the same vector space, i.e., 
\be
\operatorname{span}(\mathcal{P})=\operatorname{span}(\hat{\mathcal{P}})=\operatorname{span}(\{\mho_I([V,G])\}_{[V,G]}).
\ee
\begin{proof}
    The proof follows from the fact that there exists a bijective correspondence between $\mathcal{P}$ and $\hat{\mathcal{P}}$ by \cref{lem:linear-transformation}, and similarly, a bijective correspondence between $\hat{\mathcal{P}}$ and $\{\mho_I(V,G)\}$ by \cref{lemm:bijectiveommho}. Moreover, we have that $\mho_I(V,G)=\Phi_V\mho_I([V,G])$ for a choice of representative $V$.
\end{proof}    
\end{corollary}

\subsection{Pauli monomials form a basis of the commutant}
\label{sec:commutantCliffordgen}
In this section, building on the results of \cref{sec:Paulimon,sec:comclif}, we demonstrate that Pauli monomials form a basis for the Clifford commutant.

The next corollary shows that the set of reduced Pauli monomials as defined in \cref{def:distinctpaulimonomials} span all the commutant of the Clifford group, and are independent provided that $k\le n+1$.

\begin{corollary}[Reduced Pauli monomials form a basis of the commutant]\label{lem:linearindependency}
The span of set of reduced Pauli monomials $ \mathcal{P}$ is the commutant of the Clifford group
\begin{align}
\operatorname{span}(\mathcal{P})=\com(\mathcal{C}_n^{\otimes k});    
\end{align}
moreover, if $k\le n+1$, then $\mathcal{P}$ contains linearly independent operators and thus $\mathcal{P}$ provides a basis of $\com(\mathcal{C}_n^{\otimes k})$.
\end{corollary}
\begin{proof}
Both claims follow directly from \cref{cor:spanningeverything,th:fullcommutantnk}, as together they establish a bijective transformation between the set of operators $\mho_{I}([V,G])$, which constitute a basis for $\rank_2(G)\ge 2(m-n)$ trivially (verified for $k\le n+1$). 
 \end{proof}

\section{The Clifford commutant is generated by permutations and three fundamental operators}
\label{sec:3gen}
\subsection{The Clifford commutant is generated by only three additional elements}

In this section, building on the formalism of Pauli monomials introduced in \cref{sec:Paulimon}, we state and prove one of the main results of this work. The following theorem shows that the $k$-th order Clifford group commutant is generated by only three additional elements compared to the commutant of the full unitary group, namely the primitives corresponding to $\frac{1}{d} \sum_{P} P^{\otimes 4}$, $\frac{1}{d} \sum_{P} P^{\otimes 6}$, and, if $k$ is a multiple of $4$, also $\frac{1}{d} \sum_{P} P^{\otimes k}$. 
In this sense—as already observed for the specific case $k = 4$ in Ref.~\cite{zhu_clifford_2016}, which showed that the Clifford group fails ``gracefully'' to be a unitary 4-design—we show that the Clifford group fails gracefully to be a unitary $k$-design for \textit{any} $k$, thereby providing a natural generalization of this result.

\begin{theorem}[The $k$-th order Clifford group commutant is generated by only three additional elements]\label{th:algebraicstructurecommutantofclifford}
Let $\com(\mathcal{C}_n^{\otimes k})$ be the commutant of the Clifford group. 
\begin{enumerate}[label=(\roman*)]
    
\item All Pauli monomials are generated by 
    \begin{equation}
\Omega_4\coloneqq\sum_{P \in \mathbb{P}_n} P^{\otimes 4}, \quad 
\Omega_6\coloneqq\sum_{P \in \mathbb{P}_n} P^{\otimes 6}, \quad \text{and, if $ k $ is divisible by 4,} \quad 
\Omega_k\coloneqq\sum_{P \in \mathbb{P}_n} P^{\otimes k},
\end{equation}
and permutations.

\item All unitary Pauli monomials are generated by $\Omega_6$ and permutations.
\item Every reduced Pauli monomial $\Omega$ can be generated by the product of at most $\min(m(\Omega)+4,k+1)$ many primitives.
\end{enumerate}
\begin{proof}
First of all, let us notice that, with $\Omega_4,\Omega_6$ and $\Omega_k$ and permutations, we can generate every primitive $ \Omega(v)$ with $|v|=4,6,k$. 
Then, thanks to the result of \cref{lem:normalform}, we know that any reduced Pauli monomial $\Omega$ can be written as $\Omega=\Omega_P\Omega_U$ with $\Omega_U\in\mathcal{P}_U$ and $\Omega_P\in\mathcal{P}_P$, a unitary and projective Pauli monomial, see \cref{le:unitaryprojmon}. This result automatically implies that the commutant of the Clifford group is generated by products of primitives $\Omega(v)$ with $|v|\in[k]$. We are therefore just left to show that we can generate every primitive Pauli monomial by $\Omega_{4}$ and $\Omega_{6}$ and permutations. First, we show that we can generate $\Omega(v')$ from $\Omega(v)$ with $|v'|=|v|+4$ and two $\Omega_{6}$ (up to the adjoint action of permutations), using one extra register, as the following diagrammatic calculation shows:
    \begin{align}
    \monomialdiagram{9}{{1,2,3,4,5,6},{6,7,8,9},{1,2,3,4,5,6}}{}=\monomialdiagram{9}{{1,2,3,4,5,6},{1,2,3,4,5,7,8,9},{1,2,3,4,5,6}}{0:1}=\monomialdiagram{9}{{1,2,3,4,5,7,8,9},{1,2,3,4,5,6},{1,2,3,4,5,6}}{}=\monomialdiagram{9}{{1,2,3,4,5,7,8,9},}{}.
\end{align}
Since, in the above diagrammatic calculation, we only use the property that the middle primitive anticommmutes with the two $\Omega_{6}$, we can generate all $\Omega(v)$ with $|v|=0\pmod2$ from $\Omega_4$, $\Omega_6$, and permutations using one extra register. This only leaves $\Omega_{k}$ if $k$ is even, which cannot be generated by this method, because there is no extra register to be used, and it commutes with all other primitives. 

Let us see how to generate $\Omega_{k}$. It turns out that there are two cases.
\begin{itemize}
    \item If $k=4\mathbb{N}$, then the primitive $\Omega_{k}$ is a projector (see \cref{lem:propertiesprimitivepauli}) and, as we show below, it cannot be generated by lower order primitives. 
To proceed with the proof, let assume, towards contradiction, that we can create $\Omega_{v_k}$ by lower primitives as
\begin{align}
    \Omega_k\propto\prod_m \Omega(v_m)
\end{align}
where $|v_m| < k$, and the proportionality factor arises from the fact that the product of lower-order primitives may not be reducible (cf.\  \cref{def:reducedpaulimonomials}). Using \cref{th:graphrules}, we then move all unitary primitives to the right and result in the following expression:
\begin{align}
    \Omega_k\propto\prod_m \Omega_{P}(v'_m)\times \prod_m \Omega_{U}(v''_m)
\end{align}

where $\Omega_{P}(v_{m}') \in \mathcal{P}_P$ and $\Omega_{U}(v_{m}'') \in \mathcal{P}_U$ for every $m$. We note that this process cannot generate an $\Omega_k$ primitive, since $\Omega_k$ commutes with everything. Hence we have $|v_m'|<k$ for every $v_{m}'$. Consequently, we have $\Omega_k = \Omega_P \times \Omega_U$, where $\Omega_P$ is the product of commuting non-full-weight primitives. Notably, $\Omega_P$ may be non-reduced, meaning that $\Omega_P = d^\alpha \Omega'_P$, as stated in \cref{lem:everypaulimonomialcanbereduced}. Thus, we obtain
\be
\Omega_k \propto \Omega_P' \Omega_U.
\ee
Now, multiplying both sides by the conjugate—specifically, $\Omega_k$ on the left and $(\Omega_P' \Omega_U)^{\dag}$ on the right—we arrive 
at
\be
d\Omega_k \propto d^{\alpha} \Omega_P',
\ee
where we have used the fact that $\Omega_U$ is a unitary operator. However, 
since both $\Omega_k$ and $\Omega_P'$ are reduced, it follows that $\Omega_k = \Omega_P'$, which contradicts our initial assumption that $|v_m| < k$. This establishes the contradiction.
\item If $k = 4\mathbb{N} + 2$, then $\Omega_{k}$ is unitary and can be generated as a product of lower-order primitives. Therefore, it lies in the algebra generated by, $\Omega_{4}$, $\Omega_{6}$, and permutations as shown by the diagrammatic calculation.
\begin{align}
    \monomialdiagram{10}{{1,2,3,4,5,6,7,8,9,10},{1,2},{3,4},{5,6}}{}=\monomialdiagram{10}{ {5,6,7,8,9,10},{3,4,5,6,7,8,9,10},{1,2,5,6,7,8,9,10},{5,6}}{1:2}
    =\monomialdiagram{10}{{5,6,7,8,9,10},{3,4,7,8,9,10},{1,2,7,8,9,10},{1,2,3,4,5,6}}{}
\end{align}
The above only applies the rules of \cref{th:graphrules}. The weight of the first primitive can be arbitrary, as long as it is unitary.
\end{itemize}
The above reasoning thus shows items (i) and (ii).

Let us prove item (iii). For the minimal number of terms, we proceed as in the normal form proof, see \cref{lem:normalform}. Hence, we can write every monomial as
\begin{align}
    \Omega=\Omega_P\tilde \Omega \Omega_U\label{eq:148}
\end{align}
where $\Omega_P\in\mathcal{P}_P$, i.e., is a multiplication of even commuting primitives with no phase, $\Omega_{U}\in\mathcal{P}_U$ i.e., is a product of odd primitives and  $\tilde \Omega$ is a monomial of commuting projective primitives, which share a phase. The form in \cref{eq:148} is obtained right before Step 3 of the proof of \cref{lem:normalform}. We note that, through \cref{th:graphrules}, we can always write $\tilde{\Omega}$ as products of \textit{pairs} of commuting projective primitives with a phase. An example of this kind is provided below.
\setmonomialscale{3.5mm}
\begin{align} 
\monomialdiagram{6}{{1,2,3,4},{3,4,5,6},{1,2,4,6},{1,3,5,6}}{0:1,0:2,2:3}
 = \monomialdiagram{6}{{1,2,5,6},{3,4,5,6},{1,2,4,6},{1,3,5,6}}{0:1,2:3}\,.
\end{align}
Crucially, to obtain the product form in \cref{eq:148}, the number $r$ of primitives is $r = m(\Omega)$, which corresponds to the order of the reduced monomial $\Omega$. Indeed, is only after Step 3 of \cref{lem:normalform} that the number of primitives increases. Indeed, Step 3 consists of adding commuting permutations and then expanding the product of two even-commuting primitives sharing a phase into four unitaries. However, this approach is suboptimal in the present case, as our goal is to develop a strategy that minimizes the number of primitives. To address this, let us carefully analyze the problem. First, we identify two possible scenarios that may arise.
\begin{enumerate}[label=(\alph*)]
   \item \textbf{$\Omega_{U} \neq \mathbb{1}$}. First, we observe that if there exists a primitive $\Omega(u)$ in the decomposition of $\Omega_U$ that commutes with every element of every pair in $\tilde{\Omega}$, then, by Rule (4) of \cref{th:graphrules}, we can decompose $\tilde{\Omega}$ into a product of unitaries, resulting in $r = \tilde{m}$ with no additional cost in terms of primitives.  

We can then assume that there exists an $\Omega(u)$ in the decomposition of $\Omega_U$ that does not commute 
with at least one element of some pairs in $\tilde{\Omega}$. Note that if ${\Omega}(u)$ does not commute with both elements of a given pair, we can sum the elements of the pair, ensuring that only one of them fails to commute with $\Omega(u)$. Suppose that there are at least two pairs whose elements do not commute with $\Omega(u)$, and, without loss of generality, consider the first element of each pair as the non-commuting one. In this case, summing the first elements of each pair together and the second elements of each pair together results in one completely commuting pair and one pair where only the first element does not commute. This implies that, without loss of generality, we can assume that in $\tilde{\Omega}$ there is exactly one pair of even-commuting primitives sharing a phase that does not commute with $\Omega(u)$ (or any other odd primitive in $\Omega_U$). If this holds, we simply apply Rule (4) of \cref{th:graphrules}, adding two commuting permutations and generating four unitaries from the non-commuting pair, while using $\Omega(u)$ to unfold all other pairs at no additional cost. Consequently, this case results in a total number of primitives $r = m(\Omega) + 2$.

\item $\Omega_{U} = \mathbb{1}$. In this case, we apply Rule (4) of \cref{th:graphrules}, adding two permutations to unfold, without loss of generality, the first pair of even-commuting primitives in $\tilde{\Omega}$, thereby creating four unitaries in the process. At this stage, $\Omega_U \neq \mathbb{1}$, allowing us to proceed as in case (i). This procedure ultimately results in a total number of primitives $r = m(\Omega) + 4$.

\end{enumerate}
Let us then bound the total number of primitives by considering two distinct scenarios.

In the first scenario (A), in the worst case, we have $m(\Omega) = k - 1$, and therefore the number of primitives satisfies $r \le (k - 1) + 2 = k + 1$.

In the second scenario (B), since $\Omega_U = \mathbb{1}$ from the outset, the worst case corresponds to having all commuting primitives contained in $\Omega_P$ and $\tilde{\Omega}$. However, the maximum number of commuting primitives is $\frac{k}{2}$, and thus $r \le \frac{k}{2} + 4 \le k$ for any $k \ge 8$. Item (iii) thus follows.

For the second scenario (B), we can also construct an example that cannot be generated using only $m(\Omega) + 2$ generators unless we extend $k \rightarrow k + 2$. Such an example can be constructed for $k = 8$ to be
\begin{align}
    \monomialdiagram{8}{{1,5,6,7},{2,5,7,8},{3,5,6,8},{4,6,7,8}}{0:1,2:3}.
\end{align}
Here, no unitary primitive commutes with all four operators. Consequently, constructing it requires the product of four additional primitives. 
\end{proof}
\end{theorem}

Let us point out that, due to \cref{lem:linearindependency}, \cref{th:algebraicstructurecommutantofclifford} provides not only a set of generators for the set of Pauli monomials but also a set of generators for the Clifford commutant.

\subsection{Equivalence of basis}\label{ssec:equi_basis}
In this section, building on \cref{th:algebraicstructurecommutantofclifford}, we provide a brief argument supporting the claim that the basis of Pauli monomials, introduced in \cref{def:paulimonomials}, precisely coincides with the basis found in Ref.~\cite{gross_schurweyl_2019}. Let us first recall the main result given in Ref.~\cite{gross_schurweyl_2019}.
\begin{proposition}[Commutant basis of Ref.~\cite{gross_schurweyl_2019}] Let $T\subseteq\mathbb{F}_{2}^{k}\times \mathbb{F}_{2}^{k}$ be a subspace. Labelling the elements of $T$ as $(x,y)$ for $x,y\in\mathbb{F}_2^k$, $T$ is a Lagrangian subspace if the following conditions are met:
\begin{enumerate}[label=(\roman*)]
    \item $|x|=|y|\pmod 4$ for all $(x,y)\in T$;
    \item $\dim(T)=k$;
    \item $(1,1,\ldots,1)\in T$.
\end{enumerate}
For every $n\ge k-1$, the following operators 
\be
R(T)=\left(\sum_{(x,y)\in T}\ketbra{x}{y}\right)^{\otimes n}
\ee
form a basis of the Clifford commutant $\com(\mathcal{C}_n^{\otimes k})$. 
\end{proposition}

In the next lemma, we show that the above mentioned basis corresponds to the basis of Pauli monomials.

\begin{lemma}[Equivalence of basis] The basis $\{R(T)\}$ corresponds to the set $\mathcal{P}$.
\begin{proof}
    To prove this result, we proceed in a few steps. First, in the proof of Lemma 4.25 of Ref.~\cite{gross_schurweyl_2019}, it has been established that the product of any two operators satisfies $R(T_1)R(T_2) \propto R(T)$ for some Lagrangian subspace $ T $. 

Next, the primitive Pauli monomials $ \Omega({v_m}) $ (see \cref{def:primitivepaulimonomials}) correspond to what the authors of Ref.~\cite{gross_schurweyl_2019} define as the \textit{ antiidentity} $ \bar{\mathbb{1}}_{m} $ (see, e.g., Eq.~(1.4) in Ref.~\cite{gross_schurweyl_2019}). Hence, we have the relation $    R(\bar{\mathbb{1}}_m) = \Omega({v_m})$. In \cref{th:algebraicstructurecommutantofclifford}, we showed that at most three fundamental primitive Pauli monomials (up to permutations), namely $ \Omega_{4} $, $ \Omega_{6} $, and $ \Omega_{k} $, generate the set $ \mathcal{P} $. Since the set $ \{R(T)\} $ is closed under multiplication, it follows that the above primitives generate a subset of $ \{R(T)\} $. Therefore, we immediately obtain the inclusion $    \mathcal{P} \subseteq \{R(T)\}$. 

However, since the set $ \mathcal{P} $ contains linearly independent operators (for $ n \geq k-1 $) that span $ \operatorname{Com}(\mathcal{C}_n^{\otimes k}) $, as shown in \cref{lem:linearindependency}, and since the set $ \{R(T)\} $ also contains linearly independent operators (for $ n \geq k-1 $) that span $ \operatorname{Com}(\mathcal{C}_n^{\otimes k}) $, we must conclude that $\mathcal{P} = \{R(T)\}$.

\end{proof}
\end{lemma}

\section{Applications and examples}\label{sec:applications}
This section explores several applications and examples based on the Clifford commutant. We begin by introducing the Haar average over the Clifford group in \cref{Sec:haaraveragecliffordbeyond}. Next, we establish a connection between magic state resource theories and the Clifford commutant in \cref{Sec:magicstateresourcetheory}. We then demonstrate how the Clifford commutant serves as a useful tool to characterize the property testing of stabilizer states in \cref{sec:propertytesting}. Following this, we extend our discussion to group-symmetric states, where the group we consider is the Clifford group. Finally, we conclude with a few examples in \cref{sec:cliffordcommutantexample}, considering the Clifford commutant for values of $k$ ranging from 4 to 8.

\subsection{Haar average over the Clifford group}\label{Sec:haaraveragecliffordbeyond}
In this section, we provide the essential tools to compute the averages of the Clifford group, one of the key applications of the Clifford group commutant.

\subsubsection{Twirling over the Clifford group}
Let us provide an explicit expression for the twirling operator (see \cref{twirlingofoperatorOunitary}) over the Clifford group, for any $n$ and $k$, as a direct  of \cref{th:fullcommutantnk}.

\begin{corollary}[Twirling operator over the Clifford group]\label{cor:twirlingcliffwithmhoi} Let $n,k>0$. Consider the twirling operator, defined in \cref{def:twirlingoperatoroverthecliffordgroup}, applied on $O\in\mathcal{B}(\mathcal{H}^{\otimes k})$. Then, $\Phi_{\cl}^{(k)}(O)$ belongs to $\com(\mathcal{C}_{n}^{\otimes k})$, and can be expressed as
\begin{align}
\Phi_{\cl}^{(k)}(O)=\sum_{\mho_{I}([V,G])} \frac{\tr[O\mho_{I}([V,G])]}{\|\mho_I([V,G])\|_2^{2}}\mho_I([V,G])\,,
\end{align}
where the $2$-norm of the operators $\mho_I([V,G])$ is computed in \cref{lem:orb_ofPauli}.
\begin{proof}
    Given that the twirled operator $\Phi_{\cl}^{(k)}(O)$ belongs to the commutant of the Clifford group, it can be expressed as a linear combination of basis elements. In particular, we can use the orthogonal basis formed by the operators $\mho_I([V,G])$ for $V\in\even$ and $G\in \symf$ with $\rank_2(G)\ge2(m-n)$, see \cref{th:fullcommutantnk}. Therefore, we can express the twirled operator as
    \begin{align}
\Phi_{\cl}^{(k)}(O)=\sum_{\mho_I([V,G])}b_{\mho_I([V,G])} \mho_I([V,G])\,.
    \end{align}
To determine the coefficients, one can multiply each side of the equation by $\mho_{I}([V',G'])$ and use the orthogonality of the coefficients to obtain
\be
b_{\mho_I([V',G'])}=\frac{\tr(\Phi_{\cl}^{(k)}(O)\mho_I([V',G']))}{\|\mho_I([V',G'])\|_{2}^{2}}\,;
\ee
noticing that $\tr(\Phi_{\cl}^{(k)}(O)\mho_I)=\tr(O\mho_I)$ concludes the proof.
\end{proof}
\end{corollary}

\subsubsection{Clifford-Weingarten calculus}\label{sec:cliffordweingartencalculus}
In \cref{cor:twirlingcliffwithmhoi}, we introduced a decomposition of the twirling operator using the operator basis $\mho_I([V,G])$. Here, we express the twirling super-operators in terms of the simpler basis of Pauli monomials. Based on such a decomposition, we introduce the notion of Clifford-Weingarten calculus, which plays a key role in analytical calculations involving the twirling operator over the Clifford group—analogous to the role of Weingarten calculus for the full unitary group; see \cref{cor:weingartencalculushaar}.

Consider the set of reduced Pauli monomials $\mathcal{P}$. We know that the twirling operator, because of \cref{lem:twirlingbelongstothecommutant}, can be written as a linear combination of elements of $\mathcal{P}$
\be
\Phi^{(k)}_{\cl}(O)=\sum_{\Omega\in \mathcal{P}}c_{\Omega}(O)\Omega,
\label{clwingarteneq1}
\ee
where we drop the explicit dependence on the matrix $M$ and  the $m$-dimensional vector space $V$, introduced in \cref{def:paulimonomials}. 

\begin{lemma}[Clifford-Weingarten calculus] \label{lem:weingartencalculusclifford} Let $\Phi_{\cl}^{(k)}(O)$ be the $k$-fold twirling operator $O\in\mathcal{B}(\mathcal{H}^{\otimes k})$ and let $\mathcal{P}$ be the set of (reduced) Pauli monomials. Then it reads
\be
\Phi_{\cl}^{(k)}(O)=\sum_{\Omega,\Omega'\in\mathcal{P}}(\mathcal{W}^{-1})_{\Omega,\Omega'}\tr(O\Omega)\Omega'
\ee
where we call the coefficients $(\mathcal{W}^{-1})_{\Omega,\Omega'}$ Clifford-Weingarten functions, which can be obtained as the (pseudo-)inverse of the Gram-Schmidt matrix $\mathcal{W}_{\Omega,\Omega'}\coloneqq\tr(\Omega^{\dag}\Omega')$.
\begin{proof}
   
Similarly to what done in \cref{cor:weingartencalculushaar}, we multiply both sides of \cref{clwingarteneq1} with $\Omega^{\dag}\in \mathcal{P}$ and trace both sides:
\be
\tr(\Omega^{\dag}\Phi^{(k)}_{\cl}(O))=\sum_{\Omega'\in \mathcal{P}}c_{\Omega'}(O)\tr(\Omega^{\dag}\Omega').
\ee
Noticing that the channel $\Phi_{\cl}^{(k)}(\cdot)$ is self-adjoint and acts trivially on $\Omega^{\dag}\in\mathcal{P}$, we can rewrite the coefficients $c_{\Omega}(O)$ as
\be
c_{\Omega^{\prime}}(O)=\sum_{\Omega\in\mathcal{P}}(\mathcal{W}^{-1})_{\Omega,\Omega'}\tr(\Omega^{\dag} O),
\ee
where we have defined the (symmetric) Gram-Schmidt matrix $\mathcal{W}_{\Omega,\Omega'}\coloneqq\tr(\Omega\Omega^{\prime\dag})$.
\end{proof}
\end{lemma}
The Gram matrix is defined as the matrix generated by Hilbert-Schmidt inner products between the operators in the reduced Pauli monomial set \( \mathcal{P} \). Therefore, the Gram matrix is invertible if and only if the operators in \( \mathcal{P} \) are linearly independent. This condition holds if and only if \( k \leq n + 1 \), as established in \cref{lem:linearindependency}.

The following lemma explores the properties of the matrix $\mathcal{W}$\footnote{A closely related analysis of the commutant’s Gram matrix was performed in Lemmas 11 and 12 of~\cite{haferkamp2020QuantumHomeopathyWorks}}.

\begin{lemma}[Properties of the Gram-Schmidt matrix $\mathcal{W}$]\label{lem:propertiesofW} Let $\mathcal{P}$ be the set of reduced Pauli monomials and $|\mathcal{P}|=\dim(\com(\mathcal{C}_n^{\otimes k}))$ be the dimension of the $k$-th order commutant of the Clifford group. The matrix $\mathcal{W}$, with components $\mathcal{W}_{\Omega,\Omega'}=\tr(\Omega^{\dag}\Omega^{\prime})$, obeys the following properties:
\begin{enumerate}[label=(\roman*)]
    \item $\mathcal{W}$ is symmetric.
    \item $1\le \mathcal{W}_{\Omega,\Omega'}\le d^{k-1} ,\quad \Omega\neq \Omega'\in\mathcal{P}$.
    \item $ \mathcal{W}_{\Omega,\Omega}= d^k,\quad  \Omega\in \mathcal{P}$.
    \item $d^k-|\mathcal{P}|d^{k-1}\le\lambda(\mathcal{W})\le d^k+|\mathcal{P}|d^{k-1}$ where $\lambda(\mathcal{W})$ is any eigenvalue of $\mathcal{W}$.
    \item $\det(\diag\mathcal{W})\left(1-\frac{|\mathcal{P}|^2}{d}\right)\le \det(\mathcal{W})\le \det(\diag\mathcal{W})\left(1+\frac{2|\mathcal{P}|^2}{d}\right)$, if $n\ge k^2-3k+13$.
\end{enumerate}
\begin{proof} To prove item (i), it is sufficient to notice that $\tr(\Omega^{\dag}\Omega')$ is real, hence $(\mathcal{W}^{T})_{\Omega,\Omega'}=\mathcal{W}_{\Omega'\Omega}=\tr(\Omega^{\prime\dag}\Omega)=\tr(\Omega^{\dag}\Omega')^{*}=\tr(\Omega^{\dag}\Omega')$. Item $(ii),(iii)$ follows from \cref{lem:hilbertschmidtproductofpaulimonomials}. 
From points $(ii)$ and $(iii)$, we can lower and upper bound all the eigenvalues of $\mathcal{W}$. Denoting $\lambda(\mathcal{W})$ any eigenvalue of $\mathcal{W}$, we have, through the Gershgorin circle theorem~\cite{bhatia_matrix_1997}
\be
\min_\Omega\Big(\mathcal{W}_{\Omega,\Omega}-\sum_{\Omega'\neq \Omega}\mathcal{W}_{\Omega,\Omega'}\Big)\le \lambda(\mathcal{W})\le \max_\Omega\Big(\mathcal{W}_{\Omega,\Omega}+\sum_{\Omega'\neq \Omega}\mathcal{W}_{\Omega,\Omega'}\Big)
\ee
bounding $\sum_{\Omega'\neq \Omega}\mathcal{W}_{\Omega,\Omega'}\le |\mathcal{P}|d^{k-1}$, item $(iv)$ readily follows. 

Let us show item $(v)$. Denoting $\operatorname{diag}(\mathcal{W})$ the diagonal part of $\mathcal{W}$, and $\mathcal{W}'\coloneqq\operatorname{diag}(\mathcal{W})^{-1}\mathcal{W}-I$, where $I$ denotes the $|\mathcal{P}|\times |\mathcal{P}|$ identity matrix. We can write $\det(\mathcal{W})=\det(\diag(\mathcal{W}))\det(I+\mathcal{W}')$.  Let us bound the eigenvalues of $\mathcal{W}'$ by using the Gershgorin circle theorem.  Denoting $\lambda(\mathcal{W}')$ any eigenvalue of $\mathcal{W}'$, from item $(ii)$, it follows that 
\be
|\lambda(\mathcal{W}')|\le \max_\Omega\sum_{\Omega'\neq \Omega}\frac{\mathcal{W}_{\Omega,\Omega'}}{\mathcal{W}_{\Omega,\Omega}}\le \frac{|\mathcal{P}|}{d}.\label{eq:eigenvaluewprime}
\ee
Using the Bhatia inequality (Theorem VI.7.1 in Ref.\cite{bhatia_matrix_1997}), we have
\be
(1-|\lambda(\mathcal{W}')|)^{|\mathcal{P}|}\le \det(I+\mathcal{W}')\le (1+|\lambda(\mathcal{W}')|)^{|\mathcal{P}|}.
\ee
Let us now carefully upper and lower bound the two terms. We use $(1+x)^C\le 1+2xC$, valid for $x\le \frac{1}{2C}$, to get 
\be
(1+|\lambda(\mathcal{W}')|)^{|\mathcal{P}|}\le 1+2|\lambda(\mathcal{W}')||\mathcal{P}|\le 1+\frac{2|\mathcal{P}|^2}{d},\quad \text{if}\,\,2|\mathcal{P}|^2\le d\,.\label{eq:ineq1}
\ee
Similarly, we use $(1-x)^C\ge 1-Cx$ valid for $x\ge -1$ and $C>1$, and we get 
\be
\det(I+\mathcal{W}')\ge 1-\frac{|\mathcal{P}|^2}{d}.\label{eq:ineq2}
\ee
Hence, overall, for the inequalities \cref{eq:ineq1} and \cref{eq:ineq2} to hold, we need to require $d\ge 2|\mathcal{P}|^2$, for which a sufficient condition is $n\ge k^2-3k+13$ (where we have used the results of \cref{cor:dimcom}. We are left to compute $\det(\diag\mathcal{W})=\prod_{\Omega}\tr(\Omega\Omega)=d^{k|\mathcal{P}|}$. Putting it all together, this implies item $(v)$.
\end{proof}
\end{lemma}

The following theorem, on the other hand, examines the properties of the inverse, specifically those of the Clifford-Weingarten functions.

\begin{theorem}[Asymptotics of Clifford-Weingarten functions]\label{lem:asymptoticweingarten}
Let $\mathcal{W}$ the Gram-Schmidt matrix, and let $\mathcal{W}^{-1}$ be its inverse. Let $n\ge k^2-3k+13$. Then the properties  
   \begin{align}
   \left|(\mathcal{W}^{-1})_{\Omega,\Omega}-\frac{1}{d^k}\right|&\le \frac{6|\mathcal{P}|^2}{d^{k+1}}\,,\\
    |(\mathcal{W}^{-1})_{\Omega,\Omega'}|&\le \frac{5|\mathcal{P}|^2}{d^{k+1}}, 
    \end{align}
    hold true.
    \end{theorem}
    
\begin{proof}
By the definition of inverse, $(\mathcal{W}^{-1})_{\Omega,\Omega}=\frac{\det(\mathcal{W}^{/\Omega})}{\det(\mathcal{W})}$, where $\mathcal{W}^{/\Omega}$ denotes the $(|\mathcal{P}|-1)\times (|\mathcal{P}|-1)$ matrix obtained by $\mathcal{W}$ by erasing the row and the column corresponding to $\Omega$. If $n\ge k^2-3k +13 $, and using the same proof as for item $(v)$ of \cref{lem:propertiesofW}, it can be shown that 
\be
\det(\diag\mathcal{W}^{/i})\left(1-\frac{|\mathcal{P}|^2}{d}\right)\le\det(\mathcal{W}^{/i})\le\det(\diag\mathcal{W}^{/i})\left(1+\frac{2|\mathcal{P}|^2}{d}\right).
\ee
Notice that $\det(\diag\mathcal{W}^{/i})=d^{k(|P^{(k)}|-1)}$. Let us use item $(v)$ of \cref{lem:propertiesofW}, to upper and lower bound $(\mathcal{W}^{-1})_{\Omega,\Omega}$. First
\be\label{lem2upbound}
(\mathcal{W}^{-1})_{\Omega,\Omega}=\frac{\det(\mathcal{W}^{/\Omega})}{\det(\mathcal{W})}\le \frac{\det(\diag\mathcal{W}^{/\Omega})}{\det(\diag\mathcal{W})}\frac{\left(1+\frac{2|\mathcal{P}|^2}{d}\right)}{\left(1-\frac{|\mathcal{P}|^2}{d}\right)}\le \frac{1}{d^k}\left(1+\frac{6|\mathcal{P}|^2}{d}\right),\quad \text{if}\,\, 2|\mathcal{P}|^2\le d,
\ee
where we have used  the fact that $\frac{(1+x)}{(1-x/2)}\le 1+3x$ for $x\le 1$. Similarly for the lower bound we get
\be
\mathcal{W}_{\Omega,\Omega}^{-1}=\frac{\det(\mathcal{W}^{/\Omega})}{\det(\mathcal{W})}\ge \frac{1}{d^k}\frac{\left(1-\frac{|\mathcal{P}|^2}{d}\right)}{\left(1+\frac{2|\mathcal{P}|^2}{d}\right)}\ge \frac{1}{d^k}\left(1-\frac{4|\mathcal{P}|^2}{d}\right),\quad \text{if}\,\,2|\mathcal{P}|^2\le d.
\label{lem2llbound}
\ee
Notice that together \cref{lem2upbound,lem2llbound} prove the first bound. To prove the second bound, we make use of
\be
|\mathcal{W}^{-1}_{i,j}|&\le\|\mathcal{W}^{-1}-\diag(\mathcal{W}^{-1})\|_\infty\\&\le \|\mathcal{W}^{-1}-\diag(\mathcal{W})^{-1}\|_\infty+\|\diag(\mathcal{W})^{-1}-\diag(\mathcal{W}^{-1})\|_\infty
\\&\le\|\mathcal{W}^{-1}\|_\infty\|\mathcal{W}\diag(\mathcal{W})^{-1}-I\|_\infty+\|\diag(\mathcal{W})^{-1}-\diag(\mathcal{W}^{-1})\|_\infty\\
&=\|\mathcal{W}^{-1}\|_\infty\|\mathcal{W}'\|_{\infty}+\|\diag(\mathcal{W})^{-1}-\diag(\mathcal{W}^{-1})\|_\infty
\ee
where, similarly to \cref{lem:propertiesofW}, we have defined $\mathcal{W}'\coloneqq\mathcal{W}\diag(\mathcal{W})^{-1}-I$. Now, since $\|\mathcal{W}^{-1}\|_{\infty}\le \frac{1}{\lambda_{\min}(\mathcal{W})}$, we can exploit item $(iv)$ of \cref{lem:propertiesofW}, to bound 
\be
\|\mathcal{W}^{-1}\|_{\infty}\le\frac{1}{d^k}\left(1+\frac{2|\mathcal{P}|}{d}\right),\quad \text{if}\,\, 2|\mathcal{P}|\le d.
\ee
From \cref{eq:eigenvaluewprime} in the proof of \cref{lem:propertiesofW}, we know that $\|\mathcal{W}'\|_\infty\le \frac{|\mathcal{P}|}{d}$. We are left in upper bounding $\|\diag(\mathcal{W})^{-1}-\diag(\mathcal{W}^{-1})\|_\infty$. From point $(ii)$ of \cref{lem:propertiesofW} and point (i) above, we have
\be
\|\diag(\mathcal{W})^{-1}-\diag(\mathcal{W}^{-1})\|_\infty=\max_{\Omega}\Big(\frac{1}{\mathcal{W}_{\Omega,\Omega}}-(\mathcal{W}^{-1})_{\Omega,\Omega}\Big)\le \frac{4|\mathcal{P}|^2}{d^{k+1}}.
\ee
Putting it altogether, we finally have
\be
|\mathcal{W}_{i,j}^{-1}|\le\frac{|\mathcal{P}|}{d^{k+1}}+\frac{2|\mathcal{P}|^2}{d^{k+2}}+\frac{4|\mathcal{P}|^2}{d^{k+1}}\le \frac{5|\mathcal{P}|^2}{d^{k+1}},\quad \text{if}\,\, 2|\mathcal{P}|^2\le d\,.
\label{eqeqeq2}
\ee
To conclude the proof it is sufficient to note that $n\ge k^2-3k+13$ is sufficient to ensure $2|\mathcal{P}|^2\le d$ due to the dimension of the Clifford group \cref{cor:dimcom}.
\end{proof}

As an application of the asymptotics of Clifford-Weingarten functions, in the next lemma, we show that the only contributing terms are the ones coming from unitary Pauli monomials whenever we twirl an operator $O$ with constant operator norm.
\begin{lemma}[Unitary Pauli monomials dominate]\label{lem:unitarymonomialsdominate} Fix $n\ge k^2-3k+13$. Let $O$ be an operator with $\|O\|_{\infty}\le K$, and let $c_{\Omega}(O)$ be the coefficients of the twirling operator $\Phi_{\cl}^{(k)}(O)$ for $\Omega\in\mathcal{P}$, then
\be
|c_{\Omega}(O)|\le\frac{6K|\mathcal{P}|^{3}}{d},\quad\forall\Omega\in\mathcal{P}\setminus \mathcal{P}_U
\ee
i.e., for any non-unitary monomial $\Omega$, it holds that $|c_{\Omega}(O)|=O(2^{\frac{3}{2}k^2-n})$. 
\begin{proof}
    The proof directly descends from \cref{lem:asymptoticweingarten}. Indeed, we have
    \be
|c_{\Omega}(O)|&=\left|\sum_{\Omega'}(\mathcal{W}^{-1})_{\Omega,\Omega'}\tr(\Omega' O)\right|\le |\tr(\Omega O)|\left(\frac{1}{d^k}+\frac{6|\mathcal{P}|^2}{d^{k+1}}\right)+\frac{5|\mathcal{P}|^3}{d^{k+1}}\max_{\Omega'\neq\Omega}|\tr(\Omega'O)|.
    \ee
Now notice that if $\Omega\not\in\mathcal{P}_U$, then $\|\Omega\|_1\le d^{k-1}$. In fact, in normal form (see \cref{lem:normalform}) we can write $\Omega=\Omega_U\Omega_P$ then $\|\Omega\|_1=\|\Omega_P\|_1=\tr(\Omega_P)\le d^{k-1}$, as proven in \cref{lem:tracesofpaulimonomials}. However, if $\Omega\in\mathcal{P}_U$, we have $\|\Omega_U\|_1=d^k$. Therefore, we can bound
\be
|c_{\Omega}(O)|\le \|O\|_{\infty}\left(\frac{1}{d}+\frac{6|\mathcal{P}|^2}{d^{2}}+\frac{5|\mathcal{P}|^3}{d}\right)
\ee
where we have bounded $|\tr(\Omega O)|\le \|\Omega\|_1\|O\|_{\infty}$. Upper bounding $1+6|\mathcal{P}|^2/d\le|\mathcal{P}|^3$, which holds for $n\ge k^2-3k+13$, concludes the proof.
\end{proof}
\end{lemma}

\subsubsection{Clifford orbit of quantum states}\label{sec:clorbitquantumstates}
In \cref{Sec: orbitofhaarstates}, we analyzed the orbit of $k$-copies of a pure quantum state, $\ketbra{\psi}^{\otimes k}$, under the action of the full unitary group. As anticipated, the orbit remains the same for any reference state vector $\ket{\psi}$ and is proportional to the projector onto the symmetric subspace, $\Pi_{\sym} = \frac{1}{k!} \sum_{\pi \in S_k} T_{\pi}$. In contrast, the Clifford orbit of a pure quantum reference state $\ketbra{\psi}^{\otimes k}$ is more intricate and strongly depends on the choice of the reference state vector $\ket{\psi}$. However, the calculations for the Clifford orbit of pure quantum states are more tractable compared to those for general operators $O$.

Thanks to the symmetry $T_{\pi}\ket{\psi}^{\otimes k}=\ket{\psi}^{\otimes k}$ of pure states under permutations, we can write the twirling operator applied on $\ketbra{\psi}^{\otimes k}$ in terms of Pauli monomials modulo the equivalence class given by permutations. 
\begin{lemma}[Clifford orbit of a pure quantum state] \label{lem:clifforbitofpurestate}
Let $\ket{\psi}$ be a quantum state vector and define $\psi=\ketbra{\psi}$ its density matrix. Define the following equivalence relation $\Omega\sim\Omega'$ if and only if there exists $\pi,\sigma\in S_k$ such that $\Omega=T_{\pi}\Omega'T_{\sigma}$. Define $\mathcal{P}/S_k$ as the set of representative elements of the equivalence relation. The Clifford twirling operator $\Phi^{(k)}_{\cl}(\cdot)$, applied on $k$ copies of $\psi$ reads
\be
\Phi_{\cl}^{(k)}(\psi^{\otimes k})=\sum_{\Omega \in \mathcal{P}/ S_k}p_{\Omega}(\psi)\frac{\Pi_{\sym}\Omega\Pi_{\sym}}{\tr(\Omega\Pi_{\sym})}\label{eq:cliffordorbitstate}
\ee
where $p_{\Omega}(\psi)$ are quasi-probabilities, i.e., $\sum_{\Omega}p_{\Omega}(\psi)=1$, but in general they are not positive and $\Pi_{\sym}=\frac{1}{k!}\sum_{\pi\in S_k}\pi$. Moreover,
\be
p_{\Omega}(\psi)=\sum_{\Omega'\in \mathcal{P}/ S_k}(\mathcal{S}^{-1})_{\Omega,\Omega'}\frac{\tr(\Omega'\psi^{\otimes k})}{\tr(\Pi_{\sym}\Omega')}
\ee
where $\mathcal{S}_{\Omega,\Omega'}\coloneqq\frac{\tr(\Omega\Pi_{\sym}\Omega'\Pi_{\sym})}{\tr(\Pi_{\sym}\Omega)\tr(\Pi_{\sym}\Omega')}$.
\begin{proof}
The proof of \cref{eq:cliffordorbitstate} follows identically compared to the derivation of Clifford-Weingarten functions, see \cref{lem:weingartencalculusclifford}. 
Let us show that $p_{\Omega}(\psi)$ are only quasi-probabilities. The fact that they sum to $1$, simply follows from $\tr(\Phi_{\cl}^{(k)}(\psi^{\otimes k}))=1$. It suffices to compute the average over a Haar random state of $p_{\Omega}(\psi)$:
\be
\int\de\psi p_{\Omega}(\psi)=\sum_{\Omega'\in \mathcal{P}/ S_k}(\mathcal{S}^{-1})_{\Omega,\Omega'}\int\de\psi\frac{\tr(\Omega'\psi^{\otimes k})}{\tr(\Pi_{\sym}\Omega')}=\sum_{\Omega'\in \mathcal{P}/ S_k}(\mathcal{S}^{-1})_{\Omega,\Omega'}\frac{1}{\tr(\Pi_{\sym})}.
\ee
Now, it is sufficient to note that $\mathcal{S}_{\Omega',\mathbb{1}}=\frac{\tr(\Pi_{\sym})}{\tr(\Pi_{\sym})\tr(\Pi_{\sym})}=\frac{1}{\tr(\Pi_{\sym})}$. Hence
\be
\int\de \psi p_{\Omega}(\psi)=\delta_{\Omega,\mathbb{1}}
\ee
This demonstrates that $p_{\Omega}(\psi)$ is generally non-positive. Suppose, for contradiction, that $p_{\Omega}(\psi) \geq 0$ for all states. Given that the average is zero, this would imply $p_{\Omega}(\psi) = 0$ for all $\Omega \neq \mathbb{1}$ and for any state. Consequently, this would lead to $\Phi_{\cl}^{(k)}(\psi^{\otimes k}) = \Phi_{\haar}^{(k)}(\psi^{\otimes k})$ for any state $\psi$, meaning that the Clifford orbit—such as that of stabilizer states—forms a projective $k$-design for all $k$, which is clearly a contradiction. Therefore, $p_{\Omega}(\psi)$ must be non-positive in general.
\end{proof}
\end{lemma}

\begin{lemma}[Clifford orbit of a stabilizer state~\cite{gross_schurweyl_2019}]
\label{le:staborb}
While the Clifford orbit depends on the reference state, it attains a simple value for stabilizer states. It is possible to show that, for every $n,k\ge 1$, it holds that 
\be
\Phi_{\cl}^{(k)}\left(\ketbra{0}^{\otimes k}\right)=\frac{1}{Z_{d,k}}\sum_{\Omega\in\mathcal{P}}\Omega
\ee
with $Z_{d,k}=d\prod_{i=0}^{k-2}(d+2^i)$.
\begin{proof}
This lemma has been proven in Ref.~\cite{gross_schurweyl_2019}, which we include here for completeness. For this proof, we explicitly denote the dependence of $n$ in the twirling operator as $\Phi_{\cl}^{(k)}\left(\ketbra{0_n}^{\otimes k}\right)$. For any $n' < n$, uniformity over Clifford circuits implies that  
\be
\left\langle 0_{n-n'}^{\otimes k} | \Phi_{\cl}^{(k)}\left(\ketbra{0_n}^{\otimes k}\right) | 0_{n-n'}^{\otimes k} \right\rangle \propto \Phi_{\cl}^{(k)}\left(\ketbra{0_{n'}}^{\otimes k}\right) \label{eq:proportionality}.
\ee
Now, consider $n \geq k - 1$. From \cref{lem:weingartencalculusclifford}, we know that we can express  
\be
\Phi_{\cl}^{(k)}(\ketbra{0_n}^{\otimes k}) = \sum_{ \Omega \in \mathcal{P}} c_{n,\Omega} \Omega^{(n)},
\ee
here we introduced the notation $\Omega^{(n)}$ to make explicit the dependence from the number of qubits of the system we are considering, and we use the fact that $\tr(\Omega^{(n)} \ketbra{0_n}^{\otimes k}) = 1$, which holds for any $k$-tensor copies of a stabilizer state. We also define $c_{n,\Omega} = \sum_{\Omega'} (\mathcal{W}^{-1})_{\Omega,\Omega'}$. Since $\langle 0_{n-n'}^{\otimes k} | \Omega^{(n)} | 0_{n-n'}^{\otimes k} \rangle = \Omega^{(n')}$ due to the tensor product structure $\Omega^{(n)} = \omega^{\otimes n}$ (see \cref{lem:monomialsfactorizeonqubits}), \cref{eq:proportionality} implies that  
\be
\sum_{\Omega \in \mathcal{P}} c_{n,\Omega} \Omega^{(n)} \propto \sum_{\Omega \in \mathcal{P}} c_{n',\Omega} \Omega^{(n')}.
\ee
This necessarily implies that $c_{n,\Omega} = c_{n} c_{\Omega}$. Since this equation holds for any $n'$, in particular for $n' < k - 1$, it must hold for all $n$. We can thus express the twirling operator for any $n$ as  
\begin{equation}
\Phi_{\cl}^{(k)}(\ketbra{0_n}^{\otimes k}) \propto \sum_{\Omega\in\mathcal P} c_{\Omega} \Omega^{(n)}.
\end{equation}

Next, we show that the coefficients $c_{\Omega}$ must be equal. Consider the expectation value of a Pauli monomial $\Omega'$. On one hand, we have  
\begin{equation}
\tr(\Omega^{'(n)} \Phi_{\cl}^{(n)}(\ketbra{0_n}^{\otimes k})) = \langle 0_n | \Omega^{'(n)} | 0_n \rangle = 1.
\end{equation}
On the other hand, we have  
\be
1 = c_{n} \sum_{\Omega\in \mathcal P} c_{\Omega} \tr(\Omega^{'(n)} \Omega^{(n)}) = c_{n} d^{k} c_{\Omega'} + \sum_{\Omega \neq \Omega'\in\mathcal P} c_{\Omega} \tr(\Omega^{'(n)} \Omega^{(n)}) = d^{k} c_{n} (c_{\Omega'} + O(d^{-1})) \quad \forall \Omega'.
\ee
Taking the limit of large $d$, it follows that all coefficients $c_{\Omega'}$ must be equal. However, since $c_{\Omega'}$ is independent of $n$, the result holds for every $n$. Thus for any $n$ and $k$, we can write  
\be
\Phi_{\cl}^{(k)}(\ketbra{0}^{\otimes k}) = \frac{1}{Z_{d,k}} \sum_{\Omega \in \mathcal{P}} \Omega,
\ee
where  $Z_{d,k} = \sum_{\Omega \in \mathcal{P}} \tr(\Omega)$, which can be computed directly from \cref{lem:tracesofpaulimonomials}, using the fact that $\tr(\Omega) = d^{k - m}$ and similarly from \cref{cor:dimcom} one can obtain that the number of reduced Pauli monomials with order $m$ is given by $\binom{k-1}{m}_2\times 2^{m(m-1)/2}$. 

\end{proof}
\end{lemma}

\begin{proposition}[Clifford Weingarten notable identities]
\label{prop:weingarten}
Let $G$ be the Gram matrix of Pauli monomials with entries $G_{\Omega', \Omega} \coloneqq \Tr(\Omega'^\dagger \Omega)$, and let $\mathcal{W} \coloneqq G^+$ be its Moore–Penrose pseudoinverse (the Weingarten matrix). Then, for all $\Omega' \in \mathcal{P}$,
\begin{align}
    \sum_{\Omega \in \mathcal{P}} G_{\Omega', \Omega} = Z_{d,k}\,, \quad\quad
    \sum_{\Omega \in \mathcal{P}} \mathcal{W}_{\Omega', \Omega} = \frac{1}{Z_{d,k}}\,,
\end{align}
where $Z_{d,k} = d \prod_{i=0}^{k-2}(d + 2^i)$. Notably, both sums are independent of $\Omega'$.
\end{proposition}

\begin{proof}
Let $\psi_0 \coloneqq \ketbra{0}^{\otimes k}$. The Clifford moment operator satisfies $\Phi_{\mathrm{Cl}}^{(k)}(\psi_0) = \frac{1}{Z_{d,k}} \sum_{\Omega \in \mathcal{P}} \Omega$, as shown in Lemma~\ref{le:staborb}. Taking the Hilbert--Schmidt inner product with any $\Omega' \in \mathcal{P}$, and using that $\Phi_{\mathrm{Cl}}^{(k)}$ is self-adjoint, that $\Omega'$ lies in the commutant, and that $\Tr(\Omega' \psi_0) = 1$, gives $\frac{1}{Z_{d,k}} \sum_{\Omega \in \mathcal{P}} \Tr(\Omega'^\dagger \Omega) = 1$, establishing the first identity.

For the second identity, write $\Phi_{\mathrm{Cl}}^{(k)}(\psi_0) = \sum_{\Omega'} c_{\Omega'}\, \Omega'$, where the vector $c$ of coefficients satisfies (see Lemma~\ref{lem:weingartencalculusclifford}) $c = \mathcal{W} \mathbf{1}$, with $\mathbf{1} \in \mathbb{R}^{|\mathcal{P}|}$ the all-ones vector, since $\Tr(\Omega \psi_0) = 1$ for all $\Omega$. Equivalently, $c$ satisfies $G c = \mathbf{1}$. A particular solution is given by $c = \frac{1}{Z_{d,k}} \cdot \mathbf{1}$, since all rows of $G$ sum to $Z_{d,k}$, as shown above.
It is well known that the Moore--Penrose pseudoinverse selects the unique solution of minimal Euclidean norm. Consider any other solution $\tilde{c}$ to $G c = \mathbf{1}$. Then, by basic linear algebra, we can write $\tilde{c} = \frac{1}{Z_{d,k}} \cdot \mathbf{1} + \mathbf{v}$ for some $\mathbf{v} \in \ker(G)$. Since $G$ is symmetric, its kernel is orthogonal to its image, and $\mathbf{1} \in \mathrm{Im}(G)$, so $\langle \mathbf{1}, \mathbf{v} \rangle = 0$.
It follows that
\begin{align}
\| \tilde{c} \|_2^2 = \left\| \frac{1}{Z_{d,k}} \cdot \mathbf{1} \right\|_2^2 + \| \mathbf{v} \|_2^2 + 2 \left\langle \frac{1}{Z_{d,k}} \cdot \mathbf{1}, \mathbf{v} \right\rangle = \left\| \frac{1}{Z_{d,k}} \cdot \mathbf{1} \right\|_2^2 + \| \mathbf{v} \|_2^2 \ge \left\| \frac{1}{Z_{d,k}} \cdot \mathbf{1} \right\|_2^2,
\end{align}
with equality if and only if $\mathbf{v} = 0$. Hence, the unique minimal-norm solution is $c = \mathcal{W} \cdot \mathbf{1} = \frac{1}{Z_{d,k}} \cdot \mathbf{1}$, from which it follows that each row of $\mathcal{W}$ sums to $1/Z_{d,k}$, as claimed.
\end{proof}

\subsection{Magic-state resource theory}\label{Sec:magicstateresourcetheory}
In this section, we apply the framework developed for the Clifford commutant to a closely related topic: magic-state resource theory. As in other resource theories, magic-state resource theory focuses on defining resource monotones, which are functions that do
not increase under Clifford operations~\cite{veitch_resource_2014,howard_application_2017,seddon_quantifying_2019}. In a broader context, Clifford operations encompass not only Clifford unitary transformations but also computational basis measurements, partial traces, and operations conditioned on measurement outcomes. The presence of these general operations significantly complicates the search for computable magic monotones. As a result, most magic monotones involve a minimization procedure over the entire set of stabilizer states~\cite{liu_manybody_2022,seddon_quantifying_2019}, rendering them fundamentally unmeasurable in experimental settings and analytically intractable~\cite{heinrich_robustness_2019}.  However, stabilizer entropy has emerged as the first magic monotone that is both experimentally measurable and analytically tractable~\cite{leone_stabilizer_2022}. Nevertheless, this experimental feasibility is largely limited to pure states. This limitation is not surprising, as it is a common feature across other resource theories, such as entanglement resource theory. For this reason, in the remainder of the section, we will primarily focus on magic-state resource theory as applied to pure states. 

In the following, we will use the term \textit{magic measures} to refer to quantities that are not necessarily magic monotones. The key distinction is that magic measures are required only to be resource non-increasing under Clifford unitaries, a necessary condition for being a magic monotone but not sufficient on its own, while magic monotones must exhibit monotonic behavior under a much broader class of stabilizer operations. In our terminology, every magic monotone is also a magic measure, but the converse is not true. For a more detailed discussion, see~\cite{Leone_2024}.

To begin, let us formally define the notion of \textit{measurable} function in the context of quantum information.
\begin{definition}[Measurement strategy] A measurement strategy is a quantum algorithm  $\mathcal{A}(\rho^{\otimes k})\mapsto [0,1]$ with arbitrary classical processing. Any such strategy describes a statistical process from shot-noise statistics and potentially a randomized classical processing strategy. We assume that both the classical and quantum processing can have arbitrary circuit complexity and arbitrary auxiliary bits and qubits.
\end{definition}
\begin{definition}[Measurable function]\label{def:measurablefunctions}
    A measurable function $f$ is the expectation value of a measurement strategy $f(\rho^{\otimes k})\coloneqq\langle \mathcal{A}(\rho^{\otimes k})\rangle$, where the expectation value $\langle\cdot\rangle$ is over all possible measurement outcomes of the measurement strategy $\mathcal{A}.$
\end{definition}

Moving on to magic measures, let us introduce stabilizer entropies, which will serve as the central objects of our subsequent discussion.

\begin{definition}[Stabilizer entropy~\cite{leone_stabilizer_2022}]\label{def:stabilizerentropies} Let $\ket{\psi}$ be a pure quantum state vector and $\psi=\ketbra{\psi}$ its density matrix. Then, the stabilizer entropies $M_{\alpha}(\psi)$ are defined
\be
M_{\alpha}(\psi)=\frac{1}{1-\alpha}\log\Delta_{2\alpha}(\psi),\quad \Delta_{2\alpha}(\psi)\coloneqq\frac{1}{d}\sum_{P\in\mathbb P_n}\tr^{2\alpha}(P\psi)\,.
\ee
\end{definition}
Upon closer examination, it becomes evident that the $\Delta_{2\alpha}(\psi)$ in the definition of stabilizer entropy is, in fact, the expectation value of the primitive Pauli monomial (see \cref{def:primitivepaulimonomials}): 
\be
\Omega_{2\alpha} = \frac{1}{d}\sum_{P\in\mathbb P_n}P^{\otimes 2\alpha},\quad \Delta_{2\alpha}=\tr(\Omega_{2\alpha}\psi^{\otimes 2\alpha})\,.
\ee
This equivalence and the ordering of R\'enyi entropies~\cite{leone_stabilizer_2022}, allows us to write $\Delta_{2(\alpha+1)}\le \Delta_{2\alpha}\le \Delta_{2(\alpha+1)}^{1-1/\alpha}$.

This observation hints at a profound connection between magic-state resource theory and the commutant of the Clifford group. Specifically, in what follows, we prove that any measurable magic measure can be expressed as a Pauli polynomial, that is, a linear combination of Pauli monomials. 

Before stating the main result, let us lay down the following lemma in the context of magic-state resource theory. 
\begin{lemma}[Invariance under unitary Clifford operations]\label{lem:invarianceundercliffordmonotone}
    Let $M(\psi)$ be any magic measure. Then $M(\psi)$ is invariant under Clifford unitary operations, that is, $M(\psi)=M(C\psi C^{\dag})$ for every $C\in\mathcal{C}_n$.
    \begin{proof}
Due to their invertibility, the fact that a resource does not increase under Clifford unitaries implies that it remains invariant under Clifford unitaries.
    \end{proof}
\end{lemma}

\begin{lemma}\label{lem:measurablefunctionbelongtothecommutant}
    Every measurable function $M(\rho^{\otimes k})$ on $k$ copies, which is invariant under Clifford unitary operations, is equal to an expectation value of a positive operator $\Pi\le \mathbb{1}$ that belongs to $\com(\mathcal{C}_n^{\otimes k})$, i.e., 
    \begin{align}
        M(\rho^{\otimes k})=\tr(\rho^{\otimes k} \Pi),\quad \Pi=\sum_{\Omega\in\mathcal P} \alpha_{\Omega}\Omega 
        .
    \end{align}
    \begin{proof}
    Let us first show that any measurable function in the sense of \cref{def:measurablefunctions} can be recast as an expectation value of a positive operator $\Pi$. Notice that every classical process can be described by a quantum circuit. Quantum measurements can be simulated using an auxiliary qubit and a CNOT gate. To turn the distribution of classical outcome estimates into an expectation value, the experimenter implements the state 
    \begin{align}
        \rho=\begin{pmatrix}
            r&0\\ 0 &1-r
        \end{pmatrix}
    \end{align}
    where $r$ is the classical estimate. The final measurement is then $\pi=\ketbra{0}$. Therefore, supplying with $k'$ additional registers initiated in the state $\ket{0}$, any measurement strategy can be implemented. Since measurable functions are expectation values of measurement strategies, we can write 
    \be
    M(\rho^{\otimes k})=\tr(\rho^{\otimes k}\otimes\ketbra{0}^{\otimes k'}U\Pi U^{\dag})
    \ee
    where $U$ is a classical-quantum circuit and $\pi=\ketbra{0}$. Defining $\Pi\coloneqq \langle0^{\otimes k'}|U\pi U^{\dag}|0^{\otimes k'}\rangle$, we can thus write the measurable function $f(\rho^{\otimes k})=\tr(\Pi\rho^{\otimes k})$. Notice that since $\pi\le \mathbb{1}$, so is $\Pi$. Now we just have to prove $\Pi\in\com(\mathcal{C}_n^{\otimes k})$. Let us exploit the invariance under Clifford unitary operations
\be
M(C^{\otimes k}\rho^{\otimes k} C^{\dag\otimes k}) = f(\rho^{\otimes k}), \quad \forall C, \rho \implies \tr(C^{\dag\otimes k} \Pi C^{\otimes k} \rho^{\otimes k}) = \tr(\Pi \rho^{\otimes k}), \quad \forall C,\rho.
\ee  
By the requirement that this equality holds for any state $\rho$, it follows that the operator $\Pi$ must satisfy $\Pi = C^{\dag\otimes k} \Pi C^{\otimes k}$ for all $C\in\mathcal{C}_n$, implying that $\Pi \in \com(\mathcal{C}_n^{\otimes k})$, and thus a Pauli polynomial. This concludes the proof.
\end{proof}
\end{lemma}

\cref{lem:measurablefunctionbelongtothecommutant} says that any measurable function corresponding to a measurement strategy on $k$ copies that is also invariant under Clifford operations belongs to the $k$-th order commutant of the Clifford group. This result suggests that measurable magic measures, which must be invariant under Clifford operations, see \cref{lem:invarianceundercliffordmonotone}, are nothing more than Pauli polynomials. This result is formalized in the following theorem, the proof of which descends from the above lemma.
 
\begin{theorem}[Measurable
magic measures as Pauli polynomials]\label{th:measurablemagicmonotones}
    Any magic measures $M(\rho)$, for which there exists an unbiased estimator employing $k$ copies of the state $\rho$, is a Pauli polynomial of degree $k$.
\end{theorem}
\begin{proof}
Let $M(\rho)$ be the magic measure and let $\hat{M}(\rho^{\otimes k})$ be an unbiased estimator for $M(\rho)$. By definition of unbiased estimator, we have that $\langle \hat{M}(\rho^{\otimes k})\rangle=M(\rho)$, where the expectation value is over all the possible outcomes of the measurement strategy used on $k$ copies. Hence, $M(\rho)$ is a measurable function according to \cref{def:measurablefunctions}. Moreover, $M(\rho)$ must be invariant under the action of Clifford unitaries, as per \cref{lem:invarianceundercliffordmonotone}. Hence, by virtue of \cref{lem:measurablefunctionbelongtothecommutant}, we can write $M(\rho)=\tr(\Pi\rho)$ with $\Pi\in\com(\mathcal{C}_n^{\otimes k})$ being a Pauli polynomial of degree $k$.
\end{proof}

This result demonstrates that any experimentally measurable magic measure must be expressible as the expectation value of an operator in the $k$-th order commutant of the Clifford group, underscoring the strong connection between magic-state resource theory and the Clifford commutant. Moreover, the argument above also reveals that the minimal polynomial order for which a magic measure can exist is $k \geq 4$. For $k < 4$, the Clifford commutant coincides with the commutant of the full unitary group and, therefore, cannot contain any useful information.

While the above general results apply to possibly mixed states, in the remainder of the section, we restrict our attention to pure states. In \cref{th:measurablemagicmonotones}, we have shown that any measurable magic measure is equivalent to an expectation value of a positive Pauli polynomial. If we restrict to pure states, we can assume, without loss of generality, that such operator lives in the symmetric subspace with projector $\Pi_{\sym}$. As shown in \cref{lem:clifforbitofpurestate}, we can define the equivalence relation of Pauli monomials under left and right actions of permutations, and define the set $\mathcal{P}/S_k$ as the set of independent Pauli monomials that are permutationally inequivalent. Consequently, without loss of generality, we can write any measurable magic measure restricted on pure states as
\be
M(\psi) = \sum_{\Omega \in \mathcal{P}/S_k} c_{\Omega} \tr(\Omega \psi^{\otimes k}),
\ee  
i.e., the magic measure $M(\psi)$ is simply a linear combination of independent, permutationally nonequivalent Pauli monomials, suggesting that these monomials represent more fundamental magic measures. This observation motivates the following definition.

\begin{definition}[Generalized stabilizer purities] \label{def:generalizedstabilizerpurity} 
Let $\Omega \in \mathcal{P}/S_k$, where $ \mathcal{P}/S_k$ is the set of reduced Pauli monomials up to permutations with representative elements Pauli monomials with the least possible order $m$. The $\Omega$-stabilizer purity is defined as  
\be
\Delta_{\Omega}(\psi) \coloneqq \Re(\tr(\Omega \psi^{\otimes k})).
\ee  
Notice that $\Delta_{\Omega}(\psi) = 1$ for stabilizer states, and it is manifestly invariant under Clifford operations. Moreover, if $\Omega(v)$ is a primitive Pauli monomial, we abbreviate it as $\Delta_{\Omega(v)}(\psi)=\Delta_{|v|}(\psi)$, consistently with \cref{def:stabilizerentropies}. The real part is used because not every $\tr(\Omega \psi^{\otimes k})$ is necessarily real. An alternative definition could also involve taking the absolute value of the expectation value.
\end{definition}

In the context of qudits, Ref.~\cite{Turkeshi_2025} introduced generalized stabilizer purity in the same fashion. In the next theorem, we show that generalized stabilizer purities attains their maximum values for stabilizer states.
\begin{theorem}[Faithfulness of generalized stabilizer purities]\label{th:faithfulness}
    Let $\psi$ be a pure quantum state, then for any $\Omega$-stabilizer purity, it holds that
    \be
    \Delta_{\Omega}(\psi)\le 1
    \ee
    while $\Delta_{\Omega}(\psi)=1$ if $\psi$ is a stabilizer state.
\end{theorem}
    \begin{proof}
We divide the proof into three cases:
\begin{enumerate}[label=(\roman*)]

\item If $\Omega\in\mathcal{P}_U$ is a unitary Pauli monomial, the result follows directly from a simple H\"older inequality as
\be
\Delta_{\Omega}(\psi) \leq \|\Omega\|_{\infty} \|\psi\|_{1}^{k} \leq 1.
\ee

\item If $\Omega(V,M)\in\mathcal{P}_P$ is a projector, it corresponds to 
a vector space $V$ containing only even-weight primitives with no phases, i.e., $M = 0$. Let us now use the invariance under linear transformations, formally established in \cref{th:gaussOP}. There exists a transformation $V \mapsto VA$ that performs Gaussian elimination on the column vectors of $V$, resulting in $VA$ having a pivot structure. Since for projectors, $H = 0 \pmod{2}$, this linear operation ensures $M(A) = 0$. 

For the sake of clarity, after the transformation, we express $\Omega$ as $\Omega = \Omega(v_1) \Omega(v_2) \cdots \Omega(v_m)$, 
with the property that $(v_1)_1 = 1$, while $(v_i)_1 = 0$ for $i = 2, \ldots, m$. Denoting $T_1$ as the partial transpose on the first tensor copy, we use the invariance of the trace under partial transpose:
\be
\Delta_{\Omega}(\psi) = \tr(\psi^{\otimes k} \Omega \psi^{\otimes k}) = \tr(\psi^{T} \otimes \psi^{\otimes (k-1)} \Omega(v_1)^{T_1} \Omega_{\text{rest}} \psi^{T} \otimes \psi^{\otimes (k-1)}),
\ee
where, crucially, the partial transpose does not affect the remaining terms, denoted as $\Omega_{\text{rest}} = \prod_{i=2}^{m} \Omega(v_i)$, because of the pivot structure. 

It is straightforward to verify that $\Omega(v_1)^{T_1}$ is a unitary operator, while $\Omega_{\text{rest}}$ is proportional to a projector, i.e., $\Omega_{\text{rest}}^2 = d^{m-1} \Omega_{\text{rest}}$. We now bound $\Delta_{\Omega}(\psi)$, to get
\be
\Delta_{\Omega}(\psi) & = \frac{1}{d^{m-1}} \tr(\psi^{T} \otimes \psi^{\otimes (k-1)} \Omega(v_1)^{T_1} \Omega_{\text{rest}} \Omega_{\text{rest}} \psi^{T} \otimes \psi^{\otimes (k-1)}) \\
& \leq \frac{1}{d^{m-1}} \|\Omega_{\text{rest}} \psi^{T} \otimes \psi^{\otimes (k-1)} \Omega(v_1)^{T_1}\|_2 \|\Omega_{\text{rest}} \psi^{T} \otimes \psi^{\otimes (k-1)}\|_2 \\
& = \frac{1}{d^{m-1}} \|\Omega_{\text{rest}} \psi^{T} \otimes \psi^{\otimes (k-1)}\|_2^2,
\ee
where we applied H\"older inequality and, in the last step, used the unitary invariance of the $2$-norm to eliminate $\Omega(v_1)^{T_1}$. 

Thus, we obtain
\be
\Delta_{\Omega}(\psi) \leq \frac{1}{d^{m-1}} \|\Omega_{\text{rest}} \psi^{T} \otimes \psi^{\otimes (k-1)}\|_2^2 
= \tr(\Omega_{\text{rest}} \psi^{T} \otimes \psi^{\otimes (k-1)} \psi^{T} \otimes \psi^{\otimes (k-1)})
= \tr(\Omega_{\text{rest}}^{T_1} \psi \otimes \psi^{\otimes (k-1)}).
\ee
In the last step, we have used the pivot structure of $\Omega_{\text{rest}}$, which ensures that $\Omega_{\text{rest}}^{T_1} = \Omega_{\text{rest}}$. Iterating this argument $m-1$ more times, we arrive at: $\Delta_{\Omega}(\psi) \leq \Delta_{|v_m|}(\psi)$, where $\Delta_{|v_m|}$ is a stabilizer purity corresponding to a primitive $\Omega(v_m)$. Since it is known that $\Delta_{|v_m|} \leq 1$~\cite{leone_stabilizer_2022}, the bound holds.

\item When $\Omega = \Pi U$ is a product of a projector $\Pi$ and a unitary $U$, we can apply the same reasoning as above to obtain $\tr(\Omega \psi^{\otimes k}) \leq \tr(\Pi \psi^{\otimes k})$, which completes the proof.
\end{enumerate}
\end{proof}
Before presenting explicit examples of generalized stabilizer purities, let us compute the average $\Omega$-stabilizer purity for Haar random states. This result will aid in establishing a hierarchy of generalized stabilizer purities based on their order $m(\Omega)$, as defined in \cref{def:orderpaulimonomial}. Haar random states can be thought of as states with the maximal level of magic since they asymptotically (in the number of qubits $n$) saturate the bounds of all known magic monotones (see Refs.~\cite{liu_manybody_2022,leone_stabilizer_2022}). Consequently, generalized stabilizer purities that achieve larger values for Haar random states are the most sensitive, as they would vary within a smaller range, making them more practical for experimental measurements.

\begin{lemma}[Haar average of $\Omega$-stabilizer purity]\label{lem:haaraveragestabilizerpurity} Let $\Delta_\Omega(\psi)$ be a $\Omega$-stabilizer purity with order $m(\Omega)$. The average value over Haar random state is
\be
\int\de\psi\Delta_{\Omega}(\psi)=\frac{c_{\Omega}}{d^{m(\Omega)}}+O\left(\frac{1}{d^{m(\Omega)+1}}\right)
\ee
where $c_{\Omega}\le k!$.
\begin{proof}
    We know from \cref{Sec:haaraveragecliffordbeyond} that the average $\Omega$-stabilizer purity is simply
    \be
\int\de\psi \Delta_{\Omega}(\psi)=\tr(\Omega\Psi_{\haar}^{(k)})=\frac{k!}{\tr(\Pi_{\sym})}\sum_{\pi\in S_k}\tr(\Omega T_{\pi})=\frac{c_\Omega}{d^{m(\Omega)}}+O\left(\frac{1}{d^{m(\Omega)+1}}\right),
    \ee
where $c_{\Omega}\coloneqq|\{\Omega T_{\pi}\,:\, m(\Omega T_{\pi})=m(\Omega)\}|$ is the number of permutations not increasing the order of the monomial $\Omega$. The above statement follows from \cref{lem:tracesofpaulimonomials}, from the fact that $\frac{\tr(\Pi_{\sym})}{k!}=d^{k}+O(d^{k-1})$, and from the fact that the $\Omega$-stabilizer purity is defined to be the one within the equivalence class with the smallest order $m(\Omega)$.
\end{proof}
\end{lemma}
\cref{lem:haaraveragestabilizerpurity} tells us that $\Omega$-stabilizer purities \emph{vary} in a range depending on their order $m(\Omega)$, i.e., $[O(d^{-m(\Omega)}),1]$. From a practical point of view, it is convenient to consider those with small $m(\Omega)$ to reduce the number of measurement shots to measure magic in quantum states. This is particularly relevant for those stabilizer purities corresponding to the expectation value of unitary operators that can be efficiently measured in quantum experiments via a simple Hadamard test~\cite{nielsen_quantum_2000}.

Given that $\Omega$-stabilizer purities corresponding to unitary monomials are efficiently measurable magic measures, it is natural to expect that every measurable magic measure can be expressed as a linear combination of unitary generalized stabilizer purity. We prove this intuition to be true in the following lemma.

\begin{lemma}\label{lem:measurablemonotoneunitarymonomials}
Let $M(\psi)$ be a measurable magic monotone, for which there exists an unbiased estimator on $k$ copies of $\psi$. Then for any $n\ge k^2-3k+13$, it holds that
\be
M(\psi)=\sum_{\Omega\in\mathcal{P}_U}\alpha_{\Omega}\Delta_{\Omega}(\psi) +O(2^{\frac{3}{2}k^2-n}).
\ee
\begin{proof}
Thanks to \cref{th:measurablemagicmonotones}, we can write $M(\psi)=\tr (\hat{M}\psi^{\otimes k})$. Since $\hat{M}$ belongs to the commutant and, without loss of generality, lives in the symmetric subspace because $\psi$ is pure, we can express it as $\hat{M}=\sum_{\Omega\in\mathcal{P}_U}\alpha_{\Omega}(\hat{M})\Omega$ with $\Omega$ being Pauli monomials. Using \cref{lem:unitarymonomialsdominate}, and using that $\|\hat{M}\|_{\infty}\le 1$, we know
\be
\alpha_{\Omega}(\hat{M})\le O(2^{\frac{3}{2}k^2-n})
\ee
Since $|\Delta_{\Omega}(\psi)|\le 1$ (see \cref{th:faithfulness}), then it immediately follows that
\be
M(\psi)=\sum_{\Omega\in\mathcal{P}_U}\alpha_{\Omega}(\hat{M})\Delta_{\Omega}(\psi)+O(2^{2k^2-n})
\ee
From the condition of $M(\sigma)=1$ for every stabilizer state and recalling that $\tr(\Omega\sigma^{\otimes k})=1$ for every stabilizer state, one gets the condition $\sum_{\Omega}\alpha_{\Omega}=1$, from which it immediately follows that $\sum_{\Omega\in\mathcal{P}/S_k}\alpha_{\Omega}=1-O(2^{2k^2-n})$.
\end{proof}
\end{lemma}

\subsubsection{Examples of generalized stabilizer purities}
Having defined generalized stabilizer purity and having drawn intuition of their hierarchy, let us discuss some practical examples.

In \cref{sec:cliffordcommutantexample}, we identify all non-permutationally independent operators up to $k=8$. As previously noted, for $k < 4$, no polynomial can capture the magic of pure quantum states. At $k = 4$, the first magic measure is provided by the Pauli monomial $\Omega_{4}$, which corresponds to the $\alpha=2$ stabilizer entropy. Thus, the $\alpha=2$ stabilizer entropy represents \textit{the minimal expectation value} capable of quantifying magic. Using the graphical calculus, this equivalence can be expressed as  
\be
 \Delta_{4}(\psi) = \tr( \Omega_4 \psi^{\otimes 4}) = \tr(\monomialdiagram{4}{{1,2,3,4}}{}\,\,\,\, \psi^{\otimes 4}) = \frac{1}{d}\sum_{P\in\mathbb P_n} \tr^4(P\psi).
\ee  
where for the sake of clarity we use the shorthand notation $\Delta_x = \Omega_x$. 
For $\Delta_{4}(\psi)$, we can compute the leading order Haar average as done in \cref{lem:haaraveragestabilizerpurity}. In fact, from \cref{sec:cliffordcommutantexample} we know that $c_{\Omega}=4$, and therefore $\mathbb{E}_{\psi}[\Delta_{4}(\psi)]=4d^{-1}+O(d^{-2})$. Despite being the minimal stabilizer purity that can quantify magic and having an order $m(\Omega_4)=1$, the second stabilizer entropy has a significant limitation. For large numbers of qubits, which are often the most common case of interest, it cannot be measured efficiently. Specifically, there is no known algorithm with non-exponential sample complexity in $n$ that can measure $\Delta_{4}$ to additive precision~\cite{miller2024experimentalmeasurementphysicalinterpretation}. This inefficiency arises because the primitive $\Omega_4$ is proportional to a projector with a proportionality factor of $d$, see \cref{lem:propertiesprimitivepauli}. As a result, it has an exponential operator norm. Moreover, there are strong arguments suggesting that the efficiency of unbiased estimators for measuring the expectation values of observables depends on the operator norm of the observables. These facts motivate the search for magic measures corresponding to higher-order polynomials.

For $k = 6$, two additional elements arise that correspond to Pauli monomials. The first is $\Omega_6$, which corresponds to the third stabilizer entropy:  
\be
 \Delta_{6}(\psi) = \tr( \Omega_6 \psi^{\otimes 6}) = \tr\left(\monomialdiagram{6}{{1,2,3,4,5,6}}{}\,\,\,\, \psi^{\otimes 6}\right) = \frac{1}{d} \sum_{P\in\mathbb P_n} \tr^6(P \psi).
\ee  
The second arises from the product of two weight-4 primitives as  
\be
\Delta_{4,4}(\psi) = \tr(\Omega_{4,4} \psi^{\otimes 6}) = \tr\left(\monomialdiagram{6}{{1,2,3,4},{3,4,5,6}}{}\,\,\,\, \psi^{\otimes 6}\right) = \frac{1}{d^2} \sum_{P,Q\in\mathbb P_n} \tr^2(P \psi) \tr^2(PQ \psi) \tr^2(Q \psi),
\ee  
where we have used  the shorthand $\Omega_{4,4} \coloneqq \Omega(1,1,1,1,0,0)\Omega(0,0,1,1,1,1)$ (Note that we will use this shorthand notation from now on for all the $\Omega$s).

We can compute the Haar average of the two above explicitly, being $\mathbb{E}_{\psi}[\Delta_6(\psi)]=d^{-1}+O(d^{-2})$ thanks to its unitarity, while $\mathbb{E}_{\psi}[\Delta(\psi)]=24d^{-2}+O(d^{-3})$

While the latter is proportional to a projector, being the product of two operators proportional to projectors (see \cref{lem:propertiesprimitivepauli}, and hence cannot be measured efficiently, $\Delta_{6}(\psi)$ can be measured efficiently as it corresponds to a unitary operator. This observation motivates the introduction of a second criterion to give a hierarchy to the $\Omega$-stabilizer purities. Specifically, generalized stabilizer purities corresponding to Pauli monomials that are projectors are less appealing because of their inefficiency in measurements, whereas those corresponding to unitary Pauli monomials are more practical.

From the previous discussion, we know that for $k = 6$, there exists a unitary Pauli monomial corresponding to the third stabilizer entropy. This fact makes the $\alpha=3$ stabilizer entropy to be the minimal-order monomial that quantifies magic and can be experimentally measured efficiently. As we will see in \cref{sec:propertytesting}, this fact also has profound consequences for stabilizer property testing.

However, it is intriguing to explore whether higher-order polynomials give rise to other measures with similar properties.
Extending this to $k = 8$, as further discussed in \cref{sec:cliffordcommutantexample}, there is only one unitary Pauli monomial in addition to $\Omega_6$, that is, $\Omega_{6,6}\coloneqq \Omega(1,1,1,1,1,1,0,0)\Omega(0,0,1,1,1,1,1,1)$. Graphically, we find
\be
\Delta_{6,6}(\psi) = \tr(\Omega_{6,6}\psi^{\otimes 8}) \tr\left(\monomialdiagram{8}{{1,2,3,4,5,6},{3,4,5,6,7,8}}{}\,\,\,\, \psi^{\otimes 8}\right) = \frac{1}{d^2} \sum_{P,Q\in\mathbb P_n} \tr^2(P \psi) \tr^4(PQ \psi) \tr^2(Q \psi).
\ee  
However, since the Pauli monomial described above is the product of two primitives, we know that $m(\Omega_{6,6})=2$ and, therefore, achieves a typical value, when averaged over Haar random states, of $d^{-2}+O(d^{-3})$. 

\begin{open} As noted above, invariance under Clifford operations alone is not sufficient for a quantity to qualify as a monotone. The generalized stabilizer purities defined in \cref{def:generalizedstabilizerpurity} have not yet been proven to be monotones for magic-state resource theory. Specifically, the question of whether $\Omega$-stabilizer purity for $\Omega$ not being a primitive Pauli monomial (hence corresponding to stabilizer entropies) is non-decreasing under stabilizer operations remains an open question. See also Ref.~\cite{Turkeshi_2025}. 
\end{open}

\subsubsection{Stabilizer entropy and Bell magic}
Bell magic is a magic measure, i.e., a quantity invariant under Clifford unitaries (though it is not known to be a magic monotone), introduced in Ref.~\cite{haugbellmagic}. It has the property of being experimentally measurable via Bell difference sampling~\cite{Hangleiter_2024}. In \cref{th:measurablemagicmonotones}, we proved that any measure of this kind, denoted as $M(\psi)$, belongs to the $k$-th order commutant of the Clifford group, where $k$ corresponds to the number of copies in the unbiased estimator. In this section, we first express Bell magic in terms of Pauli monomials and then show that it is tightly upper and lower bounded by stabilizer entropies.

\begin{definition}[Bell magic~\cite{haugbellmagic}]\label{def:bellmagic} Let $\psi$ be a pure quantum state. The Bell magic is defined as
\be
B(\psi)=\sum_{P_1,P_2\in\mathbb{P}_n }\mathcal{Q}(P_1)\mathcal{Q}(P_2)\|[P_1,P_2]\|_\infty,\quad \mathcal{Q}(P_1)\coloneqq\frac{1}{d^2}\sum_{Q\in\mathbb P_n}\tr^2(Q\psi)|\tr(PQ\psi)|^2
\ee
and it obeys 1) $B(\psi)=0$ if and only if $\psi$ is a stabilizer state and 2) it is invariant under Clifford unitaries.  
\end{definition}
\begin{lemma}[Bell magic in terms of Pauli monomials]\label{lem:bellmagicpaulimonomials}
The Bell magic can be expressed as an expectation value of the unitary Pauli monomial $\Omega_{6,6}$ belonging to $\com(\mathcal{C}_n^{\otimes 8})$. In particular, we find
    \be
B(\psi)=1-\tr(\Omega_{6,6}\psi^{\otimes 8})=1-\tr\left(\monomialdiagram{8}{{1,2,3,4,5,6},{3,4,5,6,7,8}}{}\,\,\,\psi^{\otimes 8}\right).
    \ee
\begin{proof}
To prove the result, let us first observe that we can rewrite
\be
\|[P_1,P_2]\|_\infty&=1-\chi(P_1,P_2)\\
|\tr(PQ\psi)|^2&=\tr^2(QP\psi)\chi(P,Q)
\ee
Hence, we can write the Bell magic as the following expression:
    \be
B(\psi)=1-\frac{1}{d^4}\sum_{P_1,P_2,Q_1,Q_2\in\mathbb P_n}\tr^2(Q_1\psi)\tr^{2}(Q_2\psi)\tr^2(Q_1P_1\psi)\tr^2(Q_2P_2\psi)\chi(P_2,Q_2)\chi(P_1,Q_1)\chi(P_1,P_2)
    \ee
This describes a Pauli polynomial, where the latter term can be expressed with the following diagrammatic representation
\be
\frac{1}{d^4}\sum_{P_1,P_2,Q_1,Q_2\in\mathbb P_n}\tr^2(Q_1\psi)\tr^{2}(Q_2\psi)\tr^2(Q_1P_1\psi)\tr^2(Q_2P_2\psi)\chi(P_2,Q_2)\chi(P_1,Q_1)\chi(P_1,P_2)=\tr(\monomialdiagram{8}{{1,2,5,6},{5,6},{3,4,7,8},{7,8}}{0:1,2:3,1:3}\psi^{\otimes 8}).
\ee
To conclude the proof, it is sufficient to use the machinery developed in \cref{sec:graphcalc} and manipulate this diagram
\be 
\monomialdiagram{8}{{1,2,5,6},{5,6},{3,4,7,8},{7,8}}{0:1,2:3,1:3} = \monomialdiagram{8}{{1,2,5,6},{5,6},{7,8},{3,4,7,8}}{0:1,1:2,2:3} = \monomialdiagram{8}{{1,2,5,6},{5,6,7,8},{7,8},{3,4,7,8}}{0:1,1:3,2:3} = \monomialdiagram{8}{{1,2,5,6},{5,6,7,8},{7,8},{3,4}}{0:1,1:2} = \monomialdiagram{8}{{1,2,5,6},{3,4},{3,4,5,6,7,8},{7,8}}{0:2} =  \monomialdiagram{8}{{3,4},{1,2,3,4,7,8},{3,4,5,6,7,8},{7,8}}{}
\ee 
Notice that the last diagram is the product of odd primitives without phases and, as such, it is a unitary operator. Moreover, as we discussed in \cref{Sec:magicstateresourcetheory}, the operator corresponding to the expectation value of a magic measure can be projected in the symmetric subspace without loss of generality. As a consequence, the permutation can be canceled and we are left with the following identity:
\be\label{eqbell2}
\frac{1}{d^4}\sum_{P_1,P_2,Q_1,Q_2\in\mathbb P_n}\tr^2(Q_1\psi)\tr^{2}(Q_2\psi)\tr^2(Q_1P_1\psi)\tr^2(Q_2P_2\psi)\chi(P_2,Q_2)\chi(P_1,Q_1)\chi(P_1,P_2)=\tr\left(\monomialdiagram{8}{{1,2,3,4,7,8},{3,4,5,6,7,8}}{}\,\,\,\psi^{\otimes 8}\right).
\ee
The $\Omega$ in \cref{eqbell2}  under the adjoint action of permutations corresponds to the unitary Pauli monomial $\Omega_{6,6}$ and since $\psi^{\otimes 8}$ is invariant under the action of permutation, we can rearrange the diagram concluding the proof. 
\end{proof}
\end{lemma}
As speculated in \cref{Sec:magicstateresourcetheory}, any computable magic measure can be recast as the expectation value of an operator residing in the commutant of the Clifford group. Since any operator in the commutant of the Clifford group is a linear combination of Pauli monomials, which, in \textit{normal form}, can be expressed as products of primitive Pauli monomials, and since the expectation values of these primitives correspond to stabilizer entropies, it is natural to conjecture that any such computable monotone is equivalent to the stabilizer entropy. In the following theorem, we demonstrate that for Bell magic, this conjecture holds true: we establish upper and lower bounds for Bell magic in terms of the stabilizer entropy, thereby proving the universality of the stabilizer entropy in relation to Bell magic.
\begin{theorem}[Stabilizer entropy upper and lower bounds Bell magic]\label{th:stabentropyownsbellmagic} Let $B(\psi)$ the Bell magic defined in \cref{def:bellmagic} and let $M_{\alpha}=\frac{1}{1-\alpha}\log\Delta_{2\alpha}(\psi)$ be the stabilizer entropy defined in \cref{def:stabilizerentropies}. The  bounds 
\be
1-\Delta_{4}^{2}(\psi)\le B(\psi)\le 1-\Delta_{6}^{2}(\psi)
\ee
hold true, which, in terms of stabilizer entropies and \textit{additive Bell magic} defined in \cite{haugbellmagic} $-\log(1-B(\psi))$, the above bounds can be expressed in terms of the stabilizer entropy as
\be
2M_{3}(\psi)\le 2M_{2}(\psi)\le -\log(1-B(\psi))\le 4M_{3}(\psi).
\ee

\begin{proof}
    Using the result of \cref{lem:bellmagicpaulimonomials}, we can write Bell magic in terms of $\Omega_{6,6}$:
\be
B(\psi)=1-\frac{1}{d^2}\sum_{P,Q\in\mathbb P_n}\tr^2(P\psi)\tr^2(Q\psi)\tr^4(PQ\psi). = 1 - \tr[\left(\frac{1}{d}\sum_{P\in\mathbb P_n}\tr^2(P\psi) P^{\otimes 4} \psi^{\otimes 4}\right)^2]
\ee
Then, we have the following inequality between expectation values of Pauli monomials
\be
\Delta_{6,6}(\psi)=\tr\left(\monomialdiagram{8}{{1,2,3,4,5,6},{3,4,5,6,7,8}}{}\,\,\,\psi^{\otimes 8}\right)\ge \tr^2\left(\monomialdiagram{6}{{1,2,3,4,5,6}}{}\,\,\,\psi^{\otimes 6}\right)=\Delta_{6}^2(\psi).
\ee
where we used the fact that given a Hermitian operator $H$ $\tr(H^2\psi^{\otimes 4})\ge\tr^2(H\psi^{\otimes 4})$, this holds in our case because $\frac{1}{d}\sum_{P\in\mathbb P_n}\tr^2(P\psi) P^{\otimes 4} $ is a Hermitian operator with infinity norm smaller than $1$. 
Now, let us proceed with the upper bound. 
\begin{align}
\frac{1}{d^2}\sum_{P,Q}\tr^2(P\psi)\tr^2(Q\psi)\tr^4(PQ\psi)& =\frac{1}{d^2}\sum_{P,Q}\tr^2(P\psi)\tr^2(PQ\psi)\tr^2(Q\psi)\tr^2(PQ\psi) \nonumber\\ 
& \le \sqrt{\frac{1}{d^2}\sum_{P,Q\in\mathbb P_n} \tr^4(P\psi)\tr^4(PQ\psi)  }\sqrt{\frac{1}{d^2}\sum_{P,Q\in\mathbb P_n} \tr^4(Q\psi)\tr^4(PQ\psi)  }\nonumber\\ 
& = \frac{1}{d^2}\sum_{P,Q\in\mathbb{P}_n}\tr^{4}(P\psi)\tr^4(PQ\psi) = \tr\left(\monomialdiagram{8}{{1,2,3,4,5,6,7,8},{5,6,7,8}}{}\,\,\,\psi^{\otimes 8}\right)= \tr\left(\monomialdiagram{8}{{1,2,3,4},{5,6,7,8}}{}\,\,\,\psi^{\otimes 8}\right). 
\end{align}
where we have used H\"older inequality to get the 
upper bound. Putting it all together gives
\be
\Delta_{6,6}(\psi)=\tr\left(\monomialdiagram{8}{{1,2,3,4,5,6},{3,4,5,6,7,8}}{}\,\,\,\psi^{\otimes 8}\right)\le \tr\left(\monomialdiagram{8}{{1,2,3,4},{5,6,7,8}}{}\,\,\,\psi^{\otimes 8}\right)=\Delta_{4}^{2}(\psi),
\ee
which concludes the proof. 

\end{proof}
\end{theorem}
The above theorem confirms our intuition: stabilizer entropies tightly bound Bell magic from above and below, effectively an intermediate between $\Delta_4^2$ and $\Delta_6^2$. 
Bell magic gained significant attention upon its introduction because of its efficient measurability via Bell difference sampling, while the efficient measurability of stabilizer entropies has initially been unclear. However, as shown, odd stabilizer entropies are related (through the non-linear function $-\log$) to the expectation values of unitary primitive Pauli monomials, which can be directly measured as we know that Pauli monomials factorize on qubits (see \cref{lem:monomialsfactorizeonqubits}). 

\subsubsection{The triple purity}
So far, we have shown that any magic measure is a Pauli polynomial and, therefore, a linear combination of generalized stabilizer purities. These generalized stabilizer purities provide a meaningful extension of stabilizer purities (entropies), as they can be expressed as the expectation value of an operator obtained from the product of $\Omega_{2\alpha}$, whose expectation value yields $\Delta_{2\alpha}$ (see \cref{def:stabilizerentropies}). The term "entropy" in stabilizer entropy arises from the fact that these quantities serve as entropy measures for the Pauli distribution of a pure quantum state. Specifically, for any pure state vector $\ket{\psi}$, one can define a normalized distribution over Pauli operators given by $\Xi_{\psi}(P) \coloneqq \frac{\tr^2(P\psi)}{d}$, known as the \textit{characteristic distribution}. The R\'enyi entropies of $\Xi_{\psi}$ correspond, up to an offset, to stabilizer entropies. Since stabilizer entropies are magic monotones~\cite{Leone_2024}, the property of magic can be understood purely as a function of the Pauli distribution $\Xi_{\psi}$. However, not every magic measure depends solely on the Pauli distribution. This distinction becomes evident when considering generalized stabilizer purities. In particular, for $k=9$, we encounter the first measurable generalization of stabilizer purities—one that extends beyond being a mere function of the Pauli distribution.
 
\begin{definition}[Triple purity] The triple purity $\Delta_{\text{triple}}$ is defined as the expectation value of the following Pauli monomial
\begin{align}
    \Delta_{\text{triple}}(\psi)=\tr\left(\monomialdiagram{9}{{1,2,3,4,5,6},{4,5,6,7,8,9}}{}\,\,\,\, \psi^{\otimes 9}\right) =\frac{1}{d^2}\sum_{P,Q\in\mathbb P_n}\tr^3(P\psi)\tr^3(Q\psi)\tr^3(PQ\psi).
\end{align}
\end{definition}
From the definition of triple purity it becomes evident that $\Delta_{\text{triple}}$ is not a mere function of the Pauli distribution, as not only the absolute value of expectation values of Pauli operators $|\tr(P\psi)|$ matters but also their signs.

\subsection{Property testing of stabilizer states}\label{sec:propertytesting}
In this subsection, we revisit the well-studied problem of property testing for stabilizer states~\cite{montanaro2018surveyquantumpropertytesting,gross_schurweyl_2019,bu2023stabilizertestingmagicentropy,grewalImprovedStabilizerEstimation2023,arunachalam2024polynomialtimetoleranttestingstabilizer,chen2024stabilizerbootstrappingrecipeefficient,hinsche2024singlecopystabilizertesting,bao2024toleranttestingstabilizerstates}, connecting it with the commutant structure of the Clifford group. The idea of connecting property testing with the commutant—or the invariant symmetries of the problem—is well established in the literature (see, e.g.,~\cite{montanaro2018surveyquantumpropertytesting,gross_schurweyl_2019,hinsche2024singlecopystabilizertesting}). Our aim here is not to introduce new methods, but to emphasize how the characterization of the Clifford group commutant constrains the design of optimal measurement strategies in this context. In particular, we show that \emph{generalized stabilizer purities} emerge as essentially the only viable measurements for stabilizer property testing when the number of available copies is fixed.

We begin by briefly introducing the problem. Property testing of stabilizer states involves determining whether an unknown quantum state is close to or far from the stabilizer set, based on measurements on multiple copies of the state. Formally, we define the problem as follows.
\begin{problem}[Property testing of stabilizer states]  \label{problem:stabilizer-testing}
    Let $ 1 \ge \varepsilon_B \geq \varepsilon_A > 0 $. Given access to $ k $ copies of an unknown quantum state $ \rho $, assume that $ \rho $ satisfies one of the following conditions:
    \begin{itemize}
        \item $ \rho $ is $ \varepsilon_A $-close to some stabilizer state, i.e., there exists $ \phi \in \stab(n) $ such that $ \|\rho - \phi\|_1 \leq \varepsilon_A $.
        \item $ \rho $ is at least $ \varepsilon_B $-far from all stabilizer states, that is, $ \min_{\phi \in \stab(n)} \|\rho - \phi\|_1 \geq \varepsilon_B $.
    \end{itemize}
    The goal is to design a measurement strategy that acts on $ k $ copies of $ \rho $ that distinguishes these two cases with probability at least $ 2/3 $.
\end{problem}

Equivalently, the task is to find a two-outcome POVM $ \{E, I - E\} $ such that
\begin{align}
\label{eq:POVMtest}
   &\text{if } \rho \text{ is }  \varepsilon_A \text{-close to a stabilizer state, then } \tr(E \rho^{\otimes k}) \geq a, \\
   &\text{if } \rho \text{ is } \varepsilon_B \text{-far from all stabilizer states, then } \tr(E \rho^{\otimes k}) \leq b.
\end{align}
We note that the threshold values $ a  $ and $ b  $ are arbitrary, and any constants satisfying $ a \coloneqq \frac{1}{2} + \Theta(1) $ and $ b \coloneqq \frac{1}{2} - \Theta(1) $ suffice to solve the property testing problem. Moreover, if an tester is given $k$-copies of a stabilizer state or $k$-copies of a non-stabilizer state with $1/2$ probability, the probability of success of the property testing algorithm associated with the POVM $E$ is given by
\be\label{successprobability}
p_{\text{succ}}= \frac{1}{2}(a+(1-b))\,.
\ee
We observe that a natural figure of merit for the $k$-copy testing problem is the maximal success probability achievable over all possible POVMs acting on $k$ copies.
We now show that the optimal POVM must lie in the commutant. This follows from the observation that any optimal measurement strategy must respect the symmetry constraints imposed by the problem. Since the commutant consists of operators that commute with the symmetry group, it forms the natural space in which the optimal POVM resides.

\begin{proposition}[Characterization of the optimal two-outcome POVM]
    An optimal two-outcome POVM $ E $ solving the property testing problem (\cref{problem:stabilizer-testing}) acting on $ k $-copies of the unknown state must lie in the commutant.
\end{proposition}

\begin{proof}
    To prove this, we show that for any POVM $ E $ satisfying \cref{eq:POVMtest}, the POVM given by
    \begin{align}
    E' =\Phi_{\cl}^{(k)}(E) =\frac{1}{|\mathcal{C}_n|}\sum_{C \in \mathcal{C}_n}  C^{\otimes k} E C^{\dagger \otimes k}
    \end{align}
    performs at least as well as $ E $ in the testing problem.

    First, if $ \rho $ is $ \varepsilon_A $-close to a stabilizer state $ \phi $, then for any Clifford $ C $, the state $ C \rho C^\dagger $ is also $ \varepsilon_A $-close to a stabilizer state (i.e., $ C^\dagger \phi C $), due to the unitary invariance of the trace distance. This gives:
    \begin{equation}
    \tr(E' \rho^{\otimes k}) = \frac{1}{|\mathcal{C}_n|}\sum_{C \in \mathcal{C}_n}  \tr(E (C \rho C^\dagger)^{\otimes k}) \geq \frac{2}{3}.
    \end{equation}

    Similarly, if $ \rho $ is $ \varepsilon_B $-far from all stabilizer states, then $ C \rho C^\dagger $ is $ \varepsilon_B $-far from all stabilizer states for any Clifford $ C $. This gives:
    \begin{equation}
    \tr(E' \rho^{\otimes k}) \leq \frac{1}{3}.
    \end{equation}

    Therefore, we conclude that $ E' $ satisfies the property testing criteria and performs as well as $ E $, which implies that the optimal POVM must lie in the commutant of the Clifford group.
\end{proof}

A concrete family of operators in the commutant are the projectors onto stabilizer tensor powers constructed by Gross, Nezami, and Walter~\cite[Theorem~5.6]{gross_schurweyl_2019} (denoted \(\Pi_k^{\min}\) there). In the language of Pauli monomials, these projectors are linear combinations of products of primitive monomials \(\Omega(v)\). For qubits and \(k=6\), the relevant projector is proportional to \(\mathbb{1} + \Omega_6\).

\subsubsection{No efficient testing algorithm with \texorpdfstring{$\leq 5$}{ ≤ 5} copies}

We now prove that no stabilizer property testing algorithm can succeed using at most five copies of the unknown state, as also noted in Refs.~\cite{Damanik2018Optimality,gross_schurweyl_2019}.. This follows from the fact that the commutant of the Clifford group up to $k=5$ contains only primitives proportional to projectors. 

\begin{proposition}
    No quantum algorithm can solve the stabilizer property testing problem (\cref{problem:stabilizer-testing}) using at most five copies of the unknown state.
\end{proposition}

\begin{proof}
To prove the statement and thus the hardness of stabilizer testing using $k=5$ copies, we can focus on showing the hardness, for the less hard case for which $\varepsilon_A=0$ and the unknown state is assumed to be pure.

Suppose, towards contradiction, that there exists a two-outcome POVM tester $ \{E, I - E\} $ satisfying \cref{eq:POVMtest}, capable of testing stabilizer states. First, let us show that the existence of such POVM also implies the existence of a POVM capable of separating Haar random states from uniformly random stabilizer states.  First, by definition of stabilizer tester we have
\begin{equation}
    \mathbb{E}_{\psi \sim \text{STAB}} \tr(E \psi^{\otimes k}) \geq \min_{\phi\in\stab}\tr(E\psi^{\otimes k})\ge  \frac{2}{3}.
\end{equation}
On the other hand, we also have
\begin{align}
    \mathbb{E}_{\psi \sim \text{Haar}} \tr((I - E) \psi^{\otimes k}) 
    &\geq  \frac{2}{3} \cdot \Pr_{\psi \sim \text{Haar}} 
    \left( \min_{\phi \in \text{STAB}} \|\psi - \phi\|_1 \geq \varepsilon_B \right)   \\
    \nonumber 
    &= \frac{2}{3} \cdot \Pr_{\psi \sim \text{Haar}} 
    \left( 2\min_{\phi \in \text{STAB}}\sqrt{1- |\braket{\psi}{\phi}|^2 } \geq \varepsilon_B \right)   \\
     \nonumber
    &=  \frac{2}{3} \left(1 - \Pr_{\psi \sim \text{Haar}} 
    \left( \max_{\phi \in \text{STAB}} |\braket{\psi}{\phi}|^2 \ge 1 - \frac{\varepsilon_B^2}{4} \right)\right)  \\
     \nonumber
    &\geq  \frac{2}{3}\left(1 - \exp \left(0.52n^2 - (2^{n} - 1) \left( 1 - \frac{\varepsilon_B^2}{4} \right) \right)\right) \\
     \nonumber
    &= \frac{1}{2} + \Omega(1),
     \nonumber
\end{align}
where in the first step we have used  that $1-\tr(E \psi^{\otimes k}) \geq \frac{2}{3}$ if $\psi$ is $\varepsilon_B$-far from all stabilizer states, while the fourth step follows from the union bound and L\'evy's lemma (as done explicitly in Ref.~\cite{liu_manybody_2022}).
Therefore, we have
\begin{align}
    \tr\left(E \left(\mathbb{E}_{\psi \sim \text{STAB}} \psi^{\otimes k} - \mathbb{E}_{\psi \sim \text{Haar}} \psi^{\otimes k}\right)\right) 
    &= \mathbb{E}_{\psi \sim \text{STAB}} \tr(E \psi^{\otimes k}) + \mathbb{E}_{\psi \sim \text{Haar}} \tr((I - E) \psi^{\otimes k}) - 1 \\
    &= \Omega(1).
    \nonumber
\end{align}
However, this provides a contradiction, since the one-norm distance between the moment operators satisfies:
\begin{equation}
   \sup_{E: 0 < \|E\|_{\infty} \leq 1} \left|\tr\left(E \left(\mathbb{E}_{\psi \sim \text{STAB}} \psi^{\otimes k} - \mathbb{E}_{\psi \sim \text{Haar}} \psi^{\otimes k}\right)\right)\right| 
   = \left\| \mathbb{E}_{\psi \sim \text{STAB}} \psi^{\otimes k} - \mathbb{E}_{\psi \sim \text{Haar}} \psi^{\otimes k} \right\|_1=\exp(-\Omega(n))\,,
\end{equation}
i.e., it is exponentially small in $ n $ for $ k \leq 5 $, as we show in \cref{example:k5}.
\end{proof}

\subsubsection{Optimal testing by measuring the unitaries in the commutant basis}
In the previous section, we established that there is no strategy to solve the stabilizer testing problem when strictly fewer than six copies of the unknown state are available. This result has been derived by reducing the problem of testing stabilizer states to a simpler problem: testing Haar-random states against random stabilizer states, for which we showed hardness.

In this section, we first demonstrate that testing Haar-random states from random stabilizer states becomes feasible when six copies of the unknown state are provided. 
Specifically, we show that the two-outcome POVM element $ E = \frac{1 + \Omega_6}{2} $, where $ \Omega_6 \coloneqq \frac{1}{d} \sum P^{\otimes 6} $, enables the testing of Haar-random states against random stabilizer states.

First, note that $ E $ is a valid POVM element since $ \Omega_6 $ is a unitary operator, and its operator norm is equal to one (see \cref{le:unitaryprojmon}). Moreover, we have
\begin{align}
   \tr\left(E \left(\mathbb{E}_{\psi \sim \text{STAB}} \psi^{\otimes 6} - \mathbb{E}_{\psi \sim \text{Haar}} \psi^{\otimes 6}\right)\right)
    &= \frac{1}{2}\mathbb{E}_{\psi \sim \text{STAB}} \tr\left(\Omega_6 \psi^{\otimes 6}\right) - \frac{1}{2}\mathbb{E}_{\psi \sim \text{Haar}} \tr\left(\Omega_6 \psi^{\otimes 6}\right) \\
    &= \frac{1}{2} - \exp(\Omega(-n)) = \Omega(1),
    \nonumber
\end{align}
where in the second line we have used \cref{example:k6tracedistance}. Thus, we conclude that $ \Omega_6 $ allows for efficient testing of Haar-random states against random stabilizer states, i.e., $ \left\| \mathbb{E}_{\psi \sim \text{STAB}} \psi^{\otimes 6} - \mathbb{E}_{\psi \sim \text{Haar}} \psi^{\otimes 6} \right\|_1 = \Omega(1) $.

In general, we can state the following.

\begin{proposition}[Best POVM for testing given any $ k = O(n^{2-c}) $ for $c>0$ copies must be in the span of the unitary basis elements of the commutant]\label{prop:bestpovmisaunitary}
    For any $ k = O(1) $ copies of the unknown state, the best POVM which solves the property testing problem must lie in the span of the unitary elements in the commutant basis. If we have the additional assumption that the unknown state is pure, the best POVM is a linear combination of basis elements of the commutant, projected onto the symmetric subspace that are also unitary.
\end{proposition}
\begin{proof}
    This follows directly from \cref{lem:unitarymonomialsdominate} by expressing the POVM element as $E\in\com(\mathcal{C}_n^{\otimes k})$. Indeed, we can express it as $E=\sum_{\Omega\in\mathcal{P}}\alpha_{\Omega}(E)\Omega$ with $\Omega$ being a Pauli monomial. Using \cref{lem:unitarymonomialsdominate}, and using that $\|E\|_{\infty}\le 1$, we know
\be
\alpha_{\Omega}(E)\le O(2^{\frac{3}{2}k^2-n})
\ee
Therefore, we can express the expectation value of $E$ on $\psi^{\otimes k}$ as $\tr(E\psi^{\otimes k})=\sum_{\Omega}\alpha_{\Omega}(E)\tr(\Omega\psi^{\otimes k})$. 
Since $\tr(\Omega\rho)\le 1$ (see \cref{th:faithfulness}), then it immediately follows that
\be
\tr(E\psi^{\otimes k})=\sum_{\Omega\in\mathcal{P}_U}\alpha_{\Omega}(E)\tr(\Omega \psi^{\otimes k})+O(2^{2k^2-n}),
\ee
where we have used the fact the dimension of the $k$-th order commutant in this regime is $O(2^{k^2/2})$~\cite{nebe2000invariantscliffordgroups}.
Therefore, the operator $E=\sum_{\Omega\in\mathcal{P}_U}\alpha_{\Omega}(E)\Omega$ performs as well as the POVM $E$ up to an exponentially small error in $n$ (in the hypothesis that $k^2=o(n^{2-c})$ for $c>0$. This proves the statement.
\end{proof}
In contrast to Theorem~5.6 in~\cite{gross_schurweyl_2019},  Proposition~\ref{prop:bestpovmisaunitary} shows that the optimal distinguishing POVM is essentially unitary (i.e., supported on unitary Pauli monomials). This explains why the minimal number of copies required for a dimension-independent test in~\cite{gross_schurweyl_2019} corresponds to the first $k$ for which the commutant contains a non-trivial unitary element (e.g., $k=6$ for qubits).

\begin{corollary}[Best POVM for testing given $ k = 6 $ copies must be given by $ \Omega_6 $]\label{cor:testingP6}
    For $k=6$ copies of the unknown pure state, the best POVM which solves the property testing problem (even in the possibly general tolerant scenario) must be given by $ E = \alpha_1\mathbb{1} + \alpha_2\Omega_6 $.
\end{corollary}
\begin{proof}
    This follows directly from the previous proposition and the fact that $\Omega_6$ is the only unitary residing in the $k = 6$ order commutant of the Clifford group, when projected onto the symmetric subspace. For an explicit construction of the $k = 6$ order commutant, see \cref{ssec:k6}.
\end{proof}
In fact, Ref.~\cite{gross_schurweyl_2019} has shown that $ \Omega_6 $ enables testing using only $6$ copies. Using the same argument and our computation of the $k=8$ commutant, it follows that an optimal algorithm with access to $8$ copies can only use $\Omega_6$ and $\Omega_{6,6}$ associated with Bell magic.

\subsubsection{An operational meaning for stabilizer entropies}
In this section, we build on the results of the previous section to show that stabilizer property testing provides a simple and powerful operational interpretation of the $(\alpha=3)$ stabilizer entropy.

Specifically, we formulate the stabilizer testing problem as follows. A tester is given $k = 6$ copies of an unknown pure quantum state, which is either a stabilizer state $\sigma$ with probability $1/2$, or a pure (non-stabilizer) state $\ket{\psi}$ with probability $1/2$. The task of the tester is to determine which of the two cases applies. From \cref{cor:testingP6}, we know that the optimal POVM is of the form $E = \alpha_1 \mathbb{1} + \alpha_2 \Omega_6$. The condition $E \ge 0$ for a valid POVM implies $\alpha_1 \ge \alpha_2$, since $\Omega_6$ is a unitary operator (cf. \cref{lem:propertiesprimitivepauli}).

Using \cref{successprobability}, the success probability for correctly identifying the state in this scenario is given by
\begin{align}
p_{\text{succ}} &= \frac{1}{2} \tr(E \sigma^{\otimes 6}) + \frac{1}{2} \tr((\mathbb{1} - E) \psi^{\otimes 6}) \\
\nonumber
&= \frac{1}{2}[(\alpha_1 + \alpha_2) + (1 - \alpha_1) - \alpha_2 \Delta_6(\psi)] \\
&= \frac{1}{2}[1 + \alpha_2(1 - \Delta_6(\psi))],
\nonumber
\end{align}
where we recall the definition of the stabilizer purity $\Delta_6(\psi)$ from \cref{def:stabilizerentropies}, and the fact that $\Delta_6(\sigma) = 1$ for all stabilizer states $\sigma$. 

Maximizing over $\alpha_2$ under the constraint $\alpha_1 \ge \alpha_2$, we obtain the optimal success probability for the $k=6$ property testing problem:
\begin{equation}\label{eq21212}
p_{\text{succ}}^{\text{optimal}} = \frac{1}{2} + \frac{1}{4}(1 - 2^{-2M_3(\psi)}),
\end{equation}
and the optimal property testing algorithm (i.e., the one with the highest success probability) is given by $\frac{\mathbb{1}+\Omega_6}{2}$.  \cref{eq21212} endows the $(k=3)$ stabilizer entropy with a clear operational meaning, revealing a deep connection between magic-state resource theory and stabilizer property testing.

\subsection{Clifford-symmetric quantum states}\label{sec:cliffordsymmetricstates}

The findings here immediately translate into a framework of symmetric states. Group-symmetric states play an important role in quantum information
theory, prominently in entanglement theory, where states that are $U\otimes U$-symmetric go under the name of \emph{Werner states}. The commutant $ \com(U\otimes U:U\in {\cal U}(2^n))$ is spanned by the identity  $ \mathbb 1 $ and the swap operator $ T$ \cite{PhysRevA.40.4277}. Hence,
the set of Werner states is a one-dimensional set of quantum states. In fact, many interesting research questions can be decided on this set, such as their entanglement distillability. Being able to precisely characterize the commutant of symmetry groups is also instrumental for computing asymptotic entanglement measures under symmetry \cite{PhysRevLett.87.217902}. The importance of 
being able to characterize sets of group 
symmetric 
quantum states for the field of quantum information 
theory can hardly be overestimated. 

Similar statements can also be made for Clifford-symmetric states that are invariant under ${\cal C}_n$, again with
\begin{align}
    \com(\mathcal{C}_n^{\otimes k}) \coloneqq \{ O \in \mathcal{B}(\mathcal{H}^{\otimes k}) : C^{\dagger \otimes k} O C^{\otimes k} = O \, \text{ for all } C \in \mathcal{C}_n \}.
\end{align}
The set of symmetric quantum states is then nothing but the intersection of the commutant with the quantum mechanical state space, 
\begin{align}
 {\cal D}_{\cl}({\mathcal{H}^{\otimes k}}) \coloneq \com(\mathcal{C}_n^{\otimes k})\cap {\cal D}(\mathcal{H}^{\otimes k}),
 \end{align}
 which in turn is the positive semi-definite cone intersected with unit trace operators.
 Since the commutant is spanned by Pauli monomials $\Omega\in\mathcal{P}$, membership to
 the
 set $ {\cal D}_{\cl}({\mathcal{H}^{\otimes k}}) $ of 
 Clifford-symmetric quantum states is characterized by the semi-definite feasibility problem
 \cite{Boyd2004}
 \begin{eqnarray}
 	& \sum_{\Omega\in\mathcal{P}}c_{\Omega}\Omega\geq 0,\\
	&\sum_{m=0}^{k-1}d^{k-m}\sum_{\Omega\,:\, m(\Omega)=m}c_{\Omega}=1,
 \end{eqnarray}
 which can be solved for relatively small $k$, with interior point 
 methods, given the dimension  
 $\dim(\com(\mathcal{C}_n^{\otimes k}))=\exp(\Theta(k^2))$ (see \cref{cor:dimcom})
 of $\com(\mathcal{C}_n^{\otimes k})$.

\subsection{Explicit 
characterization of the commutant up to \texorpdfstring{$k=8$}{k=8}}\label{sec:cliffordcommutantexample}

In this section, we explicitly construct the commutant of the Clifford group for low values of $ k $. Specifically, we begin with $ k=4 $, which represents the first non-trivial case, and extend our analysis up to $ k=8 $. We provide explicit constructions of the Clifford-Weingarten functions (when restricted to pure quantum states), as introduced in \cref{sec:cliffordweingartencalculus}, and illustrate their applications with concrete examples.

\subsubsection{\texorpdfstring{$k=4$}{k=4}} Let us start from the first non-trivial case, i.e., $k=4$. There are two possible types of Pauli monomials to be defined. One class corresponds to permutations $S_4$. The second non-trivial element are generated by $\Omega_4$ equal to
\be
\Omega_4=\frac{1}{d}\sum_{P}P^{\otimes 4}
\ee
which is proportional to a projector $(\Omega_4)^2=d\Omega_4$. The dimension of the commutant is $\dim(\com(\mathcal{C}_n^{\otimes 4}))=30$. As such, Weingarten calculus is straightforward. Through graphical calculus, one can verify the following identities
\be
\Omega_4=\Omega_4 T_{(12)(34)}=\Omega_4 T_{(13)(24)}=\Omega_4 T_{(14)(23)}
\ee
For example, the 
first identity is verified via the following diagrammatic calculation
 \setmonomialscale{3.5mm}
       \begin{align}
        \monomialdiagram{4}{{1,2,3,4},{1,2},{3,4}}{}\quad\textbf{=}\quad\monomialdiagram{4}{{1,2,3,4},{3,4},{3,4}}{}\quad\textbf{=}\quad\monomialdiagram{4}{{1,2,3,4}}{}
    \end{align}
and the others follow identically. It follows that, in addition to permutations, there are only $6$ additional elements, that is, $\Omega_4$, $\Omega_4T_{(12)}$, $\Omega_4T_{(13)}$, $\Omega_4T_{(14)}$, $\Omega_4T_{(123)}$, $\Omega_4T_{(132)}$. 
The Gram-Schmidt matrix is a $30\times 30$ matrix, which in a compact notation reads
\be
\mathcal{W}^{(4)}=\begin{pmatrix}
    \tr(T_{\pi}T_{\sigma}) &  \tr(T_{\kappa}\Omega_4)\\ \tr(T_{\kappa}\Omega_4)  & \tr(T_{\kappa}T_{\kappa'}\Omega_4)    \end{pmatrix}
\ee
where $\pi,\sigma\in S_{4}$ runs over all the permutation group $S_{4}$, while $\kappa,\kappa'\in\{(),(12),(13),(14),(123),(132)\}$, for a total of $30$ elements. While in Ref.~\cite{leone_mathematica_2023} we explicitly provide 
the Gram-Schmidt matrix and its inverse for readers interested in utilizing them, here we apply this framework to derive the Clifford orbit of a pure quantum state.
Following \cref{lem:weingartencalculusclifford}, we can now simply compute the Clifford twirling operator $\Phi_{\cl}^{(4)}$, by computing the Clifford-Weingarten functions $(\mathcal{W}^{-1})_{\Omega,\Omega'}$. As such, one 
can express 
the Clifford twirling of any 
operator $O\in\mathcal{B}(\mathcal{H}^{\otimes 4})$ as
\be
\Phi_{\cl}^{(4)}(O)=\begin{pmatrix}
    \tr(T_{\pi}T_{\sigma}) &  \tr(T_{\kappa}\Omega_4)\\ \tr(T_{\kappa}\Omega_4)  & d\tr(T_{\kappa}T_{\kappa'}\Omega_4)    \end{pmatrix}^{-1}
    \begin{pmatrix}
    \tr(T_{\pi}O)\\\tr(T_{\kappa}\Omega_4O)
\end{pmatrix}.
\ee
One can express the Clifford orbit of a pure quantum state $\psi^{\otimes 4}=\ketbra{\psi}^{\otimes 4}$ as
\begin{align}\label{eq:4foldchannel}
\Phi_{\cl}^{(4)}(\psi^{\otimes 4})=p(\psi)\frac{\Pi_{\sym}}{\tr(\Pi_{\sym})}+(1-p(\psi))\frac{\Omega_4\Pi_{\sym}}{\tr(\Omega_4\Pi_{\sym})}
\end{align}
where the (quasi)-probability is equal to
\be
p(\psi)=\frac{(d+3)(d-\tr(\Omega_4\psi^{\otimes 4}))}{(d+4)(d-1)}.
\ee

\begin{example}[Purity fluctuations of Clifford orbit states~\cite{leone_quantum_2021,oliviero_transitions_2021}] 
    In \cref{exampleaverageentropy}, we computed starting from the average purity in a subsystem; we computed the average entropy in a subsystem. Let us now focus on the average purity fluctuations. Let us show the example for both the unitary and Clifford cases to also show the relevant differences. 
    Starting from the unitary case, we have seen in the previous section how the $4$-fold channel is built; let us use these tools to compute the average purity fluctuations. Let us start from the definition of purity fluctuations for a state $\psi$. 
    \begin{align}    
    \Delta_{\psi\sim\haar}\pur(\psi_A) \coloneqq \int \de \psi \pur^2(\psi_A) -\left(\int \de\psi \pur(\psi_A\right)^2 .
    \end{align}    
    That can also be rewritten in a more convenient form as 
    \begin{align}
    \label{purityfluc1}
    \Delta_{\psi\sim\haar}(\pur(\psi_A)) & = \tr(T_{(12)(34)}^{A}\Phi_{\haar}(\psi^{\otimes 4})) - \frac{d_A^2+d_B^2+2d_Ad_B}{(d_Ad_B+1)^2}\\
    \nonumber
    & =  \frac{\sum_{\pi}T^{A}_{(12)(34)}T_{\pi}^{A}T_{\pi}^{B}}{d_Ad_B(d_Ad_B+1)(d_Ad_B + 2)(d_Ad_B + 3)} - \frac{d_A^2+d_B^2+2d_Ad_B}{(d_Ad_B+1)^2} \\ 
     \nonumber
    & = \frac{2 \left(d_A^2-1\right) \left(d_B^2-1\right)}{\left(d_A d_B+1\right)^2 \left(d_A d_B+2\right) \left(d_A d_B+3\right)}
     \nonumber
  \end{align}    
  where we have used 
  the results of \cref{exampleaverageentropy}. 
Shown, as an example, are the 
purity fluctuations with respect to the unitary group; we are now able to show what happens when we instead consider the Clifford group. 
The purity fluctuations read 
\begin{align}
\Delta_{\psi\sim\stab_n }(\pur(\psi_A)) & = \tr(T_{(12)(34)}^{A}\Phi_{\cl}^{(4)}(\psi^{\otimes 4})) - \frac{d_A^2+d_B^2+2d_Ad_B}{(d_Ad_B+1)^2}\\ 
\nonumber
&+ \tr(T_{(12)(34)}^{A}\left(p(\psi)\frac{\Pi_{\sym}}{\tr(\Pi_{\sym})}+(1-p(\psi))(\psi)\frac{\Omega_4\Pi_{\sym}}{\tr(\Omega_4\Pi_{\sym})}\right)) - \frac{d_A^2+d_B^2+2d_Ad_B}{(d_Ad_B+1)^2} \\
\nonumber
& = \frac{\left(d_A^2-1\right) \left(d_B^2-1\right) \left(\left(d_A d_B+1\right) \tr( \Omega_4 \psi ^{\otimes 4})-2\right)}{\left(d_A d_B+1\right){}^2 \left(d_A^2 d_B^2+d_A d_B-2\right)}.
\nonumber
\end{align}
\end{example}
As a last application, we compute the trace norm between the Clifford twirling of a pure state and a Haar random state, which is exponentially vanishing.

\begin{example}[Trace distance of random Clifford orbit]\label{example:k4}
 Let us compute $\|\Phi_{\cl}^{(4)}(\psi^{\otimes 4})-\Phi_{\haar}^{(4)}(\psi^{\otimes 4})\|_1$. 
Let us start by recalling that the twirling over the whole unitary group of a pure quantum state is equivalent to \cref{Sec: orbitofhaarstates}
\begin{align}
\Phi_{\haar}^{(4)}(\psi^{\otimes 4})=\frac{\Pi_{\sym} }{\tr(\Pi_{\sym}) }.
\end{align}
Now by using \cref{eq:4foldchannel} we have that the $1$-norm can be rewritten as follows
\be
\norm{\Phi_{\cl}^{(4)}(\psi^{\otimes 4})-\Phi_{\haar}^{(4)}(\psi^{\otimes 4}) }_{1}=|1-p(\psi)|\norm{\frac{\Omega_4\Pi_{\sym}}{\tr( \Omega_4\Pi_{\sym}) }-\frac{\Pi_{\sym} }{\tr\Pi_{\sym} }}_{1}.
\ee
To proceed, let us 
notice that  given an operator in its spectral decomposition $A=\sum_{i}a_i \ketbra{a_i}{a_i}$ its $1$-norm is equal to $\norm{A}_{1}=\sum_{i}|a_{i}|$. Then, note the following: (i) $\Omega_4\Pi_{\sym}/\tr( \Omega_4\Pi_{\sym}) $ is a mixed quantum state of $\mathcal{H}^{\otimes 4}$ with purity $\tr^{-1}(\Omega_4\Pi_{\sym})$. (ii) $\Pi_{\sym}/\tr\Pi_{\sym}$ and $\Omega_4\Pi_{\sym}^{(4)}/\tr( \Omega_4\Pi_{\sym}) $ commute $[\Pi_{\sym},\Omega_4\Pi_{\sym}]=0$ and thus admit a common basis in the Hilbert space $\mathcal{H}^{\otimes 4}$. (iii) The product of $\Pi_{\sym}$ and $\Omega_1 \Pi_{\sym}$ is proportional to $\Omega_4 \Pi_{\sym}$ as we have $\Omega_4 \Pi_{\sym} \Pi_{\sym}=\Omega_4 \Pi_{\sym}$. (iv) Both $\Pi_{\sym}$ and $\Omega_4 \Pi_{\sym}$ are proportional to projectors and therefore have flat spectra.

Therefore, the eigenvalues $\lambda_{1},\lambda_2$ of $\Phi_{\cl}^{(4)}(\psi^{\otimes 4})-\Phi_{\haar}^{(4)(\psi^{\otimes 4})}$ are 
\be
\lambda_{1}=\frac{1}{\tr \Omega_4\Pi_{\sym} }-\frac{1}{\tr\Pi_{\sym} }
\ee
with multiplicity given by the rank of $\Omega_4\Pi_{\sym}$, i.e., $m_{\lambda_{1} }=\tr(\Omega_4\Pi_{\sym})$ and
\be
\lambda_{2}=-\frac{1}{\tr\Pi_{\sym} }
\ee
with multiplicity given by the difference of the ranks, i.e., $m_{\lambda_{2}}=\tr(\Pi_{\sym})-\tr(\Omega_4\Pi_{\sym})$. Therefore, we get
\be
\norm{\Phi_{\cl}^{(4)}(\psi^{\otimes 4})-\Phi_{\haar}^{(4)}(\psi^{\otimes 4})}_{1}& =(1-p(\psi))(m_{\lambda_{1} }\lambda_{1}+m_{\lambda_{2}}|\lambda_{2}|) \\ & =(1-p(\psi))
\left(2-2\frac{\tr \Omega_4\Pi_{\sym} }{\tr \Pi_{\sym} }\right) \\& =\frac{2}{d(d+1)}\left|\Delta_4(\psi)-4\right|\ee
and for large $d$, we obtain $\norm{\Phi_{\cl}^{(4)}(\psi^{\otimes 4})-\Phi_{\haar}^{(4)}(\psi^{\otimes 4})}_{1}\simeq 2\Delta_4(\psi)/d$.
\end{example}

\subsubsection{\texorpdfstring{$k=5$}{k=5}}
Let us analyze the commutant for the Clifford group for $k=5$. The cardinality of the commutant is given by $|\com(\mathcal{C}_n^{\otimes 5})|=270$. According to \cref{th:algebraicstructurecommutantofclifford}, the commutant of the Clifford group is generated by permutations and $\Omega_4$. It turns out that, other than permutations, there are only $5$ more additional elements that, on top of left and right multiplication for permutations, generate the entire commutant, which are $T_{\pi}\Omega_4 T_{\pi}$. While $\Omega_4$ is the only one with $m(\Omega)=1$, we can ask if there are others with a higher weight.
\setmonomialscale{3.5mm}
       \begin{align}
        \monomialdiagram{5}{{1,2,3,4},{2,3,4,5}}{}\quad\textbf{=}\quad\monomialdiagram{5}{{1,2,3,4},{1,5}}{}\quad\textbf{=}\Omega_4T_{(15)},\quad\quad\quad \monomialdiagram{5}{{1,2,3,4},{2,3,4,5}}{0:1}\quad\textbf{=}\quad\monomialdiagram{5}{{2,3,4,5},{1,2,3,4}}{}\quad\textbf{=}T_{(15)}\Omega_4\label{eq:5commutantgraphical}
    \end{align}
This tells us that the product of any two adjoint permutations of $\Omega_4$ for $i\in[5]$, regardless of whether a phase is included, results in another instance of $\Omega_4$, up to permutations.

Therefore, the $5$ order commutant looks very similar to the $4$ order commutant embedded in a larger space. The reason is that there is no additional register in which to embed the unitary $\Omega_{6}$ that, together with $\Omega_{4}$ and permutations, generates the commutant for any $k$.

Exploiting the fact that we do not really have a different non-trivial object in the commutant for $k=5$, but only $\Omega_4$ embedded in $k=5$ copies of the Hilbert space, it is easy to see that there are $30$ independent elements for each non-trivial $\Omega$ generated by right multiplication for permutations. This can be seen by the fact that any permutation in $S_{5}$ can be written as a permutation $\pi\in S_{4}$ times a permutation in the set $S_{5/4}\coloneqq\{e,(15),(25),(35),(45)\}$, thus generating a total of $30\times 5+5! =270$, recovering the dimension of the full Clifford group. 

 We can thus express the Clifford twirling operator $\Phi_{\cl}^{(5)}(O)$ for any $O\in\mathcal{B}(\mathcal{H}^{\otimes 5})$ as
 \be
\Phi_{\cl}^{(5)}(O)=\sum_{\pi\sigma \in S_5}\alpha_{\sigma,\pi}T_{\sigma}\Omega_4T_{\sigma^{-1}}T_{\pi}+\beta_{\pi}T_{\pi}
 \ee
where the coefficients can be determined by computing the Gram-Schmidt matrix and inverting it to get the Clifford-Weingarten functions. Although we omit the explicit computation of the Clifford-Weingarten functions in this case, we focus again on the computation of the Clifford twirling applied on $5$ tensor copies of a pure quantum state, giving explicit details for its computation.

Note that, from \cref{sec:clorbitquantumstates}, we can write
\begin{align}\label{eq:stateorbitcliffordk5}
\Phi_{\cl}^{(5)}(\psi)=p(\psi)\frac{\Pi_{\sym}}{\tr(\Pi_{\sym})}+(1-p({\psi}))\frac{\Pi_{\sym}\Omega_4\Pi_{\sym}}{\tr(\Omega_4\Pi_{\sym})}
\end{align}
which is very similar to the one for the case $k=4$, but not exactly. Indeed note that while for $k=4$ $[\Omega_4,\Pi_{\sym}]=0$, for $k=5$ we have
\be
\Pi_{\sym}\Omega_4\Pi_{\sym}\propto \sum_{\sigma}(T_{\sigma}\Omega_{4}T_{\sigma^{-1} })\Pi_{\sym}
\ee
and thus clearly $[\Omega_4,\Pi_{\sym}]\neq 0$. To find the probabilities, we proceed in the way sketched in \cref{sec:clorbitquantumstates} and write the matrix $\mathcal{S}$ down as
\be
\mathcal{S}=\frac{1}{\tr(\Pi_{\sym})}\begin{pmatrix}
1&1\\
1&\frac{\tr(\Omega_4\Pi_{\sym}\Omega_4\Pi_{\sym})\tr(\Pi_{\sym})}{\tr(\Omega_4\Pi_{\sym})^2}\\
\end{pmatrix}.
\ee
Let us compute the matrix elements. 
First of all, 
$\tr(\Pi_{\sym})=d(d+1)(d+2)(d+3)(d+4)/120$. Then, because of the structure of the symmetric group, one can write the projector onto the symmetric subspace as $\Pi_{\sym}=\frac{1}{5!}\sum_{\pi\in S_4, \sigma\in S_{5/4}}T_{\pi}T_{\sigma}$, which allows us to compute the trace right away as
\be
\tr(\Omega_4 \Pi_{\sym})=\frac{1}{5!}\sum_{\pi\in S_4, \sigma\in S_{5/4}}\tr(\Omega_4T_{\pi}T_{\sigma})=\frac{1}{5!}\sum_{\pi\in S_4, \sigma\in S_{5/4}}\tr(\Omega_4T_{\pi}\tr_5T_{\sigma})=\frac{(d+4)}{5}\tr(\Omega_4\Pi_{\sym}),
\ee
where we have used  the fact that 
$
\tr_{5}T_{4,5}=d^{-1}\sum_{P}\mathbb{1}^{\otimes3}\otimes P\otimes \tr(P)=\mathbb{1}^{\otimes4}
$. Thus, we have $\tr(\Omega_4 \Pi_{\sym})=\frac{d+4}{5}\tr(\Omega_4\Pi_{\sym})$. Similarly
\be
\tr(\Omega_4\Pi_{\sym}\Omega_4\Pi_{\sym})& =\frac{1}{5!}\sum_{\pi,i}\tr(\Pi_{\sym}\Omega_4\Omega_4T_{\pi})\\
& =\frac{1}{5}\tr(\Pi_{\sym}\Omega_4(d\mathbb{1}^{\otimes5}+T_{(15)}+T_{(25)}+T_{(35)}+T_{(45)}))\\
& =\frac{(d+4)}{5}\tr(\Omega_4\Pi_{\sym}),
\ee
where we have used the identity $\Omega_4T_{\pi}\Omega_4T_{\pi}=\Omega_{4}T_{(i5)}$ for any $i\in[2,\ldots,5]$, which we showed by the graphical calculus in \cref{eq:5commutantgraphical}. Finally, we can express the matrix $\mathcal{S}$, and its inverse
\be
\mathcal{S}=\frac{1}{\tr(\Pi_{\sym})}\begin{pmatrix}
1&1\\
1&\frac{(d+4)(d+3)}{5}\\
\end{pmatrix}\,,\quad \mathcal{S}^{-1}=\frac{d(d+1)(d+2)(d+3)(d+4)}{6(d+8)(d-1)}\begin{pmatrix}
\frac{(d+4)(d+3)}{20}&-1\\
-1&1\\
\end{pmatrix}.
\ee
Therefore, according to \cref{sec:clorbitquantumstates}, we can compute the quasi-probabilities as
\begin{align}\label{eq:p1mpk5}
\begin{pmatrix}
    p(\psi)\\
    1-p(\psi)
\end{pmatrix}=\mathcal{S}^{-1}\begin{pmatrix}
\frac{\tr(\Pi_{\sym}\psi^{\otimes 5})}{\tr(\Pi_{\sym})}\\
    \frac{\tr(\Pi_{\sym}\Omega_4\Pi_{\sym}\psi^{\otimes 5})}{\tr(\Pi_{\sym}\Omega_4)}
\end{pmatrix}=\begin{pmatrix}
        \frac{(d+3)(d+4-5\Delta_{4}(\psi)))}{(d+8)(d-1)}\\
        \frac{5[(d+3)\Delta_{4}(\psi))-4]}{(d+8)(d-1)}
    \end{pmatrix}
\end{align}
where $\Delta_{4}(\psi)$ is the stabilizer purity, see \cref{def:stabilizerentropies}.

As an application, let us compute the trace distance between $\Phi_{\cl}^{(5)}(\psi^{\otimes 5})$ and $\Phi_{\haar}^{(5)}(\psi^{\otimes 5})$ and show that $5$ copies of a stabilizer state cannot be distinguished by $5$ copies of a Haar random state.

\begin{example}[Trace distance of random Clifford orbit]\label{example:k5} Let us compute $\|\Phi_{\cl}^{(5)}(\psi^{\otimes 5})-\Phi_{\haar}^{(5)}(\psi^{\otimes 5})\|_1$. Using \cref{eq:stateorbitcliffordk5}, we can immediately write
\be
\left\|\Phi_{\cl}^{(5)}(\psi^{\otimes 5})-\Phi_{\haar}^{(5)}(\psi^{\otimes 5})\right\|_1&=\left(1-p(\psi)\right)\left\|\frac{\Pi_{\sym} }{\tr(\Pi_{\sym})}-\frac{\Pi_{\sym}\Omega_4\Pi_{\sym} }{\tr(\Pi_{\sym}\Omega_4)}\right\|\\
&\left(1-p(\psi)\right)\left(\left\|\frac{\Pi_{\sym} }{\tr(\Pi_{\sym})}\right\|+\left\|\frac{\Pi_{\sym}\Omega_4\Pi_{\sym} }{\tr(\Pi_{\sym}\Omega_4)}\right\|\right)
&\le 2(1-p(\psi))
\ee
where we have used  the fact that $p(\psi)\ge0$. Using \cref{eq:p1mpk5}, we get
\be
\left\|\Phi_{\cl}^{(5)}(\psi^{\otimes 5})-\Phi_{\haar}^{(5)}(\psi^{\otimes 5})\right\|_1\le \frac{10[(d+3)\Delta_{4}(\psi))-4]}{(d+8)(d-1)}=O(d^{-1}).
\ee
\end{example}

\subsubsection{\texorpdfstring{$k=6$}{k=6}}\label{ssec:k6}
Let us analyze the commutant for the Clifford group for $k=6$. As done before, let us consider the independent Pauli monomials that generate the Clifford commutant. We do it iteratively, by iterating over $m$, i.e., the number of primitives in the matrix $V$. For $m=0$, we only have the identity. $m=1$, we have the primitive Pauli monomials,. 
\setmonomialscale{3.5mm}
\begin{align}        \Omega_4\quad\textbf{=}\quad\monomialdiagram{4}{{1,2,3,4}}{},\quad\quad\quad\Omega_6\quad\textbf{=}\quad\monomialdiagram{6}{{1,2,3,4,5,6}}{}\quad\quad\quad \,.
\end{align}
Starting from these monomials, we need to add to the vector spaces other primitives with two premises: 1) the primitive should be independent, so the full Pauli monomial is reduced. 2) the vector space does not contain weight $2$ vector, which would correspond to a permutation, 3) avoiding repeated monomials. We get only two possible candidates 
\setmonomialscale{3.5mm}
\begin{align}\monomialdiagram{6}{{1,2,3,4},{3,4,5,6}}{}\quad\quad\quad\monomialdiagram{6}{{1,2,3,4},{3,4,5,6}}{0:1}
\end{align}
Thanks to \cref{sec:graphcalc}, we know that two even primitives overlapping evenly with a phase are equivalent to an odd primitive up to right multiplication for permutations. Therefore, the second term is permutationally equivalent to $\Omega_2$. So the only possible candidate for $m=1$ is  
\begin{align}\Omega_{4,4} \textbf{=}\quad\monomialdiagram{6}{{1,2,3,4},{3,4,5,6}}{}
\end{align}
This concludes the possible terms, since it is not difficult to see that adding an additional $\Omega_1$ term would return us still $\Omega_3$ up to permutations.
Hence, the commutant, other than permutations, is spanned by the following operators, up to left and right action of permutations
\begin{align}
\Omega_4 & =  \frac{1}{d}\sum_{P}P^{\otimes 4}\otimes \mathbb{1}^{\otimes 2}\nonumber,\\
\Omega_6& = \frac{1}{d}\sum_{P}P^{\otimes 6}\nonumber,\\
\Omega_{4,4}& =\frac{1}{d^2}\sum_{P,Q} P^{\otimes2}\otimes(PQ)^{\otimes2}\otimes Q^{\otimes2}.
\end{align}
While it is clear that multiplying the different $\Omega$ among themselves does not generate new elements of the commutant, the same is not true when considering permutations. To generate all possible elements of the commutant, we must \emph{permute} the representative elements. Thus, the most general element of the commutant for $k=6$ 
is $T_{\pi} \Omega^{(6)} T_{\sigma}$, where $T_\pi$ and $T_\sigma$ are permutation matrices. In \cref{tab: cardC6}, we show the cardinality of the different classes generated.
\begin{table}[H]
    \centering
\begin{tabular}{|c|c|}
\hline
    $\Omega$& $\#(T_{\pi}\Omega T_{\sigma})$  \\
    \hline 
    $\Omega_2$ & $720$ \\
    \hline
    $\Omega_4$ &$2700$\\
    \hline
    $\Omega_6$ &$720$\\
    \hline
    $\Omega_{4,4}$ &$450$\\
    \hline
\end{tabular}
    \caption{Cardinality of the independent class forming the basis for $\com(\mathcal{C}_n^{\otimes 6})$.}
    \label{tab: cardC6}
\end{table}
Once we have given the Clifford commutant for $k=6$, we can then introduce the action of the twirling operator. By noticing that all the terms are given by $T_{\pi}\Omega T_{\sigma}$, the twirling operator for $k=6$ can be written as 
\begin{align}
\Phi_{\cl}^{(6)}(O)=\sum_{\substack{\Omega \in \mathcal{P} \\\pi\sigma\in S_{6} } }T_{\pi}\left(\alpha_{\pi,\sigma}^{(\Omega)} \Omega \right)T_{\sigma}
\end{align}
where each of the coefficients can be computed through the inverse of the Gram-Schmidt matrix through the rules introduced in \cref{sec:cliffordweingartencalculus}.

As done before, instead of directly computing the Clifford Weingarten functions, let us focus our attention on the twirling of a pure quantum state. We can write
\begin{align}
\Phi_{\cl}^{(6)}(\psi)=\sum_{\Omega, \in \mathcal{P}/S_6 }p_{\Omega}(\psi)\frac{\Pi_{\sym} \Omega \Pi_{\sym} }{\tr(\Pi_{\sym}\Omega)}
\end{align}
where the $p_{\Omega}$s are equal to
\be
p_{\Omega}(\psi)=\sum_{\Omega^{\prime}\in \mathcal{P}/ S_6}(\mathcal{S}^{-1})_{\Omega\Omega^{\prime}}\frac{\tr(\Omega\psi^{\otimes k})}{\tr(\Pi_{\sym}\Omega)}.
\ee
One can then use the definition of $\mathcal{S}$ from \cref{sec:clorbitquantumstates} with $k=6$, leading to the computation of
\begin{align}
S = \frac{1}{\tr(\Pi_{\sym})}
\begin{pmatrix}
1 & 1 & 1 & 1 \\
1 & \frac{\tr(\Omega_4\Pi_{\sym}\Omega_4\Pi_{\sym} )\tr(\Pi_{\sym})}{\tr(\Omega_4 \Pi_{\sym})^2} & \frac{\tr(\Omega_4\Pi_{\sym}\Omega_6\Pi_{\sym})\tr(\Pi_{\sym})}{\tr(\Omega_4 \Pi_{\sym})\tr(\Omega_6 \Pi_{\sym})} & \frac{\tr(\Omega_4\Pi_{\sym}\Omega_{4,4}\Pi_{\sym})\tr(\Pi_{\sym})}{\tr(\Omega_4 \Pi_{\sym})\tr(\Omega_{4,4} \Pi_{\sym})} \\ 
1 & \frac{\tr(\Omega_6\Pi_{\sym}\Omega_4\Pi_{\sym} )\tr(\Pi_{\sym})}{\tr(\Omega_6 \Pi_{\sym})\tr(\Omega_4 \Pi_{\sym})} & \frac{\tr(\Omega_6\Pi_{\sym}\Omega_6\Pi_{\sym})\tr(\Pi_{\sym})}{\tr(\Omega_6 \Pi_{\sym})^2} & \frac{\tr(\Omega_6\Pi_{\sym}\Omega_{4,4}\Pi_{\sym})\tr(\Pi_{\sym})}{\tr(\Omega_6 \Pi_{\sym})\tr(\Omega_{4,4} \Pi_{\sym})} \\
1 & \frac{\tr(\Omega_{4,4}\Pi_{\sym}\Omega_4\Pi_{\sym} )\tr(\Pi_{\sym})}{\tr(\Omega_{4,4} \Pi_{\sym})\tr(\Omega_4 \Pi_{\sym})} & \frac{\tr(\Omega_{4,4}\Pi_{\sym}\Omega_6\Pi_{\sym})\tr(\Pi_{\sym})}{\tr(\Omega_{4,4} \Pi_{\sym})\tr(\Omega_6 \Pi_{\sym})} & \frac{\tr(\Omega_{4,4}\Pi_{\sym}\Omega_{4,4}\Pi_{\sym})\tr(\Pi_{\sym})}{\tr(\Omega_{4,4} \Pi_{\sym})^2} 
\end{pmatrix}.
\end{align}
The computation of $\mathcal S$ reduces to the computation of a few elements, and then combines them. So we have, by brute force computation, the  result
\begin{align}
\tr(\Pi_{\sym}\Omega_4)& =\frac{1}{180}d(d+5)(d+4)(d+2)(d+1),\\
\tr(\Pi_{\sym}\Omega_{4,4})& = \frac{1}{30}d(d+4)(d+2)(d+1),\\
\tr(\Pi_{\sym}\Omega_6)& = \frac{1}{6!}d(d+23)(d+4)(d+2)(d+1),\\
\tr(\Pi_{\sym}\Omega_4\Pi_{\sym}\Omega_4)& = \frac{1}{2700}d(d+1)(d+2)(d+4)
(76+13d+d^{2}),\\
\tr(\Pi_{\sym}\Omega_{4,4}\Pi_{\sym}\Omega_4)& = \frac{(d+4)}{5}\tr(\Pi_{\sym} \Omega_{4,4}),\\
\tr(\Pi_{\sym}\Omega_6\Pi_{\sym}\Omega_4)& = \tr(\Pi_{\sym}\Omega_4),\\
\tr(\Pi_{\sym}\Omega_{4,4}\Pi_{\sym}\Omega_6)& = \frac{1}{450}d(d+4)^2(d+2)^2(d+1),\\
\tr(\Pi_{\sym}\Omega_6\Pi_{\sym}\Omega_{4,4})& = \tr(\Pi_{\sym}\Omega_{4,4}). 
\end{align}
Consequently, the matrix $\mathcal S$ can be written as
\begin{align}\mathcal{S}=\frac{1}{\tr(\Pi_{\sym})}\begin{pmatrix}
   1& 1 &1&1\\
    1& \frac{(d^2+13d+76)(d+3)}{60(d+5)}&\frac{(4+d)(d+3)}{20}&\frac{(5+d)(d+3)}{23+d}\\
    1&\frac{(d+3)(4+d)}{20}&\frac{(d+5)(d+4)(d+3)(d+2)}{360}&\frac{(5+d)(d+3)}{23+d}\\
    1&\frac{(5+d)(d+3)}{23+d}&\frac{(5+d)(d+3)}{23+d}&\frac{(5+d)^2(d+3)^2}{(23+d)^2}
\end{pmatrix}.
\end{align}
Therefore, according to \cref{sec:clorbitquantumstates} we get that the quasi-probabilities are
\begin{align}\begin{pmatrix}
    p_{2} \\
    p_{4} \\
    p_{6} \\
    p_{4,4} 
\end{pmatrix}
& = 
\left[\frac{1}{\tr(\Pi_{\sym})}\begin{pmatrix}
   1& 1 &1&1\\
    1& \frac{(d^2+13d+76)(d+3)}{60(d+5)}&\frac{(4+d)(d+3)}{20}&\frac{(5+d)(d+3)}{23+d}\\
    1&\frac{(d+3)(4+d)}{20}&\frac{(d+5)(d+4)(d+3)(d+2)}{360}&\frac{(5+d)(d+3)}{23+d}\\
    1&\frac{(5+d)(d+3)}{23+d}&\frac{(5+d)(d+3)}{23+d}&\frac{(5+d)^2(d+3)^2}{(23+d)^2}
\end{pmatrix}\right]^{-1} \begin{pmatrix}
    \frac{\tr(\Pi_{\sym}\psi^{\otimes 6})}{\tr(\Pi_{\sym})} \\ 
    \frac{\tr(\Pi_{\sym}\Omega_4\Pi_{sym}^{(6)}\psi^{\otimes 6})}{\tr(\Omega_4\Pi_{\sym})} \\
    \frac{\tr(\Pi_{\sym}\Omega_6\Pi_{sym}^{(6)}\psi^{\otimes 6})}{\tr(\Omega_6\Pi_{\sym})} \\
    \frac{\tr(\Pi_{\sym}\Omega_{4,4}\Pi_{sym}^{(6)}\psi^{\otimes 6})}{\tr(\Omega_{4,4}\Pi_{\sym})} 
\end{pmatrix} \\
& = \begin{pmatrix}
    \frac{(d+5)(d+3)[(d^2+13d)-\tr(\Omega_{4,4}\psi^{\otimes 6})(d-32)-(15d+60)\tr(\Omega_4\psi^{\otimes 4})+30\tr(\Omega_6\psi^{\otimes 6})]}{(d+16)(d+8)(d-1)(d-2)}\\
\frac{15(d+5)[-4(d+4)-4(d+4)\tr(\Omega_{4,4}\psi^{\otimes 6})-(d^2+8d+52)\tr(\Omega_4\psi^{\otimes 4})-3(d+6)\tr(\Omega_6\psi^{\otimes 6})]}{(d+16)(d+8)(d-1)(d-2)}\\
\frac{720+720\tr(\Omega_{4,4}\psi^{\otimes 6})+270(d+6)\tr(\Omega_4\psi^{\otimes 4})+15(d+14)(d+1)\tr(\Omega_6\psi^{\otimes 6})}{(d+16)(d+8)(d-1)(d-2)}\\
\frac{(d+23)[-(d-32)+d(d+13)\tr(\Omega_{4,4}\psi^{\otimes 6})-(15d+60)\tr(\Omega_4\psi^{\otimes 4})+30\tr(\Omega_6\psi^{\otimes 6})]}{(d+16)(d+8)(d-1)(d-2)}
\end{pmatrix},
\nonumber
\end{align}
where we are adopting the lighter notation of $p_x \coloneqq p_{\Omega_x}$. 
Let us point out that, for the sake of completeness, the inverse matrix of $\mathcal S$ is given in Ref.~\cite{leone_mathematica_2023}.

In \cref{example:k4,example:k5}, we showed that the trace distance between five coherent copies of Haar-random states and random stabilizer states is exponentially small. In the following, we show that for $k = 6$, the trace distance becomes large. This occurs because the commutant supports the first non-trivial unitary primitive, $\Omega_{6}$.

\begin{example}[Trace distance of random Clifford orbit]\label{example:k6tracedistance} Let us compute the lower bound $\|\Phi_{\cl}^{(6)}(\psi^{\otimes 6})-\Phi_{\haar}^{(6)}(\psi^{\otimes 6})\|_1$. To do it, let us start by considering the unitary primitive $\Omega_{6}$. The expectation value over $\psi^{\otimes 6}$ corresponds to the stabilizer purity $\Delta_{6}(\psi)=\tr(\Omega_6\psi^{\otimes 6})$, see \cref{def:stabilizerentropies}. If $\psi$ is a pure quantum state vector it follows 
\be
\Delta_{6}(\psi)-\tr(\Phi_{\haar}^{(6)}\Omega_6)=\tr(\Omega_6(\Phi_{\cl}^{(6)}(\psi^{\otimes 6})-\Phi_{\haar}^{(6)}(\psi^{\otimes 6})))\le \|\Omega_{6}\|_{\infty}\|\Phi_{\cl}^{(6)}(\psi^{\otimes 6})-\Phi_{\haar}^{(6)}(\psi^{\otimes 6})\|_1
\ee
which implies that
\be
\|\Phi_{\cl}^{(6)}(\psi^{\otimes 6})-\Phi_{\haar}^{(6)}(\psi^{\otimes 6})\|_1\ge \Delta_{6}(\psi)-O(d^{-1}).
\ee
For the case of $\sigma\in\stab$, since $\Delta_{6}(\sigma)=1$ (see \cref{def:stabilizerentropies}), then $\|\Phi_{\cl}(\sigma^{\otimes 6})-\Phi_{\haar}^{(6)}(\sigma^{\otimes 6})\|_1\ge 1-O(d^{-1})$.
\end{example}

\subsubsection{\texorpdfstring{$k=8$}{k=8}}
In this section, we construct the commutant for $k=8$. We want to construct all the Pauli monomials that are independent up to left and right multiplication for permutations. The analysis is quite similar to the one done for $k=6$, and previously for $k=5$ and $k=4$. We do it iteratively, by iterating over $m$, i.e., the number of primitives in the matrix $V$. For $m=0$, we only have the identity. $m=1$, we have the primitive Pauli monomials, which are listed below.
 \setmonomialscale{3.5mm}
       \begin{align}
    \Omega_4\quad\textbf{=}\quad\monomialdiagram{4}{{1,2,3,4}}{},  
\quad\quad\quad  
\Omega_6\quad\textbf{=}\quad\monomialdiagram{6}{{1,2,3,4,5,6}}{},  
\quad\quad\quad  
\Omega_8\quad\textbf{=}\quad\monomialdiagram{8}{{1,2,3,4,5,6,7,8}}{}.
    \end{align}
Then, we use the same prescriptions used in \cref{ssec:k6} in order to build the other terms for $m>1$. Starting from $\Omega_4$, we obtain the following Pauli monomials obeying such prescriptions:
\be\label{eq:com8-1}
\monomialdiagram{6}{{1,2,3,4},{3,4,5,6}}{}\quad\quad\quad\monomialdiagram{7}{ {1,2,3,4},{4,5,6,7} }{}\quad\quad\quad\monomialdiagram{8}{{1,2,3,4},{5,6,7,8}}{}\quad\quad\quad \monomialdiagram{6}{{1,2,3,4},{3,4,5,6}}{0:1} \quad\quad\quad\monomialdiagram{7}{{1,2,3,4},{4,5,6,7}}{0:1}\quad\quad\quad\monomialdiagram{8}{{1,2,3,4},{5,6,7,8}}{0:1}\quad\quad\quad\monomialdiagram{8}{{1,2,3,4},{3,4,5,6,7,8}}{}\quad\quad\quad\monomialdiagram{8}{{1,2,3,4},{3,4,5,6,7,8}}{0:1}.
\ee
Thanks to \cref{sec:graphcalc}, we know that two even primitives overlapping evenly with a phase are equivalent to an odd primitive up to right multiplication for permutations. Therefore, the fourth and the sixth monomial in \cref{eq:com8-1} are permutationally equivalent to $\Omega_6$. Through graphical calculus, we can manipulate the fifth one as
\be
\monomialdiagram{7}{{1,2,3,4},{4,5,6,7}}{0:1}\quad\textbf{=}\quad \monomialdiagram{7}{{4,5,6,7},{1,2,3,4}}{}=T_{(17)}T_{(26)}T_{(35)}\monomialdiagram{7}{{1,2,3,4},{4,5,6,7}}{}\,\,T_{(17)}T_{(26)}T_{(35)}
\ee
which shows that the second and the fifth are (non-trivially) the same up to the action of the adjoint permutation $T_{(17)}T_{(26)}T_{(35)}$. For convenience, we can manipulate the last one to take the phase away
\be
\monomialdiagram{8}{{1,2,3,4},{3,4,5,6,7,8}}{0:1}=\monomialdiagram{8}{{1,2,5,6,7,8},{3,4,5,6,7,8}}{}\mapsto\monomialdiagram{8}{{1,2,3,4,5,6},{3,4,5,6,7,8}}{}
\label{seccom8:p4p6intop6p6}\ee
where the last arrow refers to permutational equivalence. This leaves us in defining three more Pauli monomials that are independent up to permutations
as
\be
\Omega_{4,4}\quad\textbf{=}\quad\monomialdiagram{6}{{1,2,3,4},{3,4,5,6}}{}
,\quad\quad\quad\Omega_{4,4^{1}},\textbf{=}\quad\monomialdiagram{7}{{1,2,3,4},{4,5,6,7}}{},\quad\quad\Omega_{4,4^{\scriptscriptstyle 2}}\quad\textbf{=}\quad\monomialdiagram{8}{{1,2,3,4},{5,6,7,8}}{},\quad\quad\Omega_{4,6}\quad\textbf{=}\quad\monomialdiagram{8}{{1,2,3,4},{3,4,5,6,7,8}}{},\quad\quad\Omega_{6,6}\quad\textbf{=}\quad\monomialdiagram{8}{{1,2,3,4,5,6},{3,4,5,6,7,8}}{}.
\ee
We can now start adding primitives to $\Omega_2^{(8)}$. However, we have already exhausted all the possibilities allowed by analyzing $\Omega_4$ by \cref{seccom8:p4p6intop6p6} and $\Omega_{4,6}$. 

Now, we can construct Pauli monomials with $m=3$ from $\Omega_{4,4}, \cdots, \Omega_{6,6}$. We observe that it is only sufficient to add the weight four primitives. Indeed, adding a weight $8$ primitive would necessarily make the Pauli monomial not reducible anymore, and adding a weight $6$ primitive would necessarily result in a Pauli monomial corresponding to a $V$ equivalent to one with weight 2 vectors in the span. Therefore, we can add weight 2 vectors only to $\Omega_4^{(8)},\Omega_5^{(8)},\Omega_6^{(8)}$. However, it will be sufficient to add all the possible weight 2 primitives to $\Omega_4^{(8)}$ to exhaust all possibilities. Let us first write all the possibilities as
\be
\Omega_{4,4,4}\quad\textbf{=}\quad\monomialdiagram{8}{{1,2,3,4},{3,4,5,6},{5,6,7,8}}{},\quad\quad\quad\Omega_{4,4,4^{\scriptscriptstyle 1}}\quad\textbf{=}\quad\monomialdiagram{8}{{1,2,3,4},{3,4,5,6},{1,3,5,7}}{},\quad\quad\quad\Omega_{4,4,4^{\scriptscriptstyle 2}},\quad\textbf{=}\quad\monomialdiagram{8}{{1,2,3,4},{3,4,5,6},{1,3,7,8}}{}.
\ee
Notice that we did not consider any phases yet, as it turns out that all the Pauli monomials with phases are dependent on all the others that we already defined. Indeed, since in $\Omega_{4,4,4}$ and $\Omega_{4,4,4^{\scriptscriptstyle 1}}$ all primitives commute, by virtue of \cref{sec:graphcalc}, we know that any commuting weight 2 primitive with a phase can be reduced to a weight $6$ primitive up to permutation, which has already been covered by the Pauli monomials defined above with $m=2$. However, for $\Omega_{4,4,4^{\scriptscriptstyle 2}}$, the argument is even simpler: the third and the second primitive anticommute with each other, and therefore any phases put between them can be canceled by a trivial swapping, cf.\cref{sec:graphcalc}. 

We are now left with the last case, that is, $m=4$. Following the reasoning above, the only possibility is to add a weight 4 primitive to $\Omega_{4,4,4},\Omega_{4,4,4^{\scriptscriptstyle 1}}, \Omega_{4,4,4^{\scriptscriptstyle 2}}$. However, by the same argument as before, it will be sufficient to add it to just the first one. In particular, for $\Omega_{4,4,4}$ and $\Omega_{4,4,4^{\scriptscriptstyle 1}}$ there is one possibility in each case that would result in the same Pauli monomial with $m=4$, namely
\be
\Omega_{4,4,4,4}\quad\textbf{=}\quad\monomialdiagram{8}{{1,2,3,4},{3,4,5,6},{3,4,7,8},{1,3,5,7}}{}
\ee
whereas, it is not possible to add anything to $\Omega_{4,4,4^{\scriptscriptstyle 2}}$ as it would result in a vector space with a weight 2 vector in it. By the same argument as in the case $m=3$, adding any phase would result in a lower $m$ Pauli monomial, since two weight 4 primitives sharing a phase are permutationally equivalent to a weight 6 primitive. At this point, it is easy to get convinced that no more primitive can be added to $\Omega_{4,4,4,4}$ without violating our $3$ above principles. Hence, the commutant with $k=8$ is spanned by the following operators, up to left and right action for permutations
\begin{align}
\Omega_{2}&= \mathbb{1},  \\
\Omega_{4}&=  \frac{1}{d}\sum_{P}P^{\otimes 4}\otimes \mathbb{1}^{\otimes 4},\\
\Omega_{6}&= \frac{1}{d}\sum_{P}P^{\otimes 6}\otimes \mathbb{1}^{\otimes 2},\\
\Omega_{8}&= \frac{1}{d}\sum_{P}P^{\otimes 8},\\
\Omega_{4,4}&=\frac{1}{d^2}\sum_{P,Q} P^{\otimes2}\otimes(PQ)^{\otimes2}\otimes Q^{\otimes2}\otimes I^{\otimes 2},\\
\Omega_{4,4^{\scriptscriptstyle 1} }&= \frac{1}{d^2}\sum_{P,Q} P^{\otimes3}\otimes(PQ)\otimes Q^{\otimes3},\\
\Omega_{4,4^{\scriptscriptstyle 2}}&=\frac{1}{d^2}\sum_{P,Q}P^{\otimes 4}\otimes Q^{\otimes4},\\
\Omega_{4,6}&=\frac{1}{d^2}\sum_{P,Q} P^{\otimes2}\otimes(PQ)^{\otimes2}\otimes Q^{\otimes4},\\
\Omega_{6,6}&= \frac{1}{d^2}\sum_{P,Q} P^{\otimes2}\otimes(PQ)^{\otimes4}\otimes Q^{\otimes2},\\
\Omega_{4,4,4}&= \frac{1}{d^3}\sum_{P,Q,K}P^{\otimes 2}\otimes (PK)^{\otimes 2}\otimes (QK)^{\otimes2}\otimes Q^{\otimes2},\\
\Omega_{4,4,4^{\scriptscriptstyle 1}}&=\frac{1}{d^3}\sum_{P,Q,K}P\otimes (PK)\otimes(PQ)\otimes (PQK)\otimes(QK)\otimes Q\otimes K\otimes\mathbb{1},\\
\Omega_{4,4,4^{\scriptscriptstyle 2}}&= \frac{1}{d^3}\sum_{P,Q,K}(PK)\otimes P\otimes (PQK)\otimes (PQ)\otimes Q^{\otimes 2}\otimes K^{\otimes 2},\\
\Omega_{4,4,4,4}&=\frac{1}{d^4}\sum_{P,Q,K,L} (PL)\otimes P\otimes (PQKL)\otimes (PQK)\otimes (QL)\otimes Q\otimes (KL)\otimes K.\label{eq:commutantfork8}
\end{align}
Extending the logic from the $k=6$ case, the commutant for $k=8$ is generated by permuting its representative elements, $\Omega$. Specifically, any element can be expressed as $T_{\pi} \Omega T_{\sigma}$, where $T_{\pi}$ and $T_{\sigma}$ are permutation matrices. The cardinality of the resulting classes is presented in \cref{tab: cardC8}.

\begin{table}[H]
    \centering
\begin{tabular}{|c|c|}
\hline
    $\Omega $& $\#(T_{\pi_i}\Omega T_{\sigma_i})$  \\
    \hline
    $\Omega_2$ &$40320$\\
    \hline
    $\Omega_4$ &$705600$\\
    \hline
    $\Omega_6$ &$1128960$\\
    \hline
    $\Omega_8$ &$40320$\\
    \hline
     $\Omega_{4,4}$ &$705600$\\
    \hline
    $\Omega_{4,4^{\scriptscriptstyle 1}}$ &$2822400$\\
    \hline
    $\Omega_{4,4^{\scriptscriptstyle 2}}$ &$88200$\\
    \hline
     $\Omega_{4,6}$ &$2116800$\\
    \hline
    $\Omega_{6,6}$ &$1411200$\\
    \hline
     $\Omega_{4,4,4}$ &$22050$\\
    \hline
    $\Omega_{4,4,4^{\scriptscriptstyle 1}}$ &$57600$\\
    \hline
    $\Omega_{4,4,4^{\scriptscriptstyle 2}}$ &$705600$\\
    \hline
    $\Omega_{4,4,4,4}$ &$900$\\
    \hline
\end{tabular}
    \caption{Cardinality of the independent class forming the basis for $\com(\mathcal{C}_n^{\otimes 8})$.}
    \label{tab: cardC8}
\end{table}

Once built the commutant for $k=8$, we can also write the general equation for the twirling of an operator $O$ as
\begin{align}\Phi_{\cl}^{(8)}(O)=\sum_{\substack{\Omega \in \mathcal{P} \\\pi\sigma\in S_{8} } }T_{\pi}(\alpha_{\pi,\sigma}^{(\Omega)} \Omega )T_{\sigma}
\end{align}where each of the coefficients can be computed through the inverse of the Gram-Schmidt matrix through the rules introduced in \cref{sec:cliffordweingartencalculus}. As done before, instead of directly computing the Clifford Weingarten functions, let us focus our attention on the twirling of a pure quantum state. We can write
\begin{align}
\Phi_{\cl}^{(8)}(\psi^{\otimes 8})=\sum_{\Omega \in \mathcal{P}/S_8 }p_{\Omega }(\psi)\frac{\Pi_{\sym} \Omega \Pi_{\sym} }{\tr(\Pi_{\sym}\Omega)}
\end{align}
where the $p_{\Omega }$s are equal to
\be
p_{\Omega }(\psi)=\sum_{\Omega^{\prime} \in \mathcal{P}/ S_8}(\mathcal{S}^{-1})_{\Omega\Omega^{\prime}}\frac{\tr(\Omega\psi^{\otimes k})}{\tr(\Pi_{\sym}\Omega)}.
\ee

\section{Clifford commutant in the multi-qudit scenario}
\label{sec:genquditsapp}
\paragraph{Preliminaries:} We consider a system of $n$ qudits with local dimension $q$, being a prime number. It is crucial for $q$ to be a prime number, so that the generalized Pauli operators have the same eigenvalues and the Clifford group can be defined analogously to the qubit case, as those unitary operators that map Pauli operators to Pauli operators. In the reminder of the section, we adopt the same notation that we have used  for the qubit case, and $\mathbb{P}_n$, $\mathcal{C}_n$ denotes the Pauli group and the Clifford group on $n$ qudits. Similarly, the dimension of the Hilbert space is denoted as $d\equiv q^n$. The generalized Pauli operators on $n$ qudits, that in the following will be called as \textit{Weyl operators} to adapt better to the current literature, are generated by the following $X$- and $Z$-like operators:
\begin{align}
    X\ket{i}=\ket{i+1}\quad \mathrm{and} \quad Z\ket{i}=\omega^i \ket{i}
\end{align}
where $\omega=e^{-2\pi i/q}$. We also define $\tau=\sqrt{\omega}=e^{i\pi (q^2+1)/q}$. Notice that for $q>2$, the square root can be uniquely defined as a power of $\omega$. 
The single qudit Weyl operators are given by
\begin{align}
    P=\tau^{ac}X^{a}Z^{c}\,.
\end{align}
for $a,c\in\mathbb{F}_q$. Analogously to the qubit case, we can encode each Weyl operator $P\in\mathbb{P}_n$ into a string $b(P)\in\mathbb{F}_{q}^{2n}$, which gives the rule  
\begin{align}
    PQ&=\tau^{b(P)^T Jb(Q)}K,\\
    PQ&=\omega^{b(P)^T J b(Q)}QP
\end{align}
of the product between Weyl operators.
In the above \begin{align}
J\coloneqq \bigoplus_{i=1}^{n}\begin{pmatrix}
    0&-1\\1&0
\end{pmatrix},
\end{align}
while $b(K)=b(P)+b(Q) \pmod q$. By recalling \cref{def:chidef}, and from the equations above, it follows that one can identify
\begin{align}
    \chi(P,Q)= \omega^{b(P)^T J b(Q)}.
\end{align}
Since the branch of the square root is fixed by demanding that it remains a $q$-th root of unity, we can define
\begin{align}
    \sqrt{\chi(P,Q)}\coloneqq\tau^{b(P)^T J b(Q)}\,.
\end{align}
\paragraph{Orthogonal basis for the commutant and its dimension:} Let us now outline a proof analogous to that of \cref{th:fullcommutantnk}. The key step in the proof of \cref{th:fullcommutantnk} is the definition of the equivalence class $[V, G]$. We begin by identifying the analog of $V$. Observe that for $k$ tensor copies and any $Q \in \mathbb{P}_n^{\otimes k}$, we have $b(Q) \in \mathbb{F}_q^{2nk}$, which allows us to construct $B(Q) \in \mathbb{F}_q^{k \times 2n}$ via the vectorization isomorphism. One can also prove an analog of \cref{le:uniquenessPV}. Specifically, we can decompose $B(Q)$ as $B(Q) = V B_{\boldsymbol{P}}^{T}$, where $V \in \mathbb{F}_q^{k \times m}$ and $B_{\boldsymbol{P}} \in \mathbb{F}_q^{2n \times m}$, with $m$ denoting the minimal rank. This means we can write
\begin{align}
    Q \propto P_1^{\otimes v_1} \cdots P_m^{\otimes v_m} \,,
\end{align}
where $V = (v_1, \ldots, v_m)$ and $P_1, \ldots, P_m \in \mathbb{P}_n$. As in \cref{le:uniquenessPV}, all such decompositions are related by multiplication with an invertible matrix $A \in \operatorname{GL}(\mathbb{F}_q^{m \times m})$.

We also define the \textit{weighted anticommutation} graph $G$, with associated matrix $\mathcal{A}(\boldsymbol{P})_{i,j}$, via the relation
\begin{align}
    P_i P_j - \omega^{\mathcal{A}(\boldsymbol{P})_{i,j}} P_j P_i = 0\,, \quad P_i, P_j \in \mathbb{P}_n \,.
\end{align}
From this, it follows that $\mathcal{A}(\boldsymbol{P})$ is anti-symmetric, i.e., $\mathcal{A}(\boldsymbol{P})^T = -\mathcal{A}(\boldsymbol{P})$, and $\chi(P_i, P_j) = \omega^{\mathcal{A}(\boldsymbol{P})_{i,j}}$.

We define the phase $\varphi(P)$ (as in the statement of \cref{th:fullcommutantnk}) for Weyl operators analogously to be
\begin{align}
    \varphi(P) = \frac{\tr(P T_{(k \cdots 2\ 1)})}{d} \,.
\end{align}
The condition $\varphi(P) \neq 0$ is necessary and sufficient for $\Phi_{\cl}^{(k)}(P) \neq 0$. For $\varphi(P) \neq 0$, we require $\sum_i v_i = 0 \pmod q$. For each $v \in \mathbb{F}_q$, this condition determines the last entry, effectively yielding a $(k-1)$-dimensional vector space.

Clifford operators on qudits can likewise map any list of Weyl operators to any other, provided the algebraic and anticommutation relations are preserved, as one can similarly prove an analog of \cref{lem:alg_transforms}. Hence, the equivalence class $[V, G]$ is defined identically.

This procedure defines the independent-based graph monomials $\mho_I([V,G])$, which form a basis for any $n$ and $k$, analogous to the qubit case as
\begin{equation}\label{eq:indepgraphbasedforqudits}
    \mho_I(V, G) \coloneqq \frac{1}{|S_{[V,G]}|} \sum_{\substack{\boldsymbol{P} \in \mathbb{P}_n^{\times m} \\ \mathcal{A}(\boldsymbol{P}) = G\\\boldsymbol{P} \text{ is alg. indep.} }} 
    \prod_{j=1}^m P_j^{\otimes v_j},
\end{equation}
where $S_{[V,G]}$ denotes the set of all such Weyl operators, which forms the Weyl orbit under Clifford operations.

We now consider which graphs are admissible. The strategy used in \cref{th:fullcommutantnk}—constructing pairwise noncommuting and otherwise commuting Weyl operators by multiplication—works also for $\mathbb{F}_q$. Therefore, we can replicate the same argument from \cref{th:fullcommutantnk}, sketched below. Any matrix $G$ can be brought by invertible transformations to the canonical (anti-symmetric) form:
\begin{align}
    G = \bigoplus_{i=1}^r \begin{pmatrix}
        0 & 1 \\ -1 & 0
    \end{pmatrix} \oplus 0_{m-2r, m-2r}.
\end{align}
We can describe a noncommutation graph on $m - r$ qudits by choosing appropriate Weyl operators for each noncommuting pair. Since any two algebraically independent Weyl operators on a single qudit do not commute, we get $\mathrm{rank}_q(G) \geq 2(m - n)$. Consequently, in analogy with \cref{cor:dimcom}, the dimension of the commutant is given by
\begin{align} \label{eq:dimensioncommutantcliffordgroupqudit}
    \dim(\com(\mathcal{C}_n^{\otimes k})) &= \sum_{m=0}^{k-1} \sum_{r=\max(m - n, 0)}^{\lfloor m/2 \rfloor}
    \binom{k-1}{m}_q N_q(m, r) \\
    &\simeq
    \begin{cases}
        q^{\frac{k^2 - 3k}{2}+1} & \text{if } 2n \geq k - 1, \\
        q^{2kn - 2n^2 - 3n} & \text{if } 2n < k - 1,
    \end{cases}
    \nonumber
\end{align}
where $N_q(m, r)$ (see, e.g., Ref.~\cite{carlitz1954representations}) denotes the number of $m \times m$ anti-symmetric matrices over $\mathbb{F}_q$ with rank $2r$, and the Gaussian binomial coefficient $\binom{k-1}{m}_q$ count the number of even vector subspaces of $\mathbb{F}_q$ of dimension $m$. The approximation in \cref{eq:dimensioncommutantcliffordgroupqudit} follows identically to the one in \cref{cor:dimcom}, since the structure of the expressions remains the same apart from the base of the exponentials. The bounds should also become increasingly accurate with increasing $q$.

\paragraph{Weyl monomials $(\Omega)$:} Similarly to the qubit case, alternative bases—beyond the one formed by independent graph-based monomials—can be defined. These alternative bases often offer more manageable or insightful representations. The strategy to define such bases is to directly generalize the linear transformations that connect them in the qubit setting, as explored in \cref{Sec:invertiblemapspaulimonomials}. 

Starting from \cref{eq:indepgraphbasedforqudits}, we extend the definition to include ordered lists of Pauli operators that are not necessarily algebraically independent:
\begin{equation}
    \mho(V, G) \coloneqq \frac{1}{d^m} \sum_{\substack{\boldsymbol{P} \in \mathbb{P}_n^{\times m} \\ \mathcal{A}(\boldsymbol{P}) = G}} 
    \prod_{j=1}^m P_j^{\otimes v_j}.
\end{equation}
These are precisely the qudit analogues of the “$\mho$” basis introduced in \cref{sec:mhosandstuff}. Following the approach of \cref{Sec:invertiblemapspaulimonomials}, we employ a Fourier-like transform over $\mathbb{F}_q$ to define Weyl monomials. Starting from the definition above, we proceed as follows:
\begin{align}
    \mho(V, G) &\coloneqq \frac{1}{d^m} \sum_{\substack{\boldsymbol{P} \in \mathbb{P}_n^{\times m} \\ \mathcal{A}(\boldsymbol{P}) = G}} 
    \prod_{j=1}^m P_j^{\otimes v_j} \\
    \nonumber
    &= \frac{1}{d^m} \sum_{\boldsymbol{P} \in \mathbb{P}_n^{\times m}}  
    \prod_{j=1}^m P_j^{\otimes v_j} \prod_{i<j} \delta_{\mathcal{A}(\boldsymbol{P})_{i,j} = G_{i,j}} \\
    \nonumber
    &= \frac{1}{d^m q^{m(m-1)/2}} \sum_{\boldsymbol{P} \in \mathbb{P}_n^{\times m}}  
    \prod_{j=1}^m P_j^{\otimes v_j} \prod_{i<j} \sum_{M_{i,j} \in \mathbb{F}_q} \omega^{M_{i,j} (\mathcal{A}(\boldsymbol{P})_{i,j} - G_{i,j})} \\
    \nonumber
    &= \frac{1}{d^m q^{m(m-1)/2}} \sum_{M \in \operatorname{Asym}(\mathbb{F}_q^{m \times m})} \omega^{-\sum_{i<j} M_{i,j} G_{i,j}} 
    \sum_{\boldsymbol{P} \in \mathbb{P}_n^{\times m}}  
    \prod_{j=1}^m P_j^{\otimes v_j} \prod_{i<j} \chi(P_i, P_j)^{M_{i,j}} \\
    \nonumber
    &\eqqcolon \frac{1}{q^{m(m-1)/2}} \sum_{M \in \operatorname{Asym}(\mathbb{F}_q^{m \times m})} \omega^{-\sum_{i<j} M_{i,j} G_{i,j}} \Omega(V, M),
\label{eq:fourierforqudits}
\end{align}
where we have used the identity $\frac{1}{q} \sum_{i=1}^{q} \omega^{qi} = \delta_{i0}$, and the fact that $\chi(P_i, P_j) = \omega^{\mathcal{A}(\boldsymbol{P})_{i,j}}$. Since the Fourier-like transform is invertible, it follows that the monomials $\Omega(V, M)$ form a well-defined basis of the commutant for $n \geq k-1$. From \cref{eq:fourierforqudits}, they are defined as
\begin{align}
\Omega(V, M) \coloneqq \frac{1}{d^m} \sum_{\boldsymbol{P} \in \mathbb{P}_n^m}  
P_1^{\otimes v_1} P_2^{\otimes v_2} \cdots P_m^{\otimes v_m} \times \left( \prod_{\substack{i, j \in [m] \\ i < j}} \chi(P_i, P_j)^{M_{i,j}} \right).
\end{align}

\paragraph{Primitive monomials:} Among the Weyl monomials, some are particularly simple and fundamental: the \emph{primitive Weyl monomials}, defined as
\be
\Omega(v) \coloneqq \frac{1}{d} \sum_{P \in \mathbb{P}_n} P^{\otimes v},
\ee
where \( v \in \mathbb{F}_q^k \) satisfies the constraint \( \sum_i v_i = 0 \pmod q \). 

Analogously to \cref{lem:propertiesprimitivepauli}, we can determine when these primitives are unitary or proportional to projectors. Specifically, by using the identity 
\be
\sum_{P \in \mathbb{P}_n} \chi(P, Q) = d^2 \delta_{P = \mathbb{1}},
\ee
we obtain the following classification:

\begin{itemize}
    \item If \( \sum_i v_i^2 = 0 \), then \( \Omega(v) \) is proportional to a projector of rank \( d^{k-2} \), i.e., $\Omega^2(v)=d\Omega(v)$;
    \item Otherwise, \( \Omega(v) \) is a unitary operator.
\end{itemize}

\paragraph{Substitution rules:} In \cref{sec:Paulimon}, we introduced the graphical calculus—an intuitive method for manipulating Pauli monomials, which are generated as products of primitive Pauli monomials. At the core of this calculus lies the substitution rule for Pauli monomials, formalized in \cref{th:graphrules}.

While it is possible to extend this framework to a \emph{colored} graphical calculus for qudits, such an extension can become increasingly intricate as the local dimension grows. Therefore, in this section we restrict ourselves to sketching the \emph{algebraic} version of the substitution rule, which serves as a generalization of \cref{th:gaussOP}.
The relations
\begin{align}
    P^{\otimes v} Q^{\otimes w} 
    &= \left( \sqrt{\chi(K, Q)}\, KQ \right)^{\otimes v} Q^{\otimes w} 
    = \left( \sqrt{\chi(K, Q)} \right)^{\sum_i v_i^2} K^{\otimes v} Q^{\otimes (v + w)} ,\\
    \chi(P, L) 
    &= \chi\left( \sqrt{\chi(K, Q)}\, KQ, L \right) 
    = \chi(KQ, L) 
    = \chi(K, L)\chi(Q, L),
\end{align}
hold, 
where \( P = \sqrt{\chi(K, Q)}\, KQ \). Then, the substitution rules result in the mapping:
\begin{align}
v_{i+1} \mapsto v_{i+1} + v_i, \qquad M_{\cdot, i+1} \mapsto M_{\cdot, i+1} + M_{\cdot, i} + \frac{\sum_\alpha (v_i)_\alpha^2}{2} \vec e_i,
 \qquad M_{i+1,\cdot} \mapsto M_{ i+1,\cdot} + M_{ i,\cdot} -\frac{\sum_\alpha (v_i)_\alpha^2}{2} \vec e_i.
\end{align}
This operation mirrors the substitution rule used in the qubit case. Consequently, one can employ the same algebraic expressions to compute equivalences and define the matrix \(\Lambda\), as introduced in \cref{th:gaussOP}
\begin{align}
    H^{1/2}_{i,j} &\coloneqq 
    \begin{cases}
        (V^T V)_{i,j} & \text{if } i < j, \\
        (V^T V)_{i,i} / 2 & \text{if } i = j, \\
        0 & \text{otherwise},
    \end{cases} \\
    \Lambda &\coloneqq H^{1/2} + M.
\end{align}

\paragraph{Normal form and generators:} The final aspect to address is the structure of the multi-qudit Clifford commutant from the perspective of algebra. In \cref{sec:3gen}, we demonstrated that primitive monomials generate the full commutant in the qubit case. In line with other generalizations, we now discuss how this result extends to the multi-qudit setting.

The reduced form, as introduced in \cref{def:reducedpaulimonomials}, is inherently preserved in the qudit generalization and remains well-defined. Similarly, the notion of a normal form (see \cref{lem:normalform}) carries over naturally. The proof of \cref{lem:normalform} relies solely on the substitution rules established in \cref{th:graphrules}, which can be replicated in the qudit case with analogous arguments. This is because these rules depend only on the underlying algebraic transformations, and as shown in the preceding paragraphs, the matrix \( \Lambda \) governing these transformations behaves identically in the qudit setting.

In particular, a primitive corresponding to a unitary can be shifted to the right, effectively absorbing any global phases. Additionally, two projectors that do not commute can be rewritten as a combination of a unitary and a projector. Commuting projectors with non-trivial relative phases can also be rewritten as a product of unitary primitives via conjugation with a commuting unitary.
Consequently, we can define a \emph{normal form} for Weyl monomials
\begin{align}
    \Omega = \Omega_P \times \Omega_U,
\end{align}
where \( \Omega_P \) is a product of commuting projective primitives and \( \Omega_U \) is a product of unitary primitives. This implies that the entire commutant in the multi-qudit setting can also be generated by products of primitive Weyl monomials.

\paragraph{Applications:} In \cref{sec:applications}, we employed the framework developed in the previous sections to explore scenarios in which the Clifford commutant plays a central role—most notably, in magic-state resource theory and stabilizer property testing. Just as in the qubit case, these two topics are of significant interest in the qudit setting as well. In particular, the results of \cref{Sec:magicstateresourcetheory} can be generalized to qudits, allowing for the definition of a measurable family of magic measures that meaningfully extend stabilizer entropies to multi-qudit quantum states.

We note that parts of this generalization have already been explored in Ref.~\cite{Turkeshi_2025,magni2025anticoncentrationcliffordcircuitsbeyond}, where the authors showed that generalized stabilizer entropies for qutrits ($q=3$), defined via the expectation values of primitive Weyl monomials, are not only valid magic measures (according to our definition in \cref{Sec:magicstateresourcetheory}) but also satisfy the stronger condition of being magic monotones.

For stabilizer property testing, similar conclusions to those presented in \cref{sec:propertytesting} can be drawn in the qudit context. In particular, we can determine, as a function of the local dimension $q$, the minimal number of copies required to achieve a non-exponentially vanishing success probability for testing stabilizerness. This critical threshold depends on whether the commutant—beyond the permutation operators—is spanned by unitary or projective Weyl monomials, just as discussed in \cref{sec:propertytesting}.

As an example, we can see why qudit Cliffords are not $3$-designs. This follows as non-trivial Weyl monomials exist
\begin{equation}
    \Omega(1,1,-2)=\frac{1}{d} \sum_{P \in \mathbb{P}_n} P \otimes P \otimes P^{\dagger2},
\end{equation}
For qutrits ($q=3$), we only have this $\Omega(1,1,-2)=\Omega(1,1,1)$ and permutations.  As $\Omega(1,1,1)$ is a projector, testing stabilizer properties using only $k=3$ copies is not possible for qutrits, meaning stabilizer states form an approximate state $3$-design. In contrast, for local dimension $q=5$, the same monomial $\Omega(1,1,-2)$ becomes unitary, making testing feasible with just three copies.

\section*{Acknowledgments} We thank Marcel Hinsche, Gerard Aguilar, Daniel Miller, Jayant Raul Rao, David Gross, Alioscia Hamma for useful discussions. This work has been supported by the BMBF (FermiQP, MuniQC-Atoms, DAQC), the BMWK (EniQmA), QuantERA (HQCC), the Bavarian Ministry of Science and the Arts through the Munich Quantum Valley (MQV-K8), Berlin Quantum, the Cluster of Excellence MATH+, Berlin Quantum, Millenion, PASQuanS2.1, the European Research Council (DebuQC) and the DFG (CRC 183). S.F.E.O. acknowledges support from PNRR MUR project PE0000023-NQSTI.

\appendix

\setcounter{secnumdepth}{2}
\setcounter{equation}{0}
\setcounter{figure}{0}
\setcounter{table}{0}
\setcounter{section}{0}
\renewcommand{\thetable}{A\arabic{table}}
\renewcommand{\theequation}{A\arabic{equation}}
\renewcommand{\thefigure}{A\arabic{figure}}
\renewcommand{\thesection}{A.\arabic{section}} 

\begin{center}
\textbf{\large Appendix}
\end{center}

\section{Examples of Haar average over the unitary group step by step}\label{app:haaraverage}
In this section, we are going to show an example for relatively small order $k$ of the $k$-fold twirling in \cref{twirlingofoperatorOunitary}. For a more comprehensive reference and tutorial, see also Refs.~\cite{zhang_matrix_2014,Mele_2024}. Let us start with the following definition
\be
\Gamma^{(k)}_{\haar}\coloneqq\int \de U U^{\otimes k}\otimes U^{\dag\otimes k}
\ee
It is useful to write down the relation between the $k$-fold channel  $\Phi^{(k)}_{\haar}(\cdot)$ and the $\Gamma^{(k)}_{\haar}$ as
\be
\Phi_{\haar}^{(k)}(O)=\tr_{1,\ldots,k}(\hat{T}\Gamma_{\haar}^{(k)}O\otimes \mathbb{1}).
\label{relationfoldchanneltogamma}
\ee
Note, indeed, that $\Gamma_{\mathcal{E}}^{(k)}$ is an operator on $\mathcal{H}^{\otimes 2k}$. Let us compute the operator $\Gamma_{\haar}^{(k)}\in\mathcal{B}(\mathcal{H}^{\otimes 2k})$. Let us recall that 
\be
\Phi_{\haar}^{(k)}(O)=\sum_{\pi,\sigma}(\boldsymbol{\Lambda}^{-1})_{\pi,\sigma}\tr(T_{\pi}O)T_{\sigma}=\sum_{\pi,\sigma}(\boldsymbol{\Lambda}^{-1})_{\pi,\sigma}\tr_{1,\ldots, k}(T_{\pi}\otimes T_{\sigma}O\otimes \mathbb{1}).
\label{calculationGamma1}
\ee
By comparing \cref{calculationGamma1,relationfoldchanneltogamma} which are valid for every $O\in \mathcal{B}(\mathcal{H}^{\otimes k})$, it is easy to be convinced that 
\be
\Gamma^{(k)}_{\haar}=\hat{T}\sum_{\pi,\sigma}(\boldsymbol{\Lambda}^{-1})_{\pi,\sigma}T_{\pi}\otimes T_{\sigma}
\label{expressionforGammaintermsofpermutations}
\ee
where $\hat{T}$ has been defined above as a product of elementary swaps $\hat{T}=\prod_{j=1}^{k}T_{k,k+j}$. To check the validity of \cref{expressionforGammaintermsofpermutations}, it suffices to note that $\hat{T}^2=\mathbb{1}$.

As a useful application, let us begin by considering the Haar $1$-fold channel in \cref{twirlingofoperatorOunitary}. The symmetric group $S_1=\{e \}$ is filled by just the identity. Therefore, the matrix $\boldsymbol{\Lambda}$ is just a number equal to $\Omega=\tr(\mathbb{1})=d$; consequently, we only have one Weingarten function $W=d^{-1}$ and thus it is straightforward to write the operator $\Gamma_{\haar}^{(1)}$, defined below in \cref{expressionforGammaintermsofpermutations}, as
\be
\Gamma_{\haar}=T\left(\frac{\mathbb{1}\otimes\mathbb{1}}{d}\right)=\frac{T}{d}
\ee
where $T$ is the swap operator between the two copies of $\mathcal{H}$. Now, consider the operator $A\in\mathcal{B}(\mathcal{H})$. The $1$-fold channel applied in $A$ reads
\be
\Phi_{\haar}^{(1)}(A)=\frac{\tr(A)}{d}\mathbb{1}.
\ee

The $2$-fold channel, i.e., $k=2$, is the first non-trivial example of a Haar channel that has found an enormous number of applications. In this section, we will explore just a couple of those. 

The symmetric group $S_2=\{e,(12)\}$ is given by the identity $\mathbb{1}$ and the swap operator $T$ on $\mathcal{H}^{\otimes 2}$. The matrix $\boldsymbol{\Lambda}$ reads
\be
\boldsymbol{\Lambda}=\begin{pmatrix}
\tr(\mathbb{1})& \tr(T)\\ \tr(T)& \tr(TT)
\end{pmatrix}=\begin{pmatrix}
d^2& d\\ d& d^2
\end{pmatrix}
\ee
and the inverse, whose components are the Weingarten functions, reads
\be
\boldsymbol{\Lambda}^{-1}=\frac{1}{d^2-1}\begin{pmatrix}
1& -\frac{1}{d}\\ -\frac{1}{d}& 1
\end{pmatrix}.
\ee
Thus, the operator $\Gamma_{\haar}^{(2)}$, defined later on in \cref{expressionforGammaintermsofpermutations}, reads
\be
\Gamma_{\haar}^{(2)}=\frac{T_{(13)(24)}}{d^2-1}\left(\mathbb{1}\otimes \mathbb{1}+T\otimes T-\frac{\mathbb{1}\otimes \mathbb{1}+T\otimes T}{d}\right)
\ee
where $T_{(13)(24)}$ is the permutation unitary operator corresponding to the cycle $(13)(24)\in S_{4}$. Recall in fact that the operator $\Gamma_{\haar}^{(2)}\in\mathcal{B}(\mathcal{H}^{\otimes 4})$. Multiplying  permutations one gets $T_{(13)(24)}T_{(12)}T_{(34)}=T_{(14)(23)}$ and thus
\be
\Gamma_{\haar}^{(2)}=\frac{1}{d^2-1}\left(T_{(13)(24)}+T_{(14)(23)}-\frac{T_{(13)(24)}+T_{(14)(23)}}{d}\right)
\ee
while the $2$-fold channel applied on an operator $A\in\mathcal{B}(\mathcal{H}^{\otimes2})$ reads
\be
\Phi_{\haar}^{(2)}(A)=\frac{1}{d^2-1}\left(\tr(A)\mathbb{1} -\frac{\tr(A)}{d}T-\frac{\tr(TA)}{d}\mathbb{1}+\tr(AT)T\right).
\ee

\begin{example}[Average subsystem entropy]\label{exampleaverageentropy} Over many application of the Haar orbit of pure states, we find the computation of average entropy on subsystem of \textit{typical states}. With the tool of the Weingarten calculus (cf.\  \cref{cor:weingartencalculushaar}), we can compute it right away. Let $A|B$ be a bipartition of the system of $n$ qubits, with $n_A, n_B$ qubits respectively. Let $\ket{\psi}$ an $n$-qubit state vector, then the purity of the reduced density matrix $\rho_A\coloneqq\tr_B\ketbra{\psi}$ is defined as
\be
\pur(\psi_A)\coloneqq\tr(\psi_A^{2 })=\tr(T_{A}\psi^{\otimes 2})
\ee
Let us compute the average purity by integrating over $\de\psi$
\be
\int\de\psi \pur(\psi_A)=\frac{1}{\tr(\mathbb{1}+T)}\tr(T_{A}\Pi_{\sym})=\frac{1}{\tr(\Pi_{\sym})}\tr(T_{A}+T_{B})=\frac{d_{A}^{2}d_B+d_Ad_B^2}{d(d+1)}
\ee
where we have denoted $d_{A(B)}=2^{n_{A(B)}}$. Note that we have used the fact that permutation operators are factorized on qubits, i.e., $T=T_{A}T_{B}$, where $T_{A(B)}$ acts non-identically on $A(B)$. If one considers the subsystem $A$ to be much smaller than $B$, i.e., $d_A\ll d_B$, one has
\be
\int\de\psi\pur(\psi_A)=\frac{1}{d_A}+O(d^{-1})
\ee
which tells us that, on average, states are maximally entangled, up to an error exponentially small in the total number of qubits $n$.
\end{example}

\begin{example}[Twirling of a quantum channel]\label{Sec:application2design} An interesting and non-trivial application of the $2$-fold channel is the following: consider a completely positive trace preserving map $\mathcal{K}$ with kraus operators $A_k$, i.e., $\mathcal{K}(\cdot)=\sum_{k}A_{k}(\cdot)A_{k}^{\dag}$. In the same fashion as the operator $\Gamma_{\haar}^{(1)}$, we can define
\be
\tilde{\mathcal{K}}=\sum_{k}A_{k}\otimes A_{k}^{\dag}
\label{application2designcptpmap2}
\ee
and let $\tilde{\mathcal{K}}$ act on states $\rho\in\mathcal{H}$ via the following isomorphism
\be
\mathcal{K}(\rho)\mapsto \tr_{1}(T\tilde{\mathcal{K}}\rho\otimes\mathbb{1}).
\label{application2designcptpmap}
\ee
Note that $\tilde{\mathcal{K}}$ is an operator on $\mathcal{H}^{\otimes 2}$, thus we can apply the $2$-fold Haar channel $\Phi_{\haar}^{(2)}(\cdot)$ 
\be
\Phi_{\haar}^{(2)}(\tilde{\mathcal{K}})=\frac{1}{d^2-1}\left(\tr(\tilde{\mathcal{K}})[\mathbb{1}-T/d]-\mathbb{1}+Td\right)
\ee
where we have used  the fact that $\tr(T\tilde{\mathcal{K}})=\sum_{k}\tr(A_{k}^{\dag}A_{k})=\tr(\mathbb{1})=d$ for every CPTP map. Let us act with $\Phi_{\haar}^{(2)}(\tilde{\mathcal{K}})$ on a state $\rho$ as in \cref{application2designcptpmap}
\be
\tr_{1}(T\Phi_{\haar}^{(2)}(\tilde{\mathcal{K}})\rho\otimes\mathbb{1})=\frac{\tr(\tilde{\mathcal{K}})-1}{d^2-1}\rho+\left(1-\frac{\tr(\tilde{\mathcal{K}})-1}{d^2-1}\right)\frac{\mathbb{1}}{d}
\label{everychannelisdepolirizing2designapplication}
\ee
that is a depolarizing channel $\mathcal{D}(\rho)=p\rho+(1-p)\mathbb{1}/d$ with $p= \frac{\tr(\tilde{\mathcal{K}})-1}{d^2-1}$. We just found out that, on average, over Kraus operators $A_{k}\mapsto UA_kU^{\dag}$ every channel is a depolarizing channel with probability depending on the trace of $\tilde{\mathcal{K}}$ in \cref{application2designcptpmap2}.
\end{example}

Let us now go a step further and consider the $3$-fold channel. The symmetric group $S_{3}$ contains $3!=6$ permutation operators $S_{3}= \{e,(12),(23),(13),(123),(132)\}$. Denote $T_{\sigma}$ the unitary representation of $\sigma$ on $\mathcal{H}^{\otimes 3}$ and let us compute the matrix $\boldsymbol{\Lambda}$
\be
\scalebox{0.8}{
$
\boldsymbol{\Lambda}=\begin{pmatrix}
    \tr(\mathbb{1})&\tr(T_{(12)})&\tr(T_{(23)})&\tr(T_{(13)})&\tr(T_{(123)})&\tr(T_{(132)})\\
    \tr(T_{(12)})&\tr(\mathbb{1})&\tr(T_{(132)})&\tr(T_{(123)})&\tr(T_{(13)})&\tr(T_{(23)})\\
    \tr(T_{(23)})&\tr(T_{(123)})&\tr(\mathbb{1})&\tr(T_{(132)})&\tr(T_{(12)})&\tr(T_{(13)})\\
\tr(T_{(13)})&\tr(T_{(132)})&\tr(T_{(123)})&\tr(\mathbb{1})&\tr(T_{(23)})&\tr(T_{(12)})\\
    \tr(T_{(123)})&\tr(T_{(23)})&\tr(T_{(13)})&\tr(T_{(12)})&\tr(T_{(132)})&\tr(\mathbb{1})\\
    \tr(T_{(132)})&\tr(T_{(13)})&\tr(T_{(12)})&\tr(T_{(23)})&\tr(\mathbb{1})&\tr(T_{(123)})\\
\end{pmatrix}=\begin{pmatrix}
    d^3&d^2&d^2&d^2&d&d\\
    d^2&d^3&d&d&d^2&d^2\\
    d^2&d&d^3&d&d^2&d^2\\
d^2&d&d&d^3&d^2&d^2\\
    d&d^2&d^2&d^2&d&d^3\\
    d&d^2&d^2&d^2&d^3&d\\
\end{pmatrix}
$},
\ee
where we have  computed the above matrix 
following the rule sketched in \cref{asymptoticweingarten}, i.e., $\tr(T_{\sigma})=d^{|\sigma|}$ for $\sigma\in S_{3}$ where $|\sigma|$ is the length of $\sigma$, the number of cycles in the cycle-representation. For example, we have $|e|=|(1)(2)(3)|=3$ or $|(123)|=|(132)|=2$ because $(123)=(12)(23)$ and $(132)=(23)(12)$. The inverse reads
\be
\boldsymbol{\Lambda}^{-1}=\frac{1}{d(d^3-4)(d^2-1)}\left(
\begin{array}{cccccc}
 d^2-2 & -d & -d & -d & 2 & 2 \\
 -d & d^2-2 & 2 & 2 & -d & -d \\
 -d & 2 & d^2-2 & 2 & -d & -d \\
 -d & 2 & 2 & d^2-2 & -d & -d \\
 2 & -d & -d & -d & 2 & d^2-2 \\
 2 & -d & -d & -d & d^2-2 & 2 \\
\end{array}
\right).
\ee
We avoid explicitly writing the operator $\Gamma_{\haar}^{(3)}$ and the twirling $\Phi_{\haar}^{(3)}$. We instead show an application by computing the average value of a $6$-point OTOC.

\begin{example}[Average $6$-point OTOC] Consider $P,Q\in\mathbb{P}_n$ $2$ non-identity Pauli operator and the following correlation function
\be
\otoc_{6}(U)=\frac{1}{d}\tr(P(U)QQ(U)PQ(U)P(U)PQ)
\ee
where $P(U)= U^{\dag}PU$. Note that for $U=\mathbb{1}$ one has $\otoc_{6}(U)=1$. Let us use the following trick $\tr(ABC)=\tr(T_{(132)}A\otimes B\otimes C)=\tr(T_{(123)}A\otimes C\otimes B)$ and write $\otoc_6(U)$ as
\be
\otoc_{6}(U)=\frac{1}{d}\tr(T_{(132)}U^{\dag\otimes 3}[P\otimes Q\otimes QP ]U^{\otimes 3}[Q\otimes P\otimes PQ]).
\ee
Let us average over the realizations of the unitary operator $U$ according to the Haar distribution
\be
\int\de U\otoc_{6}(U)=\frac{1}{d}\tr(T_{(132)}\Phi_{\haar}^{(3)}(P\otimes Q\otimes QP )Q\otimes P\otimes PQ ).
\ee
In order to compute the above twirling using the matrix formalism, and thus $\boldsymbol{\Lambda}^{-1}$, let us compute the following $6$-components vectors $v_{\sigma}=\tr(T_{\sigma}P\otimes Q\otimes QP)$ and $w_{\sigma}=\tr(T_{\sigma}Q\otimes P\otimes PQ)$. First note that $v_{\sigma}=w_{\sigma}$ for two non-identity Pauli operators. Thus
\be
\boldsymbol{v} =
\begin{pmatrix}
\tr(P\otimes Q\otimes QP)\\
\tr(T_{(12)}P\otimes Q\otimes QP)\\
\tr(T_{(23)}P\otimes Q\otimes QP)\\
\tr(T_{(13)}P\otimes Q\otimes QP)\\
\tr(T_{(123)}P\otimes Q\otimes QP)\\
\tr(T_{(132)}P\otimes Q\otimes QP)
\end{pmatrix}=\begin{pmatrix}
0\\
0\\
0\\
0\\
d\chi(P,Q)\\
d\chi(P,Q)
\end{pmatrix}
\ee
where $\chi(P,Q)=\frac{1}{d}\tr(PQPQ)$ and $\chi(P,Q)=1$ for commuting Pauli operators $[P,Q]=0$ and $-1$ for anticommuting Pauli operators $\{P,Q\}=0$. We can write the average $\otoc_{6}$ as
\be
\int\de U\otoc_{6}(U)=\frac{1}{d}\sum_{\pi,\sigma}\boldsymbol{\Lambda}^{-1}_{\pi,\sigma}v_{\sigma}v_{\pi}=\frac{2d^2}{(d^2-4)(d^2-1)}
\ee
cf.\ also the results in Ref.~\cite{roberts_chaos_2017}.
\end{example}

We now consider the $4$-fold channel, i.e., $k=4$. Being the symmetric group of $4$ objects of $4!$ elements, calculations with the $4$-fold channel are long and tedious. However, in this section, we outline a systematic and general method to conduct such calculation by reducing the possibility of making typos around. First, let us list the permutations in $S_{4}$ ordered by conjugacy class
\be
&\text{1-cycle:}\quad &e\\
&\text{2-cycle:} &(12), (13), (14), (23), (24), (34)\\
&\text{3-cycle:}&(123), (132), (142), (234), (243), (132), (124), (143)\\
&\text{4-cycle:}&(1234), (1324), (1423), (1342), (1243),(1432)\\
&\text{22-cycle:}&(12)(34), (13)(24), (14)(32)
\ee

For clarity, we omit to report the $\boldsymbol{\Lambda}$ matrices for $k=4$ due to their size, but it is possible to find them in the \textit{.nb} file attached to the manuscript together with its inverse. 
Let us apply the above formalism to see in what conditions a channel can be well approximated with a depolarizing channel. 

\begin{example}[Approximating a channel with the depolarizing channel] In  \cref{Sec:application2design}, we showed that on average every channel is a depolarizing channel with probability depending on the trace of $\tilde{\mathcal{K}}=\sum_{k}A_{k}\otimes A_{k}^{\dag}$. Let us compute the fluctuations around the average evaluated for an expectation value of a Hermitian operator $A$ with null trace (for simplicity). We need to evaluate the following average square distance from the average (variance)
\be
\tr(T_{(12)(34)}\Phi_{\haar}^{(4)}(\mathcal{K}^{\otimes 2})(\rho\otimes A)^{\otimes 2})-\tr(T_{(12)}\Phi_{\haar}^{(4)}(\mathcal{K})\rho\otimes A)^{2}
\label{4designapplicaation}
\ee
where $\tilde{\mathcal{K}}$ is defined in \cref{Sec:application2design}, see \cref{application2designcptpmap2}. Let us specialize the discussion to a specific, although quite general, quantum channel referred to as \textit{Pauli channel}~\cite{flammia_efficient_2020}
\be
\mathcal{K}(\rho)=\sum_{P}k_{P}P\rho P
\ee
where $P$ are Pauli operators $\sum_{P}k_{P}=1$ to preserve the trace. Although the second term in \cref{4designapplicaation} has already been computed in \cref{everychannelisdepolirizing2designapplication}, let us consider the first term and define the vectors $\boldsymbol{v},\boldsymbol{w}$ with components in $S_{4}$ defined as $v_{\sigma}=\tr(\tilde{\mathcal{K}}^{\otimes2}T_{\sigma})$ and $w_{\sigma}=\tr(T_{(12)(34)}T_{\sigma}(\rho\otimes\mathbb{1})^{\otimes2}A^{\otimes 2})$. In terms of $\boldsymbol{v},\boldsymbol{w}$ we have
\be
\tr(T_{(12)(34)}\Phi_{\haar}^{(4)}(\mathcal{K}^{\otimes 2})(\rho\otimes\mathbb{1})^{\otimes 2}A^{\otimes 2})=\sum_{\pi,\sigma}\boldsymbol{\Lambda}^{-1}_{\pi,\sigma}v_{\sigma}w_{\pi}.
\ee
Let us compute the first vector
\be
\boldsymbol{v}&=(d^4k_{\mathbb{1}}^2,d^3k_{\mathbb{1}}, d^3k_{\mathbb{1}}^2, d^{3}k_{\mathbb{1}}^2,d^{3}k_{\mathbb{1}}^2,d^{3}k_{\mathbb{1}}^2,d^{3}k_{\mathbb{1}},d^2k_{\mathbb{1}},d^2k_{\mathbb{1}},\\&d^2k_{\mathbb{1}},d^2k_{\mathbb{1}},d^2k_{\mathbb{1}},d^2k_{\mathbb{1}},d^2k_{\mathbb{1}},d^2k_{\mathbb{1}},s,s,s,d,d,d,d^2,d^2\pur(k),d^2\pur(k)),
\ee
where $s=\sum_{P,Q}k_{P}k_{Q}\tr(PQPQ)$, and $\pur(k)=\sum_{P}k_{P}^{2}$ is the purity of the probability distribution $k$.
\be
\boldsymbol{w}=(&\tr^2(\rho A),0,\tr^2(\rho A),\tr(\rho A^2),\tr(\rho A^2),\tr^2(\rho A),0,0,\tr(\rho A^2),\tr(\rho A^2),0,\tr(\rho A^2),\tr(\rho A^2),0,\tr(\rho A^2),\\& 0,\tr(\rho A^2),\tr(\rho A^2),0,0,\tr(A^2),0,\tr^2(\rho A),\tr(A^2)),
\ee
where we have taken into account that $\rho^2=\rho$ and $\tr(A)=0$. We obtain
\be
\tr(T_{(12)(34)}\Phi_{\haar}^{(4)}(\mathcal{K}^{\otimes 2})(\rho\otimes A)^{\otimes 2})-\tr(T_{(12)}\Phi_{\haar}^{(4)}(\mathcal{K})\rho\otimes A)^{2}=\tr^2(A\rho)(k_{\mathbb{1}}-k_{\mathbb{1}}^2)+O(d^{-1}).
\ee
We thus conclude that, on average, every channel behaves like a depolarizing channel and the fluctuations around the average are small in two cases: $(i)$ in the case of an almost unitary channel $k_{\mathbb{1}}-k_{\mathbb{1}}^{2}\ll 1$ or $(ii)$ for noisy channels $k_{\mathbb{1}}\ll1$.

\end{example}
\end{document}